%% file: main.tex
\documentclass[12pt]{report}

\usepackage[pdftex,dvipsnames]{xcolor}
\usepackage{amsmath,amssymb,amsthm}
\usepackage{bm}
\usepackage{booktabs}
\usepackage{cancel}
\usepackage{epigraph}
\usepackage{graphicx}
\usepackage{listings}
\usepackage{lmodern}
\usepackage{mathtools}
\usepackage{multirow}
\usepackage{natbib}
\usepackage{pdfpages}
\usepackage{placeins}
\usepackage{siunitx}
\usepackage{stmaryrd}
\usepackage{xargs}
\usepackage[colorinlistoftodos]{todonotes}
\usepackage{hyperref}
\hypersetup{
  colorlinks=true,
  citecolor=OliveGreen,
  linkcolor=RoyalBlue,
  urlcolor=NavyBlue,
  pdftitle = {Continuum Kinetic Simulations of Plasma Sheaths and Instabilities},
  pdfauthor = {Petr Cagas}
}

\usepackage{stackengine}
\stackMath
\newcommand\tsup[2][2]{%
 \def\useanchorwidth{T}%
  \ifnum#1>1%
    \stackon[-.5pt]{\tsup[\numexpr#1-1\relax]{#2}}{\scriptscriptstyle\sim}%
  \else%
    \stackon[.5pt]{#2}{\scriptscriptstyle\sim}%
  \fi%
}

\definecolor{backcolour}{rgb}{0.95,0.95,0.92}
\newcommand{\eqr}[1]{Eq.\thinspace(\ref{eq:#1})}
\newcommand{\eqrp}[1]{(Eq.\thinspace\ref{eq:#1})}
\newcommand{\fgr}[1]{Fig.\thinspace\ref{fig:#1}}
\newcommand{\tbr}[1]{Tab.\thinspace\ref{tab:#1}}
\newcommand{\ser}[1]{Sec.\thinspace\ref{sec:#1}}

\newcommand{\pfrac}[2]{\frac{\partial #1}{\partial #2}}
\newcommand{\pfracc}[2]{\frac{\partial^2 #1}{\partial #2^2}}
\newcommand{\pfraca}[1]{\frac{\partial}{\partial #1}}
\newcommand{\pfracb}[2]{\partial #1/\partial #2}

\newcommand\hcancel[2][black]{\setbox0=\hbox{$#2$}%
  \rlap{\raisebox{.35\ht0}{\textcolor{#1}{\rule{\wd0}{0.5pt}}}}#2}

\newtheorem{prop}{Proposition}

\allowdisplaybreaks[1]

\DeclareMathOperator*{\spn}{span} 

\begin{document}

\thispagestyle{empty}
\pagenumbering{roman}
\begin{center}

% TITLE
{\Large 
Continuum Kinetic Simulations of Plasma Sheaths and Instabilities
}

\vfill

Petr Cagas

\vfill

Dissertation submitted to the Faculty of the \\
Virginia Polytechnic Institute and State University \\
in partial fulfillment of the requirements for the degree of

\vfill

Doctor of Philosophy  \\
in \\
Aerospace Engineering

\vfill

Bhuvana Srinivasan, Chair \\
Colin S. Adams \\
Wayne A. Scales \\
Timothy Warburton

\vfill

30$^{\text{th}}$ July, 2018 \\
Blacksburg, Virginia

\vfill

Copyright
2018, Petr Cagas

\end{center}

\pagebreak

\thispagestyle{empty}
\begin{center}

{\large Continuum Kinetic Simulations of Plasma Sheaths and Instabilities}

\vfill

Petr Cagas

\vfill

(ABSTRACT)

\vfill

\end{center}

A careful study of plasma-material interactions is
essential to understand and improve the operation of devices where
plasma contacts a wall such as plasma thrusters, fusion devices,
spacecraft-environment interactions, to name a few.  This work aims to advance our understanding of fundamental plasma
processes pertaining to plasma-material interactions, sheath physics,
and kinetic instabilities through theory and novel numerical
simulations.   Key contributions
of this work include (i) novel continuum kinetic algorithms with novel
boundary conditions that directly discretize the Vlasov/Boltzmann
equation using the discontinuous Galerkin method, (ii) fundamental
studies of plasma sheath physics with collisions, ionization, and
physics-based wall emission, and (iii) theoretical and numerical
studies of the linear growth and nonlinear saturation of the kinetic
Weibel instability, including its role in plasma sheaths.

The continuum kinetic algorithm has been shown to compare well with theoretical predictions of
Landau damping of Langmuir waves and the two-stream instability.
Benchmarks are also performed using the electromagnetic Weibel
instability and excellent agreement is found between theory and
simulation.  The role of the electric field is
significant during nonlinear saturation of the Weibel instability,
something that was not noted in previous studies of the Weibel
instability.  For some plasma parameters, the electric field energy
can approach magnitudes of the magnetic field energy during the
nonlinear phase of the Weibel instability.

A significant focus is put on understanding plasma sheath
physics which is essential for studying plasma-material
interactions. Initial simulations are performed using a baseline
collisionless kinetic model to match classical sheath theory and the
Bohm criterion. Following this, a collision operator and volumetric
physics-based source terms are introduced and effects of heat flux are
briefly discussed.  Novel boundary conditions are developed and
included in a general manner with the continuum kinetic algorithm
for bounded plasma simulations.  A physics-based wall emission model
based on first principles from quantum mechanics is self-consistently
implemented and demonstrated to significantly impact sheath physics.
These are the first continuum kinetic simulations using
self-consistent, wall emission boundary conditions with broad
applicability across a variety of regimes.

\vfill

This research was supported by the Air Force Office of Scientific
Research under grant number FA9550-15-1-0193.

The author acknowledges Advanced Research Computing at Virginia Tech
for providing computational resources and technical support that have
contributed to the results reported within this work. URL:
http://www.arc.vt.edu

\pagebreak

\thispagestyle{empty}
\begin{center}

{\large Continuum Kinetic Simulations of Plasma Sheaths and Instabilities}

\vfill

Petr Cagas

\vfill

(GENERAL AUDIENCE ABSTRACT)

\vfill

\end{center}

An understanding of plasma physics is vital for problems on a wide
range of scales: from large astrophysical scales relevant to the
formation of intergalactic magnetic fields, to scales relevant to
solar wind and space weather, which poses a significant risk to
Earth's power grid, to design of fusion devices, which have the
potential to meet terrestrial energy needs perpetually, and electric
space propulsion for human deep space exploration. This work aims to
further our fundamental understanding of plasma dynamics for
applications with bounded plasmas.  A
comprehensive understanding of theory coupled with high-fidelity
numerical simulations of fundamental plasma processes is necessary, this then can be used to improve improve the
operation of plasma devices.

There are two main thrusts of this work.  The first thrust involves
advancing the state-of-the-art in numerical modeling.  Presently,
numerical simulations in plasma physics are typically performed either
using kinetic models such as particle-in-cell, where individual
particles are tracked through a phase-space grid, or using fluid
models, where reductions are performed from kinetic physics to arrive
at continuum models that can be solved using well-developed numerical
methods.  The novelty of the numerical modeling is
the ability to perform a complete kinetic calculation using a
continuum description and evolving a complete distribution function in
phase-space, thus resolving kinetic physics with continuum numerics.

The second thrust, which is the main focus of this work, aims to
advance our fundamental understanding of plasma-wall interactions as
applicable to real engineering problems.  The continuum kinetic
numerical simulations are used to study plasma-material interactions
and their effects on plasma sheaths.  Plasma sheaths are regions of
positive space charge formed everywhere that a plasma comes into
contact with a solid surface; the charge inequality is created because
mobile electrons can quickly exit the domain. A local electric field
is self-consistently created which accelerates ions and retards
electrons so the ion and electron fluxes are equalized. Even though
sheath physics occurs on micro-scales, sheaths can have global
consequences. The electric field accelerates ions towards the wall
which can cause erosion of the material.  Another consequence of
plasma-wall interaction is the emission of electrons.  Emitted
electrons are accelerated back into the domain and can contribute to
anomalous transport.  The novel numerical method coupled with a unique
implementation of electron emission from the wall is used to study
plasma-wall interactions.

While motivated by Hall thrusters, the applicability of the algorithms developed here
extends to a number of other disciplines such as semiconductors,
fusion research, and spacecraft-environment interactions.

\pagebreak

\textit{To my wife Krist{\'y}na and parents Kamila and Pavel}

\pagebreak

\include{acknowledgement}

\tableofcontents
\pagebreak

\listoffigures
\pagebreak

\listoftables
\pagebreak

\pagenumbering{arabic}
\pagestyle{myheadings}

\include{intro}
\include{model}
\include{benchmark}
\include{weibel}
\include{bounded_plasma}

\include{conclusion}

%---------------------------------------------------------------------
%-- Reference --------------------------------------------------------
\bibliographystyle{plainnat}
\bibliography{reference}
\addcontentsline{toc}{chapter}{Bibliography}

%---------------------------------------------------------------------
%-- Appendix  --------------------------------------------------------
\appendix

\include{appendix_postgkyl}
\include{appendix_cfi}
\include{appendix_inputs}
\include{appendix_scripts}
%\newpage
%\listoftodos[Notes]
\end{document}

%% file: acknowledgement.tex
\chapter*{Acknowledgments}

I consider myself extremely lucky to have many great people
contributing to this work.  The least I can do is to thank them all
here. 

First of all, I need to thank Dr. Kate\v{r}ina Falk for pushing me out
of my comfort zone and enabling all of this. I will be forever
grateful to her for putting me in touch with my adviser and mentor,
Professor Bhuvana Srinivasan.

And I could not have hoped for a better adviser.  I think that her
philosophy, that students need to be actively seeking help but then
get all they want, works great for me.  Her mentoring style resulted
in me having a poster at APS DPP conference three months after joining
the group.  One of the things I will remember the most will be our evening discussions in front of the white-board in our
office, which always reminded me why I love Physics.  The only
downside is that I was unable to grasp a concept of ``sending email to
the adviser to schedule a meeting'' other students keep talking about.

And it still gets better because I did not have just one but two great
advisers.  Shortly after joining the group at VT, I met Dr. Ammar
Hakim through a teleconference and he quickly became my second
adviser in all but official title.  When I first met him in person, he
introduced himself to me and my colleagues as the person whom I was
going to hate soon.  Did not happen yet, though he got close when we
were discussing color maps...  Since my first week, we have been in
contact on a daily basis discussing not just physics and programming but
also politics and mad Richard Stallman.  Ammar is also the reason why
I ended up working on software engineering projects like making
\textit{Conda} builds of a code written mostly in \textit{Lua},
compiled with \textit{LuaJIT}, and relying heavily on automatically
generated \textit{C} code from \textit{Maxima}.

Listening to other students, one can reach a conclusion that the sole
purpose of the Graduate Committee is to make students cry during various
meetings.  My meetings were quite different because I always got an
impression that the Committee members are genuinely interested in my
work.  What is more, I was interacting with them beyond the mandatory
meetings.  I was meeting Professor Colin Adams every Friday at our
plasma physics Journal Club, which he organized.  Using the expertise of
Professor Wayne Scales about plasma waves resulted in him being a
co-author of one of our publications.  And last but not least,
Professor Timothy Warburton provided a lot of valuable advice about
the numerical aspects and his discontinuous Galerkin course has
been in many aspects the best class I have ever taken.

Of course, I cannot forget my parents Kamila and Pavel who always
supported me as much as they could.  They made my growing up worriless,
nurtured my curiosity, and shaped a great deal of who I am today.  They
only, for some mysterious reason, did not want me to pursue my musical
talents but more on that later.  When I was offered the opportunity
to come to the States for the Ph.D., they gave me their unconditional
support even though it must have been difficult for them.  Without
this support and the head start I got, I would never be able to
achieve any of this work whatever my skills might be.

There are many other great people at VT who helped me stay steadily on
top of all the necessary administrative requirements and were always
very friendly to me.  I would like to namely thank Rachel Hall
Smith, Amy Burchett, Kelsey Wall, Jama Green, Cory Thompson, and 
Erin Wilson.

During my time at VT, many brilliant fellow students made their
imprint on me.  On the professional level, no other student had a bigger
influence than Jimmy Juno.  His report from a summer internship at PPPL
helped me immensely to jump-start my research with \texttt{Gkeyll}.
Since then, he has been always willing to discuss anything I needed
help with (quite often it was the US comic book culture).  When I
joined the research group at VT, there were two other graduate students, Colin Glesner
and Yang Song.  I do not think I was suffering particularly bad from
the culture shock but Colin Glesner was probably the one who helped me
the most with the problems I had.  Yang became the person sitting next
to me in our office and that resulted in an interesting special
relativity phenomenon.  On many occasions, our casual office
conversation ended after we realized it was suddenly many hours later
and it was time to go home.  As a senior grad student, Yang helped me
professionally but also significantly broadened my cultural views.  I
have participated in the recruitment of Robert Masti, the Chair of our
Beer Committee.  After he joined our group, it did not take long to
get to know him better and, nowadays, I like to think about him as
'murican brother.  Special thanks also go to Jimmy and Chirag Rathod
for reading this work and providing valuable comments.  The rest of
my group members which I did not name specifically, I would like to
thank about staying cool with me during our group meeting.  I hope
they know my comments were always motivated by just the best
intentions.

I also made some friends outside the department.  Ashley Gates and
Loren Brown became such good friends that we eventually moved right
next to them. The process of moving was not simple and they went to great
lengths to help us.  A very special place belongs to Kiaya and Matthew
Vincent who forced on us (without any pressure) one of the kittens they
were fostering.  Since then, our furry monster Binks walks over us at
night, bites everything that sticks out over the edge of our bed, and
preferentially sits on pages with some of my notes or derivations.
Long story short, I have never regretted our decision to adopt him.

It is often said that grad students should have something to do apart
from research in order to keep their sanity.  I fell in love with Jethro
Tull and that motivated me to start playing the flute.\footnote{For
  people knowing me for a long time, this is probably the most
  surprising piece of information in this whole work.}  Originally, I was
trying to learn just on my own using the \textit{Flute for Dummies}
book but then I met Ms. Elizabeth Crone from the Virginia Tech's School of
Performing Arts.  With her, my learning pace drastically improved and
I started discovering a world I barely knew there is.  But more
importantly, she made me smile every week.  I vividly remember walking
from one lesson smiling to myself even though our work was just
rejected from the Physical Review Letters.

No, I did not forget about my wife,
Krist{\'y}na.\footnote{\url{http://phdcomics.com/comics/archive.php?comicid=870}}
One thing, I will say about her, is that she likes to make long-term
plans.  And still, she dropped everything and in three months went
with me through the admission process, which usually takes foreigners
at least a year.  The fact that she was accepted to VT as well is one
of the most amazing things that ever happened to me. She loves me,
stays with me, and supports me as much as I need anywhere we are.  Her
I thank the most.

%% file: intro.tex
\chapter{Introduction}\label{chap:intro}

\epigraph{All models are wrong, but some are useful.}{\textit{George
    E. P. Box}}

In the past few decades, numerical simulations have undergone rapid
development and established themselves firmly as the third pillar of
physics along with theory and experiment.  To develop high-fidelity,
carefully benchmarked numerical simulations, models need to be built
from the bottom up, with each part tested rigorously.

This is especially important for complex problems requiring rich
physics.  One such example is the Hall thruster.  These electrostatic
plasma thrusters have been known for many decades and have been
successfully flown on spacecraft; however, there are still gaps in our
understanding of the underlying physics.  Presently there are no
global models with truly predictive capabilities because details of
the electron transport inside the Hall thruster channel are not fully
understood.  This challenge extends to plasma devices and applications
beyond Hall thrusters as well, where microscale physics can
significantly affect macroscale phenomena and the subtle interplay
remains an open research question.

The aim of this project is to leverage recent progress made in
mathematics and software engineering to develop a new,
self-consistent, physically-relevant model to advance our fundamental
understanding of plasma physics for a number of applications.  As a
result, this work constitutes a complementary blend of physics,
mathematics, and software engineering.

%---------------------------------------------------------------------
\section{Plasma and Hall Thrusters}
Hall thrusters (HT; also called Stationary Plasma Thrusters or
deceivingly Hall Effect Thrusters\footnote{When a linear conductor is
  put into magnetic field perpendicular to current flowing inside,
  Lorentz force deflects the flow and eventually leads to charging of
  its edges.  Created electric field then accelerates particles in the
  opposite direction, effectively countering the effect of the magnetic
  field.  This is called the Hall effect and, since magnetic field
  plays a crucial role in Hall thrusters, Hall effect needs to be
  avoided; this is the reason why Hall thrusters are circular rather
  than linear \citep{Boeuf2017}.}) are a type of electric propulsion
which was invented in the 1960s and first flown on the Soviet
satellite \textit{Meteor-18} on 29 December 1971 \citep{Morozov2003}.
\textit{Meteor-18} was a \SI{15}{t} satellite and the HT, developed at
the Kurchatov Institute of Atomic Energy, managed to lift the
spacecraft by \SI{15}{km} in a week, orient it, and maintain the
orbital altitude.  In the following years, Soviets and later Russians
launched many more satellites with HT on board, while in the U.S., the
development focused on an alternative electric propulsion concept:
gridded ion thrusters.  Hall thrusters were first used in Europe and
the States in the 1990s.

Nowadays, HTs are still mainly used for altitude keeping, but are also
considered as a type of propulsion for deep space journeys.  In other
words, they are used for purposes where classical chemical propulsion
is lacking.  While chemical propulsion is essential for high thrust
missions, like launches into orbit, it presents serious constrains for
traveling beyond Mars.  An unavoidable problem for most propulsion
concepts\footnote{Concepts like solar sails or experimental work like
  \cite{White2016} are not discussed in this work.} is the need to
accelerate the remaining propellant together with the useful payload.
This is demonstrated by the ideal rocket equation,
\begin{align}\label{eq:intro:rocket}
  \Delta v = c_e \mathrm{ln}\,\left(\frac{m_0}{m_f}\right),
\end{align}
where $\Delta v$ is the measure of an impulse needed for the mission,
$c_e$ is the effective exhaust velocity, $m_0$ and $m_f$ are the
initial and final masses of the spacecraft, respectively.  Instead of
the exhaust velocity, it is common to use specific impulse, $I_{sp} =
c_e/g_0$, where $g_0$ is Earth's gravitational acceleration, which can
be seen as a measure of fuel efficiency.  The mass ratio then follows
\begin{align*}
  \frac{m_0}{m_f} = \exp\left(\frac{\Delta v}{I_{sp}g_0}\right).
\end{align*}
For example, a mission defined by $\Delta v = \SI{11186}{m.s^{-1}}$
(Earth escape $\Delta v$) and a propulsion system with $I_{sp} \approx
\SI{400}{s}$, which is a reasonable $I_{sp}$ for chemical propulsion,
has the mass ratio around 18. The ratio gets even worse for more
extreme missions.  Getting \SI{10}{kg} of useful payload to the
nearest star from the Sun (around 5~light years) in \SI{5000}{years}
using a propulsion system with $I_{SP} = \SI{500}{s}$ requires fuel
which is on the order of the Earth mass (\SI{5.972e24}{kg}
\citep{Luzum2011}).  Clearly, the mass ratio can be improved by
increasing specific impulse, $I_{sp}$.  For chemical propulsion, the
source of the energy are chemical bonds in the propellant.  This
energy is used to increase the enthalpy of the propellant which is
subsequently converted into kinetic energy with a nozzle.  Therefore,
the specific impulse of chemical propulsion is limited, unless a
radically new propellant is found.  On the other hand, electrostatic
propulsion\footnote{Other types of the electric propulsion (e.g.,
  arc-jets and resistojets) heat the propellant electrically and then
  follow same principles like the chemical propulsion.
  Electromagnetic propulsion concepts use the Lorentz force.  However,
  these are not considered for this work.}  uses an electrostatic
field to accelerate ionized propellant; therefore, the exhaust
velocity is variable and can be calculated with
\begin{align*}
  c_e = \sqrt{\frac{2q\phi}{m}},
\end{align*}
where $q$ and $m$ are the propellant charge and mass, and $\phi$ is
the applied electrostatic potential. In theory, the propellant could
be accelerated to an indiscriminately high velocity. However, new
constraints posed by electric propulsion, due to needing a power
source, increase the overall mass of the system and prohibit an
unlimited acceleration.  Despite the aforementioned limitation, up to
ten-fold higher $I_{sp}$ of electric propulsion systems allows them to
maintain a better fuel efficiency compared to chemical propulsion.

Hall thrusters span a wide part of the propulsion parameter space.
Currently, their power ranges from \SI{100}{W} to \SI{30}{kW}, across
a variety of sizes, with $I_{sp}$ up to \SI{3000}{s}, and thrust
between mN and N \citep{Boeuf2017}.  Additional concepts, like nested
HTs or arrays of HTs are constitute current research \citep{Hall2017},
as the amount of available power on satellites increases.

The basic principle of HTs is simple.  Neutral gas gets ionized and is
then accelerated by an electrostatic field.  However, maintaining the
external electric field is challenging due to the complexity of the
plasma\footnote{It is not without interest, that the name ``plasma''
  was first used by Langmuir around 1927 because it reminded him how
  blood plasma (electron fluid) carries corpuscles (ions).}  created
inside.

In popular literature, plasma is referred to as a fourth state of
matter, obtained when electrons leave their atoms.  This
classification is, however, questionable since the ionization ratio
can vary and there is no clear distinction between plasma and neutral
gas, like it is with the other states of matter.  A proper definition
is more complicated:
\begin{quote}
  \textit{Plasma is a quasi-neutral mixture of electrons, ions, and
    neutral atoms and molecules in various quantum states which
    exhibit collective behavior.}
\end{quote}
Particularly, the last point of the definition is very important.  The
electrons in a plasma have much higher mobility that the other species
due to their low mass (in the simplest hydrogen plasma, the mass ratio
between electrons and ions is $\approx$1836) and can rearrange
themselves to shield an external electric field.  Considering a
simplified case of a 1D hydrogen plasma with a potential $\phi$,
Poisson's equation gives
\begin{align}\label{eq:intro:poisson}
  \nabla^2 \phi =
  -\frac{q_in_i-q_en_e}{\varepsilon_0},
\end{align}
where $n_e$ and $n_i$ are electron and ion number densities and
$\varepsilon_0$ is the vacuum permitivity.  It can be assumed that on
time-scales relevant to electrons, ions retain the density unperturbed
by the electric field, i.e., $n_i = n_0$.  The particle distribution
function of electrons is (see Chapter~\ref{chap:model} for more
details on distribution functions)
\begin{align}\label{eq:intro:fe}
  f_e(x,\,v_x) = A \exp\left(-\frac{\frac{1}{2}m_ev_x^2 + q_e\phi}{T_e}\right),
\end{align}
where $A$ is some normalization constant and $T_e$ is the electron
temperature (in eV).  Integrating \eqr{intro:fe} over $v_x$ leads to
\begin{align}
  n_e(x) = n_0 \exp\left(-\frac{q_e\phi(x)}{T_e}\right).
\end{align}
Substitution into the Poisson's equation \eqrp{intro:poisson}
gives
\begin{align*}
  \frac{d^2 \phi}{dx^2} =
  \frac{\mathrm{e}n_0}{\varepsilon_0}\left[\exp\left(-\frac{q_e\phi}{T_e}\right)
    -1\right].
\end{align*}
The small argument Taylor series expansion is then used to obtain
\begin{align*}
  \frac{d^2 \phi}{dx^2} =
  \frac{q_en_0}{\varepsilon_0}\frac{q_e\phi}{T_e}.
\end{align*}
The solution to this differential equation is
\begin{align*}
  \phi = \phi_0 \exp\left(-\frac{|x|}{\lambda_D}\right),
\end{align*}
where
\begin{align}
  \lambda_D=\sqrt{\frac{\varepsilon_0T_e}{n_eq_e^2}}
\end{align}
is a characteristic scale length called the Debye length.  For a
special case of a particle source with charge $q$, the potential as a
function of the distance $r$ follows
\begin{align*}
  \phi(r) =
  \frac{1}{4\pi\varepsilon_0}\frac{q}{r}\exp\left(-\frac{r}{\lambda_D}\right).
\end{align*}
This brings an important insight.  In a plasma, electrons collectively
rearrange themselves so the potential caused by a test particle is
shielded and drops exponentially with a scale length of $\lambda_D$,
rather than as $\sim\frac{1}{r}$.  Getting back to the definition,
quasi-neutrality can now be better described.  In a plasma, the sum of
the charges over a region much bigger than the volume corresponding to
the Debye length is zero.

Therefore, an electric field cannot be simply used to accelerate
particles in plasma, because electrons would simply rearrange
themselves to shield the external field.  There are a couple of
solutions.  Gridded ion thrusters, which were mentioned at the
beginning, have a (sometimes called screening) grid, which is adjacent
to the region where the plasma is generated, extracting ions.  These
ions are then accelerated by an electric field formed between the
extractor grid and additional accelerator grid.  The accelerating
electric field, therefore, lies entirely outside of the plasma.

Hall thrusters implement a different approach.  The plasma is located
within an annulus with radial magnetic field formed by inner and outer
magnetic coils and a magnetic circuit.  The radial magnetic field
decreases the mobility of electrons which follow the axial field.
Consequently, the axial electric field does not get shielded even
though it lies directly in the plasma region.

Even though HTs clearly work and have been flown in the real space
environments, many of their aspects are not yet fully understood.
\cite{Boeuf2017} lists these three main reasons:
\begin{enumerate}
  \item The magnetic barrier perpendicular to the cathode-anode flow
    can be subject to a variety of instabilities which can
    significantly decrease the electron confinement.
  \item Electron interactions with the wall result in electron
    emissions which alter electron transport and plasma in general.
  \item Neutral gas needs to by highly ionized for good extraction.
    Ionization also introduces additional oscillatory modes like the
    breathing mode.
\end{enumerate}

A wide range of numerical simulations exist for HTs, however, first
principles computations to self-consistently evolve all aspects of a
Hall thruster have yet to be performed.  A common practice is to
implement empirically-determined estimates of electron mobility inside
the channel into the simulation \citep{Koo2006}.

%---------------------------------------------------------------------
\section{Plasma Simulations}

There are two classical approaches to plasma simulation: particle
methods and fluid methods.

The first approach directly evolves positions and velocities using
equations of motion with electromagnetic forces. A natural way to
obtain the forces is to apply the principle of superposition on  interactions between the
individual particles.  For example in an electrostatic case, the
interactions are given by the  Coulomb force,
\begin{align}
  \bm{F}_{AB} =
  \frac{1}{4\pi\varepsilon_0}\frac{q_Aq_B}{r_{AB}^3}\bm{r}_{AB}.
\end{align}
The infinite range of electromagnetic forces requires all the particle
pairs to be accounted for; which makes for an expensive algorithm,
$\mathcal{O}(N^2)$, where $N$ is the number of particles.  A more
efficient method is to interpolate all the particles on a mesh of grid
points and calculate the electromagnetic forces on the mesh using
Poisson's equation and Ampere's law. The forces are then interpolated
back onto particle positions in neighboring grid cells.  In comparison
to the naive approach, this algorithm significantly decreases the
computational cost to $\mathcal{O}(N\mathrm{log}N)$.  The strength of
this method is that no assumptions about the particle dynamics are
made and, therefore, a wide variety of kinetic phenomena is
intrinsically included in the system.  This makes particle methods
particularly well suited for simulation of weakly collisional plasmas
where particle distributions are far from the equilibrium.

On the other hand, particle methods theoretically require simulating
unfeasible amount of particles.\footnote{\cite{Munroe2014} provides a
  charming description of the magnitude of the numbers involve.  One
  mol of a gas at standard condition has a volume of \SI{22.4}{dm^3},
  which is an imaginable amount.  Such a volume contains the
  Avogadro's number, $\mathrm{N}_A = 6.022\times10^{23}$, of
  particles.  \cite{Munroe2014} asks a question, ``How big would a mol
  of moles be?'', by which he means the Avogadro's number of moles.
  He estimates that the amount would cover the Earth surface up to
  couple tens of kilometers or form a compact body of a size of the
  Moon.}  Instead, individual particles are grouped into a fewer
number of ``macro-particles,'' which decreases computational cost but
also introduces statistical noise.  The noise can be decreased by
increasing the number of ``macro-particles,'' however, the
signal-to-noise ratio improves only as $\sqrt{N}$, where $N$ is the
number of particles per grid cell, which makes this an inefficient
proposition.

In a situation where temporal and spatial scales of
gyromotion\footnote{Circular motion of charged particles around the
  magnetic field lines.} are much shorter than the scale of interest,
the system can be reduced by integrating over one velocity component.
This approach is called gyro-kinetics.

The system of equations can be reduced even more by integration over
the two remaining velocity components, which leads to the fluid
description of the plasma.  In the fluid model, the individual particle
positions and velocities are lost and only the macroscopic quantities
like density, $n$ or bulk velocity, $\bm{u}$, are resolved.  These are
resolved by evolving the conservation equations; for example the
continuity equation (see \ser{model:fluids} for more information),
\begin{align*}
  \pfrac{n}{t} + \nabla\cdot(n\bm{u}) = 0.
\end{align*}
Fluid models are excellent tools for macroscopic plasma simulation,
because they are usually a couple of orders of magnitude faster than
kinetic models and are not affected by statistical noise.  However,
any kinetic effects that may be potentially relevant for a given
situation need to be artificially included.

This work focuses on an alternative approach -- a continuum kinetic
method which relies on directly discretizing the Vlasov equation (see
Chapter\thinspace\ref{chap:model} for more details),
\begin{align*}
  \pfrac{f}{t} +
  \bm{v}\cdot\nabla_{\bm{x}}f +
  \frac{q}{m}\left(\bm{E} +
  \bm{v}\times\bm{B}\right)\cdot\nabla_{\bm{v}}f = 0,
\end{align*}
where $f$ is the particle distribution.  Having a similar mathematical
form as other conservation equations, the same approach can be used to
solve it, e.g., finite-elements methods.  However, since the
distribution function $f$ is directly discretized, no assumptions over
the individual particle velocities are necessary.  In other words,
continuum kinetic methods provide noise-free solutions with kinetic
effects intrinsically included in the system.

%---------------------------------------------------------------------
\section{Objectives}

The goal of this work is to build and test model pieces necessary to
create a truly predictive Hall thruster model and provide a better
understanding of plasma-material interactions.  The model is intended
to be applicable to a variety of plasma configurations.  The
objectives can be further specified as
\begin{enumerate}
  \item Develop a continuum kinetic framework to study plasma sheath
    physics and benchmark this framework to studies of plasma
    instabilities.
  \item Increase the fidelity of classical sheath simulations by
    accounting for additional physics brought about through
    appropriate collision operators and ionization.
  \item Understand plasma-material interactions in Hall thrusters by
    developing a secondary electron emission boundary condition based
    on phenomenological models to study its effect on sheaths.
    Particularly test possible changes to the shape of the sheath
    potential based on the emission and collisions.  Also, an
    understanding of secondary electron emission (SEE) using a kinetic
    model can provide insight into fluid boundary conditions to
    appropriately model SEE.
\end{enumerate}

%---------------------------------------------------------------------
\section{Notes on Conventions Used}

The International System of Units (SI; Syst\`{e}me international
d'unit\'{e}s) is used throughout the work with the exception of
temperature, which is given in terms of energy, i.e., it is
assumed to by multiplied by the Boltzmann constant, $k_\mathrm{B} =
\SI{1.38064852(79)e-23}{J/K}$ \citep{mohr2016}. What is more, it is
common in plasma physics to use electron-volts as a unit of energy
instead of Joules, $\SI{1}{eV}=\SI{1.60217662e-19}{J}$.

Einstein's summation convention is used for the tensor indices. For
example, if $\mathbf{a} = \left(a_1,a_2,a_3\right)$ and $\mathbf{b} =
\left(b_1,b_2,b_3\right)$,
\begin{align*}
  a_ib_i = \sum_i^3a_ib_i = a_1b_1 + a_2b_2 + a_3b_3.
\end{align*}
Another good example is with the Levi-Civita symbol,
$\varepsilon_{ijk}$,
\begin{align*}
  \varepsilon_{1jk}a_jb_k = a_2b_3 - a_3b_2.
\end{align*}

%---------------------------------------------------------------------
\section{Notes on the Simulations}

The numerical development, simulations, and post-processing in this
work are performed using the \texttt{Gkeyll} framework and this work
has contributed to core development of
\texttt{Gkeyll}.\footnote{\url{http://gkeyll.readthedocs.io/en/latest/}}
\texttt{Gkeyll 2.0} is used as the simulation tool of choice for this
work.

In order to provide the maximum reproducibility of the presented
results, all simulation initialization files are available for the
reader together with the details of some postprocessing techniques.
\texttt{Gkeyll 2.0} build is available in the Anaconda cloud and can
be conveniently installed using the \texttt{conda} package
manager\footnote{Installing \texttt{Gkeyll 2.0} through \texttt{conda}
  most likely results in suboptimal performance, however, it is useful
  for experimentation on a new machine.  Production level run should
  use properly built code utilizing local message parsing interface
  (MPI).}
\begin{lstlisting}[language=Bash]
conda install -c gkyl gkyl
\end{lstlisting}
Assuming Anaconda is already in the \texttt{PATH}, the simulations can
then be run, for example with
\begin{lstlisting}[language=Bash]
gkyl two-stream.lua
\end{lstlisting}

%% file: model.tex
\chapter{Numerical Model and Implementation}\label{chap:model}

\epigraph{An approximate answer to the right problem is worth a good
  deal more than an exact answer to an approximate
  problem.}{\textit{John Turkey}}

This chapter describes kinetic plasma equations; both the analytic
derivation of the governing Vlasov/Boltzmann equations and its
numerical discretization.  The fluid approximation is discussed as
well.

%=====================================================================
\section{Kinetic Plasma Equation}

The word \emph{kinetic} originates from the ancient Greek \emph{kinein}
which means \emph{to move}.  Nowadays, Merriam-Webster dictionary
defines it as \emph{of or relating to the motion of material bodies
  and the forces and energy associated therewith}.  In physics,
\emph{kinetic} means that the motions of individual particles or
macro-particles are taken into account and no assumption on the
velocity distribution is done \textit{a priori}.

%---------------------------------------------------------------------
\subsection{Klimontovich Equation}

\cite{Nicholson1983} starts the derivation of the full kinetic theory
with a definition of a density of a single particle, $i$,
\begin{align*}
  N_i(t,\bm{x},\bm{v}) =
  \delta\big(\bm{x}-\bm{X}_i(t)\big)
  \delta\big(\bm{v}-\bm{V}_i(t)\big),
\end{align*}
where $\delta$ is the Dirac delta function and $\bm{X}$ and $\bm{V}$
are the Lagrangian coordinates of the particle.  Note that, even
though the function is nonzero only at the position of the particle,
it is defined over the whole phase space, i.e., the 6-dimensional
space which is a combination of the 3-dimensional configuration space
(parameterized with $\bm{x}$) and 3-dimensional velocity space
(parameterized by $\bm{v}$).  In other words, the location in phase
space provides not only the information about the physical position
but also the vector of velocity.  Consequently, the units of $N$ are
\si{m^{-6}s^{3}} rather than \si{m^{-3}} used for the classical
density.

The extension for multiple particles is then obtained as a summation
over the individual densities
\begin{align}\label{eq:model:density1}
  N_s(t,\bm{x},\bm{v}) = \sum_{i} 
  \delta\big(\bm{x}-\bm{X}_i(t)\big)
  \delta\big(\bm{v}-\bm{V}_i(t)\big).
\end{align}
From now on, the index $s$ will denote the type of the particles,
i.e., electrons, ions, etc.

The description of this distribution is not of
particular interest.  For the purposes of the kinetic theory, the
information about the evolution of the system based on the current
state is more intriguing.  Therefore, we proceed with taking the
time derivative of \eqr{model:density1},
\begin{align}\label{eq:model:density2}
\begin{aligned}
  \frac{\partial N_s(t,\bm{x},\bm{v})}{\partial t} = 
  &- \sum_{i} \bm{\dot{X}}_i\cdot\nabla_{\bm{x}}
  \delta\big(\bm{x}-\bm{X}_i(t)\big)
  \delta\big(\bm{v}-\bm{V}_i(t)\big) \\
  &- \sum_{i} \bm{\dot{V}}_i\cdot\nabla_{\bm{v}}
  \delta\big(\bm{x}-\bm{X}_i(t)\big)
  \delta\big(\bm{v}-\bm{V}_i(t)\big), 
\end{aligned}\end{align}
where $\nabla_{\bm{x}} = (\partial_{x}, \partial_{y}, \partial_{z})$
and $\nabla_{\bm{v}} = (\partial_{v_x}, \partial_{v_y},
\partial_{v_z})$.

Up until this point, the whole description was purely mathematical.
Now it is required to include physics; specifically the relation
\begin{align}\label{eq:model:motion1}
  \bm{\dot{X}}_i(t) = \bm{V}_i(t)
\end{align}
and the Lorentz force equation
\begin{align}\label{eq:model:motion2}
  m_s\bm{\dot{V}}_i =
  q_s\bm{E}^m\big(t, \bm{X}_i(t)\big) + q_s \bm{V}_i(t)
  \times \bm{B}^m\big(t,\bm{X}_i(t)\big),
\end{align}
where $m_s$ and $q_s$ are mass and charge respectively of
particle $s$. $\bm{E}^m$ and $\bm{B}^m$ represent microscopic
electric and magnetic fields from other particles (fields from the
particle itself are neglected) together with the external macroscopic
fields.  The microscopic fields satisfy Maxwell's equations,
\begin{align}
  \nabla\cdot\bm{E}^m(t,\bm{x}) &=
  \frac{\rho^m(t,\bm{x})}{\varepsilon_0},\label{eq:model:gauss} \\
  \nabla\cdot\bm{B}^m(t,\bm{x}) &= 0, \\
  \nabla\times\bm{E}^m(t,\bm{x}) &=
  -\frac{\partial \bm{B}^m(t,\bm{x})}{\partial t}, \\
  \nabla\times\bm{B}^m(t,\bm{x}) &=
  \mu_0 \bm{j}^m(t,\bm{x}) + \mu_0\varepsilon_0 \frac{\partial
    \bm{E}^m(t,\bm{x})}{\partial t},
\end{align}
where $\rho^m$ is microscopic charge density
\begin{align*}
  \rho^m(t,\bm{x}) = \sum_s q_s \int N_s(\bm{x},\bm{v},t)
  \,d\bm{v} 
\end{align*}
and $\bm{j}^m$ is microscopic current density (surface
density, $[j] = \si{A.m^{-2}}$)
\begin{align*}
  \bm{j}^m(t,\bm{x}) = \sum_s q_s \int \bm{v}
  N_s(\bm{x},\bm{v},t) \,d\bm{v}.
\end{align*}

Substituting \eqr{model:motion1} and \eqr{model:motion2} into
the evolution equation \eqrp{model:density2} gives
\begin{align*}
  \frac{\partial N_s(t,\bm{x},\bm{v})}{\partial t} = 
  &- \sum_{i} \bm{V}_i\cdot\nabla_{\bm{x}}
  \delta\big(\bm{x}-\bm{X}_i(t)\big)
  \delta\big(\bm{v}-\bm{V}_i(t)\big) \\
  &- \sum_{i}
  \frac{q_s}{m_s}\left[\bm{E}^m\big(t,\bm{X}_i(t)\big) +
  \bm{V}_i(t) 
  \times \bm{B}^m\big(t,\bm{X}_i(t)\big)\right]
  \cdot\nabla_{\bm{v} }
  \delta\big(\bm{x}-\bm{X}_i(t)\big)
  \delta\big(\bm{v}-\bm{V}_i(t)\big).
\end{align*}

Using the property of the Dirac delta function, $a\delta(a-b) =
b\delta(a-b)$, $\bm{X}$ and $\bm{V}$ can be replaced with $\bm{x}$ and
$\bm{v}$ and then the order of the summations and the gradients can be
switched,
\begin{align*}
  \frac{\partial N_s(t,\bm{x},\bm{v})}{\partial t} = 
  &- \bm{v}\cdot\nabla_{\bm{x}} \sum_{i}
  \delta\big(\bm{x}-\bm{X}_i(t)\big)
  \delta\big(\bm{v}-\bm{V}_i(t)\big) \\
  &- 
  \frac{q_s}{m_s}\left[\bm{E}^m\big(t,\bm{x}\big) +
  \bm{v}
  \times \bm{B}^m\big(t,\bm{x}\big)\right]
  \cdot\nabla_{\bm{v}}   \sum_{i}
  \delta\big(\bm{x}-\bm{X}_i(t)\big)
  \delta\big(\bm{v}-\bm{V}_i(t)\big).
\end{align*}

Finally, density \eqr{model:density1} can be back-substituted to
obtain the Klimontovich equation,
\begin{align}\label{eq:model:klimontovich}
  \frac{\partial N_s(t,\bm{x},\bm{v})}{\partial t} +
  \bm{v}\cdot\nabla_{\bm{x}} N_s +
  \frac{q_s}{m_s}\left(\bm{E}^m + \bm{v} \times
  \bm{B}^m\right) \cdot\nabla_{\bm{v}} N_s = 0.
\end{align}

Knowing the initial positions, $\bm{X}_i(t=0)$, and velocities,
$\bm{V}_i(t=0)$, of all the particles, the Klimontovich equation
\eqrp{model:klimontovich} together with Maxwell's equations
provides the full, exact description of the evolution of the plasma.

%---------------------------------------------------------------------
\subsection{Vlasov/Boltzmann Equation}

While the Klimontovich equation \eqrp{model:klimontovich} captures
the discrete nature of individual particles exactly, the collection of
Dirac Delta functions is not well suited for practical use.
Therefore, we introduce a new smooth distribution function
$f_s(t,\bm{x},\bm{v})$, which is defined as
\begin{align*}
 f_s(\bm{x},\bm{v},t) := \left\langle
 N_s(\bm{x},\bm{v},t) \right\rangle,
\end{align*}
where $\langle\cdot\rangle$ denotes the ensemble average, i.e., an
average over the all possible microstates realizing the given
macro-state.  This new distribution is the key variable of the
continuum kinetic method and represents the phase space particle
density integrated over the small volume $\Delta\bm{x}\Delta\bm{v}$
with center at $(\bm{x},\bm{v})$.  Analogously to the distribution
$f$, the ensemble averages can be defined for the fields as $\bm{E} :=
\langle\bm{E}^m\rangle$ and $\bm{B} := \langle\bm{B}^m\rangle$.

When the ensemble averages are substituted into the Klimontovich
equation \eqrp{model:klimontovich}, the Boltzmann equation is obtained,
\begin{align}\label{eq:model:boltzmann}
  \pfrac{f_s}{t} +
  \bm{v}\cdot\nabla_{\bm{x}}f_s +
  \frac{q_s}{m_s}\left(\bm{E} +
  \bm{v}\times\bm{B}\right)\cdot\nabla_{\bm{v}}f_s = \underbracket{
  -\frac{q_s}{m_s}\big\langle\left(\delta\bm{E} +
  \bm{v}\times\delta\bm{B}\right)\cdot\nabla_{\bm{v}}\delta
  N_s\big\rangle}_{\sum\left(\frac{\delta f_s}{\delta t}\right)},
\end{align}
where the residuals are defined as
\begin{align*}
  \delta N_s(\bm{x},\bm{v},t) &= N_s(\bm{x},\bm{v},t)
  - f_s(\bm{x},\bm{v},t), \\
  \delta \bm{E}(\bm{x},\bm{v},t) &=
  \bm{E}^m(\bm{x},\bm{v},t) -
  \bm{E}(\bm{x},\bm{v},t),
\end{align*}
and
\begin{align*}
  \delta \bm{B}(\bm{x},\bm{v},t) =
  \bm{B}^m(\bm{x},\bm{v},t) -
  \bm{B}(\bm{x},\bm{v},t).
\end{align*}

Note that in this form, the Boltzmann equation
\eqrp{model:boltzmann} is exact and the term on the right-hand-side
of the equation captures the intrinsically discrete effects, such as
collisions, but also effects like photo-ionization.

In a regime where discrete particle effects are negligible,
\eqr{model:boltzmann} is reduced to the Vlasov equation, which is the
center-piece of this work,
\begin{align}\label{eq:model:vlasov}
  \pfrac{f_s}{t} +
  \bm{v}\cdot\nabla_{\bm{x}}f_s +
  \frac{q_s}{m_s}\left(\bm{E} +
  \bm{v}\times\bm{B}\right)\cdot\nabla_{\bm{v}}f_s = 0.
\end{align}

It is worth noting that the Vlasov equation \eqrp{model:vlasov} is
not relativistic. The relativistic extension can be obtained by adding
appropriate Lorentz factors, $\gamma$,
\begin{align}\label{eq:model:rvlasov}
  \pfrac{f_s}{t} +
  \nabla_{\bm{x}}\cdot\left(\frac{\bm{p}}{m_s\gamma}f_s\right) +
  \nabla_{\bm{p}}\cdot\left[q_s \left(\bm{E} +
  \frac{\bm{p}}{m_s\gamma}\times\bm{B}\right)f_s\right] = 0,
\end{align}
where
\begin{align*}
  \gamma = \frac{1}{\sqrt{1-\bm{p}^2/m_s^2c^2}}.
\end{align*}
This work neglects relativistic effects.

%---------------------------------------------------------------------
\subsection{Particle Distribution Function}
\label{sec:model:distf}

An insight into what the distribution function represents is key to
understanding many figures presented in this work.  Configuration
space plots are much more common in the scientific literature than
phase space plots, hence a more in-depth discussion of distribution
functions and phase space is warranted.

First, let us illustrate the behavior of the distribution function in
phase space on a toy problem of collision-less neutral gas bouncing in
a bounded domain, which is depicted in \fgr{model:distf}.  The top
left panel captures the initial condition -- a collection of particles
distributed around $x=0$ with a bulk (average) velocity of $u_x=1$.
Since the vast majority of particles have a positive
velocity\footnote{The thermal spread of particles follows the Gauss
  distribution so there are technically some particles with negative
  velocity, however, in this case the distribution has $\sigma=0.1$.
  The region of negative velocities is farther than $10\,\sigma$ from
  the bulk velocity and is, therefore, negligible.} the particles will
propagate to the right (positive $x$).  However, due to the thermal
velocity spread, the distribution also becomes skewed.  When the
particles reach the wall at the right edge of the domain they are
elastically reflected, i.e., the magnitude of the velocity is
conserved but the sign (direction) is flipped.  Noting that the
normalized bulk velocity is 1 and the domain size is 10, it should be
expected that exactly at $t=20$ the center of the distribution returns
to the original position, which is seen in the bottom right panel of
\fgr{model:distf}.

\begin{figure}[!htb]
  \centering
  \includegraphics[width=.8\textwidth]{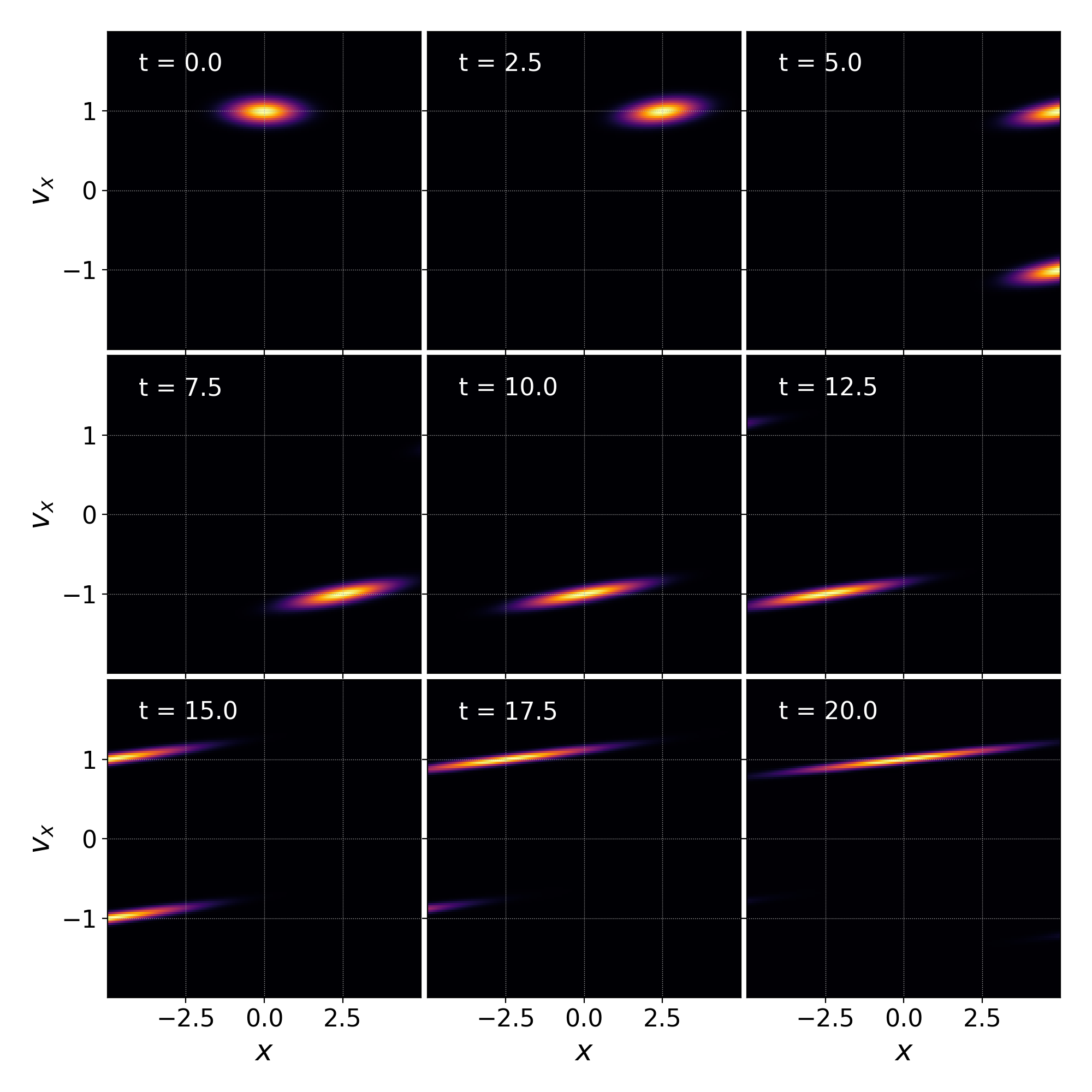}
  \caption[Example of a particle distribution
    function]{Demonstration of the evolution of the particle
    distribution function in 2D phase space (1X1V).  The top left
    panel captures the initial collection of particles with center at
    $x=0$ and with velocities thermally spread around $u_x=1$.  The
    other panels show evolution and elastic reflections
    from the walls at the edges of the domain.  At the $t=20$ (bottom
    right panel) the center of the distribution returns to the
    original position.  [Simulation input file:
      \ref{list:input:bounce}]}
  \label{fig:model:distf}
\end{figure}

Sometimes, it can be advantageous to integrate out the kinetic
information in order get the macroscopic quantities -- moments of the
distribution.  The moments have the physical meaning of density, flux,
energy, etc.\footnote{The distribution function used in this work is
  not multiplied by the mass and, therefore, the moments give the
  number density and flux rather than the mass density and
  momentum.} The particle number density is obtained as the zeroth
moment,
\begin{align}\label{eq:model:m0}
  n(\bm{x}) = \int_{\mathcal{V}} f(\bm{x},\bm{v}) \,d\bm{v},
\end{align}
where the integration is performed over the entire velocity space,
${\mathcal{V}}$. The first moment gives the particle flux,
\begin{align}\label{eq:model:m1}
  n(\bm{x})\bm{u}(\bm{x}) = \int_{\mathcal{V}} \bm{v} f(\bm{x},\bm{v}) \,d\bm{v}.
\end{align}
Note that the moment gives the conserved variable, flux, rather than
the primitive variable, bulk velocity.  In order to obtain the
primitive variables, the division by density is required,
\begin{align*}
  \bm{u}(\bm{x}) = \frac{\int_{\mathcal{V}} \bm{v} f(\bm{x},\bm{v})
    \,d\bm{v}}{\int_{\mathcal{V}} f(\bm{x},\bm{v}) \,d\bm{v}}.
\end{align*}
Note the similarity of the moment calculation to definition of an
average value in statistical math; the average value of a variable $a$
with a probability distribution $p(a)$ is calculated as $\langle a
\rangle = \int a p(a) \,da$. However, while a probability is
normalized to 1 and the distribution function here is normalized to
$n$.\footnote{While this choice is common, there are works which
  normalize particle velocity distribution to 1 as well.}

Finally, the second moment can be scaled to provide energy per unit
mass.  It can either be calculated as a scalar value
of the total energy or as the energy tensor,
\begin{align}
  \mathcal{E}(\bm{x}) &= \frac{1}{2} m\int_{\mathcal{V}} |\bm{v}|^2
  f(\bm{x},\bm{v})
  \,d\bm{v},\label{eq:model:m2}\\ \bm{\mathcal{E}}(\bm{x}) &=
  \frac{1}{2} m \int_{\mathcal{V}} \bm{vv} f(\bm{x},\bm{v}) \,d\bm{v},
\end{align}
where $\bm{vv}$ is a dyadic tensor.  The connection between the
expressions is then $\mathcal{E} =
\mathrm{Trace}(\bm{\mathcal{E}})$.

While discussing the distribution functions, it is worth mentioning
one which stands out in particular -- the Maxwellian distribution.
Many textbooks provide the following derivation,\footnote{This
  derivation is, apart from its simplicity, also of a historical
  interest because it is the argument originally given by
  \cite{Maxwell1890}.} which assumes that the probability
$f_{1D}(v_x)dv_x$, i.e. the probability of finding the particle in the
interval $\langle v_x,\, v_x + dv_x\rangle$, is independent of $v_y$
and $v_z$.  Then
\begin{align}\label{eq:model:factor}
  f(v_x,v_y, v_z)dv_x dv_y dv_z = n f_{x}(v_x) f_{y}(v_y) f_{z}(v_z) dv_x dv_y
  dv_z.
\end{align}
In a situation with no external forces, there is no preferred velocity
direction and, therefore, the distribution must depend on the velocity
only through its magnitude, $v_x^2 + v_y^2 + v_z^2$.  That means that
\begin{align*}
  n f_{x}(v_x) f_{y}(v_y) f_{z}(v_z) = f(v_x,v_y, v_z) = F(v_x^2 + v_y^2 +
  v_z^2),
\end{align*}
where $F$ is some unknown function.  Solving the equation gives
\begin{align*}
  f_{x}(v_x) = A \exp(Bv_x^2), \quad f(v_x,v_y, v_z) = n A^3
  \exp\big(B(v_x^2 + v_y^2 + v_z^2)\big).
\end{align*}
Then we can use the definitions of the moments to tie the integration
constants $A$ and $B$ with the macroscopic physical quantities:
\begin{align*}
  n = \int f \,d\bm{v} = \int n A^3 \exp\big(-B(v_x^2 + v_y^2 +
  v_z^2)\big) \,d\bm{v} \quad\Rightarrow\quad A =
  \sqrt{\frac{B}{\pi}}.
\end{align*}
From the definition of the thermal velocity:\footnote{The thermal
  velocity can be tied to the temperature through the Equipartition
  theorem, $NT/2 = 0.5 m v_{th}^2$, where $N$ is the number of degrees
  of freedom. It is worth noting, that this definition is not
  unique. For example, \cite{Chen1985} ties thermal velocity and
  temperature for $N=1$ as $mv_{th}^2 = 2T$, with gives Maxwellian
  distribution proportional to $\exp(-v_x^2/v_{th}^2)$.  The advantage
  of the definition used in this work is (apart from satisfying the
  Equipartition theorem) that the Maxwellian distribution has the
  mathematical form of the normal distribution with thermal velocity
  being the variance, $\sigma$.}
\begin{align*}
  v_{th}^2 = \frac{1}{n}\int \bm{v}^2 f \,d\bm{v} = \int (v_x^2 +
  v_y^2 + v_z^2) \sqrt{\frac{B}{\pi}}^3 \exp\big(-B(v_x^2 + v_y^2 + v_z^2)\big)
  \,d\bm{v} \quad\Rightarrow\quad B = \frac{1}{2v_{th}^2}.
\end{align*}

All together, we get the Maxwellian distribution of particles with
zero bulk velocity,
\begin{align}
  f(v_x,v_y, v_z) = \frac{n}{\sqrt{2\pi v_{th}^2}^3} \exp\left(-\frac{v_x^2 +
        v_y^2 + v_z^2}{2v_{th}^2}\right).
\end{align}
For the nonzero bulk velocity, $\bm{u}=(u_x,u_y,u_z)$,
\begin{align}\label{eq:model:maxwellian}
  f(v_x,v_y, v_z) = \frac{n}{\sqrt{2\pi v_{th}^2}^3}
  \exp\left(-\frac{(v_x-u_x)^2 + (v_y-u_y)^2 +
    (v_z-u_z)^2}{2v_{th}^2}\right).
\end{align}

However, phenomena such as inter-particle collisions would lead to the
breakdown of the assumption of independent velocity components.  The
first satisfactory derivation was performed by Boltzmann using his
H-theorem. While a brief description is provided here, the full
process is available in Chapters 3 and 4 of \cite{Chapman1970}.  The
derivation is based on more detailed description of collisions.

First we define the $H$-function,
\begin{align}\label{eq:model:H}
  H := \int f \mathrm{ln} f \,d\bm{v},
\end{align}
and its derivative,
\begin{align}\label{eq:model:derH}
  \pfrac{H}{t} = \int (1 + \mathrm{ln} f) \pfrac{f}{t} \,d\bm{v}.
\end{align}
In the absence of external forces ($\bm{F}=0$) and for uniform plasma
($\nabla_{\bm{x}}f=0$), the time derivative of the distribution
function is given only through collisions.  The process of a binary
collision can be seen as a removal of particles from phase space at
velocities $\bm{v}_1$ and $\bm{v}_2$ and a creation of ``new''
particles at $\bm{v}_1'$ and $\bm{v}_2'$.  During the process the
amount of ``lost'' particles is
\begin{align*}
  f(\bm{v}_1)f(\bm{v}_2)w_{12}\alpha_{12}
  d\bm{e}'d\bm{v}_1d\bm{v}_2d\bm{r}dt.
\end{align*}
Symmetrically,
\begin{align*}
  f(\bm{v}_1')f(\bm{v}_2')w_{12}\alpha_{12}
  d\bm{e}'d\bm{v}_1d\bm{v}_2d\bm{r}dt
\end{align*}
particles are ``created''.  The total amounts are found through
integration over the whole velocity space.  $w_{12}$ is the magnitude
of the relative velocity, $\bm{w}_{12} = w_{12}\bm{e} = \bm{v}_2
-\bm{v}_1$ and $\alpha_{12}$ is the geometric factor describing the
collision \citep{Chapman1970}.  The collision term in
\eqr{model:boltzmann} can then be described as
\begin{align}
  \left(\frac{\delta f}{\delta t}\right) = \iint \left[f(\bm{v}')
    f(\bm{v}_1') - f(\bm{v})
    f(\bm{v}_1)\right]w_{12}\alpha_{12} \,d\bm{e}'d\bm{v}_1,
\end{align}
and we can substitute into \eqr{model:derH},
\begin{align*}
  \pfrac{H}{t} = \iiint \left[1 + \mathrm{ln} f(\bm{v})\right]
  \left[f(\bm{v}') f(\bm{v}_1') - f(\bm{v})
    f(\bm{v}_1)\right]w_{12}\alpha_{12} \,d\bm{e}'d\bm{v}_1 d\bm{v}.
\end{align*}

The collision integrals have an interesting property -- since we
require $d\bm{e}'d\bm{v}_1d\bm{v}_2 = d\bm{e}d\bm{v}_1'd\bm{v}_2'$,
variables of integration can be interchanged, giving
\begin{align*}
  \iiint
  \phi(\bm{v}_1)f(\bm{v}_1')f(\bm{v}_2')w_{12}\alpha_{12}
  \,d\bm{e}'d\bm{v}_1d\bm{v}_2 = \iiint
  \phi(\bm{v}_1)f(\bm{v}_1')f(\bm{v}_2')w_{12}\alpha_{12}
  \,d\bm{e}d\bm{v}_1'd\bm{v}_2',
\end{align*}
where $\phi$ is a test function.  This corresponds to integration over
all the inverse processes.  Since the forward and inverse binary
collisions uniquely match, we can now interchange all the variables
\citep{Chapman1970},
\begin{align}\label{eq:model:itransform}
  \iiint
  \phi(\bm{v}_1)f(\bm{v}_1')f(\bm{v}_2')w_{12}\alpha_{12}
  \,d\bm{e}'d\bm{v}_1d\bm{v}_2 = \iiint
  \phi(\bm{v}_1')f(\bm{v}_1)f(\bm{v}_2)w_{12}\alpha_{12}
  \,d\bm{e}'d\bm{v}_1d\bm{v}_2.
\end{align}
Adopting the shorter notation, e.g., $f(\bm{v}_1') = f_1'$, the time
derivative of the $H$-function can be rewritten as
\begin{align*}
  \pfrac{H}{t} &= \frac{1}{4}\iiint \left(1 + \mathrm{ln}f + 1 +
    \mathrm{ln}f_1 - 1 - \mathrm{ln}f' - 1 - \mathrm{ln}f_1'\right)
  \left(f'f_1' - ff_1\right)
  w_{12}\alpha_{12} \,d\bm{e}'d\bm{v}_1 d\bm{v}\nonumber
  \\ &=\frac{1}{4}\iiint \mathrm{ln}(ff_1/f'f_1')
  \left(f'f_1' - ff_1\right)
  w_{12}\alpha_{12} \,d\bm{e}'d\bm{v}_1 d\bm{v}.
\end{align*}
Noticing that $\mathrm{ln}(ff_1/f'f_1') \left(f'f_1' - ff_1\right)
\leq 0$, we immediately obtain the Boltzmann $H$-theorem,
\begin{align}\label{eq:model:htheorem}
  \pfrac{H}{t} \leq 0.
\end{align}
What is more, $H$ is bounded below because $H = -\infty $ only if the
integral diverges. The minimal state must be given by
\begin{align}\label{eq:model:balancing}
  f'f_1' - ff_1 = 0.
\end{align}
This result is also called the Principle of absolute balancing and was
introduced by Maxwell in 1867. Taking the logarithm of
\eqr{model:balancing}, $\mathrm{ln}f'+\mathrm{ln}f_1' -
\mathrm{ln}f-\mathrm{ln}f_1 = 0$ shows that $\mathrm{ln}f$ is a
summation invariant, i.e. a quantity which sum over all the particles
is unaltered by the collisions. Another examples of summation
invariants are
\begin{align}
  \psi^{(1)} = 1, \quad \bm{\psi}^{(2)} = m\bm{v}, \quad \psi^{(3)} =
  \frac{1}{2}mv^2.
\end{align}
They correspond to the conservation of particles, momentum, and
energy, respectively, during the elastic collisions.  Any linear
combination of summation invariants is a summation invariant as well.
What is more, every summation invariant can be described as a linear
combination of the three invariants above
\citep{Chapman1970}.\footnote{Each collision is fully defined by six
  relations for six variables (twice three velocity components), but
  two of them, like the two polar angles of the line of centers at
  collision, are disposable.  Therefore, four relations (conservation
  of momentum and energy) should fully describe the encounter.}  In
other words, $\mathrm{ln}f$ can be expressed as
\begin{align*}
  \mathrm{ln}f &= \sum_i\alpha^{(i)}\psi^{(i)} \\ &= \alpha^{(1)} +
  m\left(\alpha_x^{(2)}v_x + \alpha_y^{(2)}v_y +
  \alpha_z^{(2)}v_z\right) + \frac{1}{2}m\left(v_x^2 + v_y^2 +
  v_z^2\right) \nonumber\\ &= \underbracket{\alpha^{(1)} -
    \left(\frac{\alpha_x^{(2)}}{\alpha^{(3)}}\right)^2
    -\left(\frac{\alpha_x^{(2)}}{\alpha^{(3)}}\right)^2
    -\left(\frac{\alpha_x^{(2)}}{\alpha^{(3)}}\right)^2}_{\mathrm{ln}A}
  -\nonumber\\&~~~- \underbracket{\frac{1}{2}\alpha^{(3)}m}_{B} \Bigg[
    \bigg(v_x-\underbracket{\frac{\alpha_x^{(2)}}{\alpha^{(3)}}}_{u_x}\bigg)^2
    +
    \bigg(v_y-\underbracket{\frac{\alpha_y^{(2)}}{\alpha^{(3)}}}_{u_y}\bigg)^2
    +
    \bigg(v_z-\underbracket{\frac{\alpha_z^{(2)}}{\alpha^{(3)}}}_{u_z}\bigg)^2
    \Bigg] \nonumber\\ f &= A \exp\big(-B\left[(v_x-u_x)^2 +
    (v_y-u_y)^2 + (v_z-u_z)^2\right]\big),
\end{align*}
where the constants $A$ and $B$ are obtained the same way as in the
discussion above.

To sum up, the derivation using $H$-theorem and Principle of
balancing, not only shows the form of Maxwellian distribution without
assuming the independence of velocities but also demonstrates that any
non-equilibrium mixture of particles relaxes towards it in time
through collisions!

%=====================================================================
\section{Discontinuous Galerkin Continuum Kinetic Model}
\label{sec:model:dgck}

Now that the governing equation is defined, it needs to be discretized
for computer simulations.  However, the Vlasov equation
\eqrp{model:vlasov} is not sufficient on its own because it only
describes the evolution of a single species.  Multiple species
are coupled together using fields and collisions and, in order
to self-consistently evolve the fields, additional equations are
required.  It is either Poisson's equation,
\begin{align}\label{eq:model:poisson}
  \nabla^2 \phi = - \frac{\rho_c}{\varepsilon_0},
\end{align}
for electrostatic cases or Maxwell's equations,
\begin{align}
  \pfrac{\bm{B}}{t} + \nabla\times\bm{E} &=
  0, \label{eq:model:faraday}\\ 
  \varepsilon_0\mu_0\pfrac{\bm{E}}{t}
  - \nabla\times\bm{B} &=
  -\mu_0\bm{j}, \label{eq:model:ampere}
\end{align}
for electromagnetic problems.  In Poisson's equation, $\phi$ is the
electrostatic potential, $\rho_c=\sum_s q_sn_s$ is the charge density,
and $\varepsilon_0$ is the vacuum permitivity.  The electric field,
which needs to be fed into the Vlasov equation \eqrp{model:vlasov}, is
then by definition $\bm{E}=-\nabla\phi$.  In Maxwell's equations,
$\mu_0$ is the vacuum permeability and $\bm{j} = \sum_s q_sn_su_s$ is
the current density.

This section describes the construction of the discrete Vlasov-Maxwell
system and its implementation into the
\texttt{Gkeyll}\footnote{\url{http://gkeyll.readthedocs.io}}
simulation framework \citep{Juno2017}.  While traditional methods like
finite difference method (FDM) or finite volume method (FVM) can be
used, the discontinuous Galerkin (DG) method
\citep{Reed1973,Cockburn2001} is the choice for \texttt{Gkeyll}.  DG
belongs to the family of finite-elements methods (FEM), known for
high-order accuracy and ability to handle complex geometries; however,
they also possess some traits typical for finite-volume methods like
data locality and the possibility to use limiters.  Increasing the
order of the polynomial approximation provides a sub-cell accuracy in
a similar manner as grid refinement, but usually at a fraction of the
cost \citep{Hesthaven2007}, and data locality allows for efficient
parallelization.  Merging these traits together makes DG a powerful
tool for high-performance computing.  Furthermore, certain DG methods
allow conservation of energy and with suitable modifications,
positivity and entropy decay. Most FV methods do not possess these
properties.

%---------------------------------------------------------------------
\subsection{Discrete Vlasov Equation}
\label{sec:model:discvlasov}

While it makes sense from the physics point of view to distinguish
between the position, $\bm{x}$, and the velocity, $\bm{v}$, for the
purpose of deriving the discrete form of the Vlasov equation
\eqrp{model:vlasov}, it is advantageous to rewrite it as function
of a single phase-space variable, $\bm{z}$,\footnote{It might not be
  obvious how the Vlasov equation \eqrp{model:vlasov} can be
  written in the conservative form because the Lorentz force is a
  function of $\bm{v}$. However, since
  $\left(\bm{v}\times\bm{B}\right)_i = \varepsilon_{ijk}v_jB_k$, where
  $\varepsilon_{ijk}$ is the Levi-Civita symbol,
  $$\frac{\partial\left(\bm{v}\times\bm{B}\right)_if}{\partial
    v_i} = \frac{\partial \varepsilon_{ijk}v_jB_kf}{\partial v_i} =
  \varepsilon_{ijk}v_jB_k \frac{\partial f}{\partial v_i} =
  \left(\bm{v}\times\bm{B}\right)_i \frac{\partial f}{\partial
    v_i}.$$}
\begin{align}\label{eq:model:vlasovz}
  \pfrac{f}{t} + \nabla_{\bm{z}}\cdot(\bm{\alpha}f) = 0,
\end{align}
where
\begin{align*}
  \nabla_{\bm{z}} = \left(\nabla_{\bm{x}},\nabla_{\bm{v}}\right),
\end{align*}
and
\begin{align*}
  \bm{\alpha} = \left(\bm{v},\,\frac{q}{m}\left(\bm{E} +
  \bm{v}\times\bm{B}\right)\right).
\end{align*}
For clarity, index $s$ denoting species is omitted in the
derivation.

Using the DG method, the exact distribution function,
$f(t,\bm{z})$, is discretized by a piecewise polynomial from the
space
\begin{align}
  \mathcal{S}_h^p := \left\{\psi:\psi|_{K^j} \in \mathcal{P}^p,\,\forall K^j
  \in\mathcal{T}\right\},
\end{align}
where $K^j, j=1,\ldots,N_c$ are the $N_c$ cells of the phase space
$\mathcal{T}$. In a 1D case, $\mathcal{P}^p$ is the space of the polynomials
with the polynomial order at most $p$; see the subsection
\ref{sec:model:choice} for details and for the definition of
$\mathcal{P}^p$ for higher dimensions.  The global solution is then defined
using the direct sum
\begin{align}
  f(t,\bm{z}) \approx f_h(t,\bm{z}) = \bigoplus_{j=1}^{N_c}f_h^j(t,\bm{z}).
\end{align}

There are generally two ways to describe the approximate solution,
$f_h^j(t,\bm{z})$.  The distribution function is discretized either
using the modal expression,
\begin{align}\label{eq:model:modal}
  f_h^j(t,\bm{z}) := \sum_{n=0}^{N_p-1} \widehat{f}^j_n(t)\psi_n(\bm{z}), \quad
  \bm{z} \in K^j, \psi_n \in \mathcal{S}_h^p,
\end{align}
where $N_p$ is the number of basis functions, or the nodal expression,
\begin{align}\label{eq:model:nodal}
  f_h^j(t,\bm{z}) := \sum_{n=0}^{N_p-1}
  f_n(t,\bm{z}_n^j)l^j_n(\bm{z}), \quad \bm{z} \in K^j,
\end{align}
where $l^j_n(\bm{z})$ is the Lagrange interpolation polynomial. In
1D they can be simply defined as
\begin{align}
  l^j_n(x) := \prod_{\substack{0\leq m\leq p \\m\neq n}} \frac{x -
    x_m^j}{x_n^j - x_m^j},
\end{align}
i.e., polynomial that is 1 at the $n$-node and 0 at every other node
$m$.  An example for three nodes ($p=2$; $x_0=-1$, $x_1=0$, and $x_2=1$)
is in \fgr{model:lagrange}.  The corresponding
polynomials are
\begin{align*}
  l_0(x) = x(x-1)/2, \quad l_1(x) = (1-x)(1+x), \quad l_2(x) =
  x(x+1)/2.
\end{align*}
\begin{figure}[!htb]
  \centering
  \includegraphics[width=.7\textwidth]{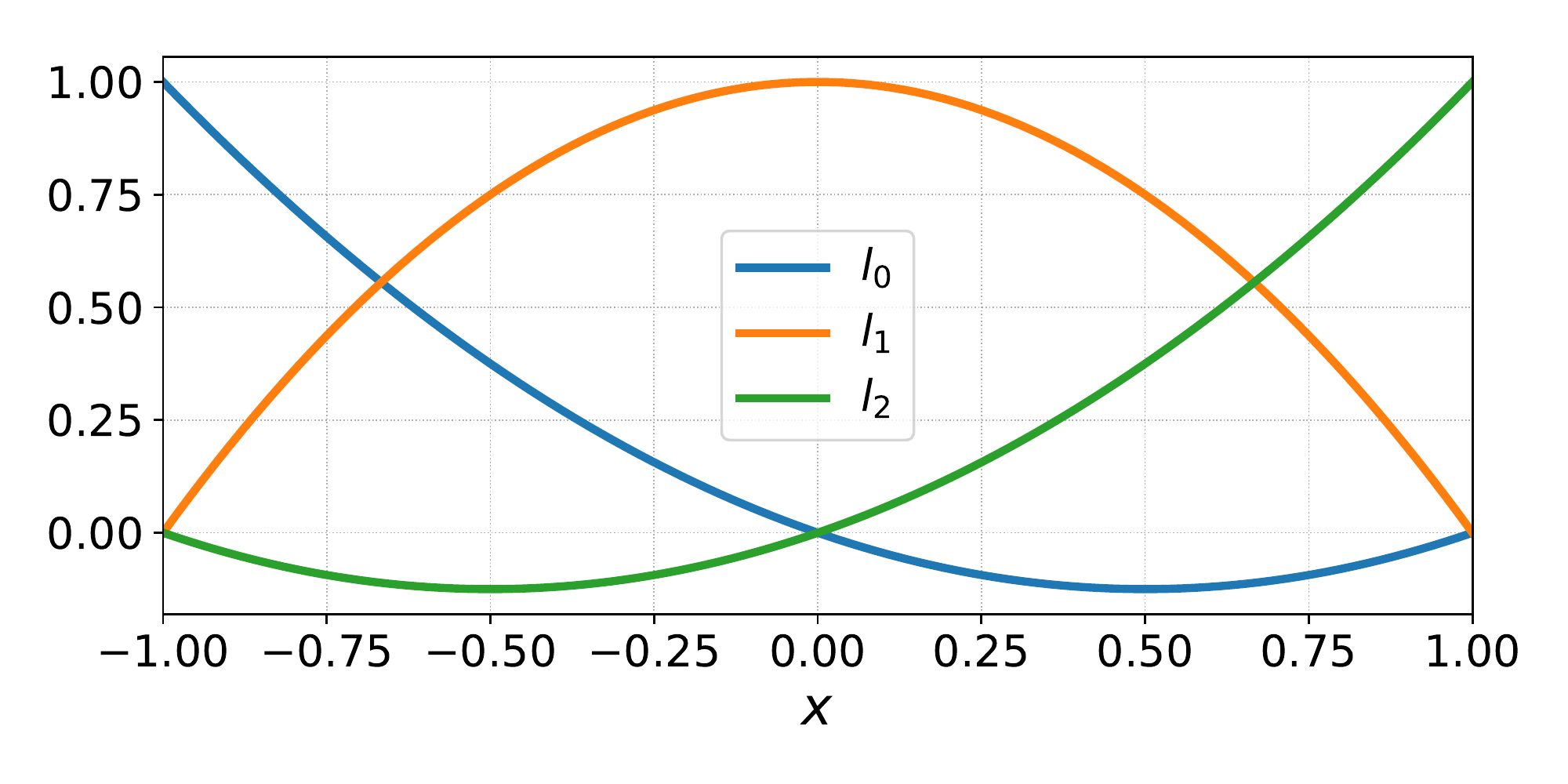}
  \caption[Example of Lagrange polynomials]{Example of Lagrange
    polynomials for three nodes at -1, 0, and 1.  Note that the
    polynomials are 1 at their corresponding node and zero at the
    other nodes. Still, they are defined over the whole interval.}
  \label{fig:model:lagrange}
\end{figure}
A good way to understand the difference between modal and nodal
description is to look at how to describe a straight line, $f(x) =
ax+b$. We can either define the modes, $a$ and $b$, or set values at
two nodes, $f(x_{0,1})$. However, it is important to note that these
two descriptions are equivalent and can be connected with the
Vandermonde matrix, $\mathcal{V}_{mn}$,
\begin{align}
  \mathcal{V}^j_{mn} \widehat{f}^j_n = f_m(\bm{z}^j_m), \quad \mathcal{V}^j_{mn} =
  \psi_n(\bm{z}^j_m).
\end{align}
Still, the two descriptions lead to different algorithms. This is one
of the big differences between \texttt{Gkeyll 1.0}, which uses the
nodal description, and the modal \texttt{Gkeyll 2.0}.

To solve the governing equation then means finding the $N_p$ unknowns
in each cell, $K^j$, that represent either nodal values of the
solution or its expansion coefficients.  In order to do this, the
approximate solution $f_h$ is required to satisfy the PDE
(Eq.\thinspace\ref{eq:model:vlasovz}) in the weak sense,
\begin{align}
  \pfrac{f_h^j}{t} + \nabla_{\bm{z}}\cdot(\bm{\alpha}^j_hf_h^j) \circeq 0.
\end{align}
This means that the equality holds when the expression is multiplied
by a test function, $\psi_t$,\footnote{The special choice of using
  basis functions as test functions separates Galerkin methods from
  other finite-element methods.} and integrated over the cell,
\begin{align*}
  \int_{K^j}\left[\pfrac{f_h^j(t,\bm{z})}{t} +
    \nabla_{\bm{z}}\cdot
    \left(\bm{\alpha}_h^j(\bm{z})f_h^j(t,\bm{z})\right)\right] \psi_t(\bm{z})
  \,d\bm{z} = 0, \quad \psi_t \in \mathcal{S}_h^p, \quad t=0,\dots,N_p-1.
\end{align*}
This give $N_p$ equations for $N_p$ unknowns.  Using the integration
\textit{per partes}, the second term can be split,
\begin{align*}
  \int_{K^j} \nabla_{\bm{z}}\cdot
  \left(\bm{\alpha}_h^j(\bm{z})f_h^j(t,\bm{z})\right) \psi_t(\bm{z})
  \,d\bm{z} =& \oint_{\partial K^j}
  \bm{\alpha}_h^j(\bm{z})f_h^j(t,\bm{z}) \psi_t(\bm{z}) \cdot d\bm{A} \\&-
  \int_{K^j} \left(\bm{\alpha}_h^j(\bm{z})f_h^j(t,\bm{z})\right)\cdot
  \nabla_{\bm{z}}\psi_t(\bm{z})\,d\bm{z},
\end{align*}
where $d\bm{A}=\bm{n}dA$ is the differential area element pointing out
of the $\partial K^j$ surface.\footnote{Note that the surface term is
  a generalization of the more common 1D form of the integration
  \textit{per partes}, $$ \int_a^b u'v dx =
  \left[uv\right]_a^b-\int_a^buv'dx, $$ which uses the Gauss theorem.}
The equation than takes the form
\begin{align*}
  \int_{K^j}\left[\pfrac{f_h^j(t,\bm{z})}{t}\psi_t(\bm{z}) -
    \left(\bm{\alpha}_h^j(\bm{z})f_h^j(t,\bm{z})\right) \cdot
    \nabla_{\bm{z}}\psi_t(\bm{z})\right] \,d\bm{z}
  = -\oint_{\partial K^j}
  \bm{\alpha}_h^j(\bm{z})f_h^j(t,\bm{z}) \psi_t(\bm{z}) \cdot d\bm{A}.
\end{align*}
However, since the solution $f_h$ is piecewise polynomial, it is
multiply defined at the interface $\partial K^j$ and a single
solution, $\bm{F}^j$, must be chosen.  This solution, known as the
numerical flux, is generally a function of the solution at both
sides,
\begin{align*}
  \bm{F}^j = \bm{F}^j(\bm{\alpha}_h^{j-}f_h^{j-},\bm{\alpha}_h^{j+}f_h^{j+}),
\end{align*}
where the $\pm$ notation indicates the evaluation just inside (-) or
outside (+) of the interface.  There are many options for the flux
function.  \texttt{Gkeyll} uses penalty fluxes defined as 
\begin{align}
  \bm{n}\cdot\bm{F}^j(\bm{\alpha}_h^{j-}f_h^{j-},\bm{\alpha}_h^{j+}f_h^{j+})
  = \frac{1}{2}\bm{n}\cdot(\bm{\alpha}_h^{j-}f_h^{j-}
    +\bm{\alpha}_h^{j+}f_h^{j+}) - \frac{A}{2}(f_h^{j+} - f_h^{j-}).
\end{align}
\texttt{Gkeyll 1.0} implements the Lax-Friedrichs flux function for
which $A=\mathrm{max}(|\bm{n}\cdot\bm{\alpha}_h^{j+}|,
|\bm{n}\cdot\bm{\alpha}_h^{j-}|)$, which leads to an up-winding
scheme. This requires the values at the interface, which are readily
available in a nodal DG algorithm, but can be expensive to calculate
in a modal DG.  Instead, in \texttt{Gkeyll 2.0} the maximal
characteristic ``speed'' is used in each direction \citep{Juno2018}.
Note the ``speed'' has units of m/s only for the configuration space
directions.  For the velocity directions, it has the units of
acceleration.

Together, this gives the weak form of the governing equation,
\begin{align}\label{eq:model:weak}
  \int_{K^j}\left[\pfrac{f_h^j(t,\bm{z})}{t}\psi_t(\bm{z}) -
    \left(\bm{\alpha}_h^j(\bm{z})f_h^j(t,\bm{z})\right) \cdot
    \nabla_{\bm{z}}\psi_t(\bm{z})\right] \,d\bm{v} = -\oint_{\partial
    K^j} \bm{F}^j(t,\bm{z}) \psi_t(\bm{z}) \cdot d\bm{A}.
\end{align}

Now we can substitute for the distribution
function,\footnote{Einstein's summation convention is used here, e.g.,
  $f_n\psi_n = \sum_{n=0}^{N_p-1}f_n\psi_n$.}
\begin{align*}
  \int_{K^j}\pfrac{\widehat{f}_{n}^j(t)}{t}\psi_n(\bm{z})\psi_t(\bm{z})
  \,d\bm{z} &= \int_{K^j} \widehat{f}_m^j(t) \psi_m(\bm{z})
  \bm{\alpha}_h^j(\bm{z}) \cdot \nabla_{\bm{z}} \psi_t(\bm{z}) \,d\bm{z}
  -\oint_{\partial K^j} \bm{F}^j(t,\bm{z}) \psi_t(\bm{z}) \cdot
  d\bm{A},\\ \pfrac{\widehat{f}_{n}^j(t)}{t}
  \underbracket{\int_{K^j}\psi_n(\bm{z})\psi_t(\bm{z})
    \,d\bm{z}}_{\mathcal{M}^j_{nt}} &= \widehat{f}_m^j(t) \int_{K^j}
  \bm{\alpha}_h^j(\bm{z}) \cdot \psi_m(\bm{z})\nabla_{\bm{z}}
  \psi_t(\bm{z}) \,d\bm{z} -\oint_{\partial K^j} \bm{F}^j(t,\bm{z})
  \psi_t(\bm{z}) \cdot d\bm{A}.
\end{align*}
The matrix $\mathcal{M}^j_{nt}=\int_{K^j} \psi_n\psi_t\,d\bm{z}$ is
called the mass matrix.  Now we can finally rearrange the equation to
usable form,
\begin{align}\label{eq:model:weak2}
  \pfrac{\widehat{f}_{n}^j(t)}{t} =
  \left(\mathcal{M}^j_{nt}\right)^{-1}\Big[\underbracket{\widehat{f}_m^j(t)
      \int_{K^j} \bm{\alpha}_h^j(\bm{z}) \cdot \big(\psi_m(\bm{z})
      \nabla_{\bm{z}} \psi_t(\bm{z})\big) \,d\bm{z}}_{\text{volume
        term}} -\underbracket{\oint_{\partial K^j} \bm{F}^j(t,\bm{z})
      \psi_t(\bm{z}) \cdot d\bm{A}}_{\text{surface term}}\Big].
\end{align}

The next question is how to evaluate the integrals. Since the basis
functions inside the integrals are polynomials, quadratures can be
used to calculate the integrals exactly. In general, quadratures
replace the definite integral from -1 to 1 with the weighted sum of
the nodal values. In 1D:
\begin{align}\label{eq:model:quadrature}
  \int_{-1}^1 f(\eta) \,d\eta \approx \sum_{i=1}^{N_q} w_if(\eta_i), \quad
  f\in\mathcal{P}^{p'}.
\end{align}
For the Gauss-Legendre (GL) quadrature, the equality is exact for $p'
\leq 2N_q - 1$, i.e.,
\begin{align*}
  \int_{-1}^1 f(\eta) \,d\eta = \sum_{i=1}^{N_q} w_if(\eta_i), \quad
  f\in\mathcal{P}^{p'}, \quad p'\leq 2N_q - 1.
\end{align*}
However, in the DG algorithm \eqrp{model:weak2}
the integrals are the product of basis functions. So, in order
to calculate the mass matrix exactly, we need at least $(2p+1)/2$
quadrature points.  The weights and nodes for the (GL) quadrature are
listed in the Table\thinspace\ref{tab:model:quadrature}.  The
generalization to higher dimensions is done using the tensor
product. For example in 2D:
\begin{align*}
  \int_{-1}^1 \int_{-1}^1 f(\eta_x,\eta_y) \,d\eta_xd\eta_y \approx \sum_{i=1}^{N_q}
  \sum_{j=1}^{N_q} w_iw_j f(\eta_{x,i}, \eta_{y,j}), \quad
  f\in\mathcal{P}^{p'},
\end{align*}
where $w_{i,j}$ are the 1D weights and $\eta_{x,i}$ and $\eta_{y,j}$
are the 1D abscissas.

\begin{table}[!htb]
 \caption[Weights and nodes for the Gauss-Legendre quadrature]{Weights
   and nodes for the Gauss-Legendre quadrature
   \eqrp{model:quadrature}.  The nodes (abscissas) are the roots
   of the Legendre polynomial $P_{N_q}(\eta)$ and the weights
   $w_i=2/[(1-\eta_i^2)(P_{N_q}'(\eta_i))^2]$ \citep{Abramowitz1985}.}
 \label{tab:model:quadrature}
 \begin{center}
   \begin{tabular}{ccc}
     \toprule $N_q$ & $\eta_i$ & $w_i$ \\ \midrule 1 & 0 & 2 \\ \midrule
     2 & $\pm\frac{1}{\sqrt{3}}$ & 1 \\ \midrule \multirow{2}{*}{3} &
     0 & $\frac{8}{9}$ \\ & $\pm\sqrt{\frac{3}{5}}$ &$
     \frac{5}{9}$\\ \midrule \multirow{2}{*}{4} &
     $\pm\sqrt{\frac{3}{7}-\frac{2}{7}\sqrt{\frac{6}{5}}}$ &
     $\frac{18+\sqrt{30}}{36}$ \\ &
     $\pm\sqrt{\frac{3}{7}+\frac{2}{7}\sqrt{\frac{6}{5}}}$&
     $\frac{18-\sqrt{30}}{36}$ \\ \midrule \multirow{3}{*}{5} & 0 &
     $\frac{128}{225}$ \\ &
     $\pm\frac{1}{3}\sqrt{5-2\sqrt{\frac{10}{7}}}$ &
     $\frac{322+13\sqrt{70}}{900}$ \\ &
     $\pm\frac{1}{3}\sqrt{5+2\sqrt{\frac{10}{7}}}$ &
     $\frac{322-13\sqrt{70}}{900}$\\ \bottomrule
   \end{tabular}
 \end{center}
\end{table}

The exact integration gets even more computationally demanding for the
volume term \eqrp{model:weak2} where we have the product of three
basis functions (due to multiplication with force terms) instead of
just two.  However, after splitting the unknowns into the expansion
coefficients (functions of just time) and the basis functions
(functions of just position) the integrals do not change in time and,
therefore, can be precomputed.  Typically, such matrices are evaluated
at the beginning of the simulation and then the matrix multiplication
is performed during the run-time.  But \cite{Juno2018} go one step
further.  Their algorithm utilizes computer algebra systems like
\textit{Maxima} or \textit{Mathematica} to evaluate the integrals
exactly and then directly generate the computational kernels with
expanded matrix multiplication.  This entirely removes the need to
compute quadratures during run-time.\footnote{At least for the volume
  and surface terms of the Vlasov equation \eqrp{model:weak}. Other
  parts of the model like collisions or boundary conditions might
  still require an exact integration which cannot be precomputed.}

Finally, we need to address the integration variables and the
integration bounds.  While the integrals in \eqr{model:weak2} are
performed over the physical coordinates, the quadratures are defined
on $I=[-1,1]^d$.  Therefore, the integrals need to be correctly
transformed.  Starting with the mass matrix, we get
\begin{align}\begin{aligned}
  \mathcal{M}^j_{nt} &= \int_{K^j} \psi_n(\bm{z})\psi_t(\bm{z})
  \,d\bm{z}, \\ &= \int_I \psi_n\big(\bm{z}^j(\bm{\eta})\big)
  \psi_t\big(\bm{z}^j(\bm{\eta})\big)
  \,\left|\frac{d\bm{z}^j}{d\bm{\eta}}\right| d\bm{\eta}, \\
  &= \int_I \widehat{\psi}_n(\bm{\eta})
  \widehat{\psi}_t(\bm{\eta})
  \,\left|\frac{d\bm{z}^j}{d\bm{\eta}}\right| d\bm{\eta}
\end{aligned}\end{align}
where
\begin{align*}
  \left|\frac{d\bm{z}^j}{d\bm{\eta}}\right| = \left| \begin{matrix} 
    \pfrac{x^j}{\eta_x} & \pfrac{x^j}{\eta_y} & \cdots \\
    \pfrac{y^j}{\eta_x} & \pfrac{y^j}{\eta_y} & \cdots \\
    \vdots & \vdots & \ddots
  \end{matrix} \right|
\end{align*}
is the Jacobian of the transformation.  Note that on uniform Cartesian
meshes, the transformation is, for example for $x$, given as
\begin{align*} 
  x^j(\eta_x) = \eta_x\frac{\Delta x}{2} + x_c^j,
\end{align*}
where $\Delta x$ is the uniform cell dimension and $x_c^j$ is the
position of the center of cell $j$. All the cross-terms are zero.
The Jacobian then simplifies to
\begin{align*}
  \left|\frac{d\bm{z}^j}{d\bm{\eta}}\right| = \frac{1}{2^d} \prod_{i=1}^d\Delta z_i
\end{align*}
Then, for the mass matrix, we can write
\begin{align}\begin{aligned}
    \mathcal{M}^j_{nt} &= \int_I \widehat{\psi}_n(\bm{\eta})
    \widehat{\psi_t}(\bm{\eta})
    \,\left|\frac{d\bm{z}^j}{d\bm{\eta}}\right| d\bm{\eta} \\ &=
    \frac{\prod_{i=1}^d\Delta z_i}{2^d} \int_I
    \widehat{\psi}_n(\bm{\eta}) \widehat{\psi_t}(\bm{\eta}) \,d\bm{\eta} \\ &=
    \frac{\prod_{i=1}^d\Delta z_i}{2^d} \widehat{\mathcal{M}}_{nt}
\end{aligned}\end{align}

For clarity, we split the volume term,
\begin{align*}
  \int_{K^j} \left[ \bm{v} \cdot \big(\psi_m(\bm{z}) \nabla_{\bm{x}}
    \psi_t(\bm{z})\big) + \frac{q}{m} \big(\bm{E}_h^j(t,\bm{x}) +
    \bm{v}\times\bm{B}_h^j(t,\bm{x})\big) \cdot \big(\psi_m(\bm{z})
    \nabla_{\bm{v}} \psi_t(\bm{z})\big)\right] \,d\bm{z}
\end{align*}
and then replace $\bm{v}^j$ by $\bm{v} - \bm{v}_c^j + \bm{v}_c^j$.
Focusing only on the first part and again assuming Cartesian mesh, we
get
\begin{align}\begin{aligned}
  \int_{K^j} \bm{v} \cdot \big(\psi_m(\bm{z}) \nabla_{\bm{x}}
  \psi_t(\bm{z})\big) \,d\bm{z} =& \int_{K^j} \big(\bm{v} - \bm{v}_c^j
  + \bm{v}_c^j\big) \cdot \big(\psi_m(\bm{z}) \nabla_{\bm{x}}
  \psi_t(\bm{z})\big) \,d\bm{z},\\ =& \int_{K^j} \big(\bm{v} -
  \bm{v}_c^j\big) \cdot \big(\psi_m(\bm{z}) \nabla_{\bm{x}}
  \psi_t(\bm{z})\big) \,d\bm{z}\\ &+ \bm{v}_c^j \cdot \int_{K^j}
  \big(\psi_m(\bm{z}) \nabla_{\bm{x}} \psi_t(\bm{z})\big)
  \,d\bm{z},\\ =& \sum_{i=1}^3 \int_{I} \frac{\eta_{i+3} \Delta
    v_{i+3}}{2} \pfrac{\widehat{\psi}_t(\bm{\eta})}{\eta_i}\frac{2}{\Delta
    v_{i+3}} \widehat{\psi}_m(\bm{\eta}) \,
  \left|\frac{d\bm{z}^j}{d\bm{\eta}}\right|d\bm{\eta}\\ &+
  \sum_{i=1}^3 (v_c^j)_i \int_{I}
  \pfrac{\widehat{\psi}_t(\bm{\eta})}{\eta_i}\frac{2}{\Delta v_{i+3}}
  \widehat{\psi}_m(\bm{\eta}) \,
  \left|\frac{d\bm{z}^j}{d\bm{\eta}}\right|d\bm{\eta}, \\ =&
  \frac{\prod_{i=1}^d\Delta z_i}{2^d} \Big[ \int_{I} \bm{\eta}_{\bm{v}}
  \cdot \big(\nabla_{\bm{\eta}_x}\widehat{\psi}_t(\bm{\eta})
  \big)\widehat{\psi}_m(\bm{\eta}) \,d\bm{\eta} \\ &+ \frac{2\bm{v}_c^j}{\Delta \bm{v}} \cdot
  \int_{I} \big(\nabla_{\bm{\eta}_x}\widehat{\psi}_t(\bm{\eta})
  \big)\widehat{\psi}_m(\bm{\eta}) \, d\bm{\eta}\Big],
\end{aligned}\end{align}
where we used index notation for clarity.  Note that the same factor
appears in front of this integral as it does in front of the mass
matrix.  The second part with the Lorentz force behaves in similar
manner after the $\bm{E}_h$ and $\bm{B}_h$ are expanded into
configuration space basis functions.

%---------------------------------------------------------------------
\subsection{Discrete Maxwell's Equations}

Unlike the Vlasov equation \eqrp{model:vlasov} Maxwell's equations are
only defined in the configuration space, $\Omega$, so the polynomial
space needs to be contracted,
\begin{align}
  \mathcal{X}_h^p := \mathcal{S}_h^p\backslash \Omega.
\end{align}
Apart from that, Faraday's \eqrp{model:faraday} and Amp\`{e}re's
\eqrp{model:ampere} laws can be rewritten in the weak form in a
similar manner to the derivation above;
\begin{align}\label{eq:model:faraday_weak}
  \int_{\Omega_j}\pfrac{\bm{B}_h^j(t,\bm{x})}{t}\varphi_t(\bm{x})
  \,d\bm{x} = \oint_{\partial \Omega_j} \left(
  \bm{E}^*_h(t,\bm{x})\varphi_t^{-}(\bm{x}) \right) \times d\bm{A} -
  \int_{\Omega_j} \bm{E}_h^j(t,\bm{x}) \times \nabla_{\bm{x}}
  \varphi_t(\bm{x}) \,d\bm{x}
\end{align}
and
\begin{align}\label{eq:model:ampere_weak}
  \varepsilon_0\mu_0 \int_{\Omega_j} \pfrac{\bm{E}_h^j(t,\bm{x})}{t}
  \varphi_t(\bm{x}) \,d\bm{x} =& \oint_{\partial \Omega_j}
  \left(\bm{B}^*_h(t,\bm{x}) \varphi^{-}_t(\bm{x}) \right) \times
  d\bm{A} + \int_{\Omega_j} \bm{B}_h^j(t,\bm{x}) \times
  \nabla_{\bm{x}}\varphi_t(\bm{x}) \,d\bm{x} \nonumber\\ &- \mu_0
  \int_{\Omega_j} \bm{j}_h^j(t,\bm{x})\varphi_t(\bm{x}) \,d\bm{x},
\end{align}
where $\varphi_t(\bm{x}) \in\mathcal{X}_h^p$ is a test function.
Similar to the discretization of the Vlasov equation, $\bm{E}_h$ and
$\bm{B}_h$ are not defined at the cell interfaces.  \cite{Juno2017}
consider two options for the flux function -- central fluxes and upwind
fluxes.  The central are simply defined as
\begin{align}\begin{aligned}
  \bm{E}_h^* &= \llbracket\bm{E_h}\rrbracket, \\
  \bm{B}_h^* &= \llbracket\bm{B_h}\rrbracket,
\end{aligned}\end{align}
where $\llbracket\cdot\rrbracket$ is the averaging operator, $\llbracket
g \rrbracket = (g^{+} + g^{-})/2$.  Upwind fluxes in the local face
coordinate system\footnote{The direction 1 is perpendicular to the
  interface and 2 and 3 are tangential to it.  Note that since the
  $\bm{E}_h$ and $\bm{B}_h$ appear in the vector product with the
  normal to the interface, only the components 2 and 3 are needed.}
are
\begin{align}\begin{aligned}
  E_{h,2}^* &= \llbracket E_{h,2}\rrbracket - \mathrm{c} \{B_{h,3}\}, \\
  E_{h,3}^* &= \llbracket E_{h,3}\rrbracket + \mathrm{c} \{B_{h,2}\}, \\
  B_{h,2}^* &= \llbracket B_{h,2}\rrbracket + \{E_{h,3}\}/\mathrm{c}, \\
  B_{h,3}^* &= \llbracket B_{h,3}\rrbracket - \{E_{h,2}\}/\mathrm{c},
\end{aligned}\end{align}
where $\mathrm{c}$ is the speed of light, and $\{\cdot\}$ is the jump
operator, $\{ g \} = (g^{+} - g^{-})/2$, \citep{Juno2017}.  While the
central fluxes are required for the Vlasov-Maxwell scheme to conserve
energy exactly (see \ser{model:conserve}), their usage can lead to
numerical instabilities \citep{Hesthaven2004}.  Therefore, only the
upwind fluxes are used through this work.

Note that evolving the fields with discrete Faraday's
\eqrp{model:faraday_weak} and and Amp\`{e}re's
\eqrp{model:ampere_weak} laws does not enforce either Gauss law,
$\nabla\cdot\bm{E} = \rho_c/\varepsilon_0,$ nor the
$\nabla\cdot\bm{B}=0$ condition. This can lead to errors, which need
to be cleaned, for example, with perfectly-hyperbolic cleaning methods
\citep{Munz2000}.  The methods work well for the $\nabla\cdot\bm{B}$
cleaning but are problematic for the electric field because the
divergence errors in $\bm{E}$ need to be self-consistently fixed in both the
Vlasov as well as Maxwell solvers.  $\nabla\cdot\bm{E}$ error cleaning
is a topic of current research for continuum kinetic methods.
However, \cite{Juno2017} looked into the evolution of
$\nabla\cdot\bm{E}-\rho_c/\varepsilon_0$ in \texttt{Gkeyll}
simulations and came to the conclusion that the correction is not
required for most problems.  Still, it is important that
$\nabla\cdot\bm{E}$ errors are not introduced by inconsistent initial
conditions.  See \ser{benchmark:landau_sim} for an example of
a simulation with initial conditions violating Gauss law.

%---------------------------------------------------------------------
\subsection{Choice of the Basis Functions and the Polynomial Space}
\label{sec:model:choice}

The choice of the basis functions, $\widehat{\psi}_n$ and
$\widehat{\varphi}_n$, is crucial for the numerical method.  As it was
discussed above \eqrp{model:weak2}, the DG algorithm requires
inverting the mass matrix.  Its condition number can therefore be used
as a good metric to assess the choice of base.

Let us start with the intuitive choice -- monomials,
\begin{align}
  \widehat{\psi}_k(\eta) &= \eta^k, \quad \forall k = 0, ..., p, \quad
  \eta \in [-1,1].
\end{align}
For example, for $p=4$ we get the following mass matrix and condition
number,
\begin{align*}
  \widehat{\mathcal{M}}_{kl} = \int_{-1}^1 \eta^k \eta^l \,d\eta
  = \begin{pmatrix} 
    2 & 0 & \frac{2}{3} & 0 & \frac{2}{5} \\ 0 &
    \frac{2}{3} & 0 & \frac{2}{5} & 0 \\ \frac{2}{3} & 0 & \frac{2}{5}
    & 0 & \frac{2}{7} \\ 0 & \frac{2}{5} & 0 &\frac{2}{7} & 0
    \\ \frac{2}{5} & 0 & \frac{2}{7} & 0 & \frac{2}{9}
  \end{pmatrix}, \quad
  \kappa^\infty(\widehat{\mathcal{M}}) = \frac{8211}{16}.
\end{align*}
The logarithm of condition number gives an estimate of how many digits
are lost in solving the linear system. For this case,
$\mathrm{log}_{10}(8211/16) \approx 2.7$.  The plot for $p=1,...,10$
is in \fgr{model:cond_number}.  Note that as the
polynomial order approaches 10 the loss in precision becomes comparable
to the single float point number precision.  This is particularly
important for simulations running on Graphical Processing Units (GPUs)
which typically run much faster with just single precision.

\begin{figure}[!htb]
  \centering
  \includegraphics[width=.9\textwidth]{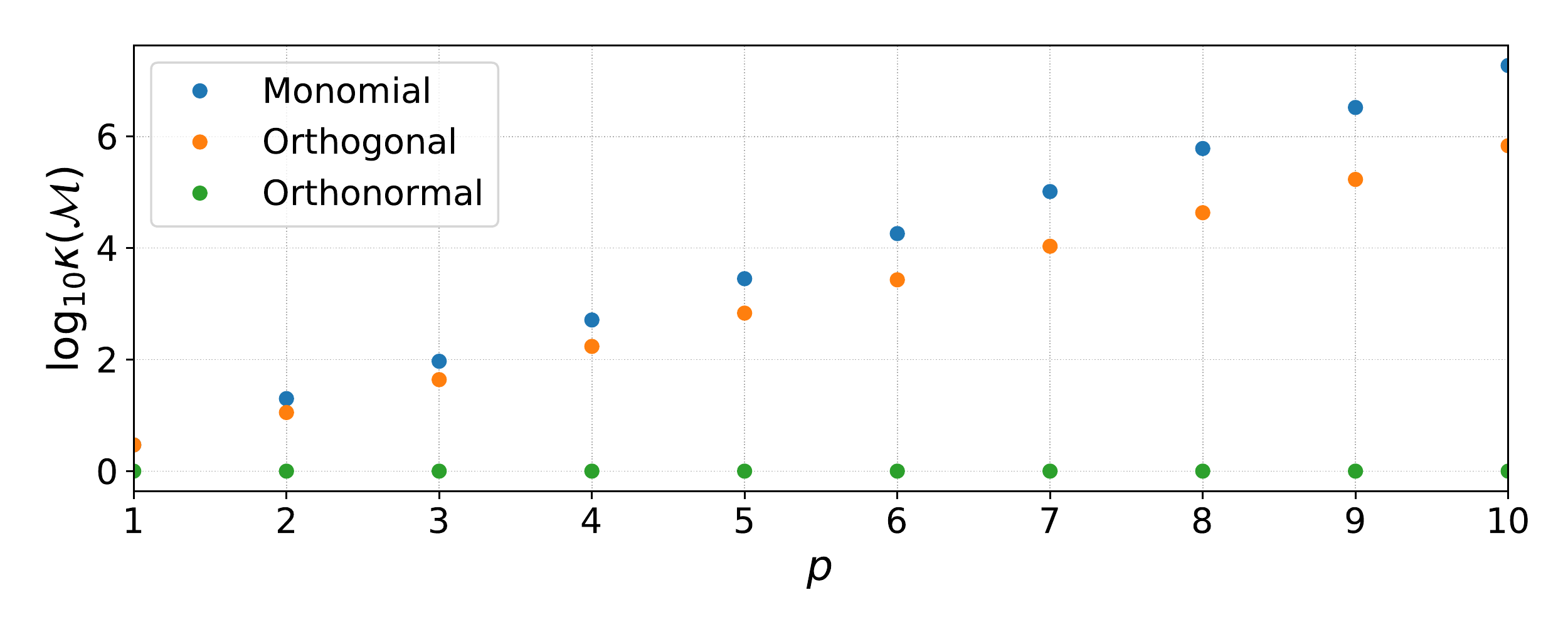}
  \caption[Condition number for mass matrices]{Logarithm of the
    condition number for mass matrices of the monomial basis,
    $\widehat{\psi}_k(\eta) = \eta^k$, and the orthogonal basis
    constructed from it using GS orthogonalization. The logarithm
    provides an estimate of digits lost in solving the linear
    system.}
  \label{fig:model:cond_number}
\end{figure}

The conditioning can be improved by constructing an orthonormal
basis.  First we need to define the inner product,
\begin{align}
  \langle f,g\rangle = \int_{-1}^1 f(\eta)g(\eta) \,d\eta,
\end{align}
and then we can use the Gram-Schmidt orthogonalization to construct an
orthogonal basis. By orthogonalization of the monomials, we
obtain:\footnote{Note that these polynomials are similar to the
  Legendre polynomials, which are typically defined using the
  recursion, $$P_0(\eta) =1,\quad P_1(\eta)=\eta,$$$$(n+1)P_{n+1}(\eta) =
  (2n+1)\eta P_n(\eta)-nP_{n-1}(\eta).$$ However, there is a difference in
  the normalization -- the Legendre polynomials are normalized so they
  are equal to $\pm1$ at the edges of $[-1,1]$, which is
  not the case for the polynomials obtained simply using the GS
  orthogonalization.}
\begin{align*}
  \widehat{\psi}_0^{OG} (\eta) &= 1 \\
  \widehat{\psi}_1^{OG} (\eta) &= \eta \\
  \widehat{\psi}_2^{OG} (\eta) &= \frac{3\eta^2-1}{3} \\
  \widehat{\psi}_3^{OG} (\eta) &= \frac{\eta(5\eta^2-3)}{5} \\
  \widehat{\psi}_4^{OG} (\eta) &= \frac{35\eta^4-30\eta^2+3}{35}
\end{align*}

\fgr{model:cond_number} shows the significant
improvement in the condition number. However, we can go one step
further and construct the orthonormal basis using
\begin{align}
  \widehat{\psi}_k^{ON}(\eta) = \frac{\widehat{\psi}_k^{OG}}{\sqrt{\langle
      \widehat{\psi}_k^{OG}, \widehat{\psi}_k^{OG}\rangle}},
\end{align}
which gives:
\begin{align*}
  \widehat{\psi}_0^{ON} (\eta) &= \frac{1}{\sqrt{2}} \\
  \widehat{\psi}_1^{ON} (\eta) &= \frac{\sqrt{3}\eta}{\sqrt{2}} \\
  \widehat{\psi}_2^{ON} (\eta) &= \frac{\sqrt{5}(3\eta^2-1)}{\sqrt{2}^3} \\
  \widehat{\psi}_3^{ON} (\eta) &= \frac{\sqrt{7}\eta(5\eta^2-3)}{\sqrt{2}^3} \\
  \widehat{\psi}_4^{ON} (\eta) &= \frac{3(35\eta^4-30\eta^2+3)}{\sqrt{2}^7}
\end{align*}
Having the orthonormal basis by definition not only guarantees the
condition number of 1 but also allows for efficient pre-generation of
computational kernels.  The efficiency comes from the fact the the
matrices constructed with the orthonormal basis are generally sparse.
This is particularly important with the precomputed machine-generated
code discussed in \ser{model:dgck}.  As are the integrals
precalculated, expanded matrix multiplications can be limited only to
non-zero elements, which significantly decreases the computational
costs.

An interesting insight is obtained by plotting the basis on the
$[-1,1]$ interval (see
\fgr{model:basis}). While monomials start merging
for higher polynomial orders, the orthonormal basis is clearly more
linearly independent and provides ``better coverage'' of the space.
\begin{figure}[!htb]
  \centering
  \includegraphics[width=.9\textwidth]{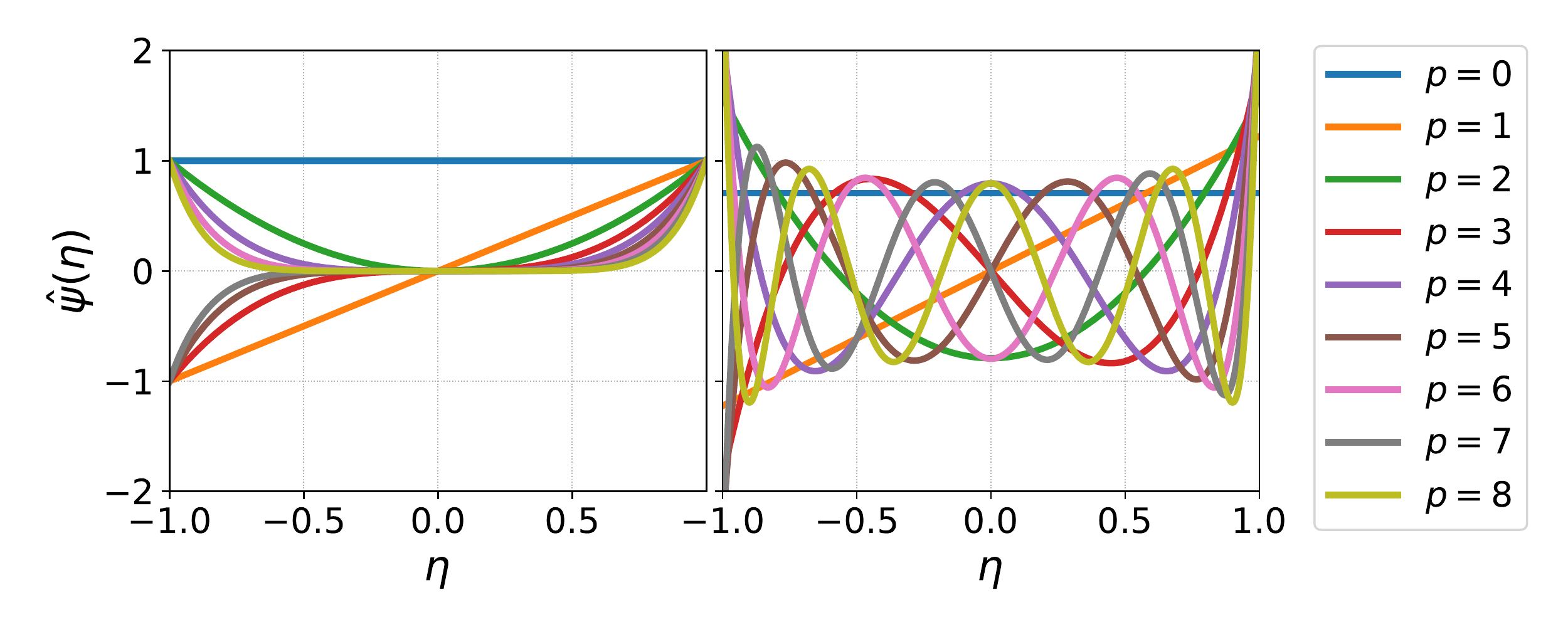}
  \caption[Comparison of the monomial and the orthonormal
    basis]{Comparison of the monomial basis (left) and the orthonormal
    basis constructed out of it using the GS orthonormalization
    (right). While monomials start merging for higher polynomial
    orders, the orthonormal basis is clearly more linearly independent
    and provides ``better coverage'' of the space.}
  \label{fig:model:basis}
\end{figure}

Finally, the generalization of the polynomial space $\mathcal{S}_h^p$
to higher dimensions needs to be discussed.  The standard approach on
a Cartesian mesh is to use the tensor product of the 1D polynomials.
For example for a 2D simulation with one configuration space
component, $\eta_x$, and one velocity component, $\eta_{v_x}$, we can
construct the space as
\begin{align}
  \mathcal{S}_h^p := \spn_{0<k,l<p} \{\eta_x^k \eta_{v_x}^l\}.
\end{align}
However, we need to keep in mind that the number of degrees of freedom
grows for the tensor product as $(p+1)^d$, where $d$ is the number of
dimensions (see Table\thinspace\ref{tab:model:tensor}).  This
exponential scaling is commonly known as the ``curse of
dimensionality''.  It is particularly problematic for the continuum
kinetic method, considering it requires discretizing up to six
dimension, and makes the tensor product difficult to use.
\begin{table}[!htb]
 \caption[Number of degree of freedom for the tensor product Lagrange
   polynomial space]{Number of degree of freedom for the tensor
   product Lagrange polynomial space, $(p+1)^d$, where $p$ is the
   polynomial order and $d$ is the number of dimensions.}
 \label{tab:model:tensor}
 \begin{center}
  \begin{tabular}{c|ccccccc}
    \toprule
    $d/p$ & 1 & 2 & 3 & 4 & 5 & 6 & 7\\
    \midrule
    2 & 4 & 9 & 16 & 25 & 36 & 49 & 64\\ 
    3 & 8 & 27 & 64 & 125  & 216 & 343 & 512\\
    4 & 16 & 81 & 256 & 625 & 1\,296 & 2\,401 & 4\,096\\
    5 & 32 & 343 & 1\,024 & 3\,125 & 7\,776 & 16\,807 & 32\,768\\
    6 & 64 & 729 & 4\,096 & 15\,625 & 46\,656 & 117\,649 & 262\,144 \\
    \bottomrule
  \end{tabular}
 \end{center}
\end{table}

For that reason, \texttt{Gkeyll} implements two reduced sets. The
first one is the Serendipity polynomial space, which is
constructed from the tensor product space by removing terms with
``super-linear'' order bigger than $p$.  The ``super-linear'' order of
a term is calculated by adding the order of each indeterminate bigger
than 1, i.e., the polynomial order of $\eta_x^3\eta_{v_x}\eta_{v_y}^2$
is 5. The 1X1V second order polynomial space can be then defined as
\begin{align*}
  \mathcal{S}_h^2 := \spn \{1,\, \eta_x,\, \eta_{v_x},\,
  \eta_x\eta_{v_x},\, \eta_x^2,\, \eta_{v_x}^2,\, \eta_x^2\eta_{v_x},
  \,\eta_x\eta_{v_x}^2,\, \hcancel{\eta_x^2\eta_{v_x}^2}\}.
\end{align*}
Still, nodal Serendipity retains the number of degrees of freedom at
the faces of each cell \citep{Arnold2011}.\footnote{Note that this is
  true only for the Cartesian grids (structured quadrilaterals).}  The
total number of degrees of freedom is
\begin{align}
  N_p = \sum_{i=0}^{\mathrm{min}(d,p/2)}
  2^{n-i}\binom{d}{i}\binom{p-i}{i},
\end{align}
which gives $N_p=8$ for $p=2$ and 1X1V \citep{Juno2017}. While this is
not a big difference in comparison to the nine degrees of freedom of
the tensor product, the scaling is much better for higher polynomial
orders and dimensions (see
Table\thinspace\ref{tab:model:serendipity}).
\begin{table}[!htb]
 \caption[Number of degree of freedom for the Serendipity polynomial
   space]{Number of degree of freedom for the Serendipity polynomial
   space, $\sum_{i=0}^{\mathrm{min}(d,p/2)}
   2^{n-i}\binom{d}{i}\binom{p-i}{i}$, where $p$ is the polynomial
   order and $d$ is the number of dimensions.}
 \label{tab:model:serendipity}
 \begin{center}
    \begin{tabular}{c|ccccccc}
      \toprule
      $d/p$ & 1 & 2 & 3 & 4 & 5 & 6 & 7\\
      \midrule
      2 & 4 & 8 & 12 & 17 & 23 & 30 & 38 \\
      3 & 8 & 20 & 32 & 50 & 74 & 105 & 144\\
      4 & 16 & 48 & 80 & 136 & 216 & 328 & 480 \\
      5 & 32 & 112 & 192 & 352 & 592 & 952 & 1\,472\\
      6 & 64 & 256 & 448 & 880 & 1\,552 & 2\,624 & 4\,256\\
      \bottomrule
    \end{tabular}  
 \end{center}
\end{table}

Apart from the Serendipity polynomial space, \texttt{Gkeyll 2.0} also
implements even less computationally expensive space -- maximal order
space. Again, for 1X1V:
\begin{align}
  \mathcal{S}_h^p = \spn_{0<k+l<p} \{\eta_x^k \eta_{v_x}^l\}.
\end{align}

After the polynomial space is chosen, the 1D process listed above can
be generalized to obtain an orthonormal basis set.  For the 1X1V
example, the inner product needs to be redefined,
\begin{align*}
  \langle f,g\rangle = \int_{-1}^1\int_{-1}^1
  f(\eta_x,\eta_{v_x})g(\eta_x,\eta_{v_x}) \,d\eta_x d\eta_{v_x}.
\end{align*}
Then, the $p=2$ 1X1V Serendipity basis\footnote{This particular
  basis set is used for many simulations through this work.}
implemented in \texttt{Gkeyll 2.0} can be calculated using the GS
orthogonalization and subsequent normalization,
\begin{align}\begin{aligned}\label{eq:model:1x1v}
  \widehat{\psi}_0(\eta_x,\eta_{v_x}) &= \frac{1}{2} \\
  \widehat{\psi}_1(\eta_x,\eta_{v_x}) &= \frac{\sqrt{3}\eta_x}{2} \\
  \widehat{\psi}_2(\eta_x,\eta_{v_x}) &= \frac{\sqrt{3}\eta_{v_x}}{2} \\
  \widehat{\psi}_3(\eta_x,\eta_{v_x}) &= \frac{3\eta_x\eta_{v_x}}{2} \\
  \widehat{\psi}_4(\eta_x,\eta_{v_x}) &= \frac{\sqrt{5}(3\eta_x^2-1)}{4} \\
  \widehat{\psi}_5(\eta_x,\eta_{v_x}) &= \frac{\sqrt{5}(3\eta_{v_x}^2-1)}{4} \\
  \widehat{\psi}_6(\eta_x,\eta_{v_x}) &= \frac{\sqrt{15}(3\eta_x^2-1)\eta_{v_x}}{4} \\
  \widehat{\psi}_7(\eta_x,\eta_{v_x}) &= \frac{\sqrt{15}\eta_x(3\eta_{v_x}^2-1)}{4} \\
\end{aligned}\end{align}

%---------------------------------------------------------------------
\subsection{Conservation Properties}\label{sec:model:conserve}

The conservation properties of the discretized Vlasov-Maxwell system
can now be assessed.  Here, only the propositions are listed with
brief comments and the reader is referred to \cite{Juno2017} for the
rigorous proofs.
\begin{prop}
  The Vlasov-Maxwell discrete scheme conserves the total number of
  particles.
\end{prop} 
\begin{proof}
  \eqr{model:weak2} holds for all the test functions, $\psi_t$.  The
  volume integral vanishes for special choice of $\psi_t$ and the
  surface integral is symmetric with the respect to the cell
  interface.
\end{proof}

\begin{prop}
  The phase-space incompressibility holds for the discrete system.
  \begin{align}
    \nabla_\mathrm{z} \cdot \bm{\alpha}_h = 0
  \end{align}
\end{prop}

\begin{prop}
  Electromagnetic energy is conserved exactly for central fluxes
  and bounded for upwind fluxes
  \begin{align}
    \sum_j\frac{d}{dt}
    \int_{\Omega_j}\left(\frac{\varepsilon_0}{2}|\bm{E}_h|^2 +
    \frac{1}{2\mu_0}|\bm{B}_h|^2\right) \,d\bm{x} \leq - \sum_i
    \int_{\Omega_j} \bm{j}_h\cdot \bm{E}_h \,d\bm{x}
  \end{align}
\end{prop}

\begin{prop}
  If $|\bm{v}|^2 \in \mathcal{S}_h^p$, Vlasov-Maxwell scheme conserves
  energy exactly for central fluxes,
  \begin{align}
    \frac{d}{dt}\sum_j\sum_s \int_{K^j}
    \frac{1}{2}m_s|\bm{v}|^2f_{h,s} \,d\bm{z} + 
    \frac{d}{dt}\sum_j \int_{\Omega_j} 
    \left(\frac{\varepsilon_0}{2}|\bm{E}_h|^2 +
    \frac{1}{2\mu_0}|\bm{B}_h|^2\right) \,d\bm{x} = 0.
  \end{align}
\end{prop}
In Remark 2, \citep{Juno2017} point out that at least piecewise
quadratic basis functions are required for $|\bm{v}|^2 \in
\mathcal{S}_h^p$.  However, they add that $|\bm{v}|^2$ can be
projected on linear basis set and the scheme will then conserve the
projected energy.

Note that the scheme does not conserve the momentum. However, the
error is often negligible \citep{Juno2017}.

\begin{prop}
  The scheme grows the discrete entropy monotonically, assuming $f_h$
  remains positive definite
  \begin{align}
    \sum_j\frac{d}{dt}\int_{K^j} - f_h \mathrm{ln}(f_h)\,d\bm{z} \geq
    0.
  \end{align}
\end{prop}

%---------------------------------------------------------------------
\FloatBarrier
\subsection{Time-stepping}

The discussion in the previous sections leads to the construction of
the governing equation in the following form \eqrp{model:weak2}
\begin{align*}
  \pfrac{f}{t} = \mathcal{L}(f,t),
\end{align*}
where $\mathcal{L}$ is the DG spatial discretization operator on the
right-hand-side of the equation.  In order to discretize this equation
in time and evolve the solution, \texttt{Gkeyll} uses the Strong
Stability Preserving (SSP) Runge-Kutta (RK) schemes
\citep{Shu2002,Durran2010}, which we describe in terms of the
first-order Euler update,\footnote{This description accurately
  captures the core structure of \texttt{Gkeyll 2.0} where the
  internal parts are written in terms of the single forward Euler
  steps and the outer control loops calls them with time steps and
  coefficients appropriate for each RK method.}
\begin{align*}
  \mathcal{F}(f,t) = f + \Delta t \mathcal{L}(f,t).
\end{align*}
The second order SSP-RK:
\begin{align}\begin{aligned}
  f^{(1)} &= \mathcal{F}\left(f^{n},t^{n}\right), \\ f^{n+1} &=
  \frac{1}{2}f^{n} + \frac{1}{2}\mathcal{F}\left(f^{(1)},t^n+\Delta
  t\right).
\end{aligned}\end{align}
The third order SSP-RK:
\begin{align}\begin{aligned}
  f^{(1)} &= \mathcal{F}\left(f^{n},t^{n}\right), \\ f^{(2)} &=
  \frac{3}{4}f^{n} + \frac{1}{4}\mathcal{F}\left(f^{(1)},t^n+\Delta
  t\right),\\ f^{n+1} &= \frac{1}{3}f^{n} +
  \frac{2}{3}\mathcal{F}\left(f^{(2)},t^n+\Delta t/2\right).
\end{aligned}\end{align}
The four stage third order SSP-RK:
\begin{align}\begin{aligned}
  f^{(1)} &= \mathcal{F}\left(f^{n},t^{n}\right), \\ f^{(2)} &=
  \frac{1}{2}f^{(1)} + \frac{1}{2}\mathcal{F}\left(f^{(1)},t^n+\Delta
  t/2\right),\\ f^{(3)} &= \frac{2}{3}f^{n} + \frac{1}{6}f^{(2)} +
  \frac{1}{6}\mathcal{F}\left(f^{(2)},t^n+\Delta t\right),\\ f^{n+1}
  &= \frac{1}{2}f^{(3)} +
  \frac{1}{2}\mathcal{F}\left(f^{(3)},t^n+\Delta t/2\right).
\end{aligned}\end{align}

The difference in between the three stage and four stage RK3 is in
their stability condition know as the Courant-Friedrich-Lewy (CFL)
condition.  \cite{Juno2018} defines the condition using so called CFL
frequency,
\begin{align*}
    \omega_i = \frac{\alpha_i}{\Delta z_i}, \quad i = 1, \ldots, d.
\end{align*}
Note that the ``velocity" has the physical meaning of velocity only
for the first three phase space dimensions (configuration space) and,
for the Vlasov-Maxwell problems, it is always the speed of light. The
other three ``velocities'' correspond to the acceleration in the
velocity space caused by the Lorentz force \eqrp{model:motion2}.
With the CFL frequency, we can write the CFL condition as
\begin{align}
    \Delta t(2p+1) \sum_{i=1}^{d}\omega_i \leq CFL, 
\end{align}
where the $CFL$ on the right-hand-side represents the so-called CFL
number.  It is 1.0 for the three stage RK3 and 2.0 for the four stage
RK3.  In other words, the four stage method is stable for twice as big
time steps in comparison to the three stage RK3.  Therefore, even
though the four stage method requires $33\%$ more work, it allows for
roughly $1.5\times$ increase in the computation speed.

Finally, it is worth noting that in practice, it is common to add a
safety margin to the CFL numbers and run the three stage RK3 with the
CFL number of 0.9 and four stage with the CFL number of 1.8.  In most
of the simulations in this work, the four stage RK3 is used.

%---------------------------------------------------------------------
\subsection{Moment Calculation}
\label{sec:model:moments}

As a final point of the discrete Vlasov-Maxwell section, the numerical
integration of the moment of the distribution function needs to be
addressed.  Since either the charge density, $\rho_c = q_s \int_{\mathcal{V}}
f_sd\bm{v}$, or the current density, $\bm{j} = q_s \int
\bm{v}f_sd\bm{v}$, are required as the sources of the Poisson's
\eqrp{model:poisson} or Amp\`{e}re's \eqrp{model:ampere}
equations respectively, the moments are required each RK step.
Therefore, they need to be evaluated both efficiently and precisely.

Starting with the number density, \eqr{model:m0} needs to be
rewritten in the discrete weak sense,
\begin{align*}
  n_h^i(t,\bm{x}) \circeq \sum_j \int_{\mathcal{V}_{ij}}
  f_h^{ij}(t,\bm{x},\bm{v}) \,d\bm{v},
\end{align*}
where the integral is performed over the velocity space of cell $ij$;
the general cell index is split into the configuration space index,
$i$ and the velocity space index, $j$.  Expanding the density and the
distribution function together with the definition of the weak
equality leads to
\begin{align*}
  \widehat{n}_n^i(t) \int_{\Omega^i} \varphi_n(\bm{x})\varphi_t(\bm{x}) \,d\bm{x}
  = \sum_j \widehat{f}_m^{ij}(t) \int_{K^{ij}} \psi_m(\bm{x},\bm{v})
  \varphi_t(\bm{x}) \,d\bm{x}d\bm{v}.
\end{align*}
Note that the number density, $n$, is expanded in terms of the same
basis functions as $\bm{E}$ and $\bm{B}$. As the next step, the
integral is converted to the logical space, similar to
\ser{model:discvlasov},
\begin{align*}
  \widehat{n}_n^i(t) \int_{I_x} \widehat{\varphi}_n(\bm{\eta}_{\bm{x}})
  \widehat{\varphi}_t(\bm{\eta}_{\bm{x}})
  \,\left|\frac{d\bm{x}^i}{d\bm{\eta}_{\bm{x}}}\right|
  d\bm{\eta}_{\bm{x}} = \sum_j \widehat{f}_m^{ij}(t) \int_{I_p}
  \widehat{\psi}_m(\bm{\eta}_{\bm{x}},\bm{\eta}_{\bm{v}})
  \widehat{\varphi}_t(\bm{\eta}_{\bm{x}})
  \,\left|\frac{d\bm{x}^i}{d\bm{\eta}_{\bm{x}}}\right|
  \left|\frac{d\bm{v}^j}{d\bm{\eta}_{\bm{v}}}\right|
  d\bm{\eta}_{\bm{x}}d\bm{\eta}_{\bm{v}},
\end{align*}
where we need to distinguish between the configuration space unit cube
$I_x$ and the phase space hyper cube $I_p$.  Assuming Cartesian mesh
and orthonormal basis function set, $\varphi$, defined in
\ser{model:choice}, the expression simplifies into
\begin{align}
  \widehat{n}_n^i(t) = \frac{\prod_{i=1}^{d_v}\Delta v_i}{2^{d_v}} \sum_j
  \widehat{\mathcal{A}}_{mn}^0 \widehat{f}_m^{ij}(t),
\end{align}
where $\widehat{\mathcal{A}}_{mt}^0$ is so-called zeroth order matrix,
\begin{align*}
  \widehat{\mathcal{A}}_{mn}^0 = \int_{I_p} 
  \widehat{\psi}_m(\bm{\eta}_{\bm{x}},\bm{\eta}_{\bm{v}})
  \widehat{\varphi}_n(\bm{\eta}_{\bm{x}})
  \,d\bm{\eta}_{\bm{x}}d\bm{\eta}_{\bm{v}}.
\end{align*}
Note that $\widehat{\mathcal{A}}_{mt}^0$ can be easily pre-computed.

The situation for the first moment is analogous to the volume term in
\ser{model:discvlasov},
\begin{align*}
  (n\bm{u})_h^i(t,\bm{x}) \circeq \sum_j \int_{\mathcal{V}^{ij}}
  \bm{v} f_h^{ij}(t,\bm{x},\bm{v}) \,d\bm{v},
\end{align*}
where $\bm{v}$ can be again replaced by $\bm{v} + \bm{v}_c^j - \bm{v}_c^j$,
\begin{align*}
  (n\bm{u})_h^i(t,\bm{x}) &\circeq \sum_j \int_{\mathcal{V}^{ij}}
  \left( \bm{v} + \bm{v}_{c}^j-\bm{v}_{c}^j\right)
  f_h^{ij}(t,\bm{x},\bm{v}) \,d\bm{v}, \\ &\circeq \sum_j
  \int_{\mathcal{V}^{ij}} \left(\bm{v} -\bm{v}_{c}^j\right)
  f_h^{ij}(t,\bm{x},\bm{v}) \,d\bm{v} + \sum_j
  \bm{v}_{c}^j\int_{K_{ij}} f_h^{ij}(t,\bm{x},\bm{v}) \,d\bm{v}.
\end{align*}
The first moment can be then expressed as
\begin{align}
  \widehat{n\bm{u}}_n^i(t) = \frac{\prod_{i=1}^{d_v}\Delta v_i}{2^{d_v}}
  \left(\frac{\Delta \bm{v}}{2} \sum_j \widehat{\bm{\mathcal{A}}}_{mn}^1
  \widehat{f}_m^{ij}(t) + \sum_j \bm{v}_{c}^j
  \widehat{\mathcal{A}}_{mn}^0 \widehat{f}_m^{ij}(t) \right),
\end{align}
where there is point-wise product between $\Delta\bm{v}$ and the
tensor $\widehat{\bm{\mathcal{A}}}_{mn}^1$,
\begin{align*}
  \widehat{\bm{\mathcal{A}}}_{mn}^1 = \int_{I_p} \bm{\eta}_{\bm{v}}
  \widehat{\psi}_m(\bm{\eta}_{\bm{x}},\bm{\eta}_{\bm{v}})
  \widehat{\varphi}_n(\bm{\eta}_{\bm{x}})
  \,d\bm{\eta}_{\bm{x}}d\bm{\eta}_{\bm{v}}.
\end{align*}

%=====================================================================
\section{Five-moment Two-fluid Model}\label{sec:model:fluids}

Even though the fluid simulations are not the focus of this work, they
are at times used for comparison. Their derivation also nicely rounds
up the discussion about the distribution function and the Vlasov
equation \eqrp{model:vlasov}.

In Section\thinspace\ref{sec:model:distf}, it was shown that taking
the moments of the distribution function leads to the macroscopic
conserved quantities.  Similarly, taking the moments of the Vlasov
equation gives the fluid conservation equations.

Starting with the zeroth moment,
\begin{align*}
  \int_{\mathcal{V}} \Big(
  \underbracket{\pfrac{f_s}{t}}_{I} +
  \underbracket{\bm{v}\cdot\nabla_{\bm{x}}f_s}_{II} +
  \underbracket{\frac{q_s}{m_s}\left(\bm{E} +
  \bm{v}\times\bm{B}\right)\cdot\nabla_{\bm{v}}f_s}_{III} 
  \Big) \,d\bm{v} = 0,
\end{align*}
we can evaluate the three terms individually.

Since the distribution function is assumed to be continuous in both
velocity space and time, the order of integration and time derivation
in the first term can be switched,
\begin{align*}
  \int_{\mathcal{V}} \pfrac{f_s}{t} \,d\bm{v} = \pfraca{t}
  \int_{\mathcal{V}} f_s \,d\bm{v} = \pfrac{n_s}{t}.
\end{align*}
The velocity and position are treated as independent variables,
therefore, $v$ can be moved into the differential operator,
\begin{align*}
  \int_{\mathcal{V}} \bm{v}\cdot\nabla_{\bm{x}}f_s \,d\bm{v} =
  \int_{\mathcal{V}} \nabla_{\bm{x}} \cdot(\bm{v}f_s) \,d\bm{v}
  =\nabla_{\bm{x}} \cdot \int_{\mathcal{V}} \bm{v}f_s \,d\bm{v} =
  \nabla_{\bm{x}} \cdot n_s\bm{u}_s.
\end{align*}
Finally, as it was discussed in \ser{model:discvlasov}, the Lorentz
force can be included in the differential operator as well,
\begin{align*}
  \int_{\mathcal{V}} \left(\bm{E} +
  \bm{v}\times\bm{B}\right)\cdot\nabla_{\bm{v}}f_s \,d\bm{v} &= \int_{\mathcal{V}}
  \nabla_{\bm{v}} \cdot \left(\bm{E} + \bm{v}\times\bm{B}\right) f_s
  \,d\bm{v} \\&= \int_{\partial \mathcal{V}} \left(\bm{E} +
  \bm{v}\times\bm{B}\right) f_s \cdot d\bm{A}
  \stackrel{\lim_{v\rightarrow\infty}f_s=0}{=}0.
\end{align*}

When the terms are put together, the zeroth moment of the Vlasov
equation \eqrp{model:vlasov} leads to the well known continuity
equation,
\begin{align}\label{eq:model:continuity}
  \pfrac{n_s}{t} + \nabla_{\bm{x}} \cdot n_s\bm{u}_s = 0.
\end{align}

Taking the first moment of the Vlasov equation gives
\begin{align*}
  \int_{\mathcal{V}} \Big(
  \underbracket{\bm{v}\pfrac{f_s}{t}}_{I} +
  \underbracket{\bm{v}\left(\bm{v}\cdot\nabla_{\bm{x}}f_s\right)}_{II} +
  \underbracket{\bm{v}\frac{q_s}{m_s}\left(\bm{E} +
  \bm{v}\times\bm{B}\right)\cdot\nabla_{\bm{v}}f_s}_{III} 
  \Big) \,d\bm{v} = 0
\end{align*}

Similar to the zeroth moment, the first term can be evaluated as
\begin{align*}
  \int_{\mathcal{V}} \bm{v}\pfrac{f_s}{t} \,d\bm{v} = \pfraca{t}
  \int_{\mathcal{V}} \bm{v}f_s \,d\bm{v} = \pfrac{n_s\bm{u}_s}{t}.
\end{align*}
The second term gives
\begin{align*}
  \int_{\mathcal{V}} \bm{v}\left(\bm{v}\cdot\nabla_{\bm{x}}f_s\right)
  \,d\bm{v} = \int_{\mathcal{V}} \bm{v}\nabla_{\bm{x}}\cdot
  (\bm{v}f_s) \,d\bm{v} = \nabla_{\bm{x}}\cdot\int_{\mathcal{V}}
  (\bm{v}\bm{v}f_s) \,d\bm{v} = \nabla_{\bm{x}}\cdot (n_s
  \langle\bm{v}\bm{v}\rangle),
\end{align*}
where the $\langle\cdot\rangle$ denotes the average value and
$\bm{v}\bm{v}$ is a dyadic tensor.  Now we can split the velocity into
the bulk velocity, $\bm{u}$, and the thermal component, $\bm{v}_{th} =
\bm{v}-\bm{u}$. Note that $\langle\bm{v}\rangle =
\langle\bm{u}+\bm{v}_{th}\rangle =
\langle\bm{u}\rangle+\langle\bm{v}_{th}\rangle = \bm{u}$. Analogously,
\begin{align*}
  \nabla_{\bm{x}}\cdot (n_s \langle\bm{v}\bm{v}\rangle) =
  \nabla_{\bm{x}}\cdot (n_s \bm{u}_s\bm{u}_s) +
  \underbrace{2\nabla_{\bm{x}}\cdot (n_s
    \bm{u}_s\langle\bm{v}_{th,s}\rangle)}_{0} + \nabla_{\bm{x}}\cdot
  (n_s \langle\bm{v}_{th,s}\bm{v}_{th,s}\rangle).
\end{align*}
After multiplying with the mass, the last part of the expression can
be identified as the pressure tensor, $\bm{P}= m_sn_s
\langle\bm{v}_{th,s}\bm{v}_{th,s}\rangle$.  The third term requires us
to use the vector identity $\nabla\cdot(\bm{AB}) =
(\nabla\cdot\bm{A})\bm{B} + (\bm{A}\cdot\nabla)\bm{B}$,
\begin{align*}
  \int_{\mathcal{V}} \bm{v}\frac{q_s}{m_s}\left(\bm{E} +
  \bm{v}\times\bm{B}\right)\cdot\nabla_{\bm{v}}f_s \,d\bm{v} =&
  \int_{\mathcal{V}}
  \bm{v}\frac{q_s}{m_s}\nabla_{\bm{v}}\cdot\left[\left(\bm{E} +
    \bm{v}\times\bm{B}\right)f_s \right] \,d\bm{v}\\ =&
  \int_{\mathcal{V}} \nabla_{\bm{v}}\cdot
  \left[\bm{v}\frac{q_s}{m_s}\left(\bm{E} + \bm{v}\times\bm{B}\right)
    f_s \right] \,d\bm{v} \\&- \int_{\mathcal{V}}
  \left[\frac{q_s}{m_s}\left(\bm{E} + \bm{v}\times\bm{B}\right) f_s
    \cdot\nabla_{\bm{v}} \right]\bm{v} \,d\bm{v} \\ =& \int_{\partial
    \mathcal{V}} \left[\bm{v}\frac{q_s}{m_s}\left(\bm{E} +
    \bm{v}\times\bm{B}\right) f_s \right] \cdot d\bm{A} \\ &-
  \int_{\mathcal{V}} \frac{q_s}{m_s}\left(\bm{E} +
  \bm{v}\times\bm{B}\right) f_s \,d\bm{v} \\ =& -
  \frac{q_sn_s}{m_s}\left(\bm{E} + \bm{u}_s\times\bm{B}\right)
\end{align*}

When everything is put together and multiplied by the mass, we obtain
the law of conservation of momentum,
\begin{align}\label{eq:model:momentum}
  m_s\pfrac{n_s\bm{u}_s}{t} +
  \nabla_{\bm{x}}\cdot\left(m_sn_s\bm{u}_s\bm{u}_s + \bm{P}\right) =
  q_sn_s\left(\bm{E} + \bm{u}_s\times\bm{B}\right).
\end{align}

Careful examination of the first two conservation equations leads to
an interesting observation -- evolution equation for a moment of the
distribution function depends on a higher moment. E.g., the evolution
of density \eqrp{model:continuity} depends on the flux and the
evolution of flux depends on the energy \eqrp{model:m2}, etc.
Therefore, in order to be useful, the system of the equations needs to
be truncated, which introduces additional approximation into the
system.  On the other hand, the fact that the equations are 3D instead
of 6D, makes solving the system significantly less expensive (even
though it consists of more equations).

The fluid model used for comparisons to kinetics in this work, the
two-fluid five-moment model,\footnote{The electron and ion equations
  are solved separately rather than merged into a single fluid as it
  is in the case of Magneto-hydrodynamic (MHD) models. The name ``five
  moment'' is in my opinion a bit misleading.  It refers to the five
  equations -- conservation of density, three equations for the
  conservation of momentum, and the conservation of energy -- and not
  to taking five moments.}  is described by \cite{Hakim2006}.
Assuming no heat flow and a scalar fluid pressure, the model is
defined as
\begin{align}\begin{aligned}
    \pfrac{n}{t} + \pfraca{x_j}\left(n u_j\right) &= 0,
    \\ m\pfrac{nu_k}{t} + \pfraca{x_j}\left(p\delta_{kj} +
    mnu_ku_j\right) &=
    nq\left(E_k+\varepsilon_{kij}u_iB_j\right),\\
    \pfrac{\mathcal{E}}{t}
    + \pfraca{x_j}\left(u_jp + u_j \mathcal{E}\right) &= qnu_jE_j,
\end{aligned}\end{align}
where $\mathcal{E} = \frac{p}{\gamma-1} + \frac{1}{2}mnu_iu_i$ and
$\gamma = \frac{5}{3}$.

%% file: benchmark.tex
\chapter{Benchmarks \& Plasma Instabilities}

\epigraph{The first principle is that you must not fool yourself and
  you are the easiest person to fool.}{\textit{Richard Feynman}}

Even though it is quite common to apply a new model directly onto the
problem of interest, it is essential to rigorously test it first.  In
the ideal case, testing is done by comparison to exact
solutions during code verification and later by validating the
simulation with experimental results \citep{Oberkampf2010}.  However,
in plasma physics, both exact analytical solutions and suitable
experimental results are rare.  Therefore, it is quite common to test
the simulation on a set of benchmark problems.  These are, typically,
simple and well understood plasma waves and instabilities.  In this
chapter, we take a look at the Landau damping of the Langmuir waves
and the two-stream instability.

%=====================================================================
\FloatBarrier
\section{Landau Damping}\label{sec:benchmark:landau}

Landau damping is responsible for collisionless wave energy
dissipation in plasmas.  During the process, electromagnetic wave
interacts with the particles, altering their velocity distribution.
Landau damping is, therefore, an intrinsically kinetic process,
which makes a good benchmark for any kinetic code.

Propagating electromagnetic wave interacts with particles in plasma by
accelerating the ones moving at lower velocities than its phase
velocity, $v_{ph}$, and slowing the faster ones; wave loses energy
from the first interaction and gains energy from the later.  Since the
particles in equilibrium plasma follow the Maxwellian distribution
(see Sec.\thinspace\ref{sec:model:distf}),
\begin{align*}
  f_M(v_x) = \frac{1}{\sqrt{2\pi v_{th}^2}}
  \exp\left(-\frac{v_x^2}{2v_{th}^2}\right),
\end{align*}
there are always more particles with lower velocity (in the absolute
value sense; see \fgr{benchmark:landau_sketch}).  Therefore, the number of
particles getting accelerated is always higher than the number of
particles getting slowed and the total electromagnetic energy of the
wave is decreasing.  This results in ``flattening'' of the
distribution function around the wave phase velocity.
\begin{figure}[!htb]
  \centering
  \includegraphics[width=0.8\linewidth]{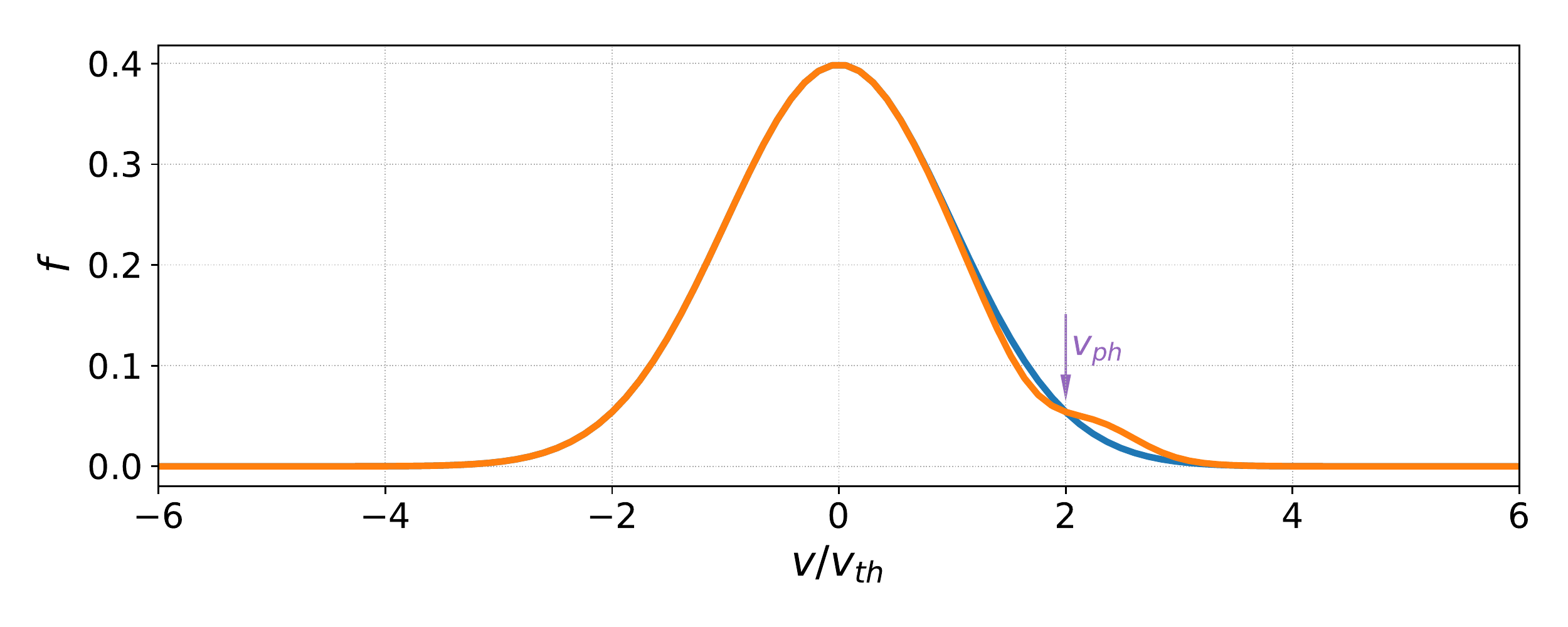}
  \caption[Ilustration of the Landau damping]{Ilustration of the
    Landau damping. A wave with phase velocity $v_{ph}=2v_{th}$ is
    interacting with an equilibrium particle distribution (blue) which
    results in the ``flattening'' of the distribution around the phase
    velocity (orange).}
  \label{fig:benchmark:landau_sketch}
\end{figure}

%---------------------------------------------------------------------
\subsection{Linear Theory}\label{sec:benchmark:landau:lintheory}

The Vlasov equation \eqrp{model:vlasov} has no exact analytical
solution, which makes it difficult to compare simulation results with
theory.  Therefore, it is usually simplified using the linear theory.
The key idea of the linear theory is to expand all variables into the
equilibrium terms and perturbations, neglecting higher order terms
(HOT),
\begin{align*}
    f = f_0 + f_1 + \text{HOT}.
\end{align*}

This approach can be then used to show the Landau damping
quantitatively in the dispersion relation of 1D electrostatic Langmuir
waves.  In this case, electrons are govern by the 1D Vlasov equation,
\begin{align}\label{eq:benchmark:elc_vlasov1D}
  \pfrac{f_e}{t} + v_x\partial_x f_e + \frac{q_e}{m_e}E_x\partial_{v_x}
  f_e = 0,
\end{align}
and ions are assumed stationary. Noting that the equilibrium part is
not a function of time and position, the equation simplifies
to\footnote{$E_{x,1}\partial_{v_x} f_1$ is a second order term and is,
  therefore, neglected.}
\begin{align*}
  \pfrac{f_{e,1}}{t} + v_x\partial_x f_{e,1} +
  \frac{q_e}{m_e}E_{x,1}\partial_{v_x} f_{e,0} = 0.
\end{align*}
Using the Fourier series, all the perturbations are assumed to be
in the form of $f_1\propto \exp\big(i(kx-\omega t)\big)$, where $k$ is a
wave-number and $\omega$ is frequency, which allows to replace the
derivatives with algebraic terms,
\begin{align*}
  -i\omega f_{e,1} + ikv_x f_{e,1} + \frac{q_e}{m_e}E_x \partial_{v_x}
  f_{e,0} = 0 \quad \Rightarrow \quad f_{e,1} =
  -i\frac{q_eE_{x,1}}{m_e}\frac{\partial_{v_x} f_{e,0}}{\omega -
    kv_x}.
\end{align*}

This distribution function perturbation is then substituted into the
linearized Poisson's equation
\eqrp{model:poisson},\footnote{Written in terms of the electric
  field.}
\begin{align}\begin{aligned}\label{eq:benchmarking:langmuir_disp0}
  ik\varepsilon_0E_{x,1} &= q_in_{i,0} + q_en_{e,0} + q_en_{e,1},
  \\ &= q_e\int_{\mathcal{V}} f_{e,1} \,d\bm{v}, \\ &=
  -i\frac{q_e^2E_{x,1}}{m_e} \int_{\mathcal{V}}\frac{\partial_{v_x}
    f_{e,0}}{\omega - kv_x} \,d\bm{v},
\end{aligned}\end{align}
where we assume that the ions are stationary and not perturbed,
$n_{i,1}=0$, the equilibrium plasma is quasi-neutral, $q_in_{i,0} +
q_en_{e,0} = 0$, and the perturbation of the density is the zeroth
moment of the perturbation of the distribution function
\eqrp{model:m0}, $n_{e,1}=\int_{\mathcal{V}} f_{e,1} \,d\bm{v}$.  Since only
electrons are evolved, the electron subscript, $e$, will be omitted
from now for clarity.

\eqr{benchmarking:langmuir_disp0} represents a form of a dispersion
relation.  Initial conditions of the system can be defined with $f_0$
and then the only remaining unknowns are $k$ and
$\omega$.\footnote{The electric field cancels out.}  In other words,
the dispersion relation is an equation describing which modes of an
initial perturbation, $f_1$, satisfy the plasma equations in a system
defined by initial conditions $f_0$.  However,
\eqr{benchmarking:langmuir_disp0} needs to be further simplified for
the practical use.

Assuming the distribution function can be factorized, the 3D integral
can be rewritten,
\begin{align*}
  \int\frac{\partial_{v_x} f_{0}}{\omega - kv_x} \,d\bm{v} = n
  \int_{-\infty}^\infty \frac{\partial_{v_x} f_{x,0}}{\omega - kv_x}
  \,dv_x \underbracket{\int_{-\infty}^\infty f_{y,0}\,dv_y}_{=1}
  \underbracket{\int_{-\infty}^\infty f_{z,0}\,dv_z}_{=1}.
\end{align*}

The dispersion relation \eqrp{benchmarking:langmuir_disp0} then
simplifies to\footnote{The sign flip is due to
  the reversed order in the denominator.}
\begin{align}\label{eq:benchmark:langmuir_disp}
  1 = \frac{1}{k^2}
  \underbracket{\frac{nq^2}{m\varepsilon_0}}_{\omega_{pe}^2}
  \int_{-\infty}^{\infty} \frac{\partial_{v_x}f_{x,0}}{v_x -
    \omega/k}\,dv_x,
\end{align}
where $\omega_{pe}$ is the plasma oscillation frequency.

However, the integral in \eqr{benchmark:langmuir_disp} is not
straightforward to evaluate analytically.  Even thought $v_x$ is real
and $\omega$ is, typically, imaginary,\footnote{Assuming $\omega =
  \omega_r + i\gamma$, all variables are $\propto \exp\big(i(kx -
  \omega t)\big) = \exp\big(i(kx - \omega_r t)\big)\exp\big(\gamma
  t\big)$. Therefore, the real part of the frequency describes the
  oscillatory behavior while the imaginary component corresponds to
  either damping or growth.} the singularity $v_x=\omega/k$ affects
the solution.  Landau was the first to point out this integral needs
to be treated as a contour integral in the complex $v_x$ plane
\citep{Dawson1961}.  A standard approach is to use the residue
theorem,
\begin{align*}
  \int_{C_1} \frac{\partial_{v_x}f_{x,0}}{v_x -
    \omega/k}dv_x + \int_{C_2}
  \frac{\partial_{v_x}f_{x,0}}{v_x - \omega/k} \,dv_x = 2\pi
  iR(\omega/k),
\end{align*}
where $C_1$ and $C_2$ are integration curves depicted in
Fig.\thinspace\ref{fig:benchmark:contourint}, and $R(\omega/k)$ is
the residuum.  However, this cannot be applied because
\begin{align*}
  \lim_{v_x \to\pm i\infty} \exp\left(-\frac{v_x^2}{2v_{th}^2}\right) \neq 0
\end{align*}
and the integral over $C_2$ does not vanish.  The only options are the
numerical integration, which will be discussed at the end of this section, or tables for the
Maxwellian distribution, for example \cite{Fried1961}.
\begin{figure}[!htb]
  \centering
  \includegraphics[width=0.7\linewidth]{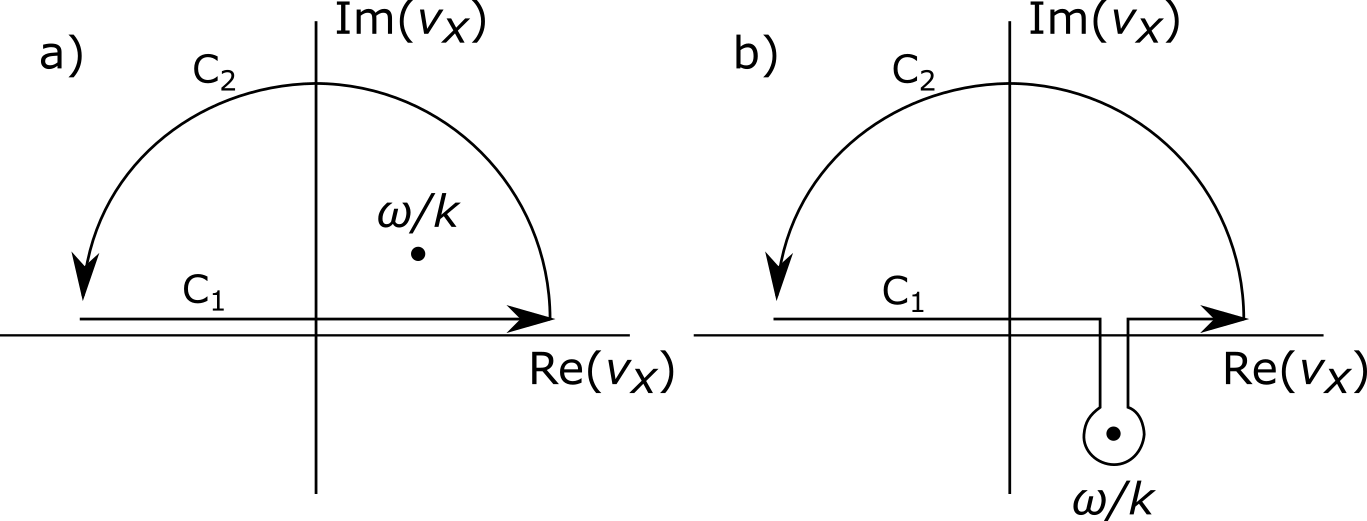}
  \caption[Integration curves for contour integrals]{Integration curves
    for the contour integrals with a singularity.  Classical approach for
    $\mathrm{Im}(\omega)>0$ a) and  $\mathrm{Im}(\omega)<0$ b).}
  \label{fig:benchmark:contourint}
\end{figure}

An approximate solution can be found for a special case of weak
growth/damping and a wave phase velocity much bigger than the thermal
velocity of the distribution, $v_{ph} \gg v_{th}$.
\eqr{benchmark:langmuir_disp} then becomes
\begin{align}\label{eq:benchmark:langmuir_disp2}
  1 = \frac{\omega_{pe}^2}{k^2}\left( PV\int_{-\infty}^\infty
  \frac{\partial_{v_x}f_{x,0}}{v_x - \omega/k} \,dv_x + i\pi
  \left.\partial_{v_x}f_{x,0}\right|_{v_x=\omega/k}\right),
\end{align}
where $PV$ stands for the Cauchy principal value.  Since $v_{ph} \gg
v_{th}$, both $f_{x,0}(v_{ph})$ and
$\partial_{v_x}f_{x,0}|_{v_{ph}}$ are small,
\begin{align*}
   PV\int_{-\infty}^\infty
   \frac{\partial_{v_x}f_{x,0}}{v_x-\omega/k}\,dv_x &\approx
   \int_{-\infty}^\infty
   \frac{\partial_{v_x}f_{x,0}}{v_x-v_{ph}}\,dv_x \\&=
   \left[\frac{f_{x,0}}{v_x-v_{ph}}\right]_{-\infty}^\infty -
   \int_{-\infty}^\infty \frac{-f_{x,0}}{(v_x-v_{ph})^2} \,dv_x \\ &=
   \int_{-\infty}^\infty \frac{f_{x,0}}{(v_x-v_{ph})^2} \,dv_x,
\end{align*}
where the integration \textit{per partes} is performed.  Noting the
definition of an average $\langle g(x) \rangle = \int f(x) g(x) dx$,
the real part of the dispersion relation
\eqr{benchmark:langmuir_disp2} becomes
\begin{align*}
  1 &= \frac{\omega_{pe}^2}{k^2}\left\langle (v_x-v_{ph})^{-2}
  \right\rangle = \frac{\omega_{pe}^2}{k^2}\left\langle v_{ph}^{-2}
  \left(1-\frac{v_x}{v_{ph}}\right)^{-2} \right\rangle =
  \frac{\omega_{pe}^2}{k^2}\left\langle v_{ph}^{-2}
  \left(1+\frac{2v_x}{v_{ph}}+\frac{3v_x^2}{v_{ph}^2}+\ldots\right)
  \right\rangle\\ &\approx \frac{\omega_{pe}^2}{k^2} v_{ph}^{-2}
  \left(1+\frac{3\left\langle v_x^2\right\rangle}{v_{ph}^2}\right)
\end{align*}
Using the Equipartition theorem, $\frac{1}{2}m_e\left\langle
v_x^2\right\rangle = \frac{1}{2}\mathrm{k}_BT_e$,
\begin{align}\label{eq:benchmark:langmuir}
  \omega_r^2 = \omega_{pe}^2 +
  \frac{\omega_{pe}^2}{\omega_r^2}\frac{3\mathrm{k}_BT_e}{m_e}k^2,
\end{align}
which is consistent with the fluid theory results of the electrostatic
electron waves \citep{Chen1985}.

If the thermal correction is small, i.e., $\omega_r \approx
\omega_{ph}$, real and imaginary terms can be simply combined,
\begin{align}
  \omega = \omega_{pe}\left(1 +
  i\frac{\pi}{2}\frac{\omega_{pe}^2}{k^2}
  \left[\pfrac{f_{x,0}}{v_x}\right]_{v_x=v_{ph}}\right).
\end{align}
Evaluating the derivative term gives for the imaginary part
\begin{align}
  \mathrm{Im}\left(\frac{\omega}{\omega_{pe}}\right) =
  \frac{\gamma}{\omega_{pe}} = -0.22\sqrt{\pi}
  \left(\frac{\omega_{pe}}{kv_{th}}\right)^3
  \exp\left(-\frac{1}{2k^2\lambda_D^2}\right).
\end{align}

As mentioned above, if the condition $v_{ph} \gg v_{th}$ is not
satisfied, the numerical integration is required.  First, it is useful
to rewrite the dispersion relation \eqrp{benchmark:langmuir_disp} in
terms of the plasma dispersion function, $Z(\zeta)$,
\begin{align}\label{eq:benchmark:dispf}
    Z(\zeta) = \frac{1}{\sqrt{\pi}}\int_{-\infty}^{\infty}
    \frac{\exp(-x^2)}{x-\zeta} \,dx.
\end{align}
Or more precisely, in terms of its derivative\footnote{Note that
  $\partial_\zeta(x-\zeta)^{-1} =
  -(x-\zeta)^{-2}\partial_\zeta(-\zeta)=(x-\zeta)^{-2}$.}
\begin{align*}
    Z'(\zeta) = \pfrac{Z(\zeta)}{\zeta} &= \frac{1}{\sqrt{\pi}}
    \int_{-\infty}^{\infty} \frac{\exp(-x^2)}{(x-\zeta)^2} \,dx, \\ &=
    -\frac{2}{\sqrt{\pi}}\int_{-\infty}^{\infty} \frac{x
      \exp(-x^2)}{x-\zeta} \,dx,
\end{align*}
where the integration \textit{per partes} is performed in the last
step.  A useful trick of adding and subtracting $\zeta$ can be used
to tie the functions together,\footnote{Note that
  $\int_{-\infty}^\infty \exp(-x^2)=\sqrt{\pi}$.}
\begin{align}\begin{aligned}\label{eq:benchmark:derdispf}
    Z'(\zeta) &= -\frac{2}{\sqrt{\pi}}\int_{-\infty}^{\infty}
    \frac{(x+\zeta-\zeta) \exp(-x^2)}{x-\zeta} \,dx, \\ &=
    -\frac{2}{\sqrt{\pi}}\int_{-\infty}^{\infty} \frac{(x-\zeta)
      \exp(-x^2)}{x-\zeta} \,dx -
    \frac{2\zeta}{\sqrt{\pi}}\int_{-\infty}^{\infty}
    \frac{\exp(-x^2)}{x-\zeta} \,dx, \\ &= -2\big(1 + \zeta
    Z(\zeta)\big).
\end{aligned}\end{align}

Substituting the derivative of the Maxwellian distribution function,
\begin{align*}
    \partial_{v_x}f_{x,0} = \frac{1}{\sqrt{2\pi v_{th}^2}}
    \left(-\frac{2 v_x}{2 v_{th}^2} \right)
    \exp\left(-\frac{v_x^2}{2v_{th}^2}\right),
\end{align*}
into the dispersion relation \eqrp{benchmark:langmuir_disp} gives
\begin{align*}
    1 + \frac{\omega_{pe}^2}{k^2}\frac{1}{\sqrt{2v_{th}^2}}
    \frac{2}{\sqrt{\pi}} \int_{-\infty}^\infty \frac{\frac{v_x}{2v_{th}^2}
      \exp\left(-\frac{v_x^2}{2v_{th}^2}\right)}{v_x - \omega/k} \,dx = 0,
\end{align*}
or, in terms of the substitution variables $s_1=v_x/\sqrt{2v_{th}^2}$
and $s_2 = \omega/k/\sqrt{2v_{th}^2}$,\footnote{Note that $ds_1 = dx /
  \sqrt{2v_{th}^2}$.}
\begin{align*}
    1 + \frac{\omega_{pe}^2}{k^2}\frac{1}{2v_{th}^2} \frac{2}{\sqrt{\pi}}
    \int_{-\infty}^\infty \frac{s_1 \exp(s_1^2)}{s_1 - s_2} \,ds_1 = 0.
\end{align*}

The Langmuir wave dispersion relation in terms of $Z$ is
then\footnote{Note that $$ \frac{\omega_{pe}^2}{v_{th}^2} =
  \frac{n_eq_e^2}{m_ev_{th}^2\epsilon_0} =
  \frac{n_eq_e^2}{T_e\epsilon_0} = \frac{1}{\lambda_D^2}$$}
\begin{align}\label{eq:benchmark:langmuir_disp3}
    1 - \frac{1}{2k^2\lambda_D^2}Z'\left(
    \frac{\omega/k}{\sqrt{2v_{th}^2}} \right) = 0.
\end{align}

The next step is to find the roots. However, since the $\omega$ is
imaginary, a good initial guess is required.  A useful way to obtain
it is to visualize the dispersion relation as a function of complex
$\omega$ and plot the zero contours of the real and imaginary parts,
see \fgr{benchmark:langmuir_disp}.  The crossings of the contours then
correspond to the roots of the dispersion relation.  However, they can
also be simply numerical artifacts.  In order to exclude these, it is
useful to add the plot of the absolute value of the dispersion
relation with values less than a set constant masked out.  Another
benefit of this approach, apart from obtaining the initial guesses, is
to get a better picture of the distribution of the roots.  The
remaining question is how evaluate the plasma dispersion function,
for example in Python.  Fortunately, there is a useful relation
\citep{Huba2004} tying it to the error function, which is included in
most of the postprocessing tools,
\begin{align*}
  Z(\zeta) = i \sqrt{\pi} \exp(-\zeta^2) \big(1+\mathrm{erf}(i\zeta)\big).
\end{align*}
Then in Python:
\begin{lstlisting}[language=Python]
import numpy as np
import scipy.special as spc

def Z(zeta):
    return 1j * np.sqrt(np.pi) * np.exp(-zeta**2) * (1 + spc.erf(1j*zeta))
def derZ(zeta):
    return -2*(1 + zeta*Z(zeta))
\end{lstlisting}

With the initial guess, the exact solution of the dispersion relation
\eqr{benchmark:langmuir_disp3} can be found, for example, using the
Newton's method.  For this method, the function and its derivative are
required,
\begin{align}\label{eq:benchmark:newton}
  F(\omega) = 1 - \frac{1}{2k^2\lambda_D^2}Z'\left(
  \frac{\omega/k}{\sqrt{2v_{th}^2}} \right), \quad
  \pfrac{F(\omega)}{\omega} = -
  \frac{1}{2\sqrt{2}k^3\lambda_D^2v_{th}}Z''\left(
  \frac{\omega/k}{\sqrt{2v_{th}^2}} \right),
\end{align}
where
\begin{align*}
  Z''(\zeta) = -2\big(Z(\zeta) + \zeta Z'(\zeta)\big).
\end{align*}

%---------------------------------------------------------------------
\subsection{Numerical Simulation}\label{sec:benchmark:landau_sim}

This subsection is focused on the \texttt{Gkeyll} tests of Landau
damping.  As it is the case for the rest of this work, a second order
modal Serendipity basis is used.

One option for initializing these simulations is to create
uniform electron and ion populations with Maxwellian velocity
distributions and an electric field following:
\begin{align*}
    E_{x,1}(t=0) = A \frac{q_en_e}{k\varepsilon_0} \sin(kx),
\end{align*}
where $A$ is the amplitude of the initial wave.  The diagnostic
variable, used to assess the damping, is the electric field squared,
summed over the whole domain, $E_x^2$, which is proportional to the
electric field energy.  A reasonable expectation on the evolution is that
periodic electron oscillations are superimposed on the decaying
exponential.  However, the results in
\fgr{benchmark:landau_wrong_IC} look different.

\begin{figure}[!htb]
  \centering
  \includegraphics[width=0.8\linewidth]{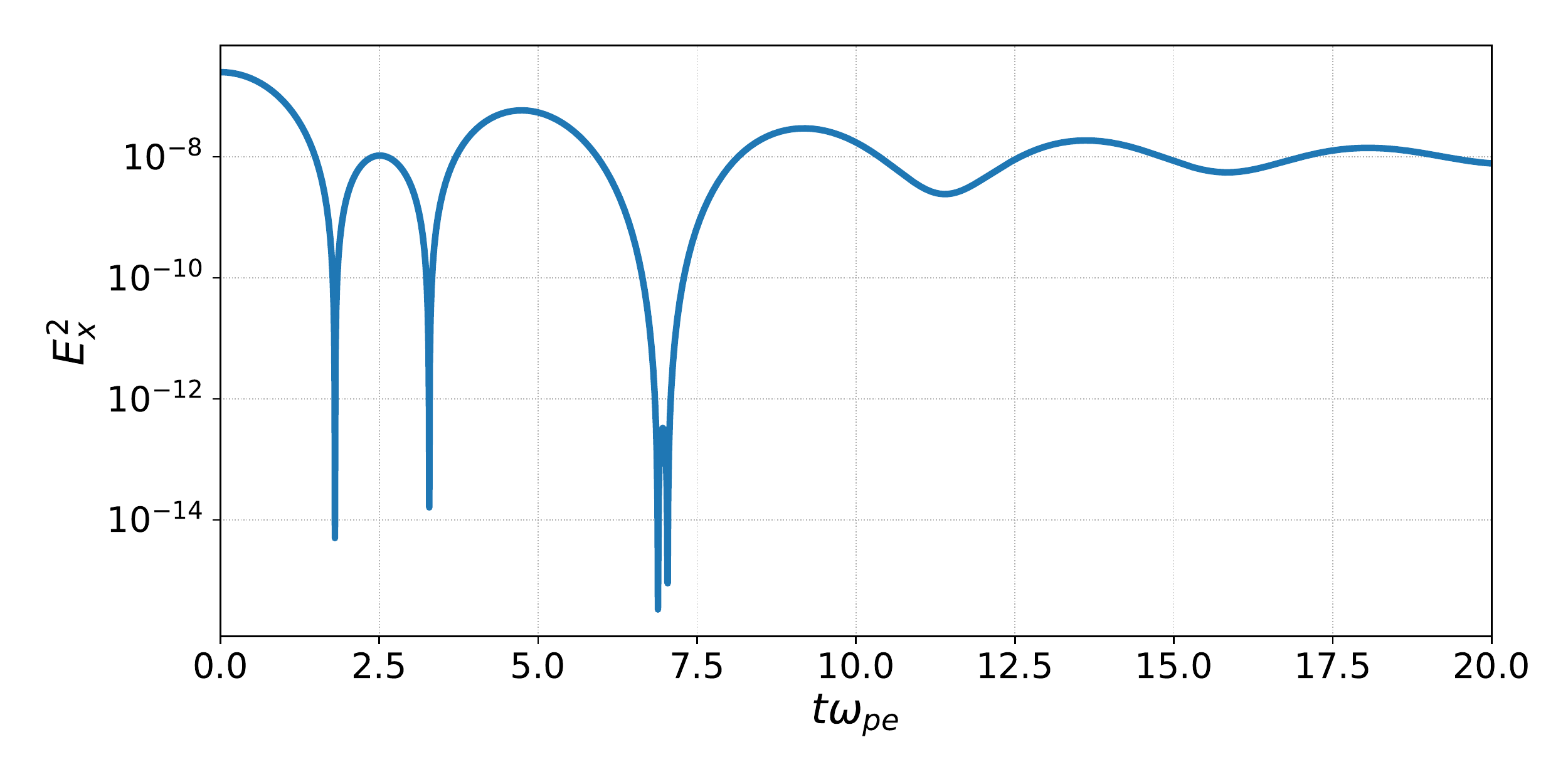}
  \caption[Landau damped electric field energy with an unphysical IC]{Landau
    damped electric field energy, $E_x^2$ with an unphysical initial
    condition -- only electric field with a periodic perturbation,
    $Aq_en_e/(k\varepsilon_0) \sin(kx)$.  Note that the evolution does
    not follow the expected profile of periodic plasma oscillations
    superimposed on the decaying exponential.}
  \label{fig:benchmark:landau_wrong_IC}
\end{figure}

The reason for this is, that the initial conditions violate the Gauss
law \eqrp{model:gauss}.  To fix this, we need to match the
derivative of $E_{x,1}(t=0)$ with respect to $x$ to the initial charge
density,
\begin{align*}
    \rho_c(t=0) = A q_en_e \cos(kx).
\end{align*}
This can be simply implemented in the \texttt{Gkeyll} input file:
\begin{lstlisting}[language={[5.1]Lua}, firstline=61,
  lastline=65, firstnumber=61]
      init = function (t, xn)
         local x, vx = xn[1], xn[2]
         return  maxwellian1D(n_e, vd_e, vth_e, vx) *
            (1 + A * math.cos(kNumber*x))
      end,
\end{lstlisting}

With the fixed initial conditions, the simulation (full listing in
\ref{list:benchmark:landau}) produces the expected results; see
\fgr{benchmark:landau_damping}.  The figure also shows the exponential
fit (green dashed line) to the envelope.  The fitting points of the
envelope were calculated as local maxima (marked with orange points in
\fgr{benchmark:landau_damping}).  It is important to note that the fit
was performed correctly as an exponential fit to data rather than a
linear fit to the logarithm.  The latter option is mathematically
questionable and generally overestimates the effects of the machine
precision errors.  One also needs to be careful about the factor of
two since the linear theory gives the growth or damping of just $E$
but $E^2$ is used for the fitting.  Therefore, it is convenient to
define the fitting function as $a \exp(2\gamma t)$.  This can be easily
done using the Python's
\texttt{scipy.optimize.curve\_fit}\footnote{\url{https://docs.scipy.org/doc/scipy/reference/generated/scipy.optimize.curve_fit.html}}
function.
\begin{lstlisting}[language=Python]
import numpy as np
import scipy.optimize as opt
import matplotlib.pyplot as plt

def exp2(x, a, b):
    return a*np.exp(2*b*x)

params, covariance = opt.curve_fit(exp2, t, E2, p0=(1.4, -0.15))
plt.plot(t, exp2(t, *params))
\end{lstlisting}

\begin{figure}[!htb]
  \centering
  \includegraphics[width=0.8\linewidth]{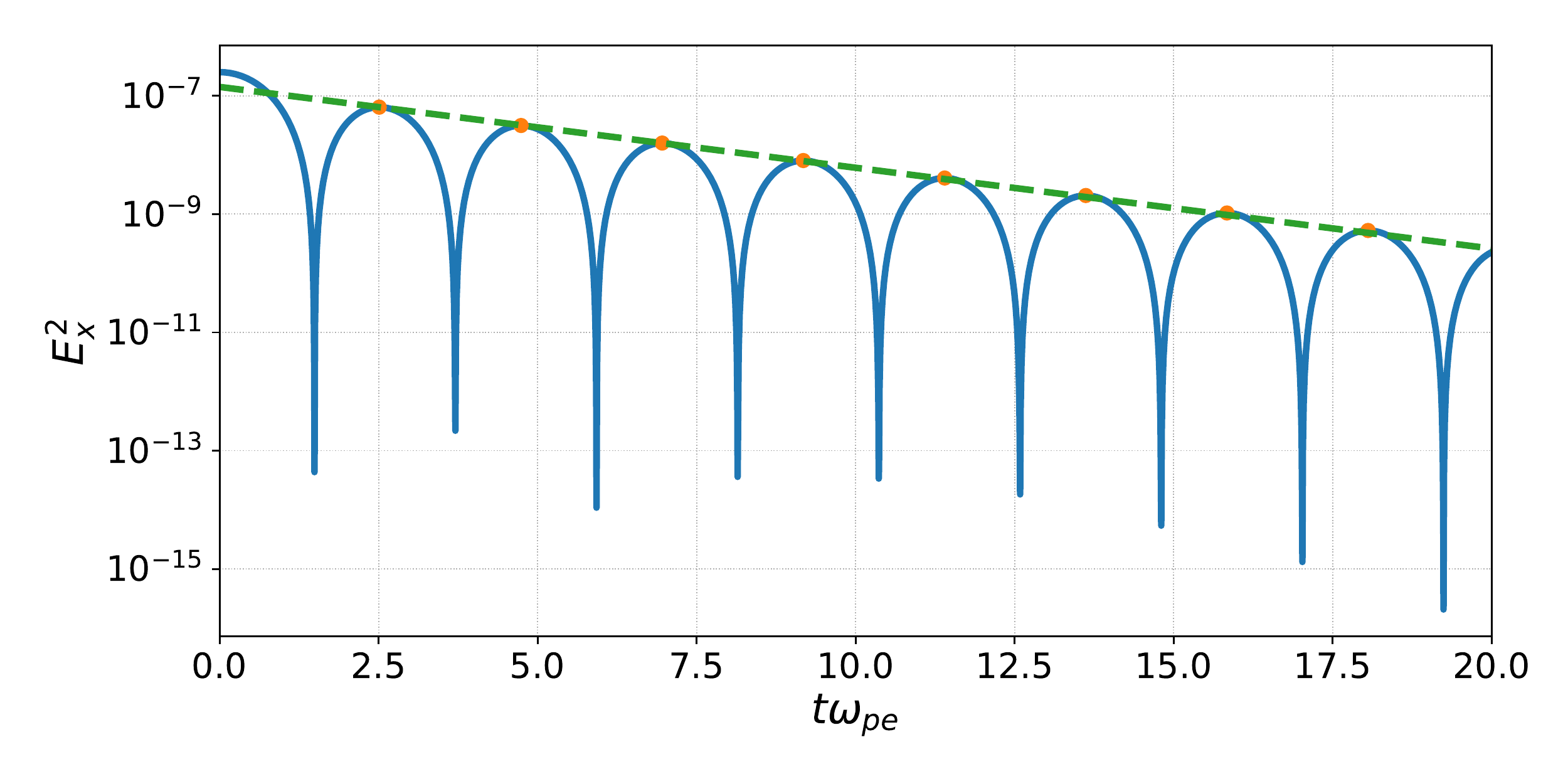}
  \caption[Landau damped electric field energy with a physical
    IC]{Landau damped electric field energy (blue line) from a
    simulation with physical initial condition -- electric field has a
    periodic perturbation, $Aq_en_e/(k\varepsilon_0) \sin(kx)$, and
    the electron density is calculated to satisfy the Gauss law
    \eqrp{model:gauss}.  $k\lambda_D$ for this simulation is
    0.5. The green dashed line is the exponential fit to the envelope.
    [Simulation input file: \ref{list:benchmark:landau}]}
  \label{fig:benchmark:landau_E2}
\end{figure}

As it was described above, it is useful to plot the dispersion
relation \eqrp{benchmark:langmuir_disp3} in order to obtain the
initial guess for the root finding and to get a better idea about the
roots.  The plot is in \fgr{benchmark:langmuir_disp}.  The
figure clearly shows two roots for $\omega_r \approx \pm
1.4\,\omega_{pe}$ which correspond to the left and right propagating
Langmuir waves. For both of these modes $\gamma \approx
-0.15\,\omega_{pe} < 0$ signifies the Landau damping.
\begin{figure}[!htb]
  \centering
  \includegraphics[width=0.9\linewidth]{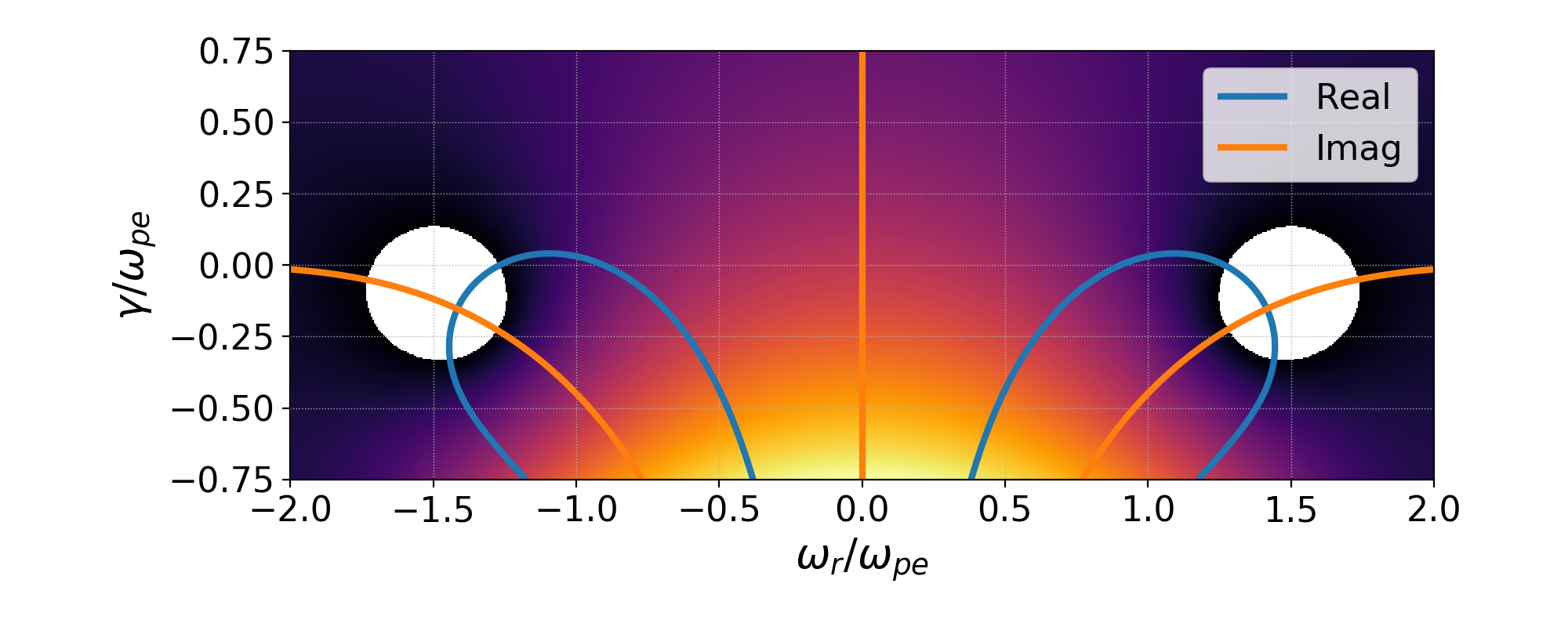}
  \caption[Roots of the Langmuir wave dispersion relation]{The
    absolute value of the left-hand-side of the Langmuir wave
    dispersion relation (logarithm of it; values less than 0.5 are
    masked out) \eqrp{benchmark:langmuir_disp3} together with the
    zero-contours of the its real and imaginary parts (blue and orange
    lines, respectively).  The contour crossings mark the solution of
    the dispersion relation, i.e., the complex frequencies of
    electrostatic waves which satisfy the Vlasov equation
    \eqrp{model:vlasov} with the initial Maxwellian distribution of
    electrons.}
  \label{fig:benchmark:langmuir_disp}
\end{figure}

Using either of these initial guesses, the Newton-Raphson root finding
algorithm\footnote{For example \texttt{scipy.optimize.newton}
  \url{https://docs.scipy.org/doc/scipy/reference/generated/scipy.optimize.newton.html}}
  with the functions in \eqr{benchmark:newton} gives for the set
  parameters the linear theory prediction of
\begin{align*}
  \omega/\omega_{pe} = 1.4157 - 0.1534i
\end{align*}

The comparison of the growth rates obtained from the linear theory and
the growth rates obtained from the simulations is in
\fgr{benchmark:landau_damping}.  Note that the numerical model
is solving the full Vlasov equation \eqrp{model:vlasov} while the
theory is a linear approximation.  Therefore, the differences on the
order of $1\%$ are not surprising.
\begin{figure}[!htb]
  \centering
  \includegraphics[width=0.8\linewidth]{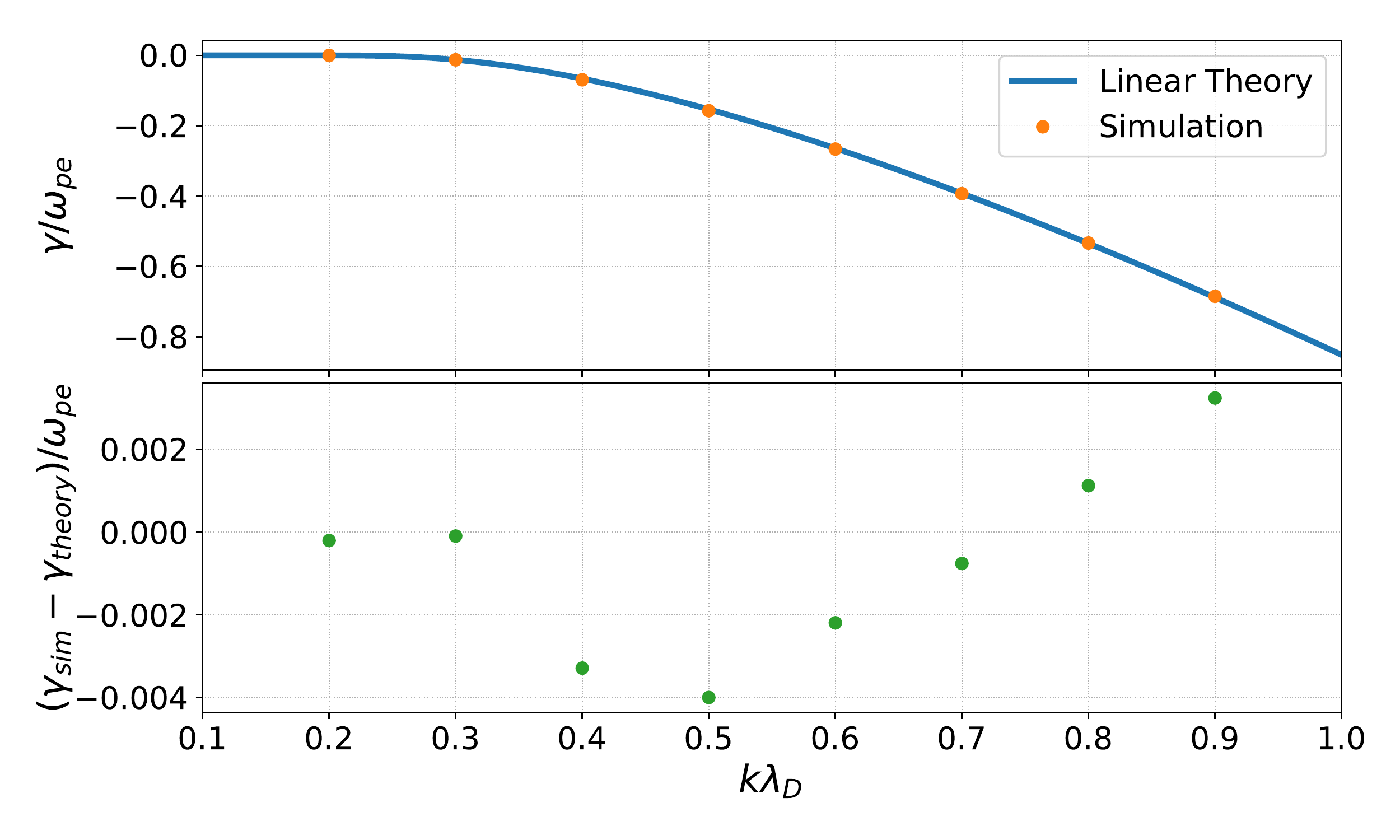}
  \caption[Landau damping rates]{The comparison of the Landau damping
    rates from the linear theory \eqrp{benchmark:langmuir_disp3}
    and the rates fitted to the simulation results.  The rates are
    plotted on top of each other in the top panel and while their
    difference is in the bottom panel.}
  \label{fig:benchmark:landau_damping}
\end{figure}

%---------------------------------------------------------------------
\FloatBarrier
\subsection{Energy Transfer and Field-particle Correlation}

For the comparison with the linear theory, the amplitude of the
Langmuir wave needs to be small ($A = 1\times10^{-4}$ is used for the
demonstration in the previous subsection). Otherwise, the nonlinear
Landau damping might dominate the process.  However, the goal of this
subsection is to investigate the energy transfer between the particles
and fields (waves) and, therefore, the amplitude needs to be
increased.  \fgr{benchmark:langmuir_distf} shows the initial and
final conditions for the Langmuir wave simulation with the amplitude
of $10\time10^{-1}$.  Note, that with such an amplitude, the
correction on the initial density to match the electric field becomes
clearly visible (left panel of \fgr{benchmark:langmuir_distf}).

\begin{figure}[!htb]
  \centering
  \includegraphics[width=0.8\linewidth]{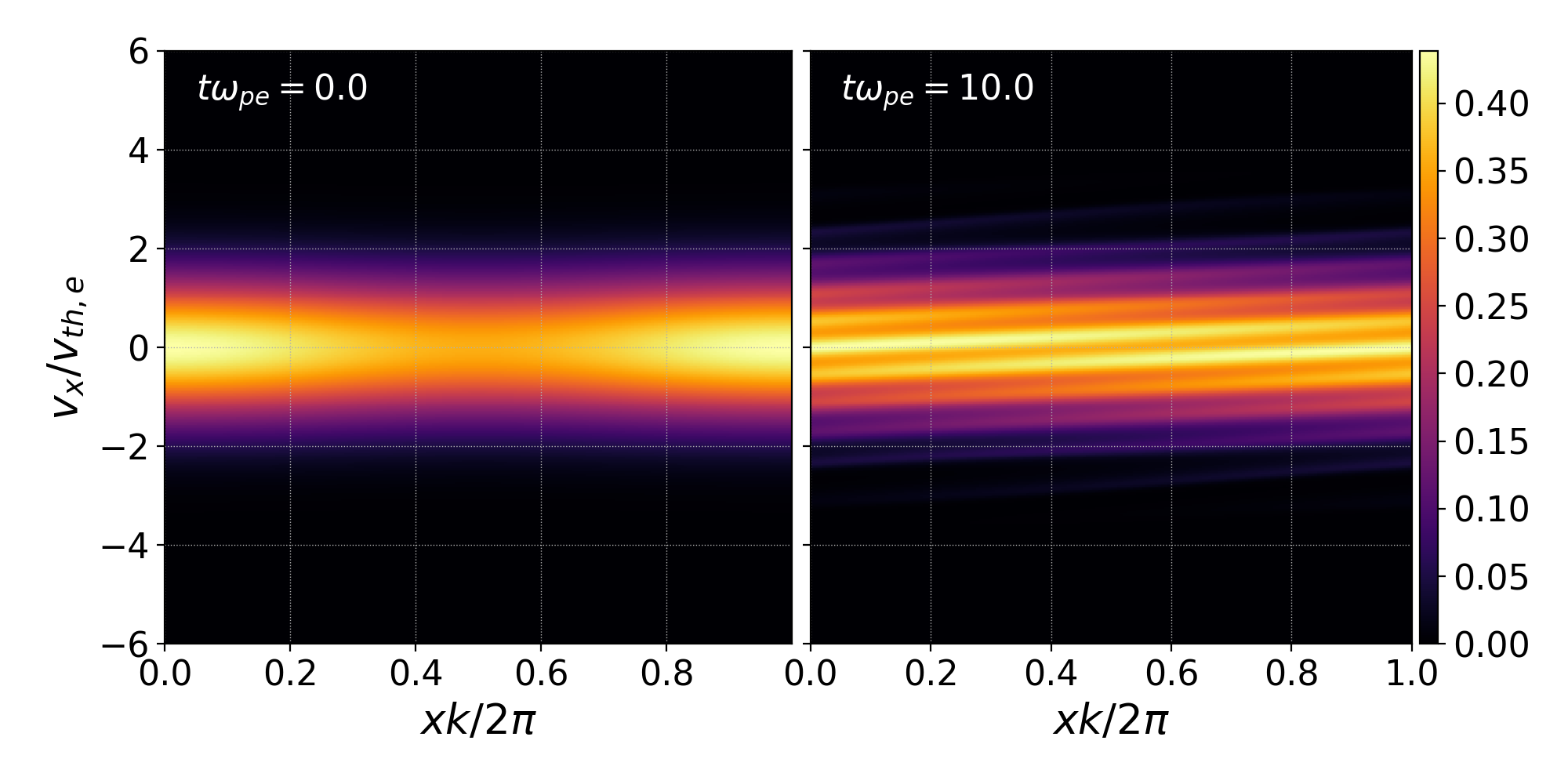}
  \caption[Electron distribution function in the Langmuir wave
    simulation]{Electron distribution functions in the Langmuir wave
    simulation. The left panel shows the initial conditions; note the
    cosine profile which is introduced to satisfy the Gauss law
    \eqrp{model:gauss}.  The right panel depicts the final time of
    $10/\omega_{pe}$.  Instead of the expected flattening of the
    distribution function around the wave phase velocity, the results
    are dominated by the electron waves and the phase space mixing.}
  \label{fig:benchmark:langmuir_distf}
\end{figure}

The wave-particle interaction is governed by the force term in the
Vlasov equation,
\begin{align*}
  \frac{\partial f}{\partial t} +
  \underbracket{v_x\partial_x f}_{\substack{\text{ballistic}\\\text{term}}} + 
  \underbracket{\frac{q}{m}E_x\partial_{v_x}
  f}_{\text{force term}} = 0.
\end{align*}
However, the force term is responsible for both the secular energy
transfer connected with damping and the oscillatory energy
transfer. What is more, the latter is typically dominant (see
\fgr{benchmark:langmuir_distf_cut}a).  It is the periodic energy
transfer connected with the damped linear motion of the waves
\citep{Howes2017}.  One way to recover the signature of the underlying
secular energy transfer between particles and fields is to integrate
the results over the whole spatial domain; see the flattening of the
distribution exactly at the predicted location in
\fgr{benchmark:langmuir_distf_cut}b).

\begin{figure}[!htb]
  \centering
  \includegraphics[width=0.8\linewidth]{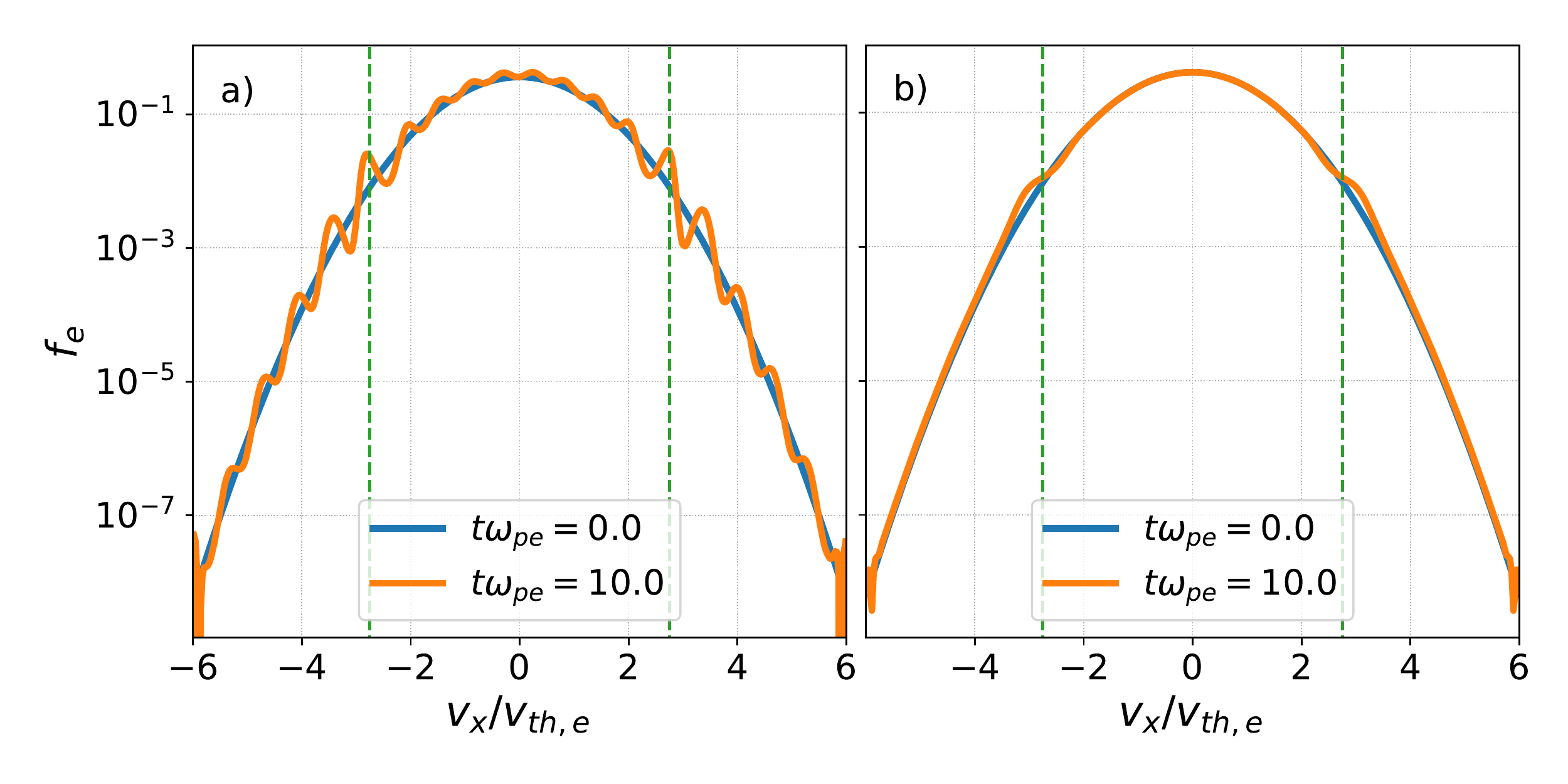}
  \caption[Cut and average of the electron distribution function in
    the Langmuir wave simulation]{Velocity profiles of the
    distribution function in the Langmuir wave simulation.  The left
    panel (a) shows the cut at the middle of the domain,
    $f_e(x=0,v_x)$, while the right panel (b) shows the average over
    $x$, $\langle f_e(x,v_x)\rangle_x$.  The green dashed lines show
    the theoretical location, where the flattening should occur, $v_x
    = (\omega_{pe}^2 + \frac{3}{2}k^2v_{th}^2)/k$ \citep{Chen1985}.
    Note that the cut (a) corresponds exactly to
    \fgr{benchmark:langmuir_distf} while the average (b) clearly
    shows the flattening. }
  \label{fig:benchmark:langmuir_distf_cut}
\end{figure}

To illustrate this, we multiply the 1D electrostatic Vlasov equation
\eqrp{benchmark:elc_vlasov1D} by $\frac{1}{2}mv_x^2$ and
integrate over the full domain,
\begin{align*}
  \underbracket{\pfraca{t}\iint\frac{1}{2}mv_x^2f \,dv_xdx}_{\text{I}}
  + \underbracket{\int \partial_x \left( \int \frac{1}{2}mv_x^3f\,dv_x
    \right)dx}_{\text{II}} + \underbracket{\int E_x \int \frac{1}{2} qv^2
  \partial_{v_x} f \,dvdx}_{\text{III}} = 0.
\end{align*}
The first term represents the change of the microscopic kinetic energy
of the particles,
\begin{align*}
  W := \int_{-L/2}^{L/2}\int_{-\infty}^{\infty}\frac{1}{2}mv_x^2f
  \,dv_xdx.
\end{align*}
The second (ballistic) term is the perfect differential in $x$, and, therefore,
disappears on periodic domains or when the distribution function
approaches zero at infinity.  The third (force) term can be integrated
\textit{per partes},
\begin{align*}
  \int E \int \frac{1}{2} qv^2 \partial_{v_x} f \,dvdx = - \int E_x \int
  qvf \,dvdx = -\int jE\,dx.
\end{align*}
Then using the electrostatic Ampere-Maxwell law we get
\begin{align*}
  \varepsilon_0\pfrac{E_x}{t} = -jE_x \quad \Rightarrow \quad
  \pfraca{t}\left(\varepsilon_0\frac{E_x^2}{2}\right) = -jE_x.
\end{align*}
Putting everything together we see that in the integral sense the
energy is conserved between the fields and particles,
\begin{align}
  \pfraca{t} \left( \int \varepsilon_0\frac{E_x^2}{2}\,dx + W \right) =
  0.
\end{align}

However, global integrations are neither precise nor always possible.
For example in the astrophysical plasma applications, where the energy
transfer is important for understanding of the turbulence in the
heliosphere \citep{Howes2017}, the experimental data come from
spacecrafts and, therefore, represent only single-point measurements.
In the computational plasma simulations, the whole spatial profile is
available, however, the energy transfer might be localized in phase
space and this information is lost during the integration.
\cite{Klein2016} address this problem by introducing a novel
diagnostic techniques -- the field-particle correlation (FPC).  FPC is
then discussed in greater detail by \cite{Howes2017} and
\cite{Klein2017}.

We can take a closer look on the energy balance by splitting the
distribution function into the equilibrium part\footnote{That is the
  Maxwellian distribution as it is shown in the subsection
  \ref{sec:model:distf}.} and the rest,
\begin{align*}
  f(t,x,v_x) = f_M(v_x) + \delta f(t,x,v_x).
\end{align*}
Note that, while this approach seems similar to the linear theory, no
assumption is made about the relative magnitudes of the parts. Then,
\begin{align*}
  \pfrac{W}{t} = -\frac{1}{2}m\iint
  v^2\Big(\underbracket{v\partial_x\delta f}_{\text{I}} +
  \underbracket{\frac{q}{m}E_x \partial_{v_x} f_M}_{\text{II}} +
  \underbracket{\frac{q}{m}E_x \partial_{v_x} \delta f}_{\text{III}}
  \Big)\,dvdx.
\end{align*}
The first term is still the perfect differential and disappears.  The
difference is in the force term.  Since the equilibrium part (II) is
not a function of $x$, it can be rearranged into the form of perfect
differential and it disappears as well.  What is more,
$v_x^2\partial_{v_x}f_M$ is an odd function and evaluates to zero through
the integration.  Only the non-equilibrium term contributes to the
wave-particle interaction.

\begin{figure}[!htb]
  \centering
  \includegraphics[width=0.8\linewidth]{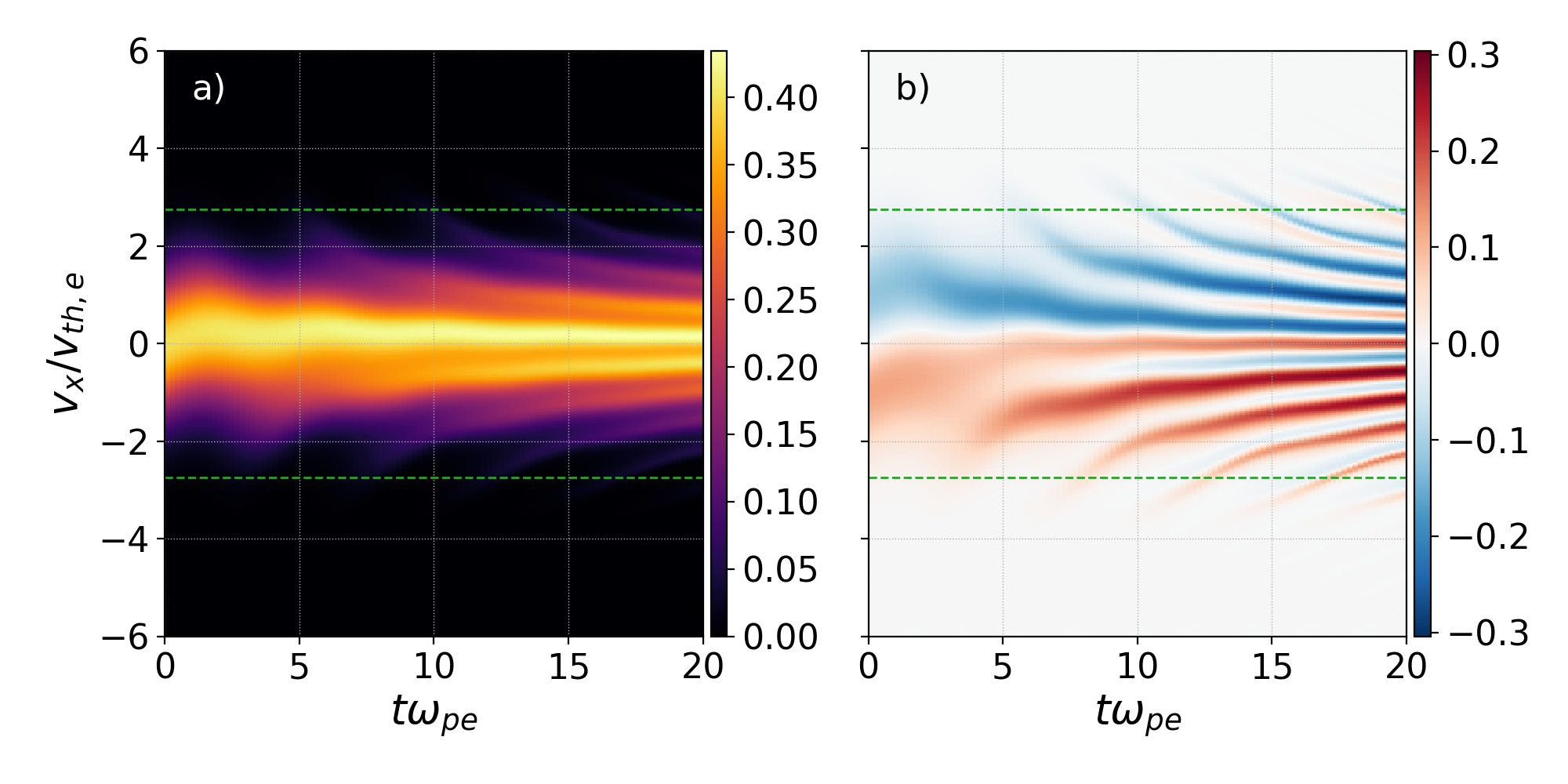}
  \caption[Evolution of a spatial cut of the distribution function and
    its derivative]{Evolution of a spatial cut at $x_0k/2\pi = 0.25$
    of the distribution function (panel a) and its derivative (panel
    b). The theoretical resonant velocity is marked in green dashed lines.}
  \label{fig:benchmark:langmuir_distf_evolution}
\end{figure}

Noticing that, \cite{Klein2016} define the single point field-particle
correlation function, which is the direct measure of the energy
transfer, as
\begin{align}\label{eq:benchmark:fpc}
  C(t, v_x, N) := - \frac{qv_x^2}{2}\frac{1}{N\Delta t}
  \sum_{i=0}^N \pfrac{\delta f(t+i\Delta t, x_0 , v_x)}{v_x} E_x(t+i\Delta
  t, x_0),
\end{align} 
where the $\Delta t$ is the time interval between the individual
frames.  At this point, it is worth stressing out that the continuum
kinetic methods are well suited for the application of the FPC.  When
applied with methods affected by the statistical noise, for example
PIC, calculating the gradient of the distribution function might cause
problems.  What is more, the DG version of the continuum kinetic method
provides a high order gradient.  The time evolution of the
distribution function and its gradient are in
\fgr{benchmark:langmuir_distf_evolution}.  Both are for fixed
$x_0k/2\pi = 0.25$, where there is the maximum of the initial electric
field.

\begin{figure}[!htb]
  \centering
  \includegraphics[width=0.8\linewidth]{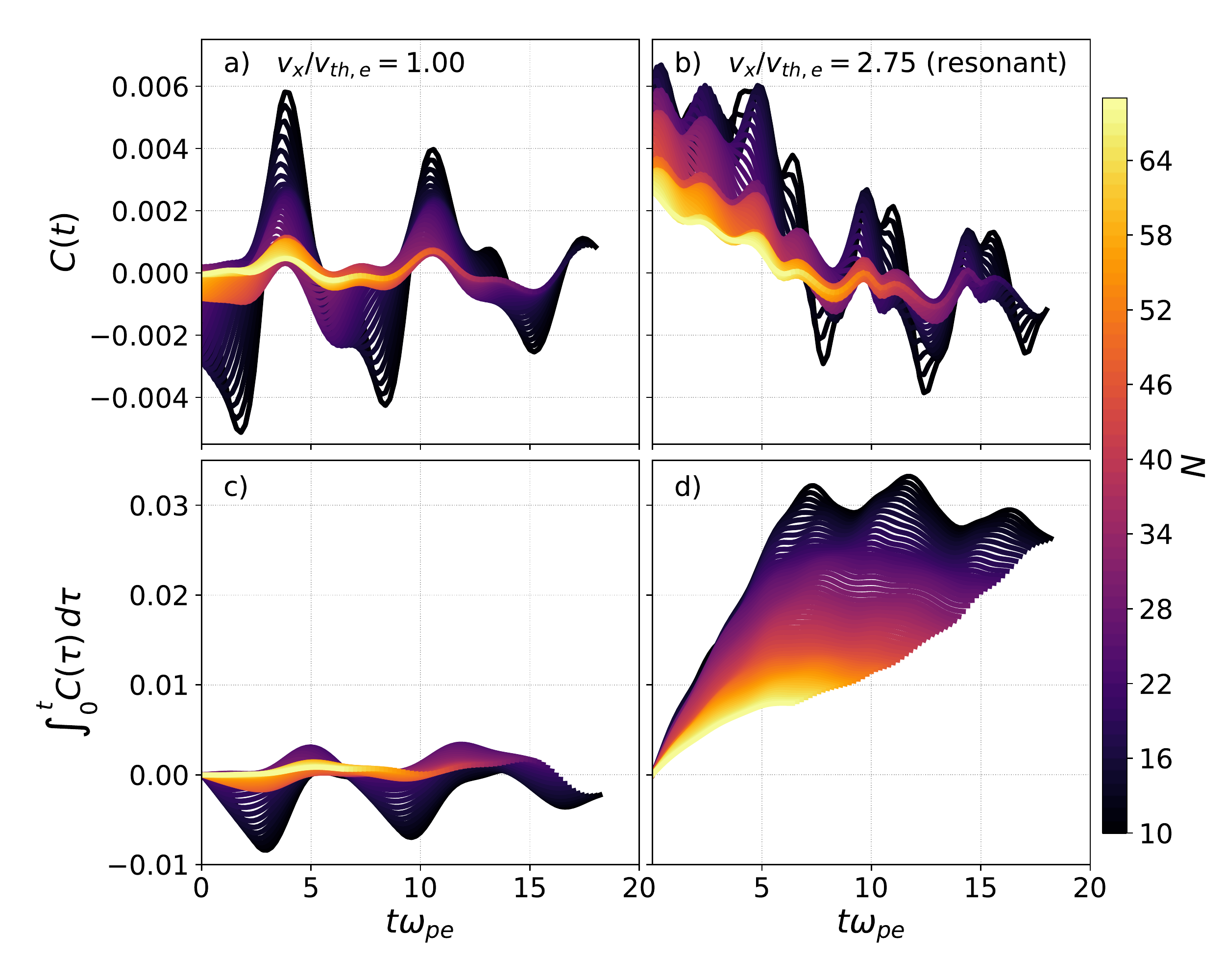}
  \caption[Localization of the energy transfer and the total
    transfered energy]{Field-particle energy transfer at $x_0k/2\pi =
    0.25$ for a non-resonant velocity $v_x/v_{th,e}=1.00$ (first
    column; a and c) and the resonant velocity $v_x/v_{th,e}=2.75$
    (second column; b and d).  Color-coded are the numbers of frames,
    $N$, used for the calculation of $C$ \eqrp{benchmark:fpc}.  The
    bottom row shows the integral $\int_0^t C(\tau)\,d\tau$ which
    represents the total energy transferred during the course of the
    simulation; positive sign signifies the transfer from the field
    energy to the particle energy.}
  \label{fig:benchmark:langmuir_C_interval}
\end{figure}

The key part of FPC is the time averaging; the $N$ needs to be chosen,
so the $N\Delta t$ corresponds to the oscillatory period.  An example
of different averaging intervals is in
\fgr{benchmark:langmuir_C_interval}.  The first row shows the
$C$ directly with the color-coded $N$.  Note that the profiles change
drastically as the $N$ increases but asymptotes as the interval
approaches the oscillation period.  The second row shows the integral
of the FPC function, $\int_0^t C(\tau)\,d\tau$, i.e., the total
accumulated energy.  Here the difference between the FPC for the
resonant and non-resonant velocities becomes clear.  While there is no
net transfer for $v_x/v_{th,e}=1.00$, there is a clear signature of
the electric field damping at the resonant velocity (the positive sign
represents the energy being transferred from the field to particles).
A full Python script highlighting the usage of FPC to create
\fgr{benchmark:langmuir_C_interval} is available in
\ref{list:scripts:fpc}.

Finally, \fgr{benchmark:langmuir_C_full} shows the full velocity
profile of FPC at $x_0k/2\pi = 0.25$ and with $N=50$.  Each single
frame is outputted once per $1/\omega_{pe}$, i.e., the FPC is averaged
over the interval of $50/\omega_{pe}$.  The highlighted energy
transfers happen exactly around the predicted velocity, $v_x =
(\omega_{pe}^2 + \frac{3}{2}k^2v_{th}^2)/k$.  What is more, the change
in the direction of the transfer is a clear indication of a resonant
process, which is captured only in the kinetic theory
\citep{Klein2016}.  Damping processes and instabilities captured just
by the fluid theory, for example the two stream instability discussed
in the next subsection, has different FPC signature.

\begin{figure}[!htb]
  \centering
  \includegraphics[width=0.7\linewidth]{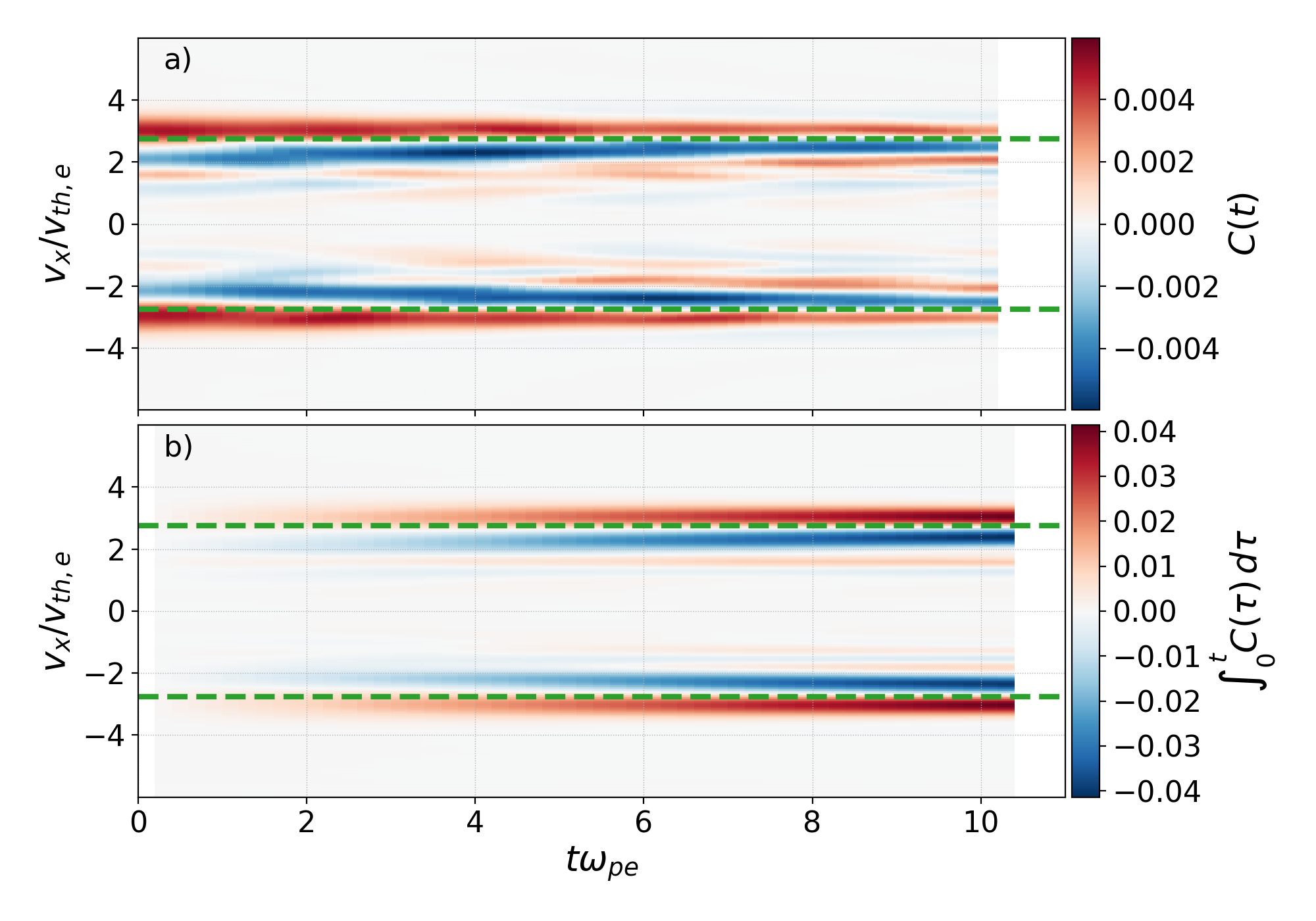}
  \caption[Field-particle correlation for the Langmuir
    waves]{Field-particle correlation for the Langmuir waves.  The
    first row shows directly $C$ \eqrp{benchmark:fpc} at
    $x_0k/2\pi = 0.25$, while the bottom row
    represents the integral $\int_0^t C(\tau)\,d\tau$. Unlike the
    \fgr{benchmark:langmuir_C_interval}, this figure captures
    the full velocity space with the resonant velocities highlighted
    with dashed green lines.}
  \label{fig:benchmark:langmuir_C_full}
\end{figure}

%=====================================================================
\FloatBarrier
\section{Two-stream Instability}\label{sec:benchmark:two-stream}

The two-stream instability is a classical textbook problem (see
\cite{Chen1985} Chapter 6.6) where the kinetic energy of the two
counter-streaming electron beams is converted into the electric field
energy.

Here, it serves as another benchmark problem and as a stepping stone
for the more complex Weibel instability in the next chapter.

%---------------------------------------------------------------------
\subsection{Linear Theory}\label{sec:benchmark:two-stream:lintheory}

The linear theory derivation follows similar steps as in
Sec.\thinspace\ref{sec:benchmark:landau:lintheory} with the exception
of the equilibrium distribution function.  Here, it consists of two
populations with bulk velocities $\pm u$,
\begin{align*}
    f_{x,0} = \frac{1}{2}\frac{1}{\sqrt{2\pi v_{th}^2}}
    \exp\left(-\frac{(v_x+u)^2}{2v_{th}^2}\right) +
    \frac{1}{2}\frac{1}{\sqrt{2\pi v_{th}^2}}
    \exp\left(-\frac{(v_x-u)^2}{2v_{th}^2}\right).
\end{align*}
Note the factors $\frac{1}{2}$ which ensure that the integrated
density is still $n_e$.  The derivative is then
\begin{align*}
    \partial_{v_x}f_{x,0} = \frac{1}{2}\frac{1}{\sqrt{2\pi v_{th}^2}} \left[
    \left(-\frac{2 (v_x+u)}{2 v_{th}^2} \right)
    \exp\left(-\frac{(v_x+u)^2}{2v_{th}^2}\right) +
    \left(-\frac{2 (v_x-u)}{2 v_{th}^2} \right)
    \exp\left(-\frac{(v_x-u)^2}{2v_{th}^2}\right) \right].
\end{align*}

Consequently, the dispersion relation,
\begin{align*}
    1 + \frac{1}{2}\frac{\omega_{pe}^2}{k^2}\frac{1}{\sqrt{2v_{th}^2}}
    \frac{2}{\sqrt{\pi}} \int_{-\infty}^\infty
    \frac{\frac{v_x-u}{2v_{th}^2}
      \exp\left(-\frac{(v_x-u)^2}{2v_{th}^2}\right) +
      \frac{v_x+u}{2v_{th}^2}
      \exp\left(-\frac{(v_x+u)^2}{2v_{th}^2}\right)}{v_x - \omega/k}
    \,dx = 0,
\end{align*}
requires a modified substitution $s_1=(v_x\pm u)/\sqrt{2v_{th}^2}$ and
$s_{2} = (\omega/k\pm u)/\sqrt{2v_{th}^2}$.

The final dispersion relation in terms of the plasma dispersion
function is
\begin{align}\label{eq:benchmark:two-stream_disp}
    1 - \frac{1}{4k^2\lambda_D^2} \left[ Z'\left(
    \frac{\omega/k - u}{\sqrt{2v_{th}^2}} \right) + Z'\left(
    \frac{\omega/k + u}{\sqrt{2v_{th}^2}} \right) \right] = 0.
\end{align}

%---------------------------------------------------------------------
\subsection{Numerical Simulation}
\label{sec:benchmark:two-stream:sims}

The full listing of the baseline two-stream simulation is available in
\ref{list:benchmark:two-stream}.  It is initialized with two electrons
beams with density $n_e=0.5$, bulk velocity $u_e=\pm1.0$, and the
thermal velocity $v_{th,e}=0.2$ in the dimensionless units; see the
top left panel of \fgr{benchmark:two-stream_overview} for the plot of
the distribution function.  \fgr{benchmark:two-stream_overview} also
includes the overview of the distribution function evolution.  Note
that the distribution function remains without noticeable changes for
the majority of the simulation run and then change rapidly together
with the exponential growth of the instability.  It also captures the
decrease of the beam kinetic energy (the bulk shifts to lower $|v_x|$
in the phase space plots) as it is transformed into electromagnetic
energy.

\begin{figure}[!htb]
  \centering
  \includegraphics[width=0.9\linewidth]{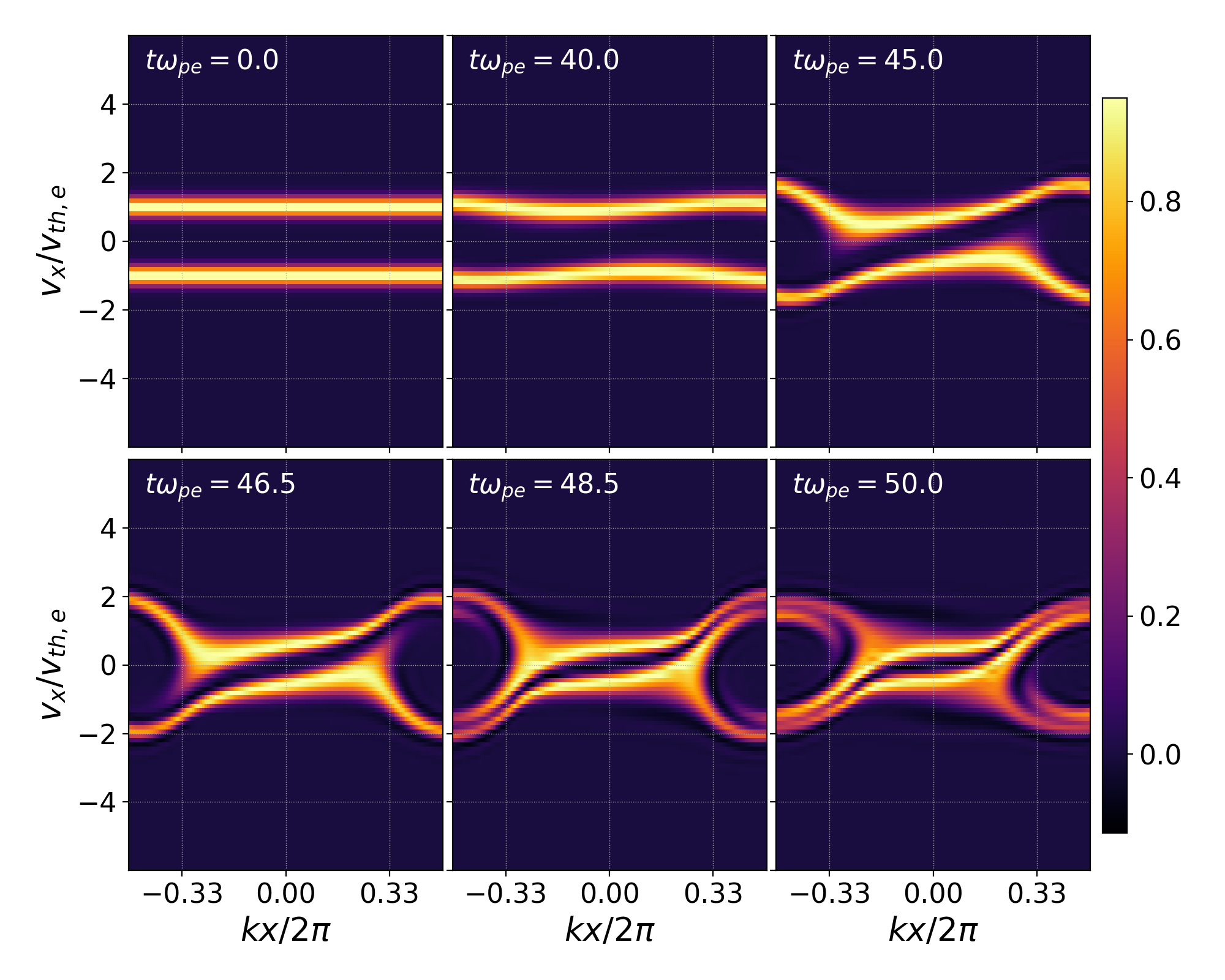}
  \caption[Evolution of the two-stream instability]{Evolution of the
    electron distribution function during the two-stream
    instability. Note the decrease of the beam kinetic energy (the
    bulk shifts to lower $|v_x|$ in the phase space plots) as is
    transformed into the electric field energy.  [Simulation input
      file: \ref{list:benchmark:two-stream}]}
  \label{fig:benchmark:two-stream_overview}
\end{figure}

The evolution of the integrated $E_x^2$, which is a proxy for the
electric field energy, is captured in
\fgr{benchmark:two-stream_energy}; the linear plot in the top
panel and the semi-logarithmic plot in the bottom one.  It
demonstrates a couple of key points.

Firstly, as it was mentioned in the Landau damping section, it is
important to fit an exponential function to linear data rather than a
straight line to a logarithm of the data.  Simulations typically
capture other modes which are quickly damped (see the roots in
\fgr{benchmark:two-stream_disp}) but could cause fluctuations at
the beginning of a simulation.  As seen in
\fgr{benchmark:two-stream_energy} for $0 < t\omega_{pe} < 15$,
these fluctuations can be exaggerated by the logarithm.

\begin{figure}[!htb]
  \centering
  \includegraphics[width=0.9\linewidth]{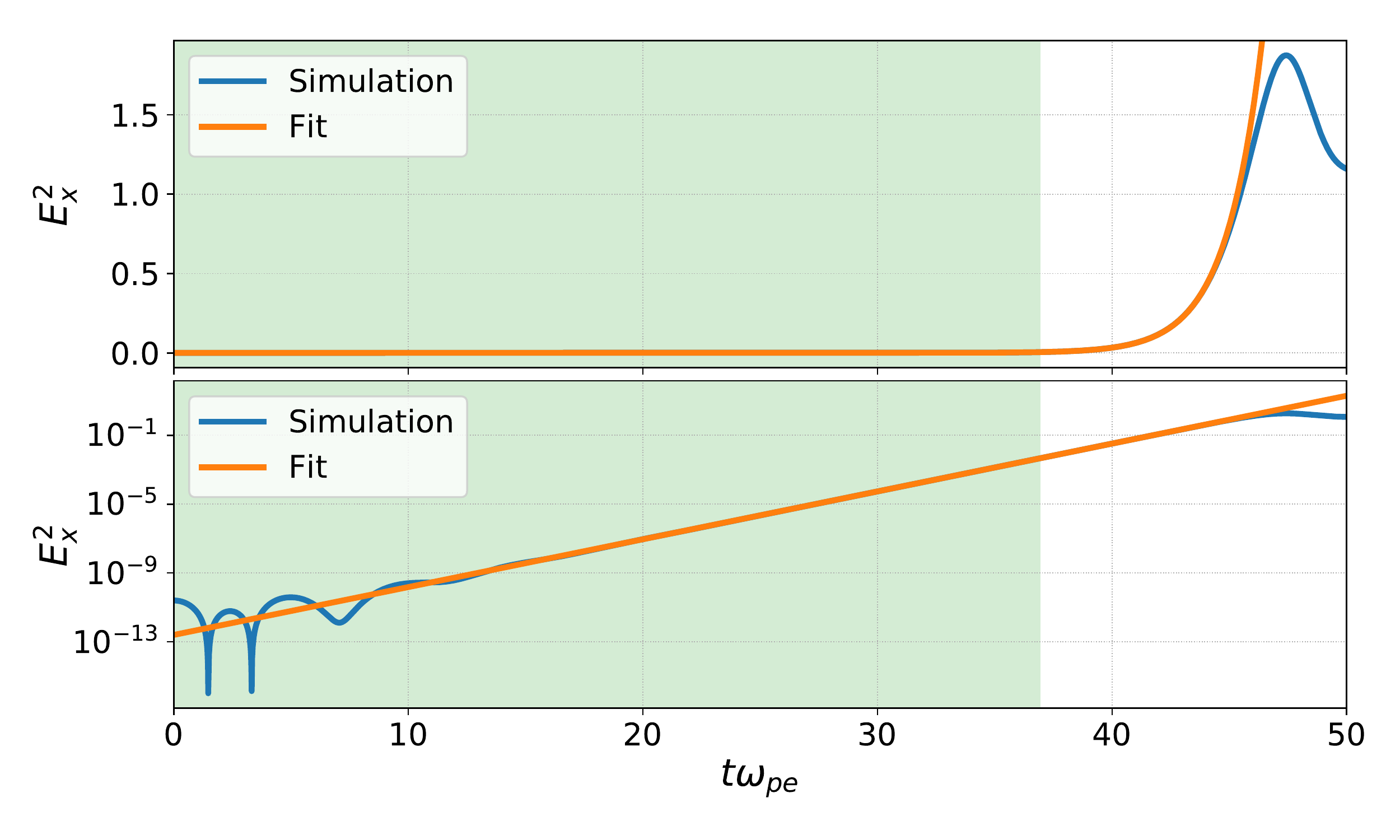}
  \caption[Evolution of field energy in the two-stream
    simulation]{Evolution of the electric field energy proxy, $E_x^2$
    in the two-stream simulation; both the linear (top) and
    semi-logarithmic (bottom) plots.  The blue lines represent the
    simulation data and the orange are the ``best'' exponential fit
    obtained using the sweeping fitting which maximizes the $R^2$
    \eqrp{benchmark:R2}.  The green background denotes the ``best''
    fitting interval.  [Simulation input file:
      \ref{list:benchmark:two-stream}]}
  \label{fig:benchmark:two-stream_energy}
\end{figure}

Secondly, there is the question of the fitting interval.  The
exponential growth prediction is not valid in the non-linear regime.
However, there is no clear boundary between the regimes.  A common
approach is to select the fitting interval by inspection, which causes
issues with the reproducibility of the obtained results.  What is
more, seemingly negligible change of the interval can result in
surprisingly different growth rate (see below).

In order to obtain reasonable and reproducible data, the sweeping
fitting is used here.  The exponential fit is performed on a
continuously increasing interval and the best fit is uniquely chosen
based on the coefficient of determination, $R^2$, which is obtained as
follows,
\begin{align}\label{eq:benchmark:R2}
  R^2 := 1-\frac{SS_{res}}{SS_{tot}} = 1 - \frac{\sum_i
    \big(y_i-y(x_i)\big)^2}{\sum_i (y_i-\overline{y})^2},
\end{align}
where $\overline{y}=\frac{1}{n}\sum_{i=1}^{n}y_i$.  Particularly for
this problem, $y(x_i) = a \,\mathrm{exp}(2\gamma x_i)$.  The green
shaded area in \fgr{benchmark:two-stream_energy} shows the fitting
region with the highest $R^2$ for this particular problem.  The
sweeping fit can be, for example, performed using the following code:
\begin{lstlisting}[language=Python]
import numpy as np
import scipy.optimize as opt
import postgkyl as pg

data = pg.GData('two-stream_fieldEnergy_')
t = data.peakGrid()[0]
Ex2 = data.peakValues()[..., 0]
def exp2(x, a, b):
    return a*np.exp(2*b*x)
bestN = 0.0
bestParams = (1.0, 0.3)
bestR2 = 0.0

for n in range(100, 7000):
    xn = t[:n]
    yn = Ex2[:n]
    params, cov = opt.curve_fit(exp2, xn, yn, bestParams)
    residual = yn - exp2(xn, *params)
    ssRes = np.sum(residual**2)
    ssTot = np.sum((yn - np.mean(yn))**2)
    R2 = 1 - ssRes/ssTot
    if R2 > bestR2:
        bestR2 = R2
        bestParams = params
        bestN = n
\end{lstlisting}

As a counter-example, the fitting region in
\fgr{benchmark:two-stream_energy_guess} is chosen by inspection.
It results in $\gamma/\omega_{pe}=2.86$ instead of
$\gamma/\omega_{pe}=3.20$ obtained from the sweeping fit, which is
approximately $10\%$ difference.
\begin{figure}[!htb]
  \centering
  \includegraphics[width=0.9\linewidth]{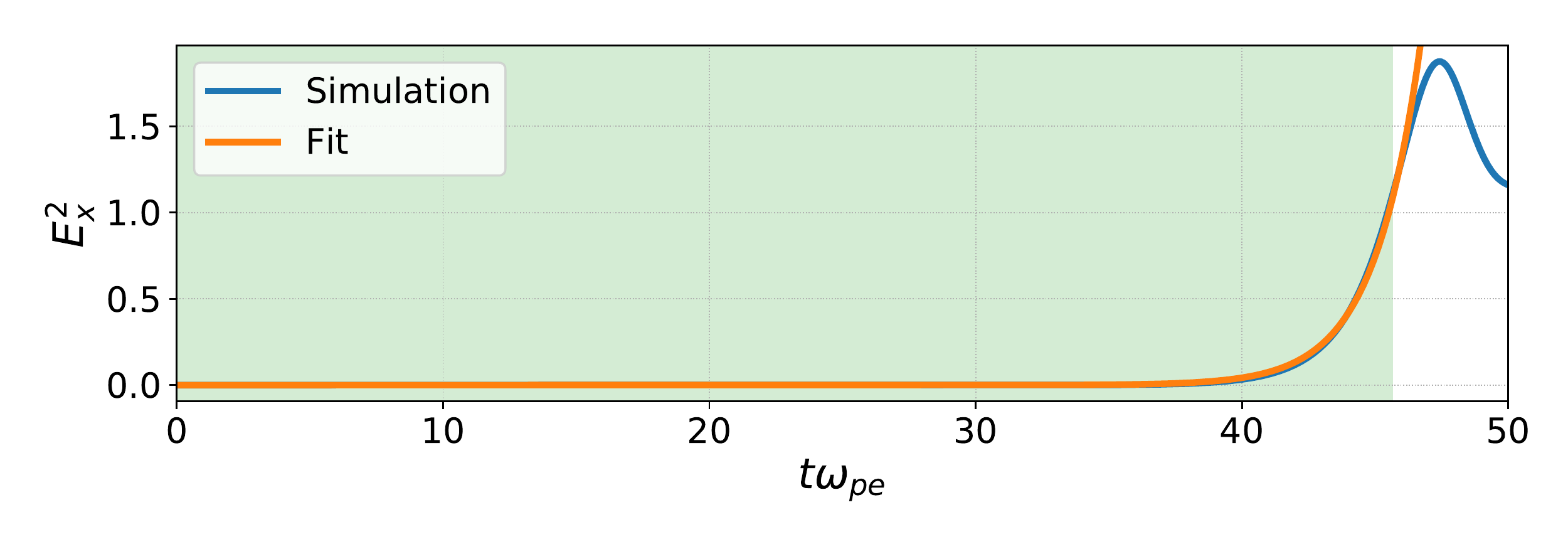}
  \caption[Evolution of field energy in the two-stream simulation
    2]{The same energy evolution as in
    \fgr{benchmark:two-stream_energy} with the difference of the
    fitting region being selected by inspection rather than by
    sweeping fit.  It results in $\approx10\%$ difference in the
    obtained growth rate.  [Simulation input file:
      \ref{list:benchmark:two-stream}]}
  \label{fig:benchmark:two-stream_energy_guess}
\end{figure}

Now that the fitting procedure has been specified, the simulation
results can be compared to the theory.  Similar to Landau damping
section we start with plotting the dispersion relation
\eqrp{benchmark:two-stream_disp}.  \fgr{benchmark:two-stream_disp}
clearly shows the single purely growing mode ($\omega_r=0.0$) of the
electron two-stream instability.  Apart from the growing mode, there
is wide range of damped oscillatory modes that are partially
responsible for the initial oscillation of energy (see
\fgr{benchmark:two-stream_energy}).

\begin{figure}[!htb]
  \centering
  \includegraphics[width=0.8\linewidth]{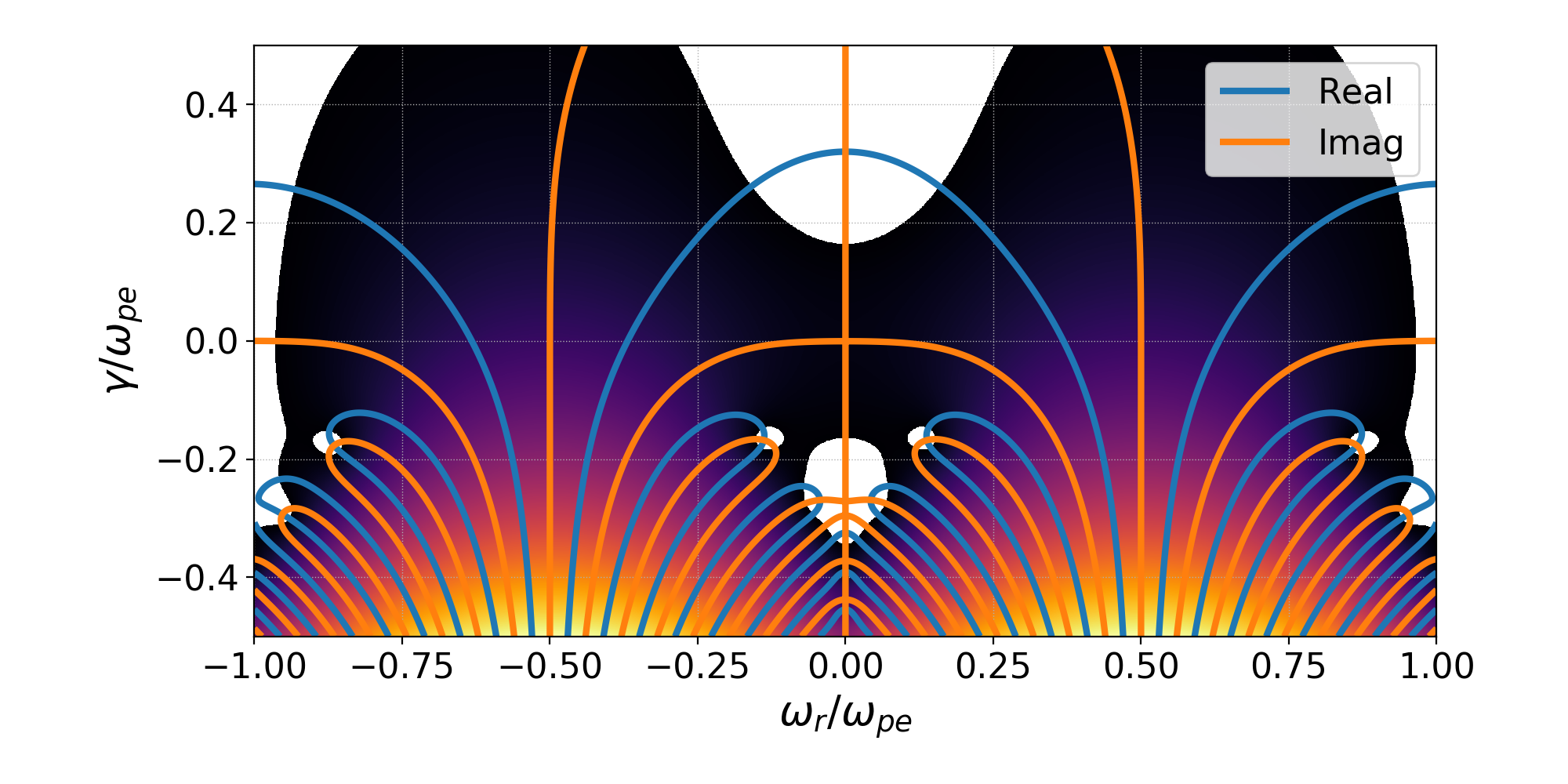}
  \caption[Roots of the two-stream instability dispersion
    relation]{The absolute value of the left-hand-side of the two
    dispersion relation (logarithm of it; values less than 2.0 are
    masked out) \eqrp{benchmark:two-stream_disp} together with the
    zero-contours of its real and imaginary parts (blue and orange
    lines, respectively).  The contour crossings mark the solution of
    the dispersion relation, i.e., the complex frequencies of
    electrostatic waves which satisfy the Vlasov equation
    \eqrp{model:vlasov} with the initial conditions of the simulation
    \ref{list:benchmark:two-stream}.}
  \label{fig:benchmark:two-stream_disp}
\end{figure}

The comparison of the linear theory growth predictions from the dispersion
relation \eqrp{benchmark:two-stream_disp} and the growth rates
obtained by fitting the simulation data is in
\fgr{benchmark:two-stream_gamma}.  The results show a good match
within $0.3\%$ for this range of $k\lambda_D$.

\begin{figure}[!htb]
  \centering
  \includegraphics[width=0.8\linewidth]{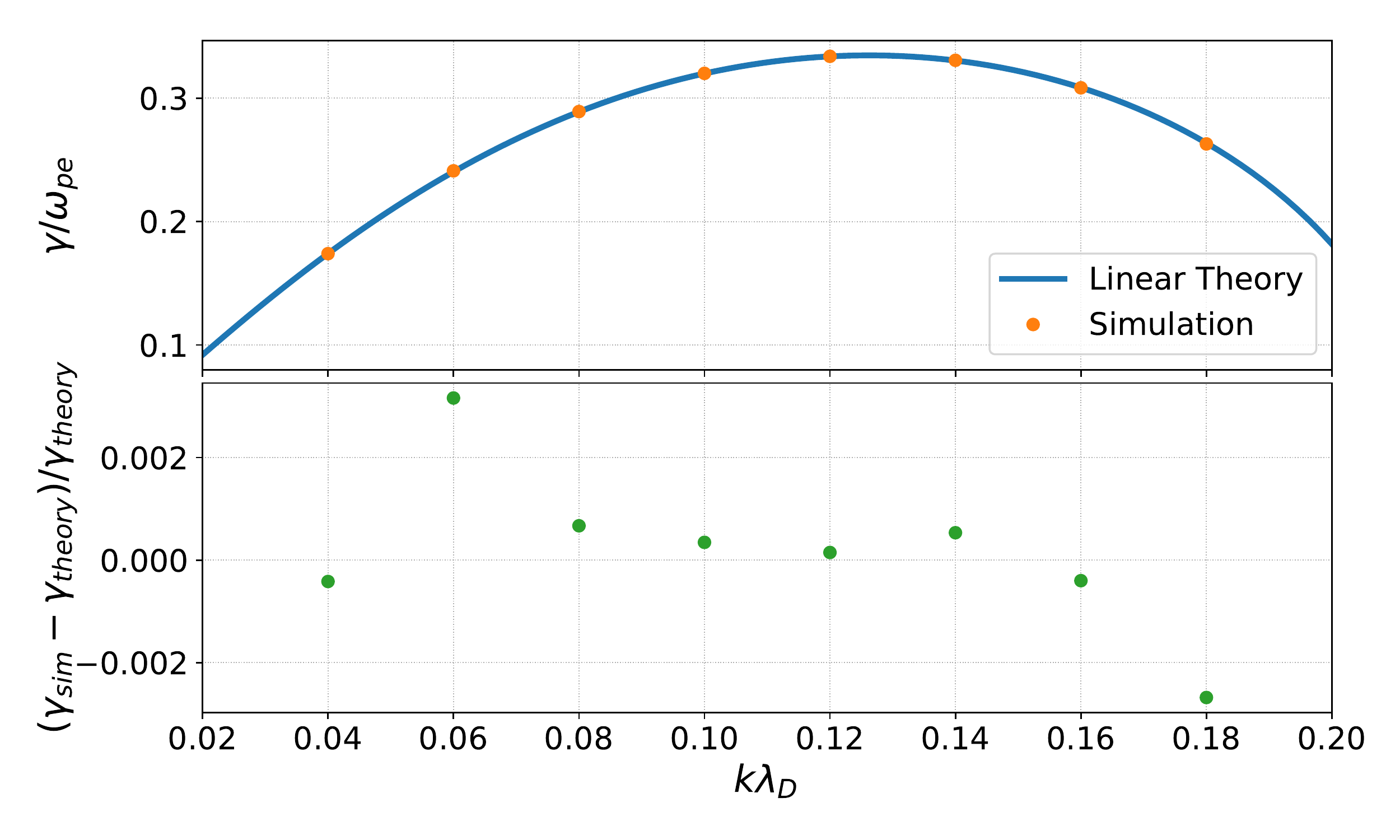}
  \caption[Growth rates of the two-stream instability]{Comparison of
    the two-stream instability growth rates obtained from the linear
    theory \eqrp{benchmark:two-stream_disp} and from the
    simulations (top) for the range of $k\lambda_D$ values. The bottom
    plot shows the relative error is less than $0.3\%$.}
  \label{fig:benchmark:two-stream_gamma}
\end{figure}

%---------------------------------------------------------------------
\subsection{Velocity Space Resolution}

The two-stream instability is an interesting problem for the velocity
space convergence test.  The input file used in the previous
discussion [\ref{list:benchmark:two-stream}] is modified to run with
8, 16, 32, 64, 128, and 256 velocity cells.  Note that 32 velocity
cells were used for the comparison with the linear theory and provided
a good match.  \fgr{benchmark:two-stream_convergence_f} shows the
comparison of the distribution function at $t\omega_{pe} = 50.0$
(analogous to the last panel in
\fgr{benchmark:two-stream_overview}).

\begin{figure}[!htb]
  \centering
  \includegraphics[width=0.9\linewidth]{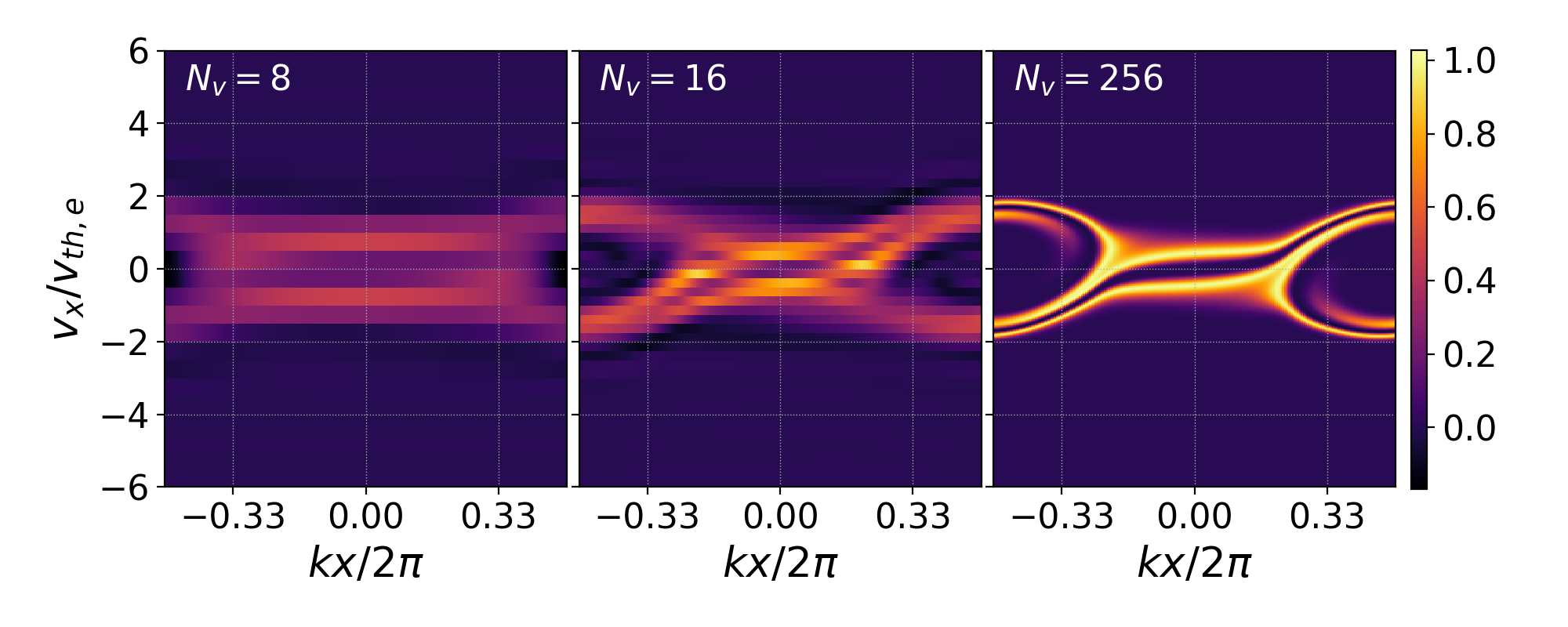}
  \caption[Test of a velocity space resolution required for two-stream
    instability]{Comparison of the electron distribution functions
    from two-stream instability simulation at $t\omega_{pe} = 50.0$.
    The panels capture results for different velocity space
    resolutions, $N_v$ -- 16, 32, and 256 cells.  An analogous plot
    for $N_v=32$ is in \fgr{benchmark:two-stream_overview}.}
  \label{fig:benchmark:two-stream_convergence_f}
\end{figure}

Increasing the velocity resolution clearly improves the problems with
negative values of the distribution function around the sharp
gradients.\footnote{It is worth noting that even though the negative
  values of distribution function are clearly nonphysical, they do not
  result in a crash of the simulation.  This is different from fluid
  simulations where the negative density produces complex sound speed,
  $\sqrt{p/\rho}$, breaking the run.}  In order to quantify the
convergence, \fgr{benchmark:two-stream_convergence} shows the
the electric field growths.  Perhaps surprisingly, the results for just
16 velocity cells appear to be on top the higher resolution data for
the linear growth phase.  The fitted values of $\gamma$ are in
Tab.\thinspace\ref{tab:benchmark:two-stream_growth}.

\begin{figure}[!htb]
  \centering
  \includegraphics[width=0.9\linewidth]{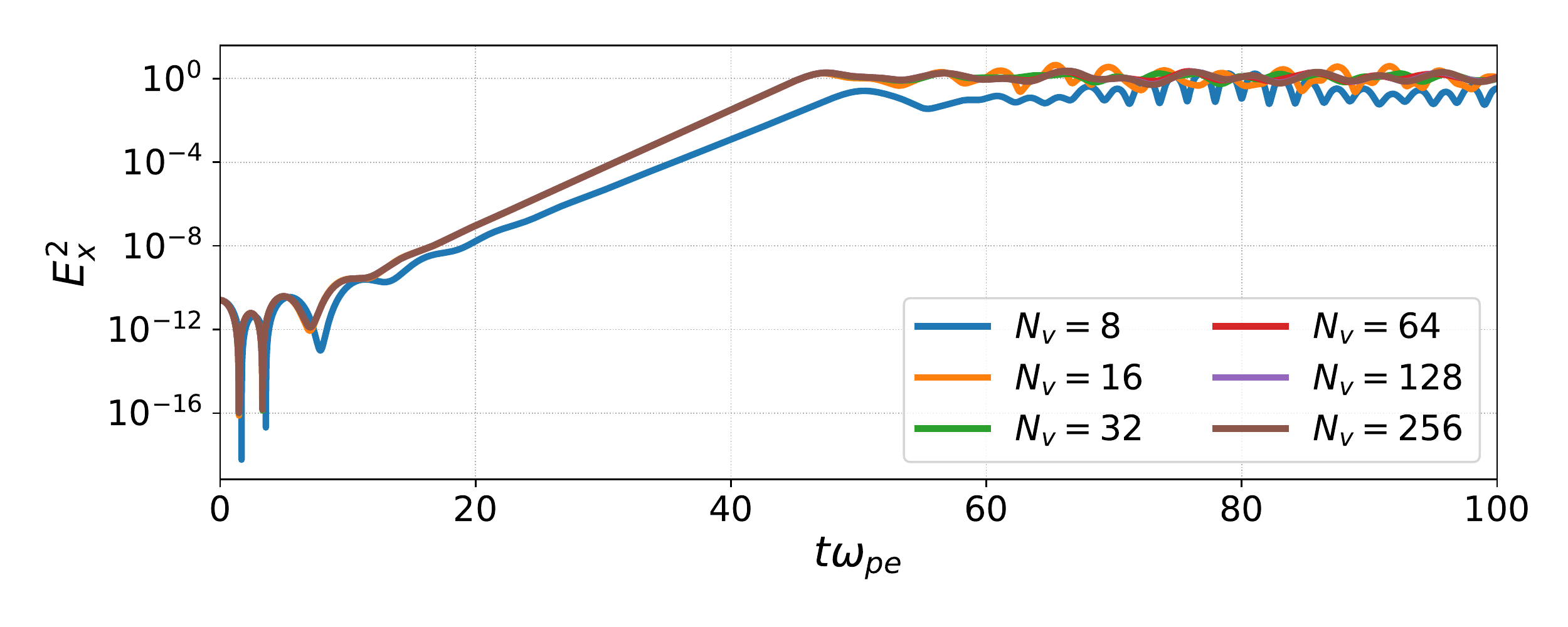}
  \caption[Two-stream instability linear growth phase convergence
    test]{Comparison of the integrated electric field energy growth in
    during two-stream instability for different velocity space
    resolutions, $N_v$ -- 8, 16, 32, 64, 128, and 256 cells.}
  \label{fig:benchmark:two-stream_convergence}
\end{figure}

\begin{table}[!htb]
 \caption[Two-stream instability growth rates based on velocity
   resolution]{Comparison of the two-stream instability growth rates,
   $\gamma$, for $k=0.5$ and $v_{th}=0.2$
   [\ref{list:benchmark:two-stream}] based on the velocity space
   resolution, $N_v$.}
 \label{tab:benchmark:two-stream_growth}
 \begin{center}
   \begin{tabular}{c|cccccc}
     \toprule 
     $N_v$ & 8 & 16 & 32 & 64 & 128 & 256 \\ \midrule
     $\gamma/\omega_{pe}$ & 0.2776 & 0.3193 & 0.3201 & 0.3199 & 0.3199
     & 0.3199 \\
     \bottomrule
   \end{tabular}
 \end{center}
\end{table}

Finally, it should be stressed that the profiles match for $N_v
\geq 16$ only during the linear growth phase but show significant
discrepancies in the non-linear part (note the logarithmic scale in
\fgr{benchmark:two-stream_convergence}).  Therefore, for example
for studies of the non-linear phase, higher velocity cell resolution is
required.

%=====================================================================
\FloatBarrier
\section{Additional Tests}

The electromagnetic benchmark with the Weibel instability ended up
with previously unreported findings and is discussed in more detail in
the next chapter.

For additional benchmarking, the reader is referred to
\cite{Juno2017} where the authors provide studies of electrostatic
shocks and  Orszag-Tang vortex.

%% file: weibel.tex
\chapter{Weibel Instability}\label{sec:weibel}

\epigraph{This quote was taken out of context.}{\textit{Randall Munroe}}

The Weibel instability\footnote{Note that especially in the regime
  when drift velocities are larger than thermal velocities, this
  instability is also referred to as the current filamentation
  instability (CFI).} (WI) \citep{Fried1959, Weibel1959} has been
studied as a leading mechanism for the origin and growth of magnetic
fields for a number of laboratory \citep{Califano1997, Okada2007,
  Silva2002, Fox2013} and astrophysical plasma \citep{Lazar2009,
  Ghizzo2017} applications.  WI can generate a large magnetic field
from no initial field and can amplify a small existing one by many
orders of magnitude.  Hence, the WI has generated a significant amount
of interest in the laboratory and astrophysics communities in recent
years and a comprehensive study of the growth and nonlinear saturation
of the WI is critical to estimate the saturated magnetic field
magnitudes that may be achieved.

Originally, our study of WI was meant as an electromagnetic benchmark,
since both Landau damping of Langmuir waves and the two-stream
instability are electrostatic problems.  However, while  simulations
of hot plasma confirm
the dominant role of the magnetic trapping during the saturation
\citep{Davidson1972}, simulations of colder beams show new
results in which an electrostatic potential develops and plays a
critical role in saturating the WI along with the magnetic
potential\footnote{Note that in this context, the magnetic potential
  does not refer to the vector potential, $\bm{A}$, but rather to the
  integral of the magnetic part of the Lorentz force, $\int q
  \left(\bm{v}\times\bm{B}\right)_x\,dx$.}  \citep{Cagas2017c}.

Furthermore, the WI is directly relevant to the plasma sheaths
discussed in the next chapter \citep{Tang2011}.

%---------------------------------------------------------------------
\FloatBarrier
\section{Description of the Instability}
\label{sec:weibel:desc}

Here we describe the physics behind WI using two counter-streaming
electron beams, similar to the two-stream instability
(Sec.\thinspace\ref{sec:benchmark:two-stream}).  In order to
excite the two-stream instability, the initial perturbation of the
unstable equilibrium must be parallel to the beam bulk velocities of
the two populations.  For WI, the perturbation is required in the
perpendicular direction. The situation is depicted in
\fgr{weibel:description}.  Initially, the two electron populations are
counter-streaming along the $y$-axis with the bulk velocities $\pm
u_y$. Note that the $\pm$ is used to label these two
populations through out this chapter.\footnote{In the figures, the `+'
  population is depicted in blue and the `-' population in orange.}
With constant uniform density along the $x$-axis, the currents
match perfectly and there are no net currents in the domain.  However,
when the system is perturbed by magnetic field $B_z$, the magnetic
part of the Lorentz force \eqrp{model:motion2}, $q\bm{v}\times\bm{B}$,
starts accelerating the beams in opposite directions.  As a result,
the beams start filamenting and the currents no longer cancel each
other.  Due to this effect, we refer to the magnetic part of the
Lorentz force as the filamentation force, $F_f^{\pm}$.  Net currents
along the $y$-axis then create a magnetic field with the same
orientation as $B_z$ \eqrp{model:ampere} and the cycle repeats.  This
positive feedback loop is the reason for the exponential growth of the
magnetic field.

\begin{figure}[!htb]
  \centering
  \includegraphics[width=0.5\linewidth]{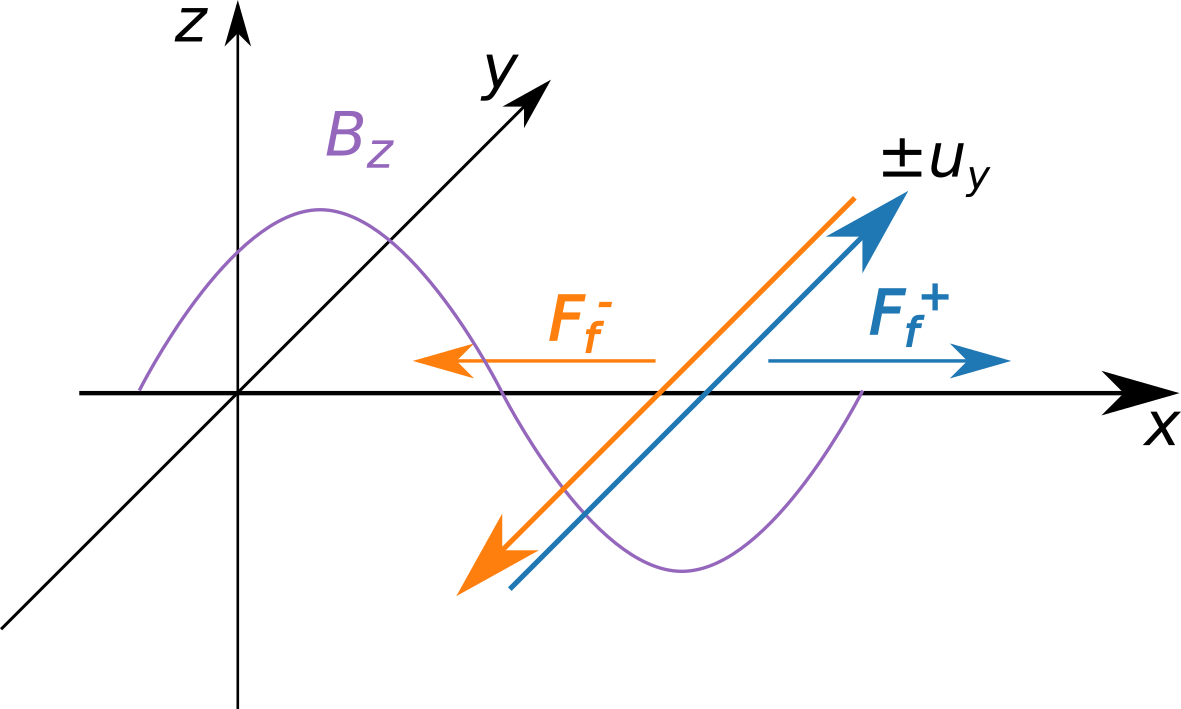}
  \caption[Drawing of the Weibel instability]{Drawing the WI setup and
    mechanism.  The $x$-axis is resolved in the simulation.  The two
    electron beams are initialized with the uniform bulk velocities
    $\pm u_y$.  Note that the colors for the $\pm$-populations are
    consistently used through this chapter.}
  \label{fig:weibel:description}
\end{figure}

It is worth noting that the initial perturbation does not need to be applied to the 
magnetic field.  Perturbing the density or bulk velocity of
one population would result in the net currents that self-consistently create the magnetic field.

As continuum kinetic simulations are noise-free, a
perturbation initialized purely in the perpendicular direction remains
perpendicular.  This allows us to focus the analysis on a single isolated
type of instability but does not reflect the real situation where the
perturbations are combined, i.e., the two-stream instability growing
together with WI.\footnote{Later in this Chapter, a secondary two-stream-like instability is discussed which growths together with WI. This secondary instability is a direct consequence of the flows introduced by the Lorentz force; in other words, it is a consequence of WI.  The two-stream instability mentioned in this paragraph is an additional electrostatic instability, which would require a perturbation along the $y$-direction.}  The analysis of their interplay is discussed, for
example, by \cite{Lazar2009} and is also a topic of current research
in the \texttt{Gkeyll} collaboration.

Finally, it is important to point out that the description of the
instability using counter-streaming beams is equivalent to the original
\cite{Weibel1959} description for a population with anisotropic
temperature.  \fgr{weibel:init} shows the initial conditions of the
simulations described in \ser{weibel:sims}.  All three cases are
initialized with Maxwellian beams, however, in the case with the
highest temperature (\fgr{weibel:init}a), the beams merge together to form a single anisotropic population.
\begin{figure}[!htb]
  \centering
  \includegraphics[width=0.9\linewidth]{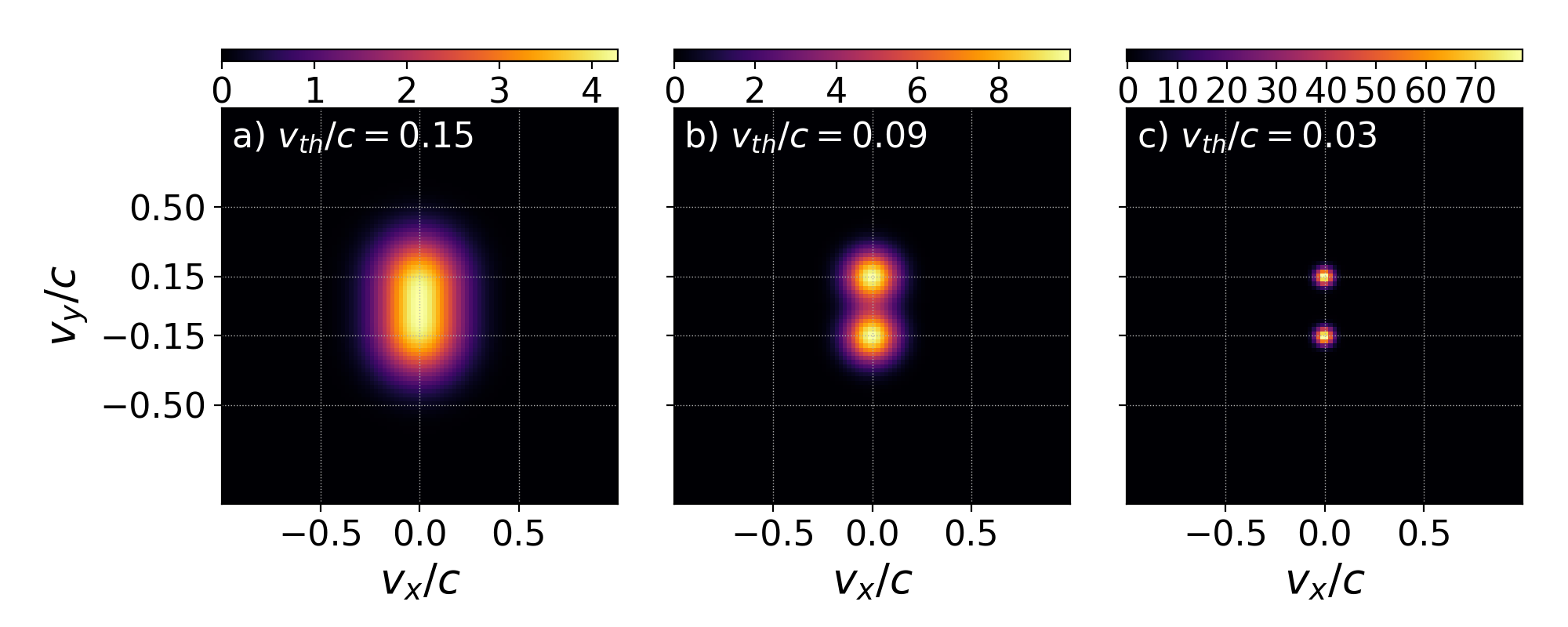}
  \caption[Initial conditions of the Weibel instability]{The initial
    conditions for the WI simulations showing the $v_x-v_y$ cuts of
    the distribution functions (distribution functions are initialized to be 
    uniform in $x$).  The bulk velocities and densities of the beams
    are fixed while the temperatures vary.}
  \label{fig:weibel:init}
\end{figure}

%---------------------------------------------------------------------
\FloatBarrier
\section{Linear Theory}
\label{sec:weibel:lintheory}

In the previous electrostatic cases
(\ser{benchmark:landau:lintheory} and
\ser{benchmark:two-stream:lintheory}), the Vlasov
equation \eqrp{model:vlasov} is linearized and combined with
Poisson's equation \eqrp{model:poisson}.  Here, in order to capture
electromagnetic effects,  \eqr{model:vlasov} is combined with the
linearized  Amp\`{e}re's \eqrp{model:ampere} law,\footnote{By
  definition,
  $$\nabla\times\bm{B}_1 =
  i\begin{pmatrix}k_yB_{z,1}-k_zB_{y,1}\\ k_zB_{x,1}-k_xB_{z,1}\\ k_xB_{y,1}-k_yB_{x,1} \end{pmatrix}.$$
  Assuming $\bm{B}_1=(0,0,B_{z,1})$ and $\bm{k} = (k_x,0,0)$, we can
  limit the discussion only to the $y$-component.}
\begin{align}\label{eq:weibel:linampere}
  -ik_xB_{z,1} = \mu_0q \left(\int v_yf_{1}^{+}\,d\bm{v} +
  \int v_yf_{1}^{-}\, d\bm{v}\right) - \frac{i\omega}{c^2}
  E_{y,1}.
\end{align}
Elimination of the fields\footnote{This requires a few algebraic steps
  and tricks.  The full process is in the Appendix \ref{app:weibel}.}
from \eqr{weibel:linampere} gives the following kinetic dispersion
relation,
\begin{align}\label{eq:weibel:disp}
 \frac{1}{2} - \frac{\omega_{pe}^2}{c^2k_x^2} \left[\zeta
  Z(\zeta) \left(1+\frac{u_y^2}{v_{th}^2}\right)
  + \frac{u_y^2}{v_{th}^2}\right] - \frac{v_{th}^2}{c^2}\zeta^2 = 0,
\end{align}
where
\begin{align*}
  \zeta = \frac{\omega/k_x}{\sqrt{2v_{th}^2}}.
\end{align*}

In order to compare with the literature, the cold fluid limit can be
calculated using the asymptotic expansion of $Z(\zeta)$ for large
$\zeta$, $|\zeta|\gg1$ \citep{Huba2004},
\begin{align*}
  Z(\zeta) = i\sqrt{\pi} \sigma \exp(-\zeta^2) - \zeta^{-1} \left(1 +
  \frac{1}{2\zeta^2} + \frac{3}{4\zeta^4} + \frac{15}{8\zeta^6} +
  \ldots \right), \quad \sigma =
  \begin{cases}
    0 & \gamma > |\omega_r|^{-1}\\
    1 & |\gamma| < |\omega_r|^{-1}\\
    2 & \gamma < -|\omega_r|^{-1}
    \end{cases}.
\end{align*}
Neglecting HOT gives
\begin{align*}
  \frac{1}{2} - \frac{\omega_{pe}^2}{c^2k_x^2} \left[
    \zeta\left(i\sqrt{\pi} \exp(-\zeta^2) -
    \frac{1}{\zeta}-\frac{1}{2\zeta^2}\right) \left(1 +
    \frac{u_y^2}{v_{th}^2}\right) + \frac{u_y^2}{v_{th}^2} \right] -
  \frac{v_{th}^2}{c^2}\zeta^2 &\approx 0, \\
  \frac{1}{2} - \frac{\omega_{pe}^2}{c^2k_x^2} \left(-1 -
  \frac{u_y^2}{v_{th}^2} - \frac{2v_{th}^2}{2\omega^2/k_x^2} -
  \frac{2u_y^2v_{th}^2}{2v_{th}^2\omega^2/k_x^2} +
  \frac{u_y^2}{v_{th}^2} \right) - \frac{\omega^2/k_x^2}{2c^2} &\approx
  0.
\end{align*}
Finally, rearranging the terms leads to
\begin{align}\label{eq:weibel:cold-dispersion}
  \frac{\omega^4}{2\omega_{pe}^2k_x^2c^2} -
    \left(\frac{1}{2\omega_{pe}^2} - \frac{1}{c^2k_x^2}\right)\omega^2
    - \frac{u_y^2}{c^2} \approx 0,
\end{align}
which corresponds exactly to the Eq.\thinspace12 in
\cite{Califano1997}.

\begin{figure}[!htb]
  \centering
  \includegraphics[width=0.8\linewidth]{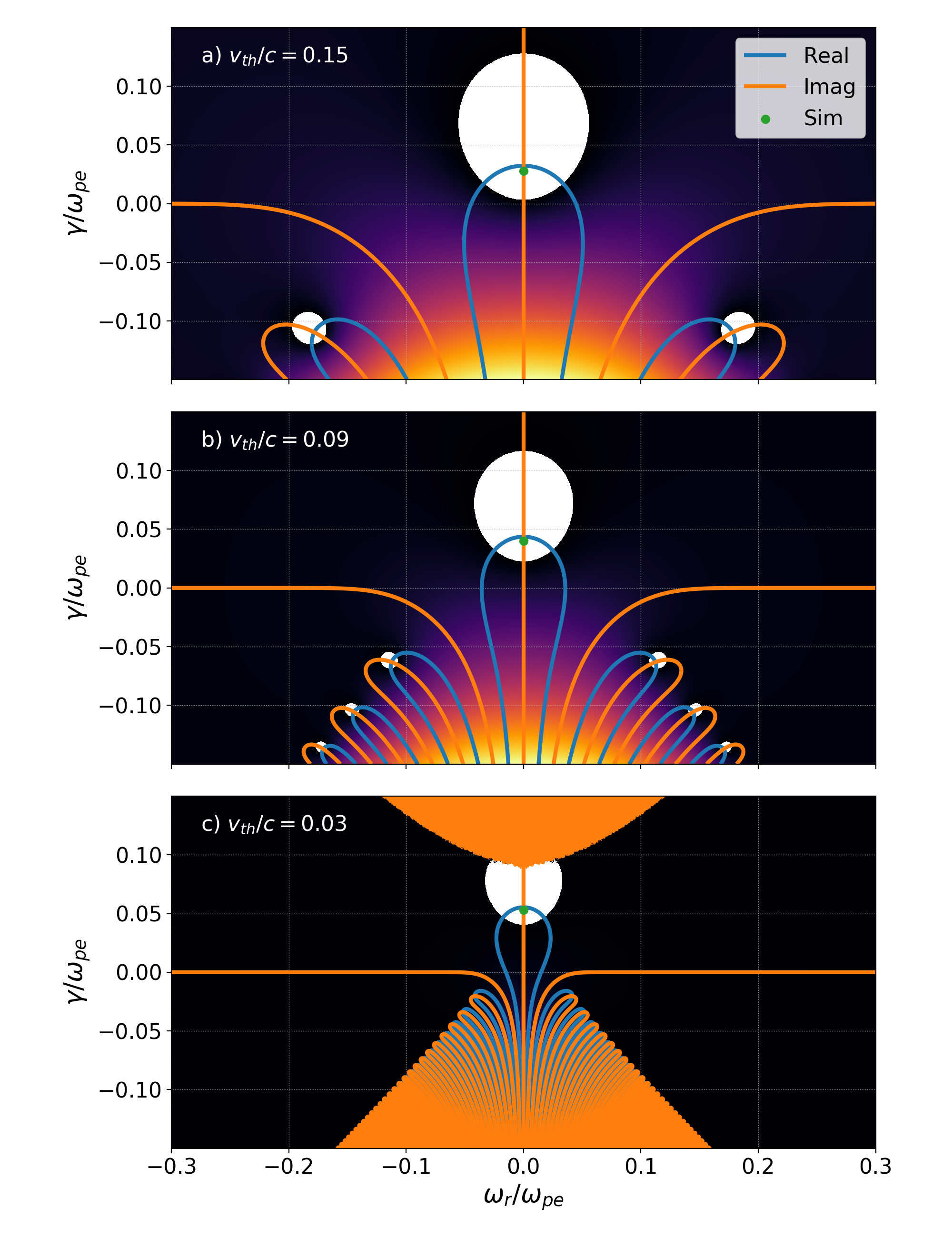}
  \caption[Roots of the Weibel instability dispersion relation]{Plots
    of the WI for three thermal velocities, $v_{th}=$ 0.3, 0.1, and
    0.04.  In all the cases, there is a single growing mode together
    with family of dampened oscillatory modes.  Note that in this
    range, the growth rate is decreasing with the temperature of the
    beams.}
  \label{fig:weibel:disp}
\end{figure}

%---------------------------------------------------------------------
\FloatBarrier
\section{Numerical Simulations}
\label{sec:weibel:sims}

Following the discussion in \ser{weibel:desc}, the WI can be captured
in a 1X2V continuum kinetic simulation where we resolve $x$ and evolve
the $v_x$ and $v_y$ velocity components.  The situation corresponds to
\fgr{weibel:description}.  The full listing of one of the simulations
used in this chapter is in Appendix~\ref{list:weibel:weibel}. As shown
in \fgr{weibel:init}, the the simulations are run with three different
initial thermal velocities, while the other parameters (density, bulk
velocity, and the initial perturbation) are kept the same.  The first
case uses $v_{th} = u_y$, which is relevant to the classical WI
configuration, while the other two feature distinct electron streams
with $v_{th} < u_y$.

The initial uniform but unstable equilibrium is disrupted with a
perturbation in $B_z$, given by
\begin{align*}
  B_z(x) = A \sin(k_xx),
\end{align*}
while all the other quantities are uniform in $x$.  The $A$ used in
this work is $10^{-4}$.  $E_x$, $E_y$, $E_z$, $B_x$, and $B_y$ are all
initialized to zero.  The configuration space is chosen to be
periodic, spanning from 0 to $2\pi/k_x$, therefore, a single period of
the initial perturbation is captured.

\subsection{Linear Growth}

Plots of the kinetic dispersion relation \eqrp{weibel:disp} in
\fgr{weibel:disp} are reminiscent of the two-stream instability
dispersion (\fgr{benchmark:two-stream_disp}).

Similar to \ser{benchmark:two-stream:sims}, a ``sweeping fit''
is required to obtain growth rates from the simulation data in a
reproducible way.  The fits are in \fgr{weibel:energy_h},
\fgr{weibel:energy_m}, and \fgr{weibel:energy_l}. The extracted growth
rates are plotted as a green dots on top of the dispersion relation
plots (\fgr{weibel:disp}) and also listed in \tbr{weibel:growths}.  In
order to compare the rates to the theory quantitatively, the
Newton-Raphson root finding algorithm is used again.  For that we
require the derivatives,
\begin{gather*}
 F(\omega) = \frac{1}{2} - \frac{\omega_{pe}^2}{c^2k_x^2}
 \left[\zeta(\omega) Z\big(\zeta(\omega)\big)
   \left(1+\frac{u_y^2}{v_{th}^2}\right) +
   \frac{u_y^2}{v_{th}^2}\right] -
 \frac{v_{th}^2}{c^2}\zeta^2(\omega), \\ \pfrac{F(\omega)}{\omega} = -
 \frac{\omega_{pe}^2}{\sqrt{2}v_{th}c^2k_x^3}
 \left(1+\frac{u_y^2}{v_{th}^2}\right)
 \left[ Z\big(\zeta(\omega)\big)+\zeta(\omega)Z'\big(\zeta(\omega)\big)\right]
 - \frac{2v_{th}\zeta(\omega)}{\sqrt{2}v_{th}c^2k_x}.
\end{gather*}

\tbr{weibel:growths} summarizes the simulation growth rates and the
theoretical predictions for $k=0.4$ and $v_{th}/c=0.15$, $0.09$, and
$0.03$. Note that the discrepancy between the simulation and theory is
much bigger in comparison to Landau damping and the two-stream
instability. The  explanation for this is offered at the end of this
Chapter in \ser{weibel:phase}.

\begin{table}[!htb]
 \caption[Weibel instability growth rates]{Comparison of the WI growth
   rates, $\gamma$, for $k=0.4$ and $v_{th}/c=0.15$, $0.09$, and
   $0.03$ [\ref{list:weibel:weibel}].  Theoretical values are
   calculated with the Newton method from \eqr{weibel:disp} while the
   simulation results are obtained using the ``sweeping fit''.}
 \label{tab:weibel:growths}
 \begin{center}
   \begin{tabular}{c|ccc}
     \toprule 
     $v_{th}/c$ & 0.15 & 0.09 & 0.03 \\ \midrule
     $\gamma_{sim}/\omega_{pe}$ & 0.0278 & 0.0402 & 0.0529  \\
     $\gamma_{theor}/\omega_{pe}$ & 0.0322 & 0.0436 & 0.0553  \\
     \bottomrule
   \end{tabular}
 \end{center}
\end{table} 

\begin{figure}[!htb]
  \centering
  \includegraphics[width=0.8\linewidth]{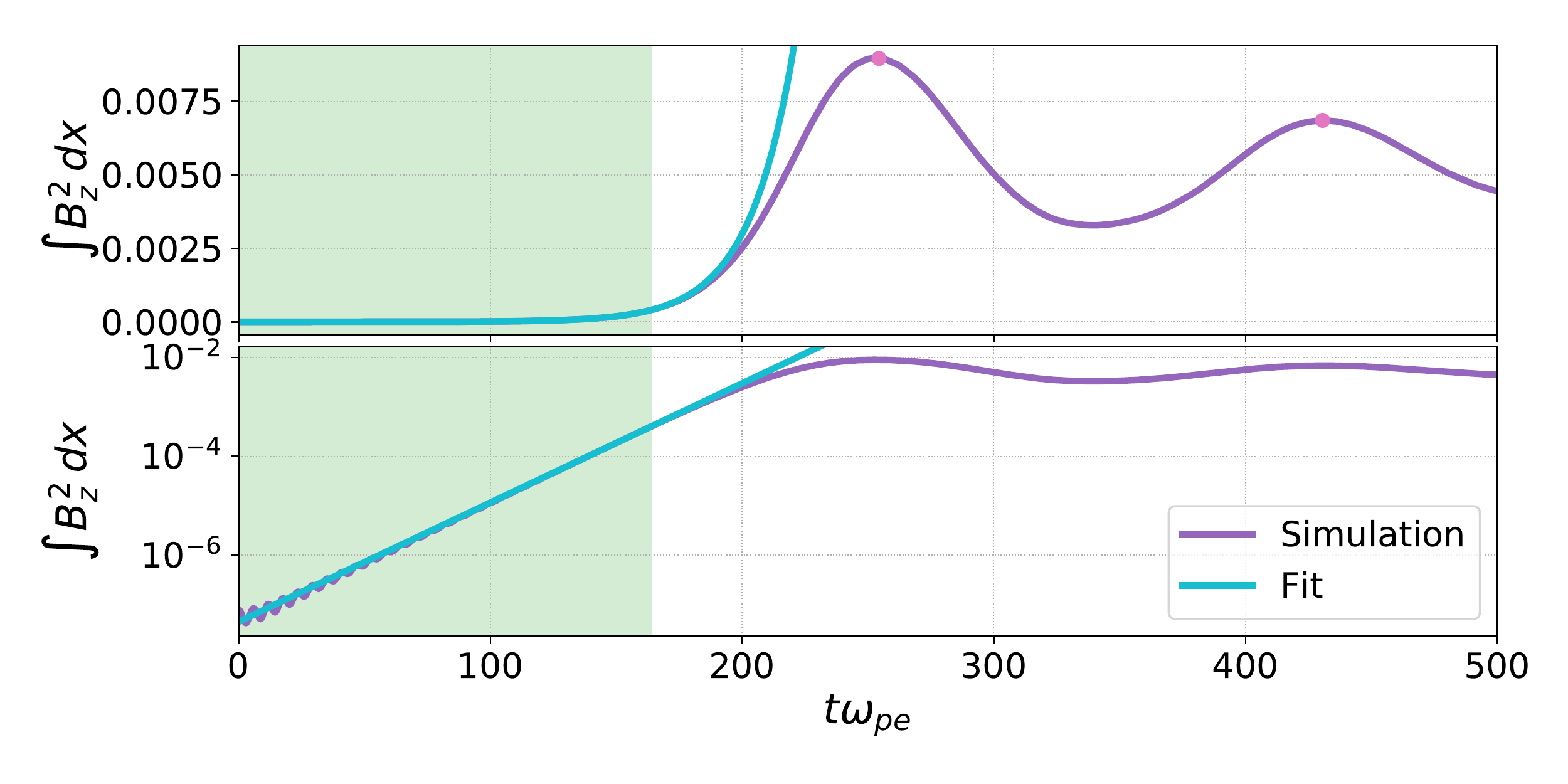}
  \caption[Total energy evolution in WI (high temperature)]{Linear and
    semi-logarithmic plots of the magnetic field energy proxy,
    $B_z^2$, evolution during WI (high temperature; $v_{th}/c=0.15$).
    Figures show the simulation data together with the best result of
    the ``sweeping fit''.  The region of the best fit is highlighted
    with green shading.  In the nonlinear phase, the energy maxima are
    highlighted with pink dots.}
  \label{fig:weibel:energy_h}
\end{figure}
\begin{figure}[!htb]
  \centering
  \includegraphics[width=0.8\linewidth]{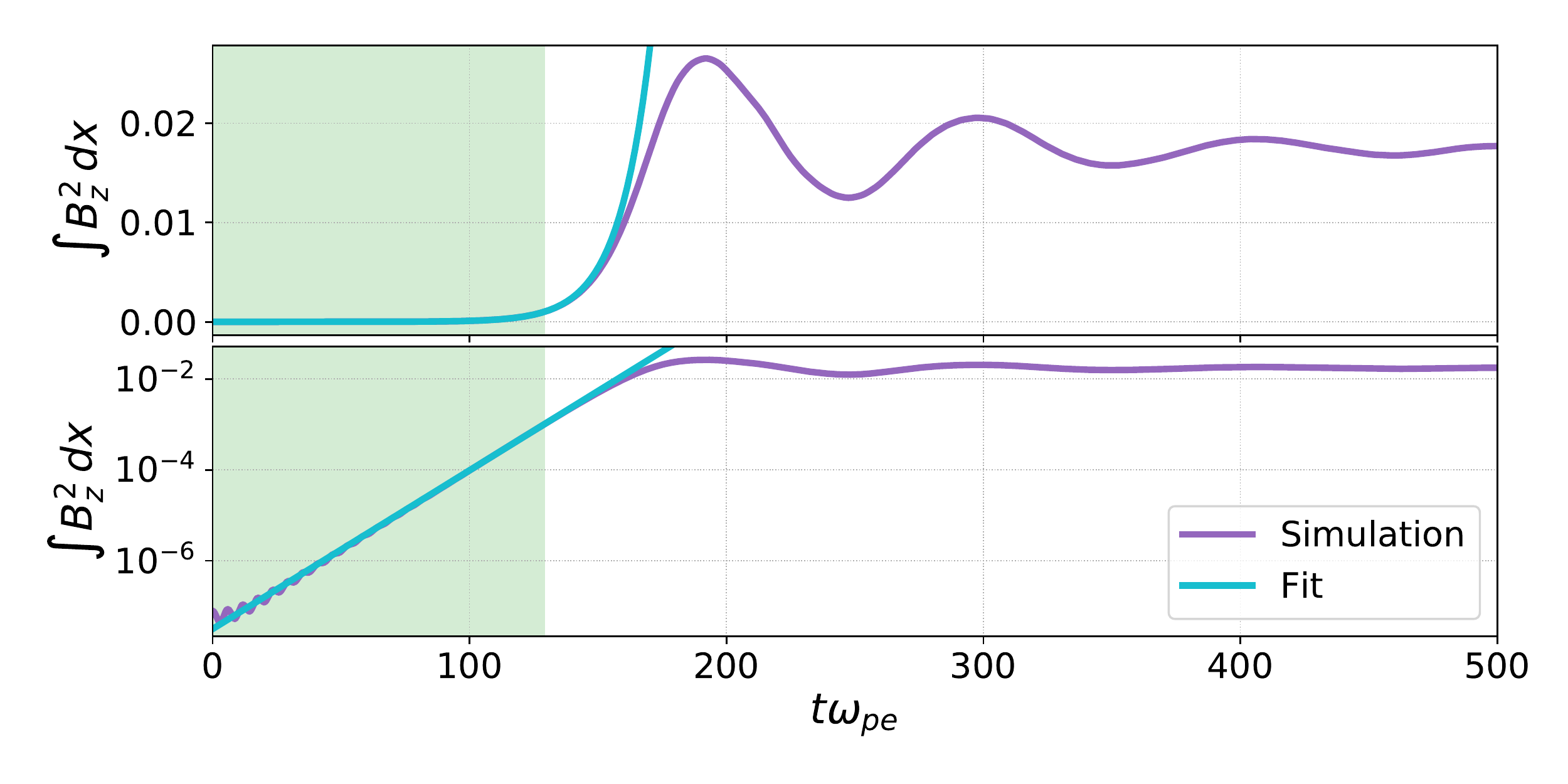}
  \caption[Total energy evolution in WI (intermediate
    temperature)]{Linear and semi-logarithmic plots of the magnetic
    field energy proxy, $B_z^2$, evolution during WI (intermediate
    temperature; $v_{th}/c=0.09$).  Figures show the simulation data
    together with the best result of the ``sweeping fit''.  The region
    of the best fit is highlighted with green shading.}
  \label{fig:weibel:energy_m}
\end{figure}
\begin{figure}[!htb]
  \centering
  \includegraphics[width=0.8\linewidth]{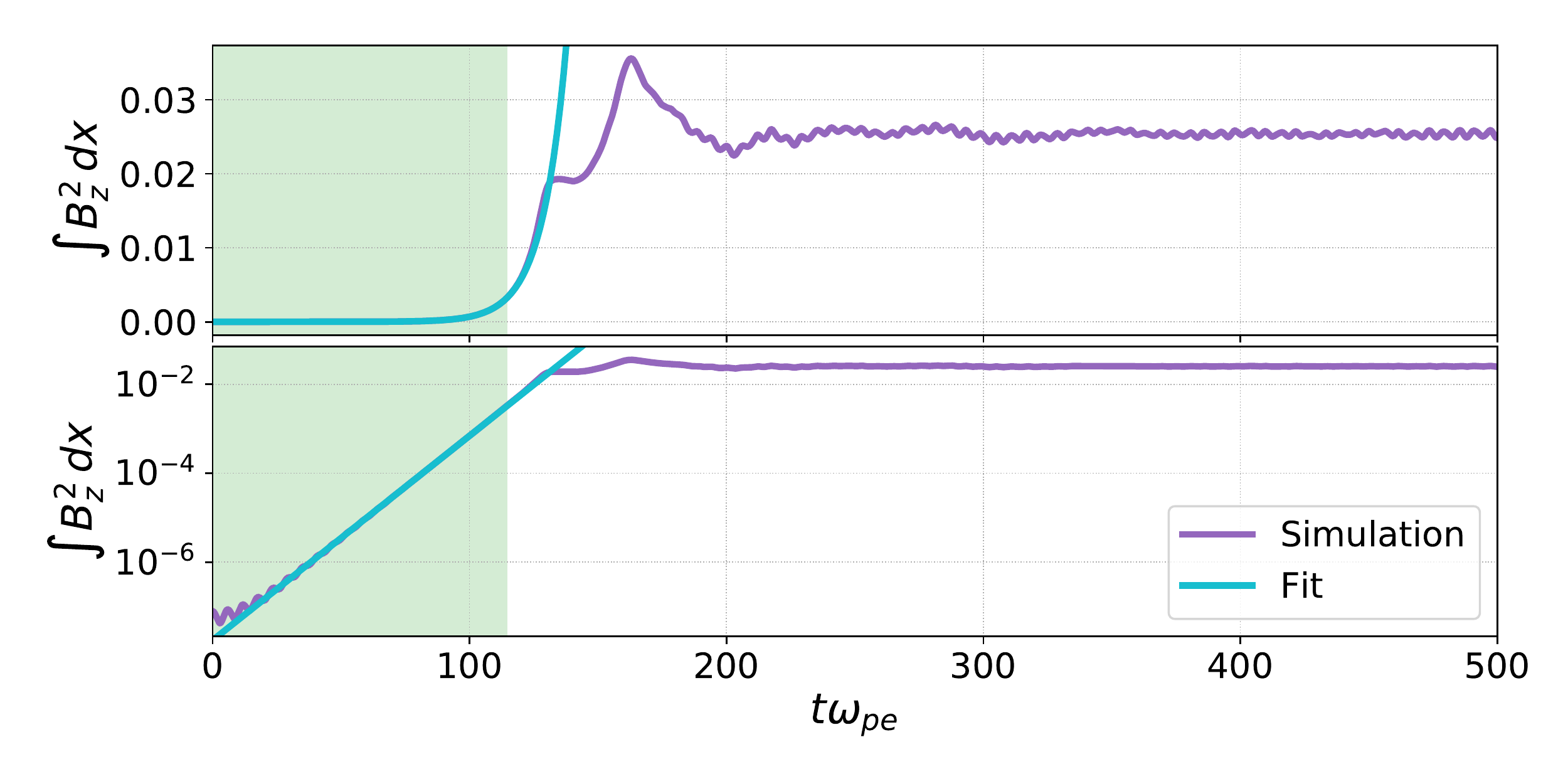}
  \caption[Total energy evolution in WI (low temperature)]{Linear and
    semi-logarithmic plots of the magnetic field energy proxy,
    $B_z^2$, evolution during WI (low temperature; $v_{th}/c=0.03$).
    Figures show the simulation data (blue) together with the best
    result of the ``sweeping fit'' (orange).  The region of the best
    fit is highlighted with green shading. [Simulation input file:
      \ref{list:weibel:weibel}]}
  \label{fig:weibel:energy_l}
\end{figure}

\FloatBarrier
\subsection{Nonlinear Saturation}
\label{sec:weibel:saturation}

In order to understand the saturation, it is necessary to discuss the
evolution of the WI in more detail.  \fgr{weibel:description} in \ser{weibel:sims} shows the
 effects of the filamentation force on each beam; `+'
population is accelerated in positive $x$-direction and `-' population
in the negative.  However, this is only true in the region where
$B_z<0$. \fgr{weibel:description2} extends the description to the region with $B_z>0$.  In this complete picture,\footnote{Note that the
  domain is periodic.} the effects of the filamentation force result in particles with $u_y <0$ building up with one part of the domain and particles with $u_y> 0$ in the other part.

\begin{figure}[!htb]
  \centering
  \includegraphics[width=0.6\linewidth]{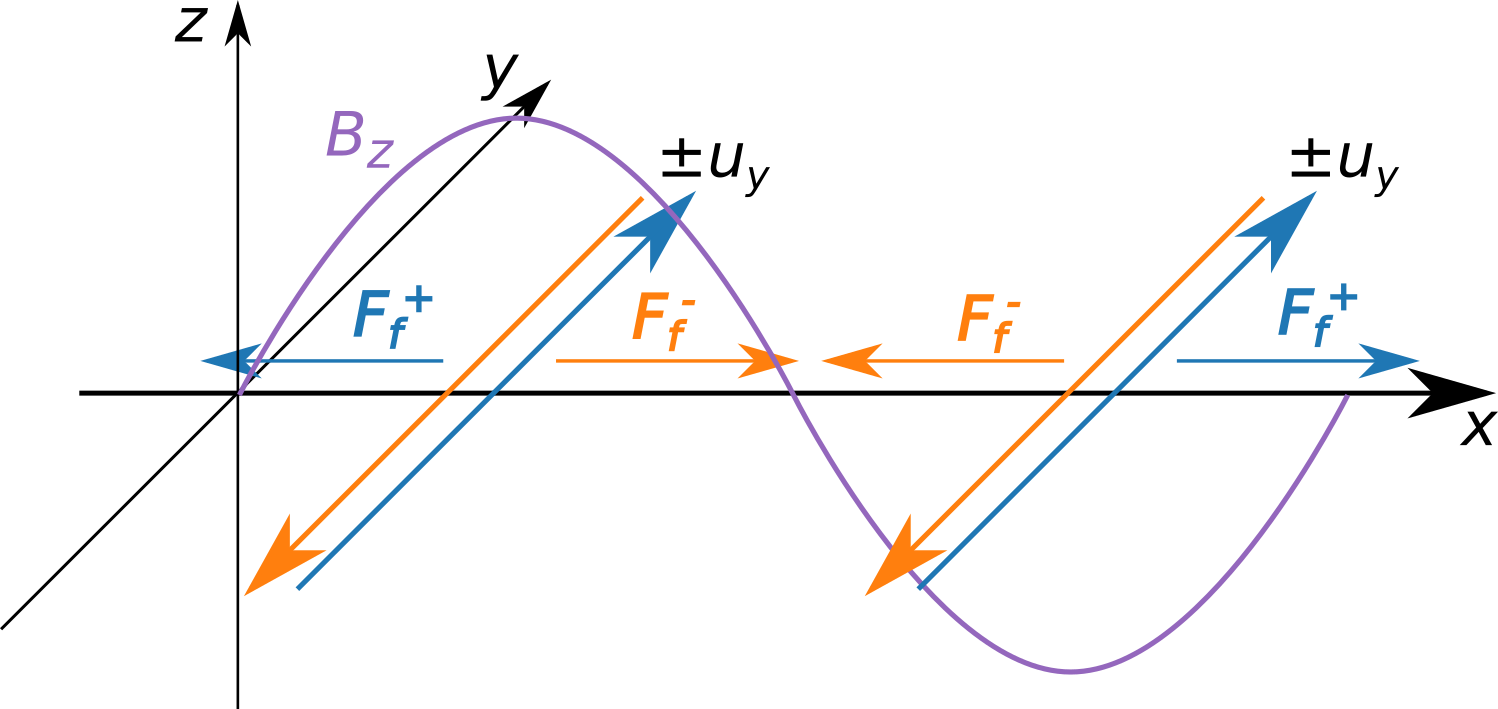}
  \caption[Extended drawing of the Weibel instability]{Extension of
    \fgr{weibel:description} to the full configuration space.  Note
    that the `-' population (orange) is accelerated to the middle
    while the `+' particles are accelerated towards the ``edges'' (the
    domain is periodic in $x$).}
  \label{fig:weibel:description2}
\end{figure}

The filamentation is eventually stopped when the particles become
trapped in the potential wells that form and the instability saturates.
This trapping can be seen in the periodic behavior of the nonlinear
phase of the instability (\fgr{weibel:energy_h}).  After the local
maxima are found (pink dots in \fgr{weibel:energy_h}), the frequency can
be estimated as
\begin{align*}
  \frac{\omega}{\omega_{pe}} \approx \frac{2\pi}{176} \approx 0.036.
\end{align*}
This compares well to the theoretical magnetic bounce frequency
\citep{Davidson1972}
\begin{align}\label{eq:weibel:bounce}
  \frac{\omega_B}{\omega_{pe}} = \sqrt{k\frac{q}{m}u_yB_z} \approx
  0.035,
\end{align}
when the values of $k=0.4$, $u_y=0.15$, and $B_z=0.02$ are used.  Note
that this calculation provides only an estimate, because it is, for
example, unclear what $B_z$ and $u_y$ to use as these are functions of
$x$.  Still, the close agreement with \cite{Davidson1972} provides an
indication that magnetic trapping is the primary mechanism for WI
saturation in this regime.

The spatio-temporal evolution is summarized in
\fgr{weibel:evolution_h}.  The top panel is a copy of
\fgr{weibel:energy_h} to provide a context for time.  The depicted density
is the total density in the simulation, i.e., the sum of both
populations, and does not provide particular insight.  Note that the
magnetic field follows the initialized profile
(\fgr{weibel:description2}) during the evolution.
\begin{figure}[!htb]
  \centering
  \includegraphics[width=0.7\linewidth]{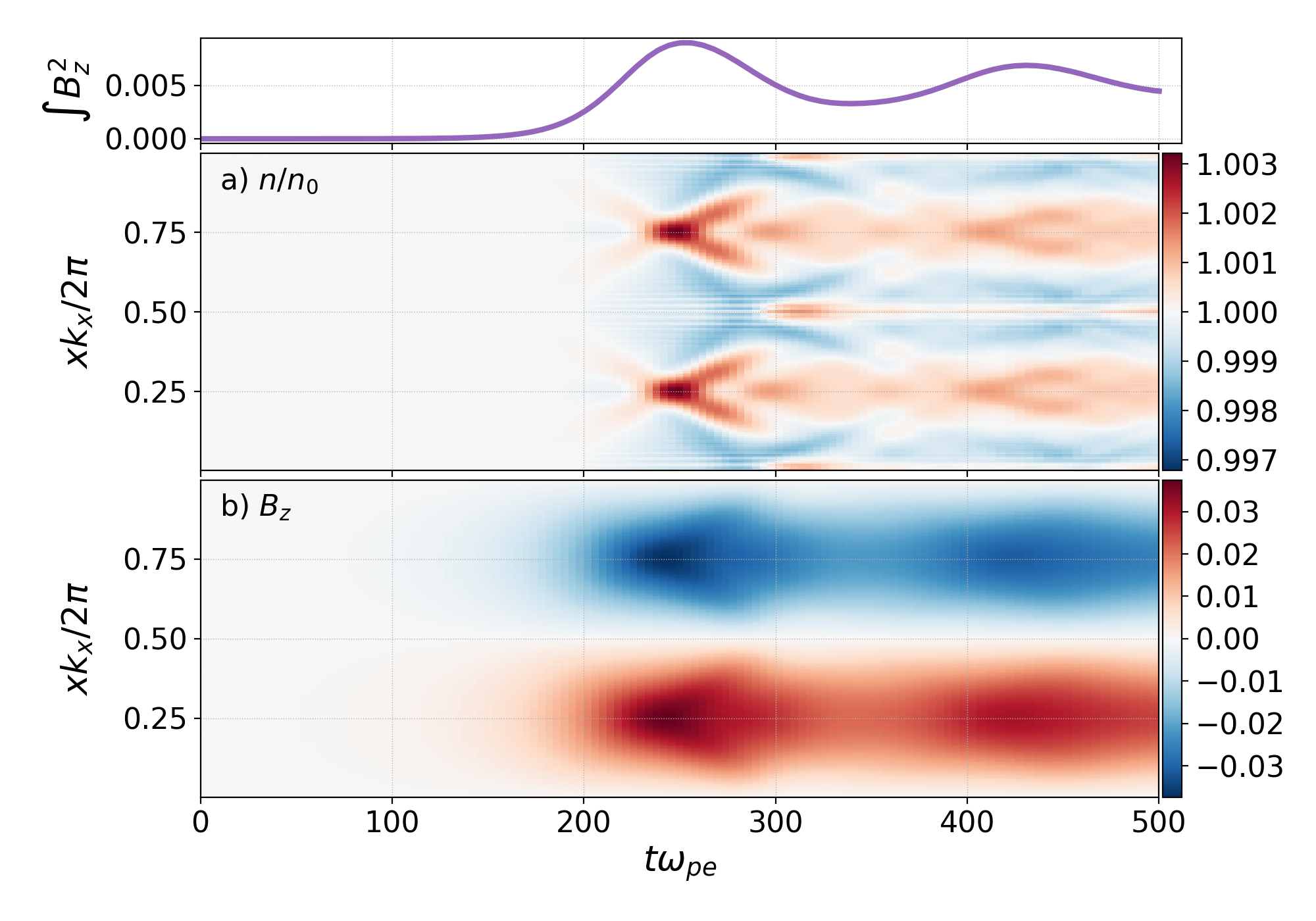}
  \caption[Evolution of density and magnetic field profiles (high
    temperature case)]{Evolution of density and magnetic field
    profiles during WI (high temperature case; $v_{th}/c=0.15$). The
    top panel is a copy of \fgr{weibel:energy_h} to provide time
    context.  Panel (a) is the evolution of density; note that the
    difference from the initial condition is highlighted.  Panel (b)
    depicts the formation of the magnetic field.}
  \label{fig:weibel:evolution_h}
\end{figure}

The low temperature case (see \fgr{weibel:init}c and
\fgr{weibel:energy_l}) behaves differently.  While the magnetic field
energy in \fgr{weibel:energy_h} saturates and oscillates with the
magnetic bounce frequency, the energy in the low temperature case
(\fgr{weibel:energy_l}) undergoes the saturation in two steps.  The
profile flattens around $t\omega_{pe}\approx 130$ and then begins
growing again.  Corresponding to this two step saturation is also the
growth of higher modes in the Fourier spectra\footnote{There are many
  online resources for the Fourier transformation; I would recommend
  the \texttt{scipy.fftpack} tutorial
  \url{https://docs.scipy.org/doc/scipy/reference/tutorial/fftpack.html}.
  Note, however, that unlike in this reference, the spectra presented
  here are not normalized.  Instead, absolute values corresponding to
  the real frequencies are plotted (the zero frequency is disregarded
  as well).}  (see \fgr{weibel:fft}).  Note the negligible
contribution of the higher modes in the high temperature case.  These
results suggest additional processes beyond
magnetic trapping.
\begin{figure}[!htb]
  \centering
  \includegraphics[width=0.7\linewidth]{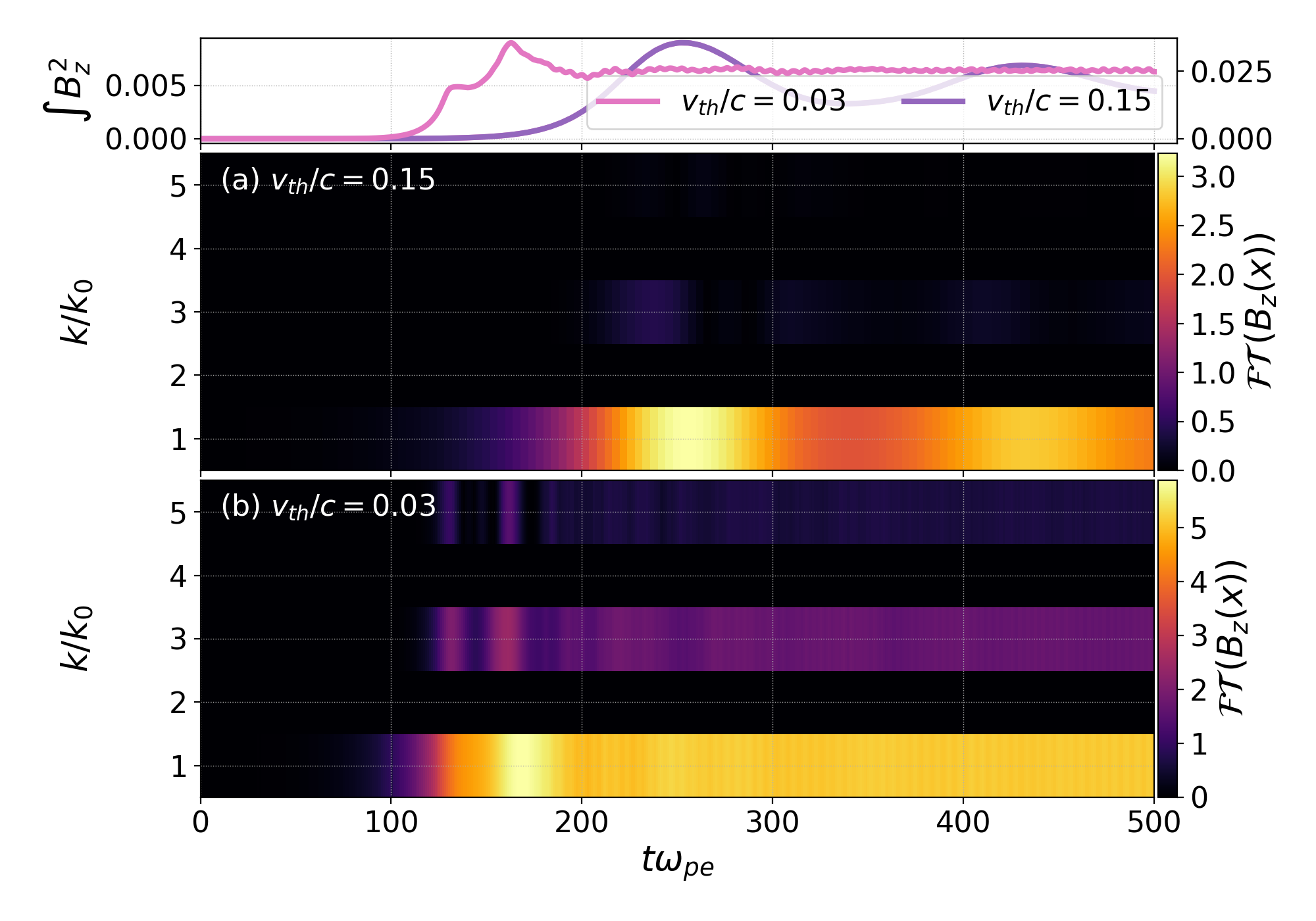}
  \caption[Evolution of the Fourier spectra of the magnetic
    field]{Evolution of the Fourier spectra of the magnetic field for
    the high temperature case (a) and the cold beam case (b).  Note
    that the higher modes, normalized to the initial perturbation, are
    significant only in the later one.}
  \label{fig:weibel:fft}
\end{figure}

The nonlinear phase magnetic
energy for the lower temperature case (\fgr{weibel:energy_l})
consists of small oscillations superimposed over the lower frequency
signal, in contrast to the magnetic bounce frequency of
\fgr{weibel:energy_h}.  This is confirmed by a temporal Fourier
transformation of the energy (\fgr{weibel:nonlinear}) for
$t\omega_{pe}>200$.  The first peak, roughly for frequency
$\omega/\omega_{pe} \approx 0.06$, corresponds well with the bounce
frequency \eqrp{weibel:bounce}.  But more interestingly, there is also
a strong signal around the plasma oscillation frequency (responsible for
the fine structure in \fgr{weibel:energy_l}), which possibly suggests
an effect of the electric field.
\begin{figure}[!htb]
  \centering
  \includegraphics[width=0.7\linewidth]{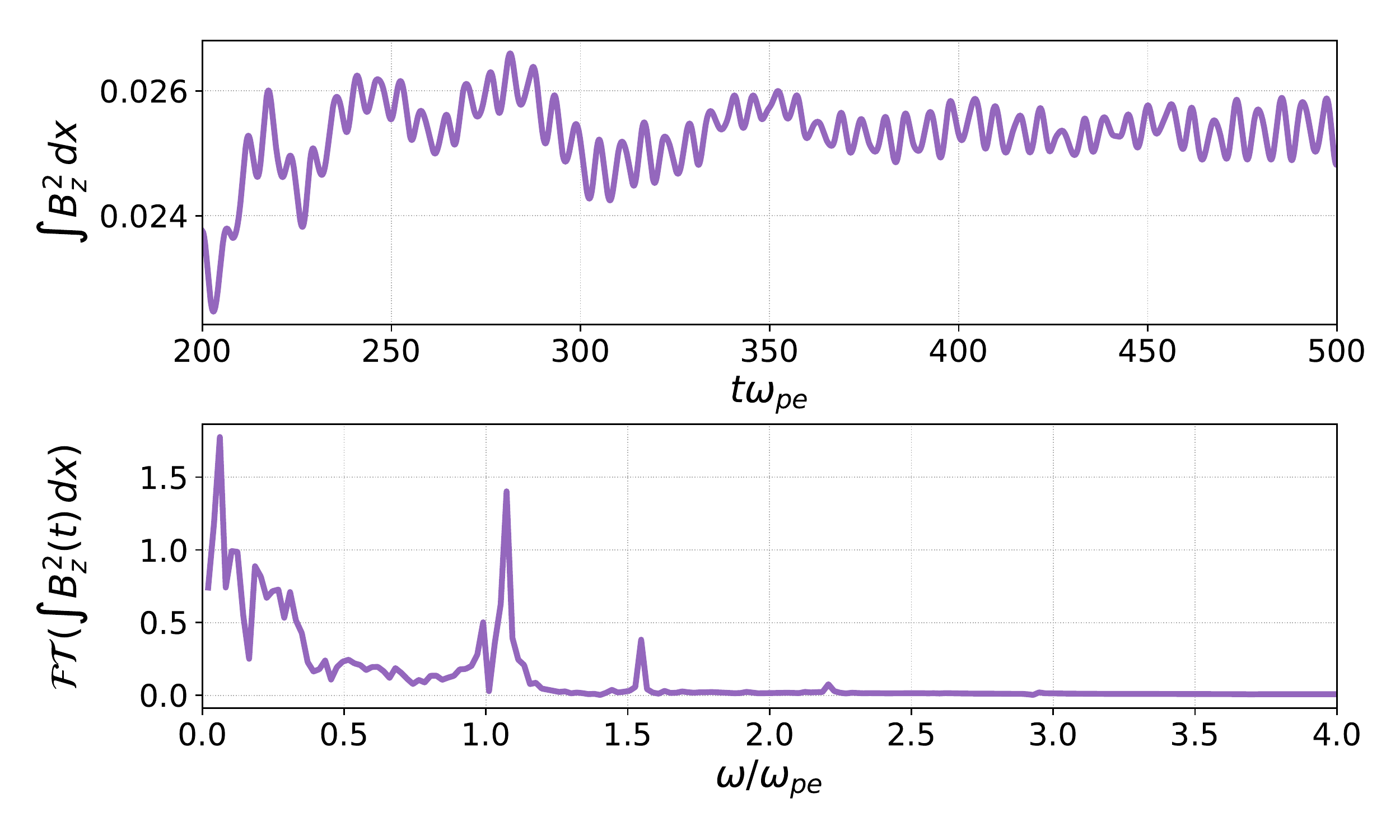}
  \caption[Temporal FFT of nonlinear magnetic field energy (low
    temperature case)]{Magnetic field energy evolution in the
    nonlinear phase ($t\omega_{pe}>200$) for the low temperature case
    ($v_{th}/c=0.15$).  Top panel shows the detail of
    \fgr{weibel:energy_l} for the relevant time interval.  The bottom
    panel shows the corresponding Fourier transformation (absolute
    value) normalized to the plasma frequency.}
  \label{fig:weibel:nonlinear}
\end{figure}

As mentioned previously, calculating the number density
simply as the moment of the distribution \eqrp{model:m0} gives the
total density, which does not provide any particularly interesting insight in this case
(note that the change from the initial uniform density in
\fgr{weibel:evolution_h} is on the order of 0.3\%).  For the next
step, it is necessary to divide the distribution function into the `+'
and `-' populations.  This allows for study of individual densities
and enables calculating the average Lorentz force acting on
each population.  A careful method needs to be constructed to separate the populations.
Taking a lineout along $v_y =0$, i.e.,
defining the `+' populations as $f(x,v_x,v_y>0)$, does not work
because the tail of the population extends into the $v_y < 0$ domain
as the particles slow down (see \fgr{weibel:division}), significantly
distorting the results.  Another option is to initialize the
simulation with two distribution functions rather than one. However,
this increases the computational cost as the Vlasov equation
\eqrp{model:weak2} needs to be solved twice. What is more, as the
populations can undergo mixing, they should be defined by the current
state rather than the initial conditions.  Alternatively, a third
option can be used. First, the $x-v_y$ profile is obtained by
integrating the distribution over $v_x$,\footnote{This is common
  practice for phase space data in order to limit the number of
  dimensions.} $\int f(x,v_x,v_y)\,dv_x$. Then, 1D slices for each $x$
can be fit to a double Maxwellian,
\begin{align*}
  f(v_y) = \frac{n_1}{\sqrt{2\pi v_{th,1}^2}}
  \exp\left(-\frac{(v_y-u_1)^2}{2v_{th,1}^2}\right) + \frac{n_2}{\sqrt{2\pi
      v_{th,2}^2}} \exp\left(-\frac{(v_y-u_2)^2}{2v_{th,2}^2}\right),
\end{align*}
and the boundary between he populations is selected as the $v_y$ between the $u_1$ and $u_2$
for which the $f(v_y)$ has the minimum.\footnote{Note that this
  process is only guaranteed to work for the discrete beam case. While
  the fitting works for the high temperature case as well, the
  distribution function does not necessarily have a minimum in
  $(u_1,u_2)$.}  Then the moments can be calculated using only a
subset of $v_y$ space which provides the densities and velocities for
the two populations separately.  \fgr{weibel:division} shows the
boundary for the distribution function at $t\omega_{pe} = 90$ and the
corresponding densities of the `$\pm$' populations.
        
\begin{figure}[!htb]
  \centering
  \includegraphics[width=0.7\linewidth]{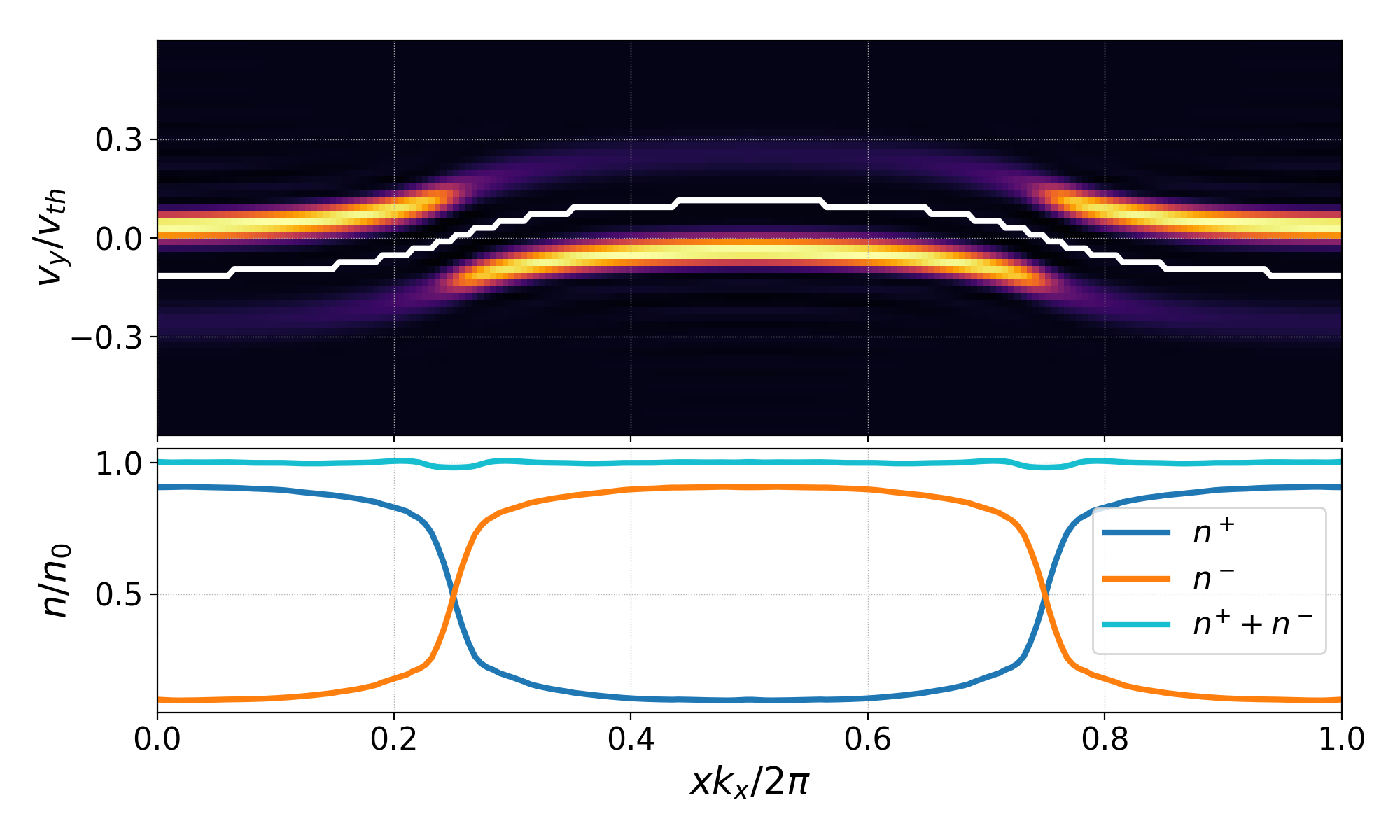}
  \caption[Separation of the $u_y>0$ and $u_y<0$ populations]{Example
    of the distribution function integrated over the $v_x$, $\int
    f(x,v_x,vy)\,dv_x$, i.e., $x-v_y$ profile (top panel;
    $v_{th}/c=0.03$ and $t\omega_{pe}=90$). The profile is fitted by
    double Maxwellian for each $x$ and the local minimum between the
    peaks is selected as the boundary (highlighted in the top
    panel). The bottom panel shows the densities of the two
    populations calculated as moments of the distribution function
    above and below the boundary in $v_y$.  Note that the sum of
    densities does not significantly differ from the initial density,
    $n_0$.}
  \label{fig:weibel:division}
\end{figure}

With information about the moments of the individual populations,
we can take a closer look at the nonlinear saturation process for the
low temperature case ($v_{th}/c=0.03$).  A detailed picture, capturing a
subset of \fgr{weibel:energy_l} limited to $70<t\omega_{pe}<150$, is
in \fgr{weibel:uber}.  As usual, the top panel contains the energy
evolution in order to provide the time reference.  However, unlike in
the previous figures, it also includes the evolution of the electric
field energy proxy, $\int E_x^2\,dx$, which is negligible in the
higher temperature cases.  Note that the electric field energy peaks
at the first saturation point of the magnetic field energy.  Panels
(b) and (c) show the electric, $E_x$, and magnetic, $B_z$, field
profiles, respectively.  Their maximal magnitudes are comparable but,
while the magnetic field roughly keeps the initialized profile, the
electric field is located at the ``boundaries'' between the two
populations (see \fgr{weibel:division}).  The panel (d) shows just the
sign of the $x$-component of the averaged total force acting on the
`+' population, i.e., $F_x(x) = q(E_x+u_y^+B_z)$.  The effects of the
force are then reflected in panel (e) which shows the bulk
$x$-velocity profile.  Panel (f) illustrates the filamentation of the
density, which is initially uniform in $x$ and is later limited to half of
the domain.  Finally, panel (g) shows the evolution of the bulk
$y$ velocity, $u_y^+$.  Note that, similar to the two-stream
instability, the velocity is decreasing (in the regions with
particles) as the kinetic energy is converted into the magnetic field
energy during the course of the instability.

The first part of the evolution corresponds exactly to the description
in \fgr{weibel:description2}.  Due to $B_z$ and $u_y$ the
filamentation force, $F_f^+=qu_y^+B_z$, forms and is positive for
$x\in\left(0.5\frac{2\pi}{k_x},\,\frac{2\pi}{k_x}\right)$ and negative
in the other half. This force accelerates particles along the $x$
direction, causing the filamentation of the density and, consequently,
the increase of the net currents in the domain and the growth of the
magnetic field.  In the high temperature case, this occurs until
the vast majority of particles are trapped in their respective magnetic potential wells.  However, in this case, there is an additional
effect, as the previous analysis suggests.  The strong currents at
the ``boundaries'' between the two populations ($xk_x/2\pi=0.25$ and
0.75) are the source of the $E_x$.  The orientation of $E_x$ is such that
it decreases the flows, and it eventually grows enough to dominate the
filamentation force (\fgr{weibel:uber}d).  The ``filamentation flow''
then quickly stops (\fgr{weibel:uber}e) and the growth of $B_z$
saturates.  However, without the currents the $E_x$ decays at the
boundary and the instability restarts.  Note that right after the
saturation, the electric field does not just stop the flow but
reverses it in some small regions, resulting in a slight decrease of the magnetic field energy while also halting the instability growth.

\begin{figure}[!htb]
  \centering
  \includegraphics[width=0.75\linewidth]{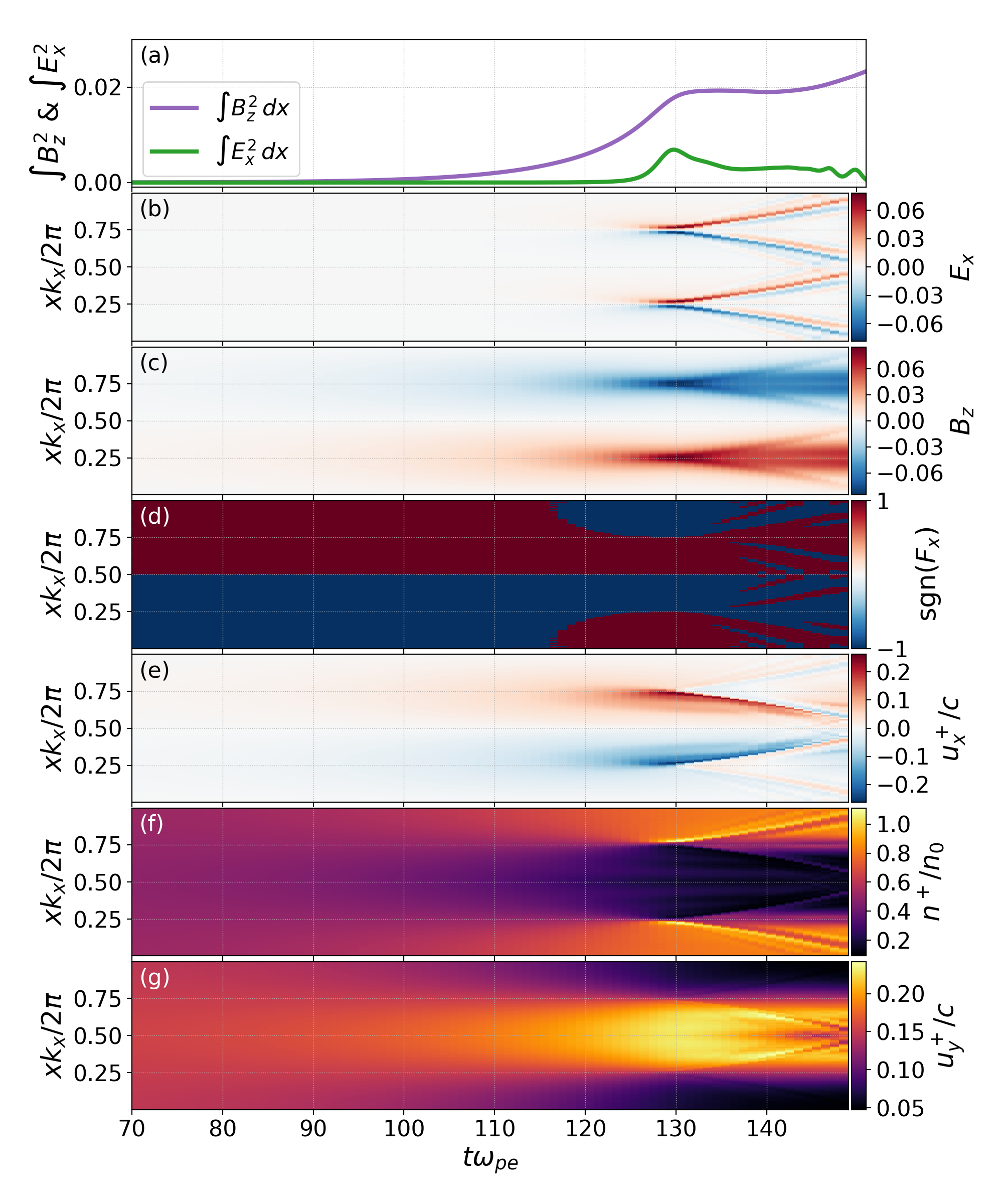}
  \caption[Nonlinear saturation of WI]{Detailed spatio-temporal plots
    of the nonlinear saturation of WI for the low temperature case
    ($v_{th}/c=0.03$); the bottom four panels are limited to the `+'
    population. The top panel (a) shows the energy evolution
    to provide a time reference, however, here it includes the
    electric field energy as well.  Panels (b) and (c) show relevant
    electric and magnetic fields, respectively.  Panel (d) captures
    the sign of the total force acting on the `+' population in the
    $x$-direction.  Panel (e) then depicts the evolution of the
    $u_x^+$ due to this force.  Panel (f) captures the density of the
    `+' population, and, finally, panel (g) shows the decrease of the
    initial bulk velocity, $u_y$.}
  \label{fig:weibel:uber}
\end{figure}

A closer look at the growth of the electric field (\fgr{weibel:energyE_l})
reveals an interesting fact. Preceding the first saturation, there is
a significant increase in the growth rate of $E_x$.  One possible explanation
of this enhancement is a secondary instability.

\begin{figure}[!htb]
  \centering
  \includegraphics[width=0.8\linewidth]{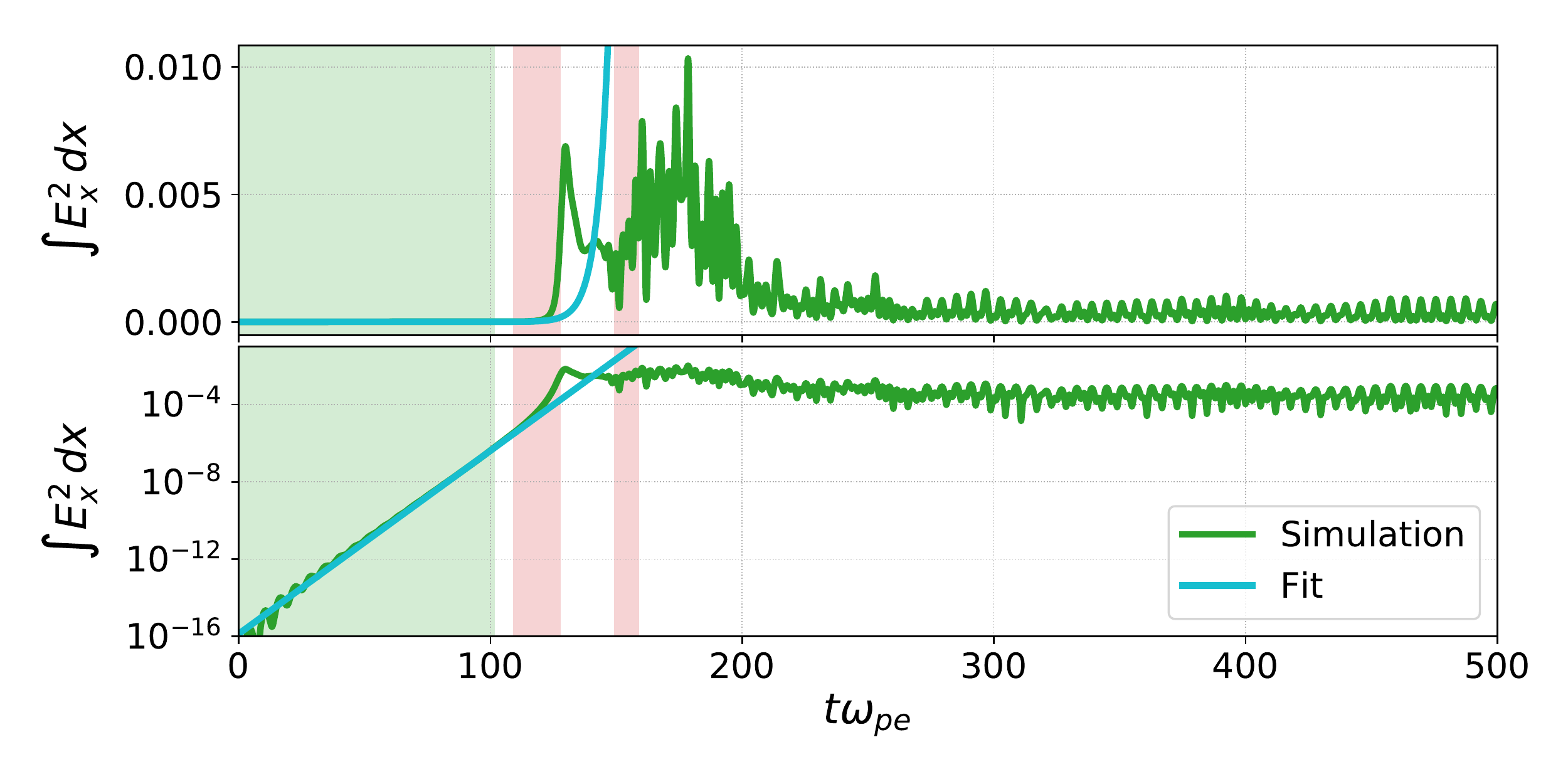}
  \caption[Growth of the electric field energy in the low temperature
    case]{Growth of the electric field energy proxy in the low
    temperature case ($v_{th}/c=0.03$). Note the significant increase
    in the growth rate immediately preceding the first saturation.
    The green background marks the fitting region; the red background
    denotes the region where a two-stream-like instability can
    potentially growth (see \fgr{weibel:reversals} for more details).}
  \label{fig:weibel:energyE_l}
\end{figure}

The generated electric field lies in the $x$-direction and so do the
counter-streaming velocities $u_x^\pm$; therefore, the two-stream or
two-stream-like instability, discussed in \ser{benchmark:two-stream},
is a potential candidate for the secondary instability.  Focusing on
$xk_x/2\pi=0.25$, corresponding to a maximum of $B_z$, a time
evolution of the counter-streaming velocities can be obtained (see
\fgr{weibel:reversals}).  The velocities are then substituted into the
dispersion relation \eqrp{benchmark:two-stream_disp}.  The regions
where the drift velocities, $u_x^\pm$, are high enough for the two-stream
instability to grow are marked in red in both \fgr{weibel:energyE_l}
and \fgr{weibel:reversals}.  The first red region is in good agreement with the region of the enhanced electric field growth.

\begin{figure}[!htb]
  \centering
  \includegraphics[width=0.8\linewidth]{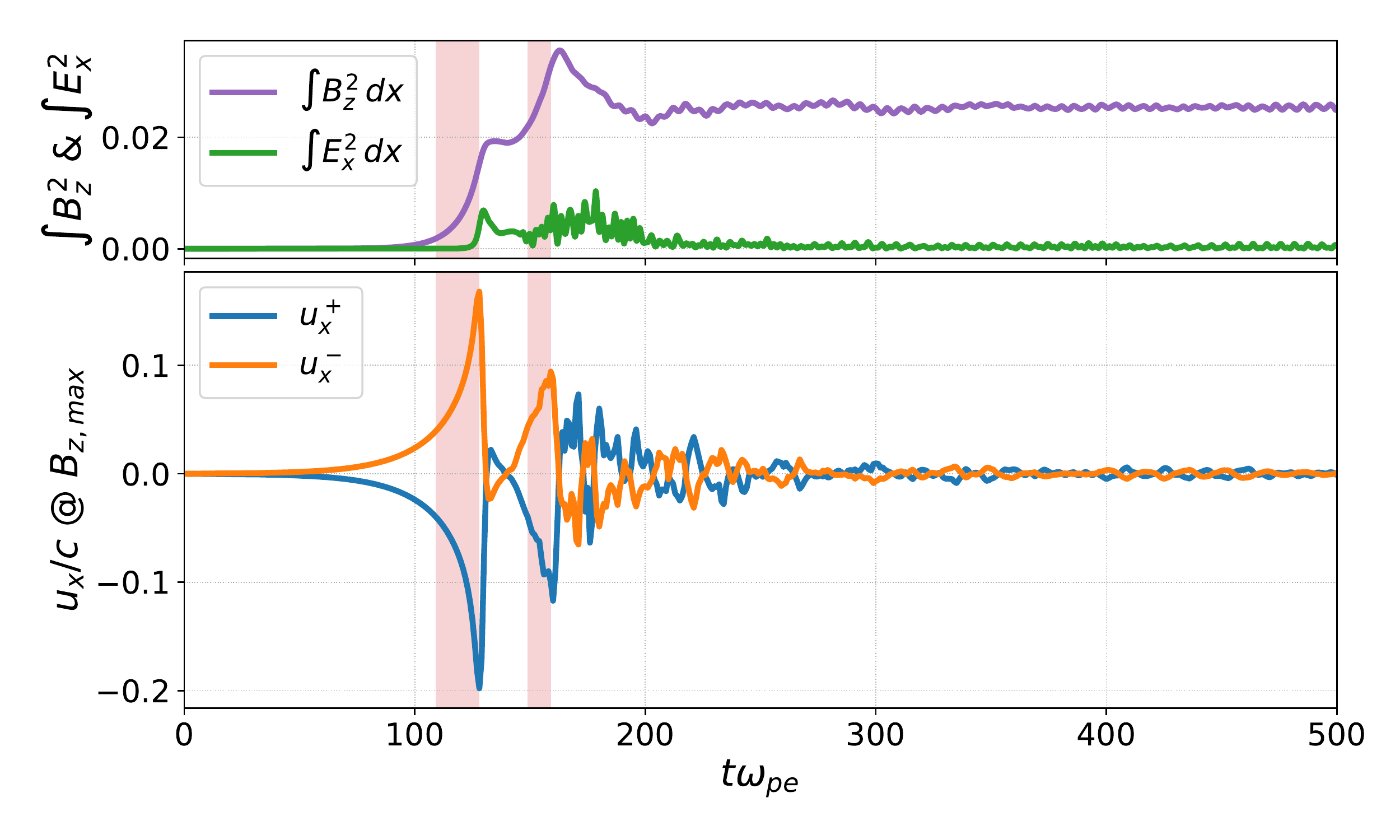}
  \caption[Counter-streaming velocities in the $x$-direction]{Detailed
    look at the $u_x^\pm$ counter-streaming velocities at
    $xk_x/2\pi=0.25$, where there is the maximum of $B_z$.  The
    red-marked area where a two-stream instability can grow based on
    the dispersion relation \eqr{benchmark:two-stream_disp}.  Note
    that for this $x$, the density becomes more filamented when
    $u_x^{-}>0$ and $u_x^{+}<0$ (see \fgr{weibel:uber}) and the
    magnetic field growths.  When the particles stream in the opposite
    direction, the magnetic field decreases.}
  \label{fig:weibel:reversals}
\end{figure}

In summation, the simulations confirm magnetic trapping as the main
mechanism for the saturation in the high temperature case.  However,
in the cold beam case, the self-consistently generated electric field
plays a significant role in saturating the instability
\citep{Cagas2017c}.  Its explosive growth is potentially connected
with a secondary two-stream-like instability.  Based on the kinetic
dispersion relation, the two-stream instability has suitable growing modes in the regions of
maximum magnetic field.  The electric field saturates the instability
before the magnetic trapping becomes dominant and also alters the
evolution in the nonlinear regime.

\FloatBarrier
\subsection{Phase-space and Temperature Evolution}
\label{sec:weibel:phase}

To round up the chapter, we briefly discuss the phase space evolution
of the distribution function and address the discrepancy in the growth
rates.  \fgr{weibel:phase} depicts the distribution function at three
times for the low temperature case ($v_{th}/c=0.03$).  In order to
visualize the 3D distribution, $f(x,v_x,v_y)$ is integrated over $v_y$
and $v_x$ to show the $x-v_x$ (left column) and $x-v_y$ (middle
column) profiles separately.\footnote{In \fgr{weibel:init}, a similar
technique is used to visualize the $v_x-v_y$ profiles.}  The first
row shows the initial conditions with only the thermal spread of the
$x$-velocities and the $\pm u_y$ bulk velocities in the $y$-direction
(see \fgr{weibel:init}c for the complementary $v_x-v_y$
profile).  The second row captures the solution at
$t\omega_{pe}=114$, which is approximately the end of linear growth phase (based
on the ``sweeping fit''; see \fgr{weibel:uber}).  The beams slow down in the $y$-direction, as the kinetic energy is
converted into the magnetic field energy, and the filamentation force
is accelerating particles along $x$.  Finally, the bottom row shows
the last frame of the simulation run, where the kinetic energy is
depleted and the system settles into a stable equilibrium disturbed only
by electron oscillations.

\begin{figure}[!htb]
  \centering
  \includegraphics[width=0.8\linewidth]{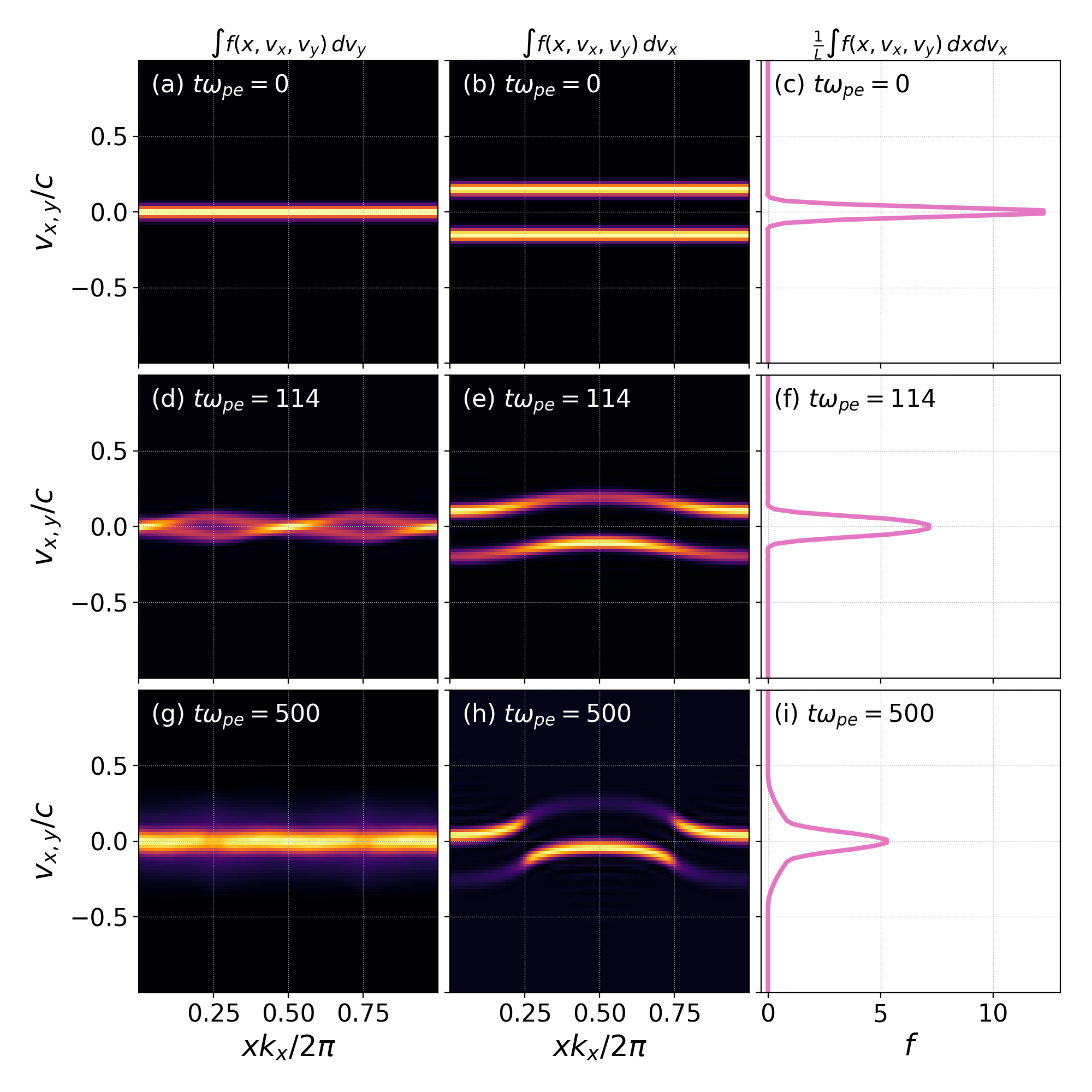}
  \caption[Phase space evolution of the distribution function during
    WI]{Phase space evolution of the distribution function during the
    WI for the low initial temperature ($v_{th}/c=0.03$).  In order to
    visualize 3D distribution $v_y$ and $v_x$ are integrated out
    (first and second column respectively).  The last column is
    integrated over $v_y$ and averaged over $x$ to obtain the average
    $v_x$ profile.  The rows show in order the initial conditions, the
    ``end of the linear growth phase'' at $t\omega_{pe}=114$, and the
    last frame of the run ($t\omega_{pe}=500$).}
  \label{fig:weibel:phase}
\end{figure}

The right column in \fgr{weibel:phase} shows 1D slices of the
distribution function integrated over $v_y$ and averaged
over $x$, i.e., the average $v_x$ profile of the distribution.  As the
instability progresses this profile becomes wider; the temperature
increases.  This temperature increase is noticeable even during the
``linear growth phase'' (the green region in \fgr{weibel:energy_l}).
Detailed comparison of the distribution functions is in
\fgr{weibel:temperatures}.  Fitting the Maxwellian
distribution\footnote{The temperature has a meaning only as the width
  or more precisely the variance, $\sigma$, of the Maxwellian
  distribution.}  (light gray lines in \fgr{weibel:temperatures})
reveals the temperature increase from $v_{th}/c=0.0300$ to
$v_{th}/c=0.0534$.  This has an important consequence for the growth
rate.  As is established in \ser{weibel:lintheory}, the growth rate
of WI decreases with the temperature.  If we compare the growth rate
obtained from the simulation (\fgr{weibel:energy_l}) with the linear
theory predictions for the initial conditions and the end of the
``linear growth'' interval, we get:
\begin{align*}
  \gamma_{theory}(v_{th}/c=0.0300) &= 0.05533\,\omega_{pe} \\
  \gamma_{simulation} &= 0.05296\,\omega_{pe} \\
  \gamma_{theory}(v_{th}/c=0.0534) &= 0.05122\,\omega_{pe}
\end{align*}
The growth rate obtained from the simulation is close to the average
of the two theoretical predictions, which very well explains the
discrepancy seen in \fgr{weibel:disp}.

\begin{figure}[!htb]
  \centering
  \includegraphics[width=0.8\linewidth]{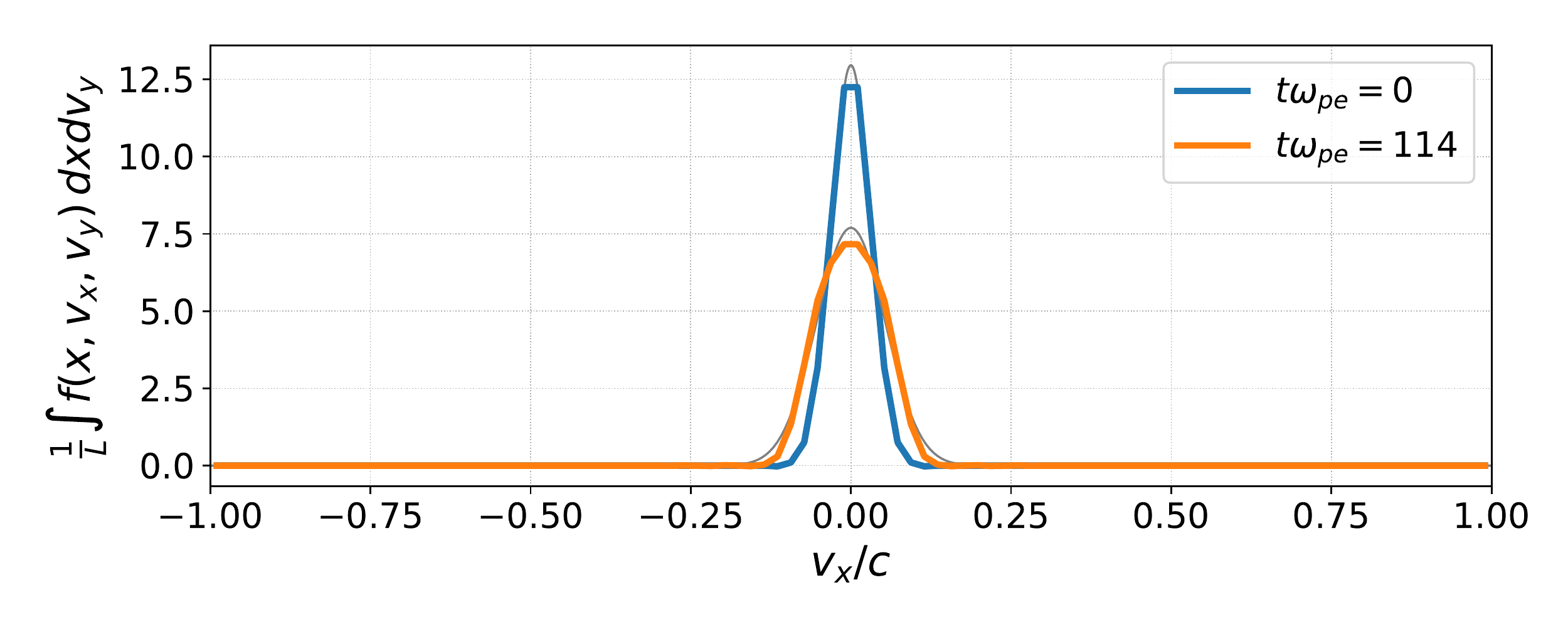}
  \caption[Temperature evolution during the linear phase of
    WI]{Temperature evolution of the linear phase of WI where the 1D
    $v_x$ profiles from \fgr{weibel:phase} for $t\omega_{pe}=0$ and
    $114$ are plotted over each other.  The profiles are on top of the
    Maxwellian fits (light gray) used to obtain the thermal
    velocities.}
  \label{fig:weibel:temperatures}
\end{figure}

%% file: bounded_plasma.tex
\chapter{Bounded Plasma Simulations}\label{sec:bounded}

\epigraph{This quote is often falsely attributed to Mark
  Twain.}{\textit{Randall Munroe}}

Interactions between plasma and solid surfaces are the center-piece of
this work.  This chapter starts with the discussion of the classical
plasma sheath, first in the collision-less regime and then with
collisions and ionization.  The last part deals with plasma-material
interactions (PMI).

%=====================================================================
\FloatBarrier
\section{Classical Plasma Sheaths}
\label{sec:bounded:class}

When plasma is contained by walls, the boundaries behave as sinks.
Electrons, as the lightest species in a plasma, are quickly absorbed
into the wall, which leads to the creation of a typically positive
space charge region called plasma sheath \citep{Robertson2013}.  The
charge then gives rise to a potential barrier, which works to equalize
fluxes to the wall.  Even though sheath physics has been studied since
early the works of \cite{Langmuir1923}, some processes remain to be
fully understood.  Additionally, despite the relatively small width of
the sheath region, which is typically on the order of a Debye length,
\begin{align*}
  \lambda_D = \sqrt{\frac{\varepsilon_0 T_e}{n_e q_e^2}}, 
\end{align*}
sheaths play an important role in particle, momentum, energy, and heat
transfers and surface erosion, which can, in turn, have global effects
on the plasma.  Furthermore, field-accelerated ions and hot electrons
are known to cause an emission from the solid surface that can further
alter the system.  Therefore, the sheath must be self-consistently
included and resolved in numerical simulations to resolve bounded
plasmas.  This significantly affects the computational cost of
simulations, because the scale length of the system is usually several
orders of magnitude higher than the Debye length.  Usually, the effect
of the sheath is mimicked with ``sheath boundary conditions'', often
constructed from very simple flux balance arguments or making
assumptions like cold ions and no surface effects \citep{Loizu2012}.
Hence, first-principle simulations of the sheath are needed to both
validate and further develop the simple models as well as to
understand the global kinetic effects of sheaths on the bulk plasma.

%---------------------------------------------------------------------
\FloatBarrier
\subsection{Brief Introduction to the Plasma Sheath Theory}
\label{sec:bounded:sheath:theory}

In order to derive the plasma sheath equations, we start with rather
strong assumptions \citep{Chen1985}. First of all, with the two-scale
description of \cite{Langmuir1923}, the domain is divided into the
quasi-neutral part where $n_e=n_i=n_0$ and the non-neutral sheath with
monotonically decreasing potential, $\phi$.  The cold ions are assumed
to enter the sheath region with a non-zero velocity $u_{i,0}$.  Then,
from the conservation equations,
\begin{align*}
  \underbracket{n_0u_{i,0} = n_i(x)u_i(x)}_{\text{Conservation of
      mass}}, \quad \underbracket{\frac{1}{2}m_iu_{i,0}^2 =
    \frac{1}{2}m_iu_i(x)^2 + q_i\phi(x)}_{\text{Conservation of
      energy}},
\end{align*}
we get
\begin{align}\label{eq:bounded:ni}
  n_i(x) = n_0 \left( 1 - \frac{2q_i\phi(x)}{m_i
    u_{i,0}^2} \right)^{-\frac{1}{2}},
\end{align}
for $x$ inside the sheath.  Inertia of the electrons is neglected and
they are assumed to instantly follow the electric field,\footnote{This
  assumption is known as the Boltzmann electrons.}
\begin{align}\label{eq:bounded:ne}
  n_e(x) = n_0 \exp\left(-\frac{q_e\phi}{T_e}\right).
\end{align}

These results are substituted into Poisson's equation
\eqrp{model:poisson},
\begin{align}\label{eq:bounded:poisson}
  \pfracc{\phi(x)}{x} = -\frac{n_e(x)q_e + n_i(x)q_i}{\varepsilon_0} =
  -\frac{n_0}{\varepsilon_0} \left[
    q_e\exp\left(-\frac{q_e\phi}{T_e}\right) + q_i\left( 1 -
    \frac{2q_i\phi(x)}{m_i u_{i,0}^2} \right)^{-\frac{1}{2}} \right].
\end{align}
The following substitution,
\begin{align*}
  \chi := \frac{q_e\phi}{T_e}, \quad \xi := \frac{x}{\lambda_D} = x
  \sqrt{\frac{n_0q_e^2}{\varepsilon_0T_e}}, \quad M :=
  \frac{u_{i,0}}{\sqrt{ZT_e/m_i}},
\end{align*}
then simplifies \eqr{bounded:poisson} into\footnote{Note
  that$$\pfracc{\chi}{\xi}= \pfraca{\xi}\pfrac{\chi}{\xi} =
  \pfraca{\xi}\left( \pfrac{\chi}{x}\pfrac{x}{\xi} \right) =
  \pfraca{x}\left( \pfrac{\chi}{x}\pfrac{x}{\xi} \right)\pfrac{x}{\xi}
  = \pfracc{\chi}{x}\left(\pfrac{x}{\xi}\right)^2.$$}
\begin{align}\label{eq:bounded:poisson2}
  \pfracc{\chi}{\xi} =
  Z\left(1+\frac{2\chi}{M^2}\right)^{-\frac{1}{2}}-\exp(-\chi),
\end{align}
where $q_i = -Zq_e$.  \eqr{bounded:poisson2} needs to be integrated
twice to obtain usable profiles.  First, the equation is multiplied by
$\pfracb{\chi}{\xi}$ and then integrated from zero to $\xi$,
\begin{gather*}
  \int_0^\xi \pfracc{\chi}{\xi'} \pfrac{\chi}{\xi'} \,d\xi' =
  \int_0^\xi Z\left(1+\frac{2\chi}{M^2}\right)^{-\frac{1}{2}}
  \pfrac{\chi}{\xi'} \,d\xi' - \int_0^\xi \exp(-\chi) \pfrac{\chi}{\xi'}
  \,d\xi',\\ \frac{1}{2}\left[\left(\pfrac{\chi}{\xi'}\right)^2\right]_0^\xi
  = Z M^2 \left[\left(1+\frac{2\chi}{M^2}\right)^{\frac{1}{2}}
    \right]_0^\xi + \left[\exp(-\chi)\right]_0^\xi.
\end{gather*}
Since $\xi=0$ is at the boundary between the sheath and quasi-neutral
plasma, a natural choice for the potential is $\chi(\xi=0) :=
0$. The assumption of no net fields in the quasi-neutral plasma leads to
$\pfracb{\chi}{\xi'}|_{\xi=0}:=0$.  The equation then simplifies to
\begin{align}\label{eq:bounded:poisson3}
  \frac{1}{2}\left.\left(\pfrac{\chi}{\xi'}\right)^2\right|_\xi = Z
  M^2 \left[\left(1+\frac{2\chi}{M^2}\right)^{\frac{1}{2}}
    -1 \right] + \exp(-\chi) - 1.
\end{align}
Unfortunately, the second integration cannot be done analytically.
However, \eqr{bounded:poisson3} can still provide very interesting
insight.  The left-hand-side of \eqr{bounded:poisson3} is positive for
all $\xi$ and so must be the right-hand-side.  Using the Taylor series
expansion for $\chi \ll 1$ gives
\begin{gather*}
  ZM^2\left[1 + \frac{\chi}{M^2} - \frac{1}{2}\frac{\chi^2}{M^4} +
    \ldots -1 \right] + 1 -\chi + \frac{1}{2}\chi^2 + \ldots -1 \geq 0 \\
  \frac{1}{2}Z\chi^2\left(1-\frac{1}{M^2}\right) \geq 0.
\end{gather*}
The inequality is satisfied for $M^2>1$.  Back-substituting for $M$,
finally gives the well known Bohm sheath criterion \citep{Bohm1949},
\begin{align}\label{eq:bounded:bohm}
    u_{i,0} \geq u_B = \sqrt{\frac{ZT_e}{m_i}}.
\end{align}
Surprisingly, even with the assumptions mentioned above, the Bohm
criterion applies to conditions beyond these assumptions, with errors
within 20-30\% \citep{Bohm1949}.

The Bohm criterion requires ions to be accelerated in the presheath to
the speed of ion acoustic waves \citep{Riemann1990}.  The underlining
physical reason is illustrated in \fgr{bounded:bohm}.  As the ions are
accelerated towards the wall, the ion density decreases.  However, if
the ions entering the sheath are not fast enough, $M<1$, they undergo
relatively high acceleration and the density drops significantly.  For
$M<1$, the density becomes smaller that the electron density, which
would produce a potential with opposite sign, making shielding
impossible \citep{Riemann1990}.
\begin{figure}[!htb]
  \centering
  \includegraphics[width=0.8\linewidth]{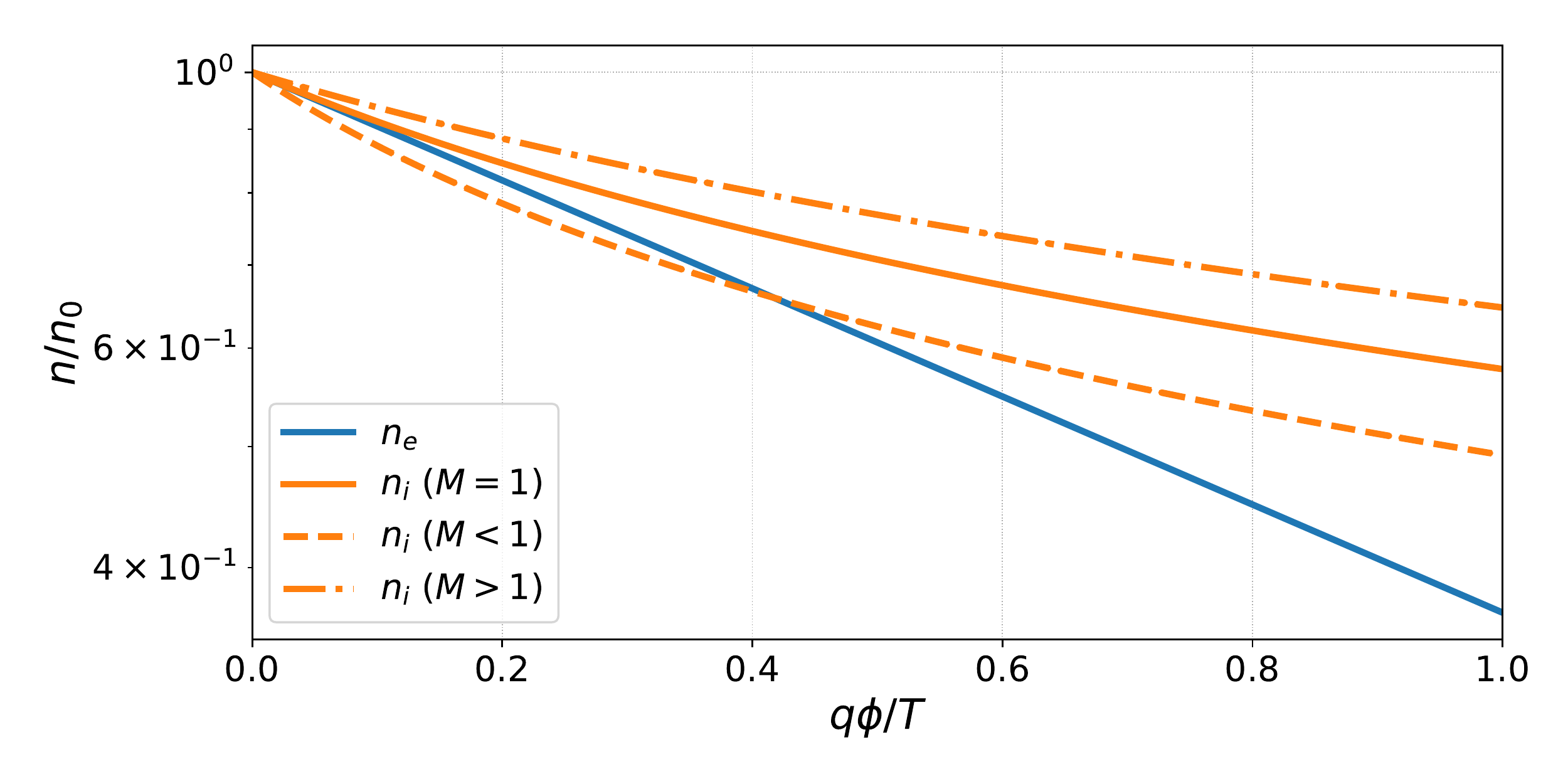}
  \caption[Electron and ion densities in plasma sheath based on
    $M$]{Normalized electron \eqrp{bounded:ne} and ion
    \eqrp{bounded:ni} densities as a function of $\chi=q_e\phi/T_e$.
    The ion density is shown for three values of
    $M=u_{i,0}/\sqrt{ZT_e/m_i}$.  The region of negative space charge between 0 and 0.4 $q\phi/T$ is would prevent the shielding.}
  \label{fig:bounded:bohm}
\end{figure}

The kinetic version of the criterion, relaxing the constrains on both
electron and ion distributions, was first formulated by
\cite{Boyd1959} and later extended by \cite{Riemann1990}.
\cite{Lieberman2005} provide it in the following form,
\begin{align}
  -\frac{q_e}{m_i} \int_0^\infty \frac{1}{v^2} f_i(v) \,dv \leq
  \left.\pfrac{n_e}{\phi}\right|_{\phi=0},
\end{align}
where $v = |v_x|$.  Substituting in the relation for Boltzmann
electrons \eqrp{bounded:ne} and the distribution function for the cold ions, $f(x,v) = n_i(x)\delta\big(v-u_i(x)\big)$, we get
\begin{align*}
  -\frac{q_e}{m_i} \frac{n_i(x)}{u_i(x)^2} \leq -\frac{q_e}{T_e} n_e(x).
\end{align*}
Noting the $n_e=n_i$ at the sheath edge, we recover the original Bohm
criterion (for singly charged ions).\footnote{Keep in mind that
  $q_e<0$.}

This so-called ``kinetic Bohm criterion'' is discussed and generalized
in several papers \citep{Allen1976, Bissell1987, Riemann1990,
  Riemann1995, Fernsler2005, Riemann2006} and its applicability on
different plasma distribution functions is further discussed by
\cite{Baalrud2011, Riemann2012}, and \cite{Baalrud2012}.

\cite{Baalrud2014} present an alternative approach based on the fluid
moment hierarchy.  They emphasize, that the concept of a sheath edge
in Langmuir's description is connected strictly to the charge density
and therefore should be independent of the plasma model.\footnote{An
  example of a description dependent on plasma model is the
  Child-Langmuir formula for the space-charge-limited current, $$j =
  \frac{4}{9} \sqrt{\frac{2|q_e|}{m_i}}\frac{\varepsilon_0|
    \phi_{w}|^{\frac{3}{2}}}{d^2},$$where $\phi_w$ is the potential at
  the wall and $d$ is the sheath width} Instead, the authors suggest
identifying the sheath edge using a threshold for the normalized
charge density $\bar{\rho}_c = (n_i-n_e)/n_i$.  However, in a real
situation where $\lambda_D/L \neq 0$ this transition is not abrupt,
hence, arbitrary values must be chosen.  By taking the expansion of
$\rho$ with respect to $\phi$, the quantitative form of the sheath
condition is derived,
\begin{align}\label{eq:bounded:edge}
  \left|\pfrac{n_i}{x}\right| \leq
  \left|\pfrac{n_e}{x}\right|.
\end{align}
From the steady-state conservation of mass, $\nabla_x(n_su_s) = S_s$,
where $S_s$ is a source or sink term,\footnote{$S_s$ comes from
  the integration of the RHS of \eqr{model:boltzmann}, $S_s = \int
  (\delta f_s/\delta t)\, d\bm{v}$.} we get
\begin{align*}
  \pfrac{n_s}{x} = - \frac{n_s}{u_s} \pfrac{u_s}{x} + S_s.
\end{align*}
Substituting the result into \eqr{bounded:edge}, \cite{Baalrud2014}
obtain an alternative version of the Bohm criterion,
\begin{align}\label{eq:bounded:bohm2}
  u_{i,0} \geq \sqrt{\frac{T_e +
      T_i - m_eu_e^2}{m_i}}.
\end{align}

%---------------------------------------------------------------------
\FloatBarrier
\subsection{Baseline Numerical Simulations}
\label{sec:bounded:sheath:sims}

In this section, the simplest continuum kinetic simulations are
presented to demonstrate the above-described features of the plasma
sheath theory.  These results were also published by
\cite{Cagas2017s}.

Unlike in the previous chapters, the simulations are run with SI
units,\footnote{SI units are required for the plasma material
  interactions, which are often based on empirical data; see
  \ser{bounded:pmi:r} for more details.} but the result are still
presented in normalized form.  For the first set of simulations (full
listing is available in \ref{list:bounded:cls}), we set $n_{i,0} =
n_{e,0} = \SI{1e17}{m^{-3}}$ and $T_e = 10\,T_i = \SI{10}{eV}$.  This
gives following plasma parameters,
\begin{align*}
  v_{th,e0} &= \sqrt{\frac{T_e}{m_e}} \approx \SI{1.33e6}{m s^{-1}}
  \\ v_{th,i0} &= \sqrt{\frac{T_i}{m_i}} \approx \SI{9.79e3}{m s^{-1}}
  \\ u_B &= \sqrt{\frac{T_e}{m_i}} \approx \SI{3.09e4}{m s^{-1}}
  \\ \omega_{pe} &= \sqrt{\frac{n_eq_e^2}{\varepsilon_0m_e}} \approx
  \SI{1.78e10}{s^{-1}} \\ \lambda_D &=
  \sqrt{\frac{\varepsilon_0T_e}{n_eq_e^2}} \approx \SI{7.43e-5}{m}
\end{align*}

\fgr{bounded:cls_distf} shows the initial and final distribution
functions of electrons and ions.  The simulation is initialized with
Maxwellian distributions for both species with the thermal velocities
listed above and no drifts. For electrons, the velocity space is set
to span $\langle -6v_{th,e0},\, 6v_{th,e0}\rangle$. Since the ions
undergo acceleration, their velocity space is chosen as $\langle
-6u_B,\, 6u_B\rangle$\footnote{This is probably unnecessarily
  conservative.}  rather than $\langle -6v_{th,i0},\,
6v_{th,i0}\rangle$.  The configuration space spans from
$\pm128\,\lambda_D$ and ends with ideally absorbing walls on both
sides, i.e., the distribution is set to zero at the
boundary.\footnote{Technically, only the part incoming from the wall
  should be set to zero but since \texttt{Gkeyll} uses upwinding
  fluxes (\ser{model:discvlasov}), the outgoing part of the ghost cell
  does not play any role.}

\begin{figure}[!htb]
  \centering
  \includegraphics[width=0.9\linewidth]{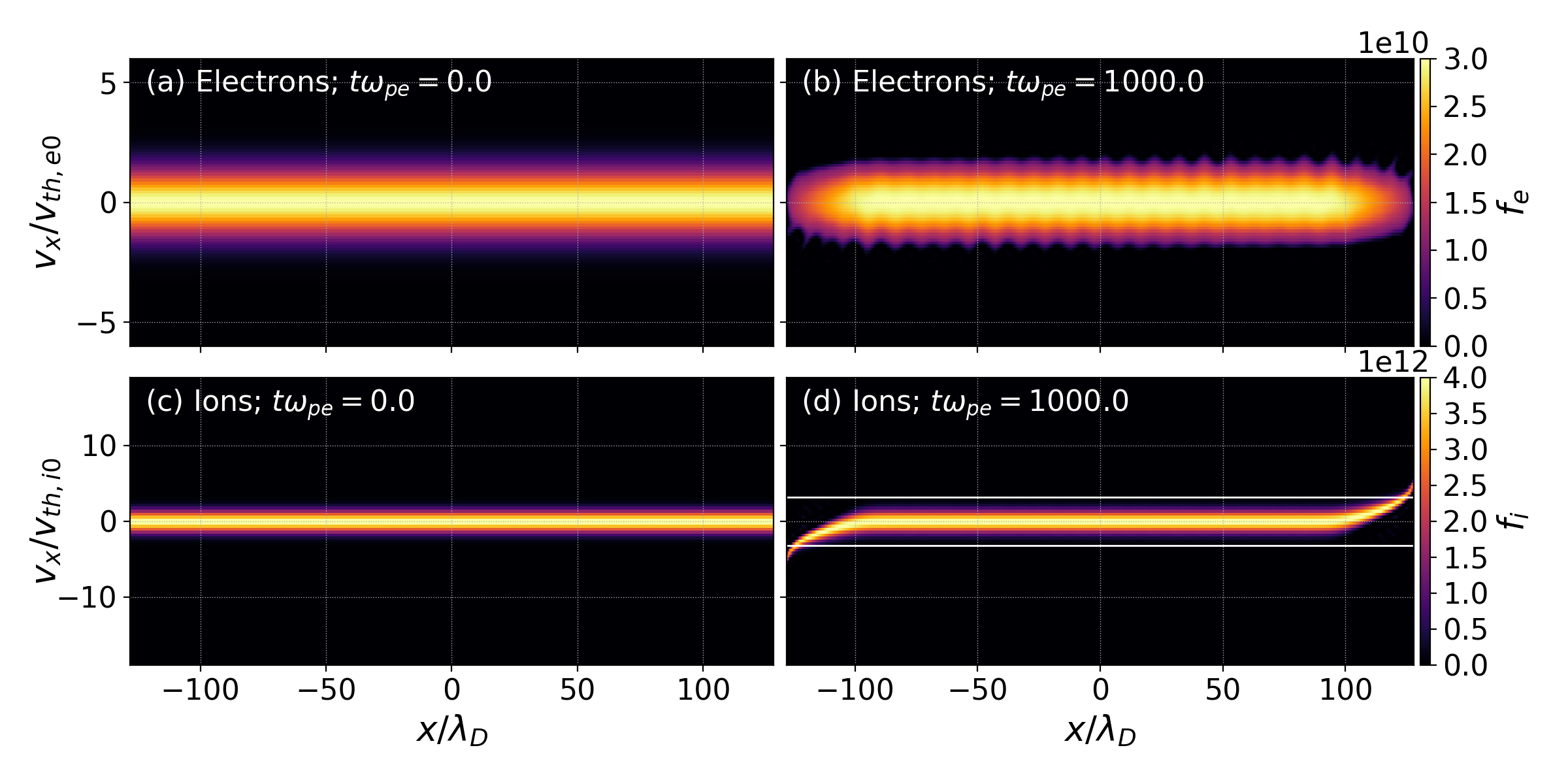}
  \caption[Initial and final distributions in a collisionless sheath
    simulation]{Initial and final distribution functions for electrons
    (top) and ions (bottom) in a collision-less sheath simulation
    [\ref{list:bounded:cls}].  The first column (a and c) captures the
    initial conditions while the right column shows the solution at
    $t\omega_{pe} = 1000$.  The white lines in panel (d) mark the Bohm
    velocity \eqrp{bounded:bohm}.}
  \label{fig:bounded:cls_distf}
\end{figure}

The white lines in \fgr{bounded:cls_distf}d mark the Bohm velocity
\eqrp{bounded:bohm} in terms of the ion thermal velocity.  Note that
the ion distribution self-consistently evolves so they reach the Bohm
velocity, $u_B$, a few Debye lengths from the wall.  On the other
hand, the electron distribution function (\fgr{bounded:cls_distf}c) is
repulsed from the wall by the sheath electric field.  There is
negligible electron density right next to the walls and we also see
the temperature decrease due to decompression cooling \citep{Tang2011}
(narrowing of the distribution).  The electron distribution function
is notably affected by oscillations which will be addressed later.

The temporal evolution from the initial conditions to
$t\omega_{pe}=1000$ is in \fgr{bounded:cls_evolution}. From top to
bottom, it shows the electron (a) and ion (b) densities, the sheath
electric field (c), ion bulk velocity (d), and the electron thermal
velocity\footnote{As was mentioned above, the definition of the
  thermal velocity and temperature in general is questionable for
  non-Maxwellian distributions.  Here, the thermal velocity is
  calculated using the moments $n$, $nu$, and $\mathcal{E}$
  (Eq.\thinspace\ref{eq:model:m0}, \ref{eq:model:m1}, and
  \ref{eq:model:m2}) as$$v_{th} =
  \sqrt{(2\mathcal{E}/m-(nu)^2/n)/n}.$$} (d). The figure captures
$28\,\lambda_D$ near the left wall of the simulation; therefore the
minus velocity and electric field are pointing towards the wall.

\begin{figure}[!htb]
  \centering
  \includegraphics[width=0.8\linewidth]{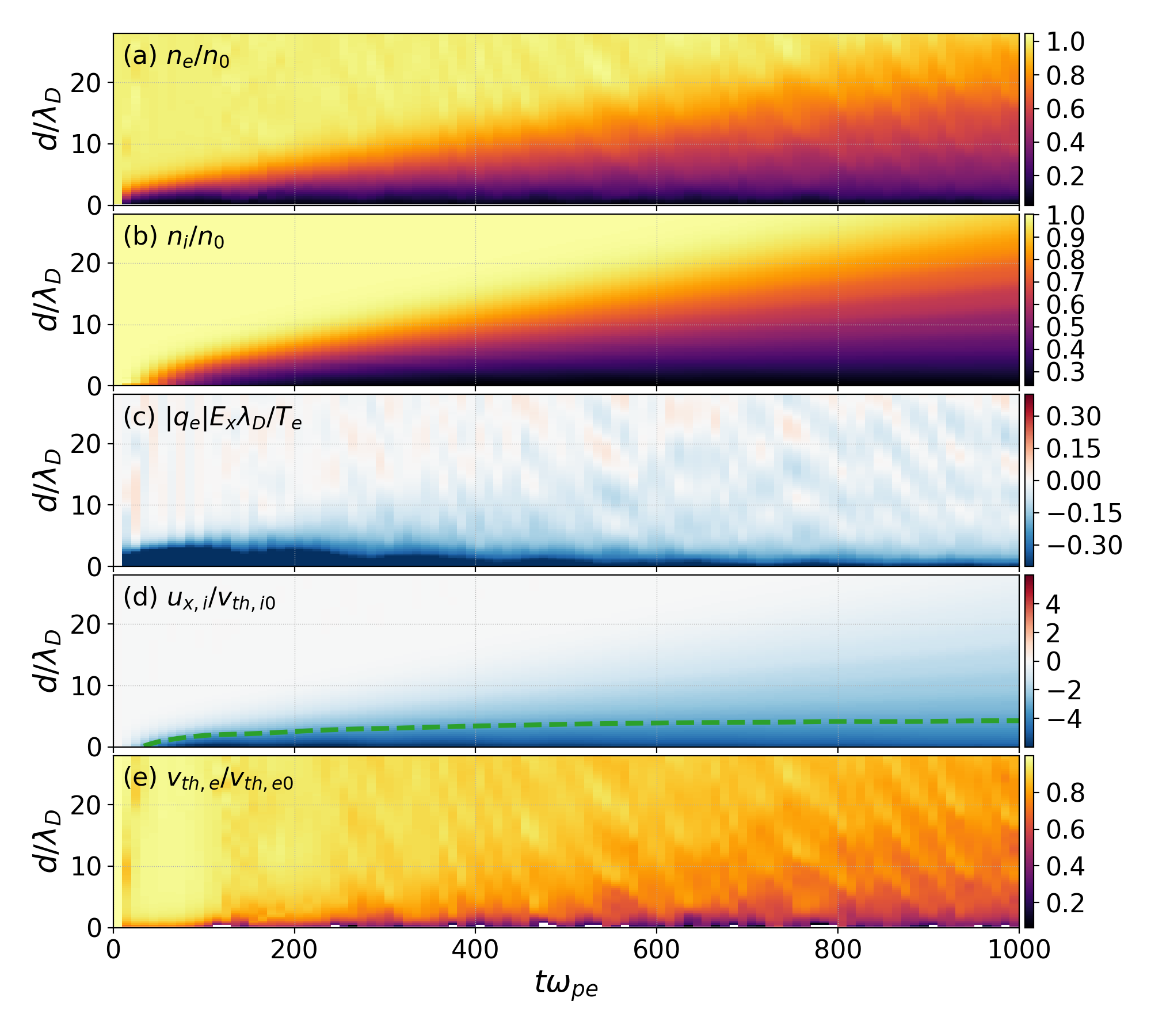}
  \caption[Temporal evolution of the plasma sheath]{Temporal evolution
    of the densities (electrons top and ions bottom), normalized
    electric field (c), ion bulk velocity (d), and the electron
    thermal velocity (e) in the region near the left wall of the
    \fgr{bounded:cls_distf} ($d=0$ is dirrectly at the left wall). The
    green contour in panel (d) marks the Bohm velocity
    \eqrp{bounded:bohm}. [Simulation input file:
      \ref{list:bounded:cls}]}
  \label{fig:bounded:cls_evolution}
\end{figure}

There are a few things to point out in \fgr{bounded:cls_evolution}.
Firstly, the vertical blocks (for each horizontal axis location) do not
correspond to individual time-steps but rather to the output frames.
In this particular case, each block represents $10/\omega_{pe}$; the
first one captures the initial conditions.  The simulation uses a
hydrogen plasma, i.e., ions are 1836$\times$ heavier than electrons.
The ion temperature is also set to be 10 times lower than electron
temperature.  Therefore, the electrons react much faster than ions and
their density drops to almost zero in the first few Debye lengths
within the first recorded frame.  The currents then act as a source of
electric field (\fgr{bounded:cls_evolution}c) limiting the outflow of
electrons.  The ion population naturally reacts to this field as well
but due to higher inertia, it takes a few tenths of plasma oscillation
periods to develop ion velocity directed to the wall
(\fgr{bounded:cls_evolution}d) to eventually reach the Bohm velocity
\eqrp{bounded:bohm}, marked by the green dashed contour in
\fgr{bounded:cls_evolution}d.  Note that as the ion supersonic flow
develops, the electron and ion fluxes start equalizing and the
electric field magnitude decreases resulting in a narrower region of
non-zero field.  As the electrons leave the domain, they undergo the
decompression cooling \citep{Tang2011} which is captured in
\fgr{bounded:cls_evolution}e.

Finally, there is the question of the oscillations noticeable in all
the electron variables and in the electric field.  Such oscillations
are reported in the literature; for example, \cite{Lieberman2005}
report oscillations of the electrostatic potential during numerical
simulations.  In order to diagnose the signal, we can construct a
dispersion diagram, i.e., the relation between the frequencies and
wavelengths.  It is done by performing the Fourier transformation
twice.  The spatio-temporal $E(t,x)$ is first transformed into
$E'(\omega,x)$ and then then $E''(\omega,k)$.  The result, normalized
to $\omega_{pe}$ and $\lambda_D$, is in \fgr{bounded:cls_dispersion}.
The figure shows the logarithm of the data which exaggerates the
noise.  Over-plotted are theoretical dispersion relations for Langmuir
waves \eqrp{benchmark:langmuir}, $\omega^2 = \omega_{pe}^2 +
\frac{3}{2}v_{th}^2 k^2$, for the initial thermal velocity and half of
it.  The green line, corresponding to the lower temperature near the
wall, fits the data well.\footnote{Sadly the lines are making the
  profile harder to see.}  It suggests that the observed oscillations
are Langmuir waves launched into the system by the system rapidly
adjusting to physical state from the initial uniform conditions.
There are two ways to minimize the effects of these waves.  Additional
damping mechanisms (see \ser{bounded:col}) can be included to diminish
the waves in time and the simulation can be initialized with
approximate initial conditions to reduce the excitation.

\begin{figure}[!htb]
  \centering
  \includegraphics[width=0.8\linewidth]{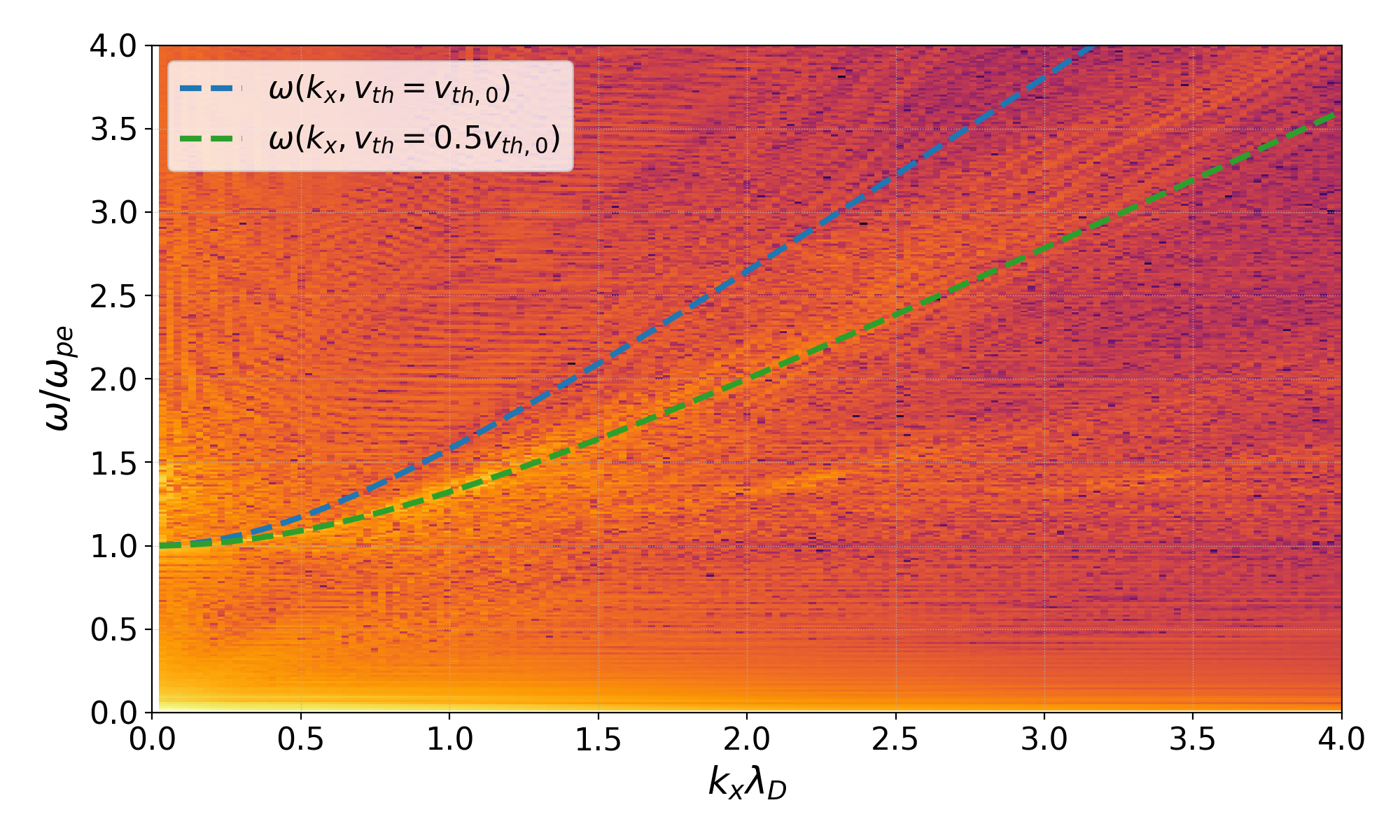}
  \caption[Dispersion diagram of the electric field]{Dispersion
    diagram of the electric field in the sheath simulation with
    uniform initial conditions.  Also included are the dispersion
    relations line plots for Langmuir waves \eqrp{benchmark:langmuir},
    $\omega^2 = \omega_{pe}^2 + \frac{3}{2}v_{th}^2 k^2$, for the
    initial thermal velocity (blue line) and half of it (green line).}
  \label{fig:bounded:cls_dispersion}
\end{figure}

%---------------------------------------------------------------------
%\FloatBarrier
\subsection{Approximate Initialization}
\label{sec:bounded:sheath:init}

Initializing the simulation with an approximate solution decreases the
rapid excitation of the electron Langmuir waves at the beginning of a
run.  \cite{Robertson2013} describes a simplified ODE model based on
the assumptions of mono-energetic ions, Boltzmann electrons, and
uniform ionization source, $S$, over the whole domain.  The validity
of the last assumption is arguable, because the ionization rate should
depend on the electron number density \citep{Meier2012}, which changes
significantly in the sheath region and is generally decreasing in the
presheath; however, that is not important for this purpose.  The model
here only serves to provide an initial guess and the simulation then
relaxes using the full kinetic equations.

The \cite{Robertson2013} model consists of the three ordinary
differential equations (ODE) and uses the following normalization,
$\tilde{x} = x/\lambda_D$, $\tilde{\phi} = |q_e|\phi/T_e$,
$\tilde{u}_i = u_i/u_B$, $\tilde{E}=|q_e|E\lambda_D/T_e$,
$\tilde{S}=S\lambda_D/(n_0u_B)$, $\tilde{n}_i=n_i/n_0$,
$\tilde{n}_e=n_e/n_0$, and the collisional momentum transfer
$\tilde{\nu}_c = \lambda_D\nu_c/u_B$.  The momentum equation then
becomes
\begin{align}
 \frac{d\tilde{u}_i(\tilde{x})}{d\tilde{x}} =
 \frac{\tilde{E}(\tilde{x})}{\tilde{u}_i(\tilde{x})} -
 \frac{\tilde{S}}{\tilde{n}_i(\tilde{x})} -\tilde{\nu}_c =
 \frac{\tilde{E}(\tilde{x})}{\tilde{u}_i(\tilde{x})} -
 \frac{\tilde{u}(\tilde{x})}{\tilde{x}} -\tilde{\nu}_c,
\end{align}
where we use $\tilde{j} = \tilde{S}\tilde{x}$.\footnote{Coming from
  the continuity equation for uniform source,$$\pfraca{x}(nu)=S \quad
  \Rightarrow\quad nu = Sx.$$} The Poisson equation is
\begin{align}
  \frac{d\tilde{E}(x)}{d\tilde{x}} = \tilde{n}_i(\tilde{x}) -
  \tilde{n}_e(\tilde{x}) =
  \frac{\tilde{S}\tilde{x}}{\tilde{u}_i(\tilde{x})} -
  \exp\big(\tilde{\phi}(\tilde{x})\big),
\end{align}
and, to close the system, we need
\begin{align}
  \frac{d\tilde{\phi}(\tilde{x})}{d\tilde{x}} = -\tilde{E}(\tilde{x}).
\end{align}
The system diverges for $x\rightarrow0$, therefore, the integration
starts at $\Delta x$.  The initial conditions are then $\tilde{u}_{i0}
= \tilde{S}\Delta\tilde{x}$, $\tilde{E}_0 =
2\tilde{S}^2\Delta\tilde{x}$, and $\tilde{\phi}_0 =
-\tilde{S}^2(\Delta\tilde{x})^2$.

Disregarding the collisional transfer, the system can be numerically
solved, for example, using Python's \texttt{odeint}
module:\footnote{\url{https://docs.scipy.org/doc/scipy/reference/generated/scipy.integrate.odeint.html}}
\begin{lstlisting}[language=Python]
import numpy as np from scipy.integrate import odeint

L = 128
numX = 128
dx = L/numX
S = 0.54 / L
def robertson(y, x, S):
    phi, E, u = y
    dydx = [-E, S*x/u - np.exp(phi), E/u - u/x]
    return dydx
y0 = [-S**2*dx**2, 2*S**2*dx, S*dx/(np.exp(-S**2*dx**2)+2*S**2)]
x = np.linspace(dx, L, numX)

sol = odeint(robertson, y0, x, args=(S,))
\end{lstlisting}

Profiles obtained from the model are shown in
\fgr{bounded:cls_robertson}.  Note that since only $u$, $\phi$, and
$E$ are evolved, the densities are calculated using
$\tilde{n}_i=\tilde{S}\tilde{x}/\tilde{u}(\tilde{x})$ and $\tilde{n}_e =
\exp(\tilde{\phi})$.

\begin{figure}[!htb]
  \centering
  \includegraphics[width=0.8\linewidth]{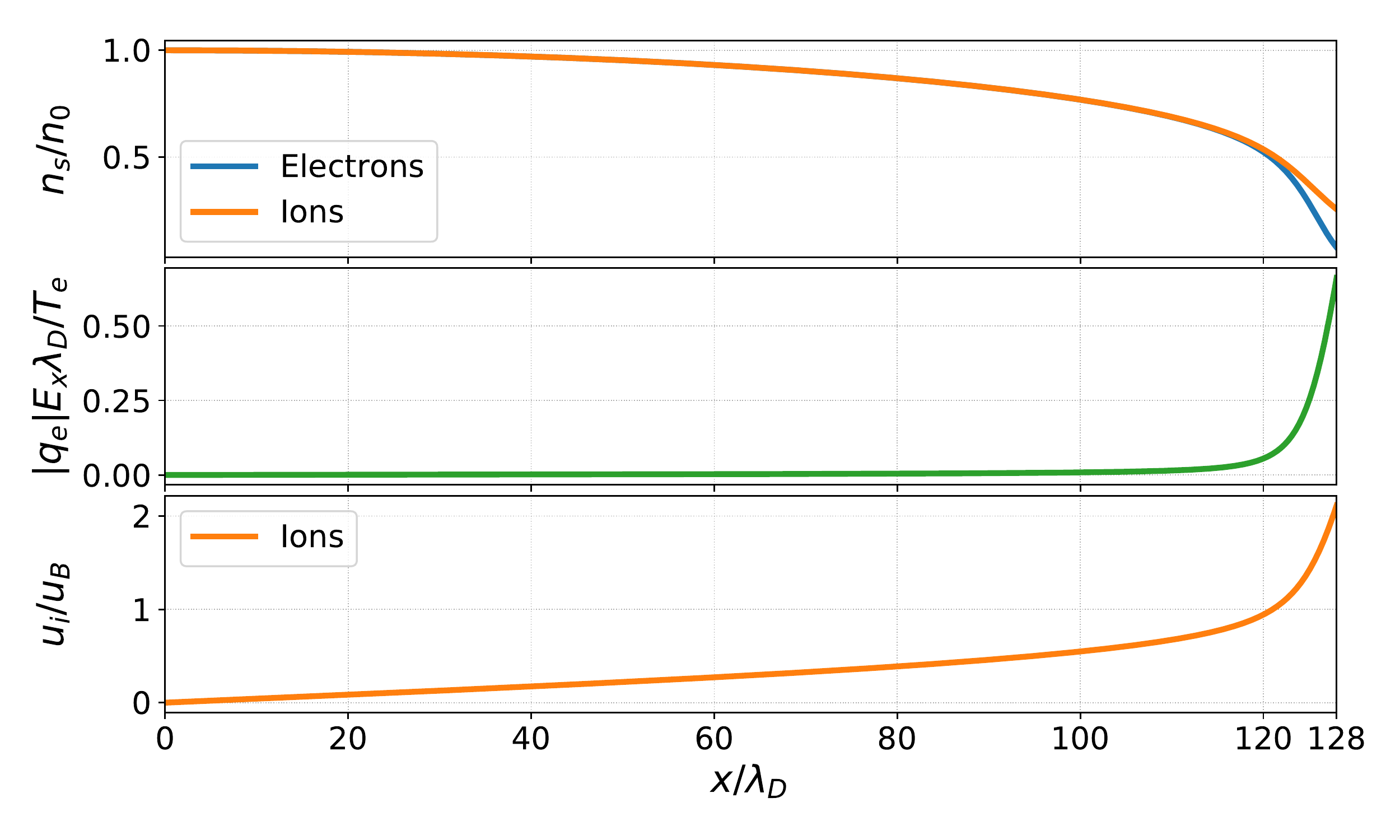}
  \caption[Sheath profiles from the Robertson model]{Electron and ion
    densities, resulting electric field, and ion bulk velocity
    profiles calculated using the \cite{Robertson2013} ODE sheath
    model.  Note the velocity is reaching the Bohm velocity around
    $8\,\lambda_D$ from the wall.}
  \label{fig:bounded:cls_robertson}
\end{figure}

After the profiles are precomputed, they can be stored and used for
the initialization of the electron and ion distribution functions and
the electric field.  In \fgr{bounded:cls_comparison}, we compare two
electron distribution functions after $t\omega_{pe}=100$ initialized
with uniform and approximate initial conditions.  Note that even
though the wave excited at the beginning of the simulation is still
present in both cases, it is more pronounced in the case initialized
with the uniform densities and zero fields (see the highlighted part
of the distribution functions in \fgr{bounded:cls_comparison}). Note
that as the simulations are captured early in time in order to clearly
capture the waves, density and temperature profiles are quite
different.  While the left plot is initialized with $n_e =
\SI{1e17}{m^{-3}}$ everywhere, the density in the right plot has this
value only in the middle and is decreasing towards both walls, as it
is shown in \fgr{bounded:cls_robertson}.

\begin{figure}[!htb]
  \centering
  \includegraphics[width=0.9\linewidth]{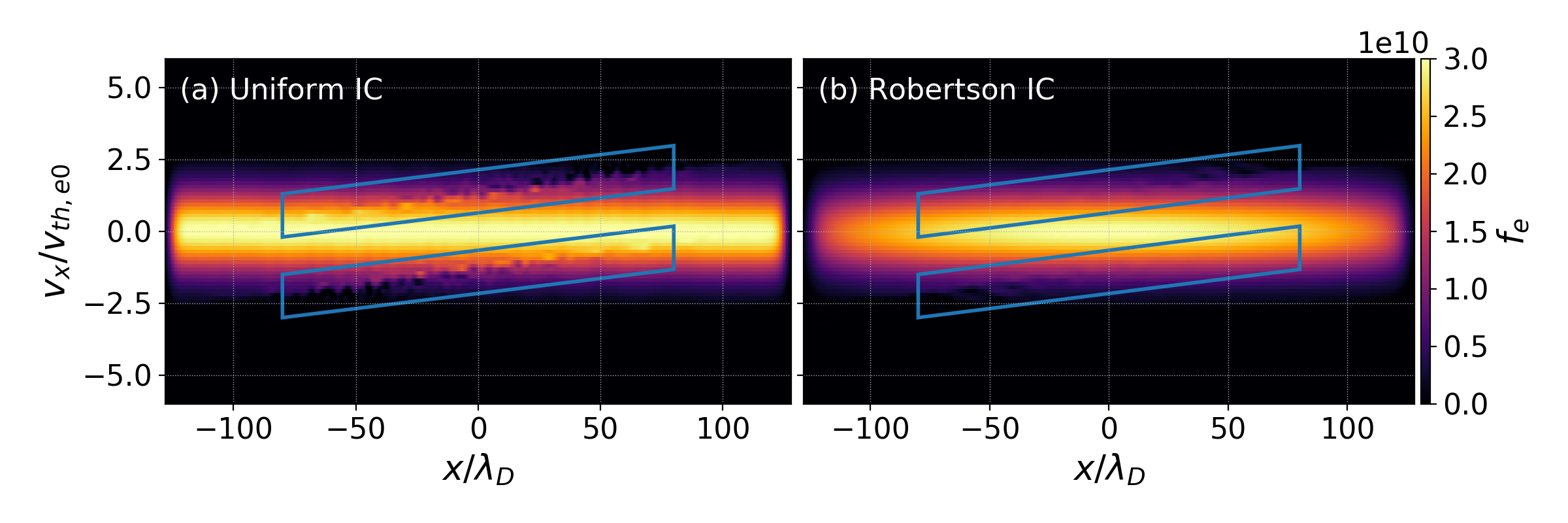}
  \caption[Comparison of the distribution functions with uniform and
    precomputed IC]{Comparison of the electron distribution functions
    with uniform initial conditions (left) and the initial conditions
    precomputed using the \cite{Robertson2013} model (right) early in
    time ($t\omega_{pe}=100$). Highlighted are the perturbations of
    the distribution functions caused by the propagating Langmuir
    waves.}
  \label{fig:bounded:cls_comparison}
\end{figure}

%=====================================================================
\FloatBarrier
\section{Collisions and Ionization}
\label{sec:bounded:col}

The discussion in \ser{bounded:class} focused on the sheath edge
resulting in derivation of the Bohm criterion \eqrp{bounded:bohm}.
The baseline simulation (\fgr{bounded:cls_evolution}) confirms the
concept that ions are accelerated to the Bohm velocity several Debye
lengths from the wall.  However, as there are no sources, plasma is
depleting and a steady-state is never reached.  In order to reach a
steady-state, collisions and particle sources need to be added.

This is confirmed by analysis done by \cite{Riemann1990}.  He uses the
same normalized equations as in \ser{bounded:sheath:theory} but
applies them to the presheath instead of the sheath edge.  This
results in the following inequality,
\begin{align}
  \pfrac{e}{\xi}-\pfrac{\chi}{\xi} <
  \frac{1}{\tilde{j}_i}\pfrac{\tilde{j}_i}{\xi},
\end{align}
where $\tilde{j}_i$ is the normalized ion current, $\tilde{j}_i =
\sqrt{m_i/2T_e}/n_0 j_i$, and $e$ is the normalized energy, $e =
\frac{1}{2}m_iu_i^2/T_e$.  This inequality is satisfied if either ion
density increases approaching the wall, $\pfracb{\tilde{j}_i}{\xi}>0$,
or ions experience friction, $\pfracb{e}{\xi}<\pfracb{\chi}{\xi}$.  In
other words, for a steady-state sheath, collisions and/or ionization
are required in the presheath.

%---------------------------------------------------------------------
%\FloatBarrier
\subsection{Collisions}
\label{sec:bounded:col:col}

To balance the loss of high-energy electrons to the walls, collisions
must be included to replenish the electron tails if steady-state is to
be achieved. These collisions, however, should be infrequent enough
that the collisional mean-free-path is much longer than the sheath
width, allowing for proper simulation of collisionless sheaths.  This
work uses a simple Bhatnagar-Gross-Krook (BGK) operator
\citep{Bhatnagar1954}
\begin{align}\label{eq:bounded:bgk}
  S_{coll,s} = \nu_{\text{coll},s}\left(f_{M,s}-f_s\right),
\end{align}
where $f_{M,s}$ is a Maxwellian distribution function constructed
using the first three moments of $f_s$ and $\nu_{coll,s}$ is the
collision frequency.

The form of the BGK operator is the direct consequence of the
discussion in \ser{model:distf} where we show that a collection of
particles naturally relaxes towards the Maxwellian distribution. Note
that because the Maxwellian distribution is constructed from the first
three moments of $f_s$, total density, momentum, and energy are
conserved.

\subsubsection{Benchmarking the BGK Operator}

A simple benchmark can be performed by allowing a non-Maxwellian
distribution to relax to a Maxwellian using a collision operator.  In
the absence of fields and for distributions that are uniform in $x$,
the Boltzmann equation \eqrp{model:boltzmann} simplifies to
\begin{align*}
  \pfrac{f_s}{t} =  \nu_{\text{coll},s}\left(f_{M,s}-f_s\right).
\end{align*}
\fgr{bounded:bgk_relax} plots the particle distribution as a function
of velocity and time. The initial distribution is defined as
\begin{align*}
  f(v_x) = \begin{cases}
    1/3, & |v_x| \leq 1.5\,v_{th} \\
    0, & |v_x| > 1.5\,v_{th}
  \end{cases}
\end{align*}
The initial number density is therefore 1.0, bulk velocity is zero,
and the square of the thermal velocity is
0.75.\footnote{$\int_{-1.5}^{1.5} v_x^2/3\,dv_x = 0.75$} In the bottom
two panels, there are evolutions of integrated first moment and the
thermal part of the second moment, i.e., number density and
$nv_{th}^2$; the latter is a proxy for thermal energy.  Both moments
are steadily dropping even though they are supposed to be constant.
To understand why, we need to take a closer look at how the operator
is implemented in the code.

\begin{figure}[!htb]
  \centering
  \includegraphics[width=0.8\linewidth]{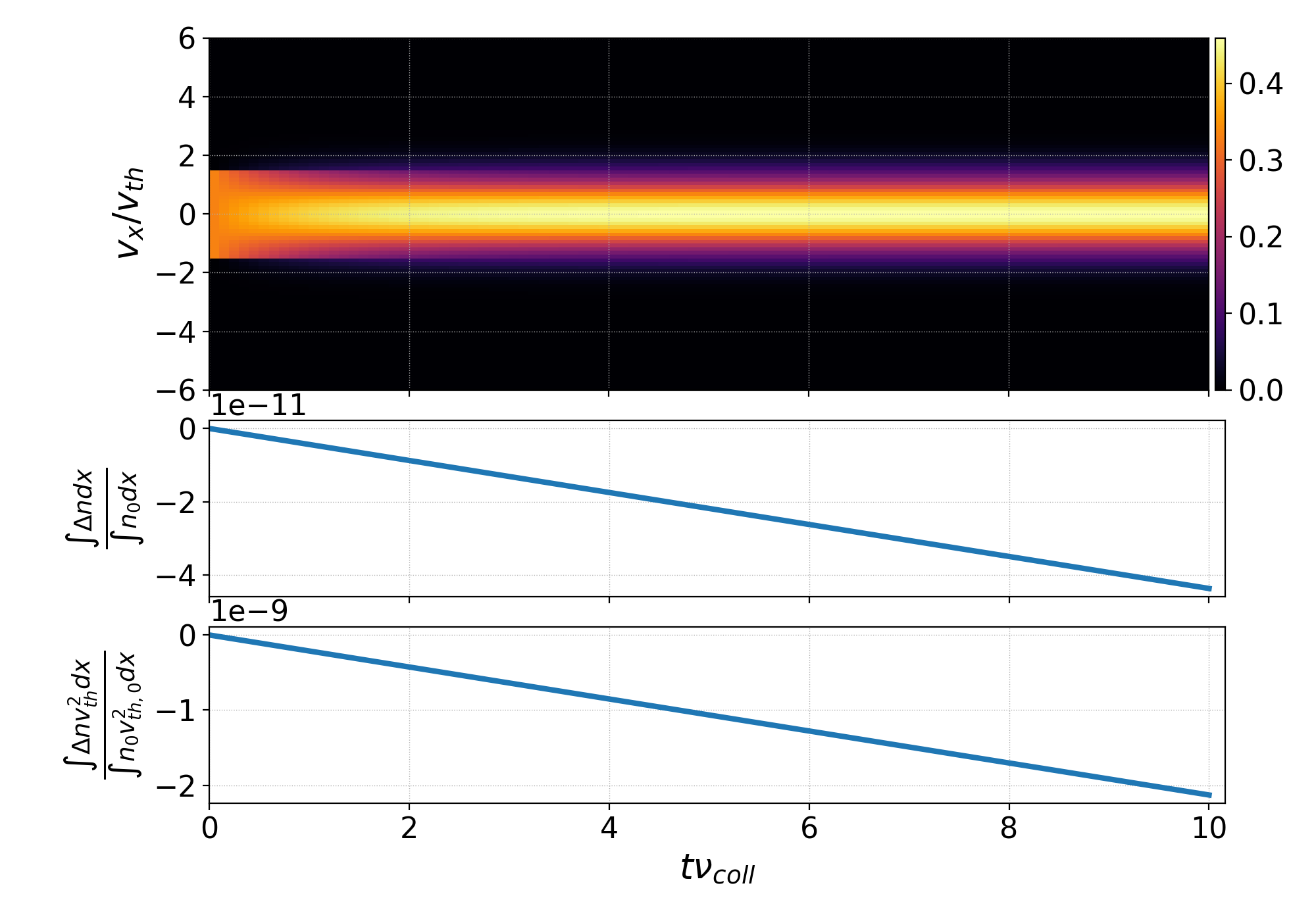}
  \caption[Relaxation due to the BGK operator]{Evolution of initially
    square distribution function to a the Maxwellian distribution due
    to the BGK collision operator (top panel).  Middle and bottom
    panels show corresponding development of total integrated density
    and thermal energy.}
  \label{fig:bounded:bgk_relax}
\end{figure}

Each time step, the distribution function moments are calculated using
the process described in \ser{model:moments}.  This calculation is
exact.  The error associated with the moment conservation is because
the Maxwellian distribution \eqrp{model:maxwellian}, which is
constructed from the calculated moments, must be expanded in terms of
the basis function.  As the distribution is an exponential function,
the expansion into a polynomial basis always results in an
error.\footnote{It is possible to use exponential functions as a basis
  but they are not used in \texttt{Gkeyll} and will not be discussed in
  this work. See, for example, \cite{Stolz1998} or
  \cite{Weniger1983}.}  In other words, if we construct a Maxwellian,
expand it into a polynomial basis, and recalculate the moments, the
final moments do not match the original values.  There is a
possibility to adjust the constructed Maxwellian distribution so its
moments match the original ones.  This is a current topic of research
in the \texttt{Gkeyll} collaboration and will be published in the near
future.

\fgr{bounded:bgk_relax_compare} provides a closer look into the
initial (blue) and final (green) lineouts of the distribution
function, i.e., distribution as a function of only velocity for fixed
$x$.  To better assess the effect of the BGK operator, Maxwellian
distribution \eqrp{model:maxwellian} constructed from the initial
parameters ($n=1.0$, $u=0.0$, and $v_{th}^2=0.75$) is included as well
(orange line below the green one).

\begin{figure}[!htb]
  \centering
  \includegraphics[width=0.8\linewidth]{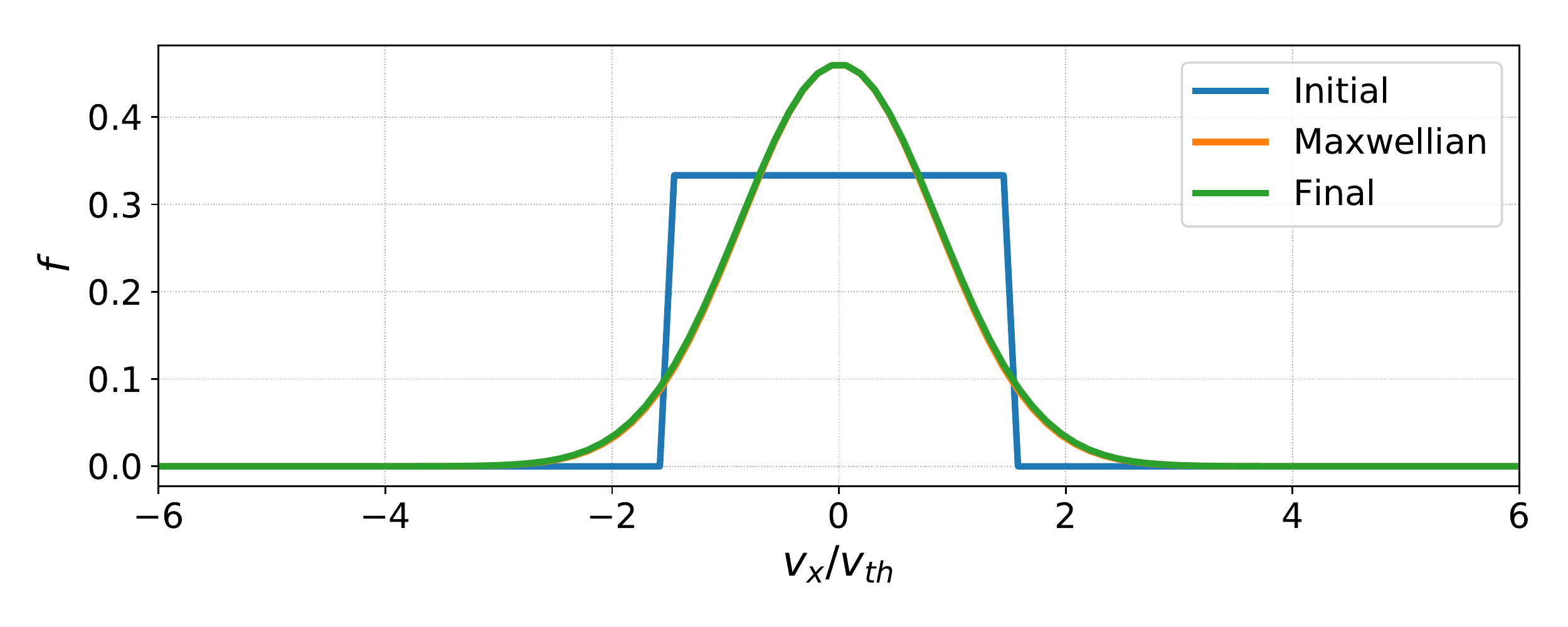}
  \caption[Comparison of initial and BGK relaxed
    distributions]{Comparison of lineouts of initially square and
    relaxed distribution functions (blue and green respectively).
    Maxwellian distribution \eqrp{model:maxwellian} constructed from
    the initial parameters ($n=1.0$, $u=0.0$, and $v_{th}^2=0.75$) is
    included as well (orange).}
  \label{fig:bounded:bgk_relax_compare}
\end{figure}

The Vlasov code coupled with the BGK operator is also benchmarked to
the \cite{Sod1978} shock tube problem.  Simulations are run with
classical parameters on the left and right side of the discontinuity
(Tab.\thinspace\ref{tab:bounded:sod}).
\begin{table}[!htb]
  \centering
    \caption[Initial parameters for the \cite{Sod1978} shock
      tube]{Initial parameters for the \cite{Sod1978} shock tube.}
  \begin{tabular}{ c|ccc }
    \toprule
    & $n$ & $u$ & $p=nv_{th}^2$ \\ \midrule
    Left & 1.0 & 0.0 & 1.0 \\
    Right & 0.125 & 0.0 & 0.1 \\\bottomrule
  \end{tabular}
  \label{tab:bounded:sod}
\end{table}

Note that for a 1X1V Vlasov simulation, pressure is given as
$p=nv_{th}^2$; therefore, effective $\gamma=c_p/c_v$ is
3.\footnote{$p/(\gamma-1) = \frac{1}{2}nv_{th}^2$} A fluid description
intrinsically assumes a Maxwellian distribution of particles, which
corresponds to a BGK operator with $\nu_{\text{coll}} \rightarrow
\infty$; in a kinetic code, the collisionality can be set arbitrarily.
In this case, the collision frequency is defined through Knudsen
number, $\mathit{Kn} = \lambda_{\text{mfp}}/L$, where
$\lambda_{\text{mfp}}$ is the mean-free-path.  Collision frequency is
then $\nu_{\text{coll}} = v_{\text{th},L}/\mathit{Kn}$.  A set of
simulations with different $\mathit{Kn}$ is shown in
\fgr{bounded:bgk_sod_comparison} together with the exact solutions of
the Euler equations in black.  For relatively high Knudsen number,
i.e., low collisionality, the solution is closer to a combination of
two rarefaction waves.  For high collisionality, the kinetic solution
matches the Euler prediction very well. Furthermore, fluid simulations
of a shock tube typically suffer from Gibbs phenomena (oscillations on
both sides of the shock) resulting in the need to use artificial
viscosity, filters, or limiters.  Interestingly, natural damping in
the kinetic model removes these features automatically and we see only
a minor undershoot in bulk velocity
(\fgr{bounded:bgk_sod_comparison}b) and temperature
(\fgr{bounded:bgk_sod_comparison}d).

\begin{figure}[!htb]
  \centering
  \includegraphics[width=0.8\linewidth]{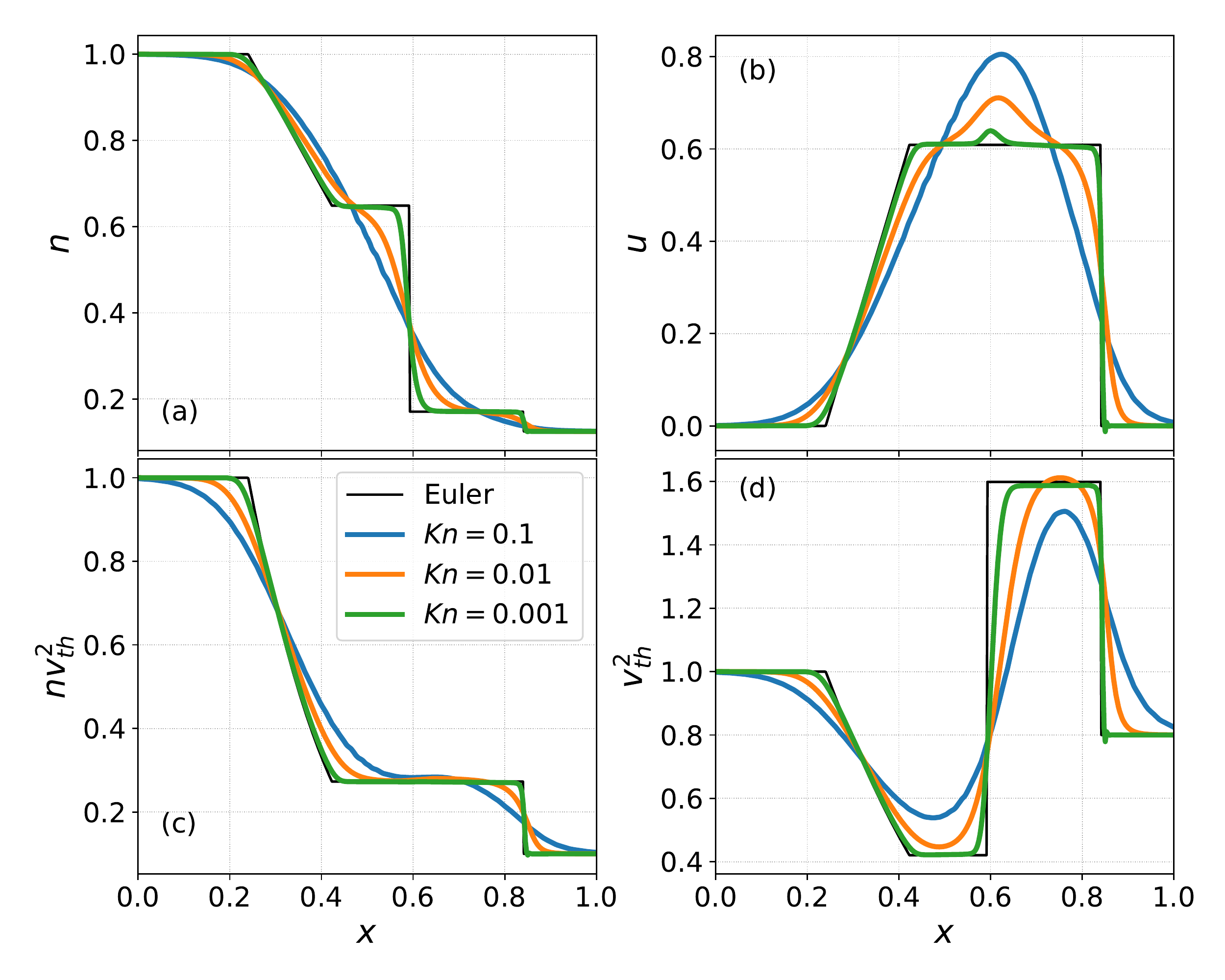}
  \caption[\cite{Sod1978} shock tube profiles for different
    collisionality]{Density (a), bulk velocity (b), pressure (c), and
    temperature (d) profiles for \cite{Sod1978} shock tube at
    $t=0.15$.  Plotted are three cases with varying Knudsen number,
    $Kn=\lambda_{\text{mfp}}/L$ together with exact Euler solution
    (black).}
  \label{fig:bounded:bgk_sod_comparison}
\end{figure}

The kinetic model also provides an interesting insight into shock
behavior through phase space plots.  The evolution is captured in
\fgr{bounded:bgk_sod}.  The top left frame shows the initial
conditions; it is clear that the left part of the domain contains a
population which is both much denser and hotter.  Without any
collisions, the evolution would look similar to the bounce example
(\fgr{model:distf}) in \ser{model:distf}; due to different speeds the
whole population would start ``tilting''.  In this case, however, the
BGK operator is continuously pushing the distribution towards a
Maxwellian and the shock and contact discontinuity form on the
right-hand-side of the domain.  The rarefaction wave on the
left-hand-side can be clearly explained as a ``lack'' of high speed
left-propagating particles as they are not replenished from the colder
right-hand-side.  Full listing of this case is available in the
appendix [\ref{list:bounded:sod}].

\begin{figure}[!htb]
  \centering
  \includegraphics[width=0.9\linewidth]{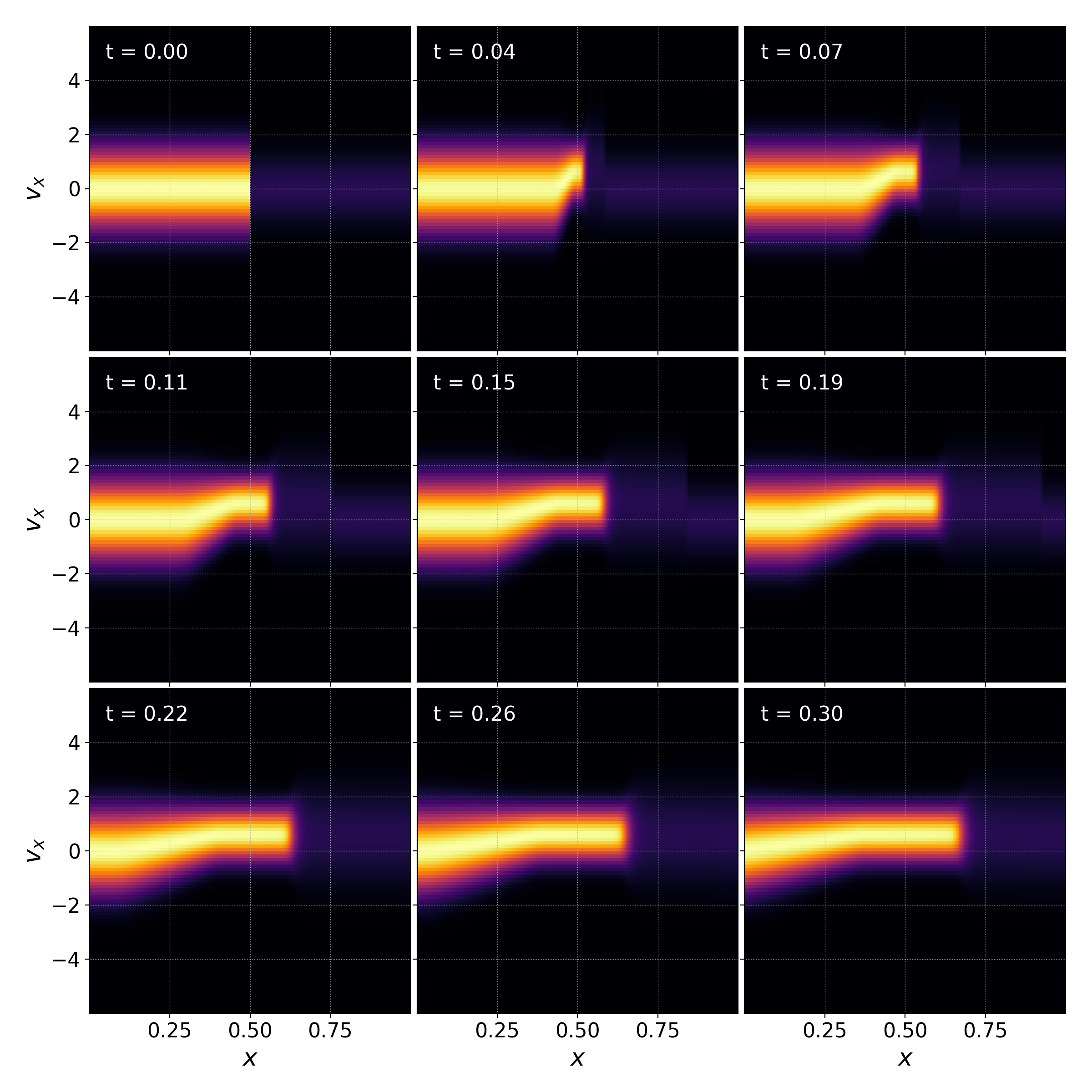}
  \caption[Phase space evolution of \cite{Sod1978} shock]{Phase space
    evolution of particle distribution function in \cite{Sod1978}
    shock tube.  Knudsen number is set to 0.001.  Full input file
    listing for this case is available in the appendix
    [\ref{list:bounded:sod}].}
  \label{fig:bounded:bgk_sod}
\end{figure}

\subsubsection{Collision Frequency}

In the previous examples, the collision frequency,
$\nu_{\text{coll}}$, is set manually\footnote{By setting the
  \texttt{collFreq} keyword of the collision object for each species.}
and is constant in both time and space.  This is useful for
benchmarking, however, for physics-relevant simulations, the collision
frequency must change together with plasma.  For that, \texttt{Gkeyll}
implements the \cite{Shi2017} formula
\begin{align}\label{eq:bounded:nu}
  \nu_{\text{coll}, s} =
  \frac{q_s^4}{6\sqrt{2}\pi^{3/2}\varepsilon_0^2m_s^2}\frac{n_s}{v_{th,s}^3}\mathrm{ln}(\Lambda),
\end{align}
where $\mathrm{ln}(\Lambda)$ is the Coulomb logarithm
\citep{Braginskii1965},\footnote{The formula comes from
  \cite{Braginskii1965} which is in cgs units. For that reason, number
  density must be converted from m$^{-3}$ to cm$^{-3}$.}
\begin{align*}
  \mathrm{ln}(\Lambda) = \begin{cases}
    23.4-1.15\,\mathrm{log}(n\times10^{-6}) +3.45\,\mathrm{log} \,T, & T <
    \SI{50}{eV} \\
     25.3-2.3\,\mathrm{log}(n\times10^{-6}) +3.45\,\mathrm{log} \,T, & T >
     \SI{50}{eV}
  \end{cases}
\end{align*}

\subsubsection{Cross-species Collisions}

Until this point, the BGK operator included collisions between the
same species only.  For collisions between different species, which
would introduce drag into system, the operator needs to be extended.
\cite{Greene1973} provides the following system for electrons and
ions,
\begin{align}\begin{aligned}
    S_e &= \nu_{\text{coll},ee} (f_{M,e}-f_e) +
    \nu_{\text{coll},ei}(f_{M,te}-f_e), \\
    S_i &= \nu_{\text{coll},ii} (f_{M,i}-f_i) +
    \nu_{\text{coll},ie}(f_{M,ti}-f_i),
\end{aligned}\end{align}
where $\nu_{\text{coll},ee} \sim \nu_{\text{coll},ei}$ and
$\nu_{\text{coll},ei}$ and $\nu_{\text{coll},ie}$ differ by mass ratio
to capture the fact that ions are only weakly affected by collisions
with electrons.

The ``cross-Maxwellians'' $f_{M,te}$ and $f_{M,ti}$ are defined using
the following combined moments,
\begin{align*}
  n_{te} =& n_e, \quad n_{ti} = n_i, \\
  \bm{u}_{te} =& \frac{1}{2}(\bm{u}_e+\bm{u}_i) -
  \frac{1}{2}\beta(\bm{u}_e-\bm{u}_i),\\
  \bm{u}_{ti} =& \frac{1}{2}(\bm{u}_i+\bm{u}_e) -
  \frac{1}{2}\beta(\bm{u}_i-\bm{u}_e),\\
  T_{te} =& \frac{m_eT_i+m_iTe}{m_i+m_e} -
  \beta\frac{m_e}{m_i+m_e}(T_e-T_i) +
  \frac{1}{6}(1-\beta^2)\frac{m_em_i}{m_i+m_e}(\bm{u}_e-\bm{u}_i)^2 +
  \\
  &\frac{1}{12}(1+\beta)^2\frac{m_i-m_e}{m_e+m_i}m_e(\bm{u}_e-\bm{u}_i)^2,\\
  T_{ti} =& \frac{m_eT_i+m_iTe}{m_i+m_e} -
  \beta\frac{m_e}{m_i+m_e}(T_i-T_e) +
  \frac{1}{6}(1-\beta^2)\frac{m_em_i}{m_i+m_e}(\bm{u}_e-\bm{u}_i)^2 +
  \\
  &\frac{1}{12}(1+\beta)^2\frac{m_e-m_i}{m_e+m_i}m_e(\bm{u}_e-\bm{u}_i)^2,
\end{align*}
where $\beta$ is arbitrary \citep{Greene1973}.  It should be pointed
out that the ``cross-moments'' are symmetric and, therefore, could be
used for arbitrary species.  However, the calculation of
``cross-temperatures'' does not guarantee positive values and can
easily result in a crash of the simulation.  This typically happens
when applying the contribution of lighter species onto heavier ones,
e.g., contribution from electron to ions. While using the
cross-species collision terms with \texttt{Gkeyll}, these
contributions are usually ignored due to the assumption that there is
negligible impact of these collisions on the heavier species.

A demonstration of the BGK operator with cross collision terms is in
\fgr{bounded:bgk_cross}, which shows interaction of relatively cold
particle population with non-zero bulk velocity reacting with a warmer
population.  Within a few collision periods, the populations reach an
equilibrium state.

\begin{figure}[!htb]
  \centering
  \includegraphics[width=0.8\linewidth]{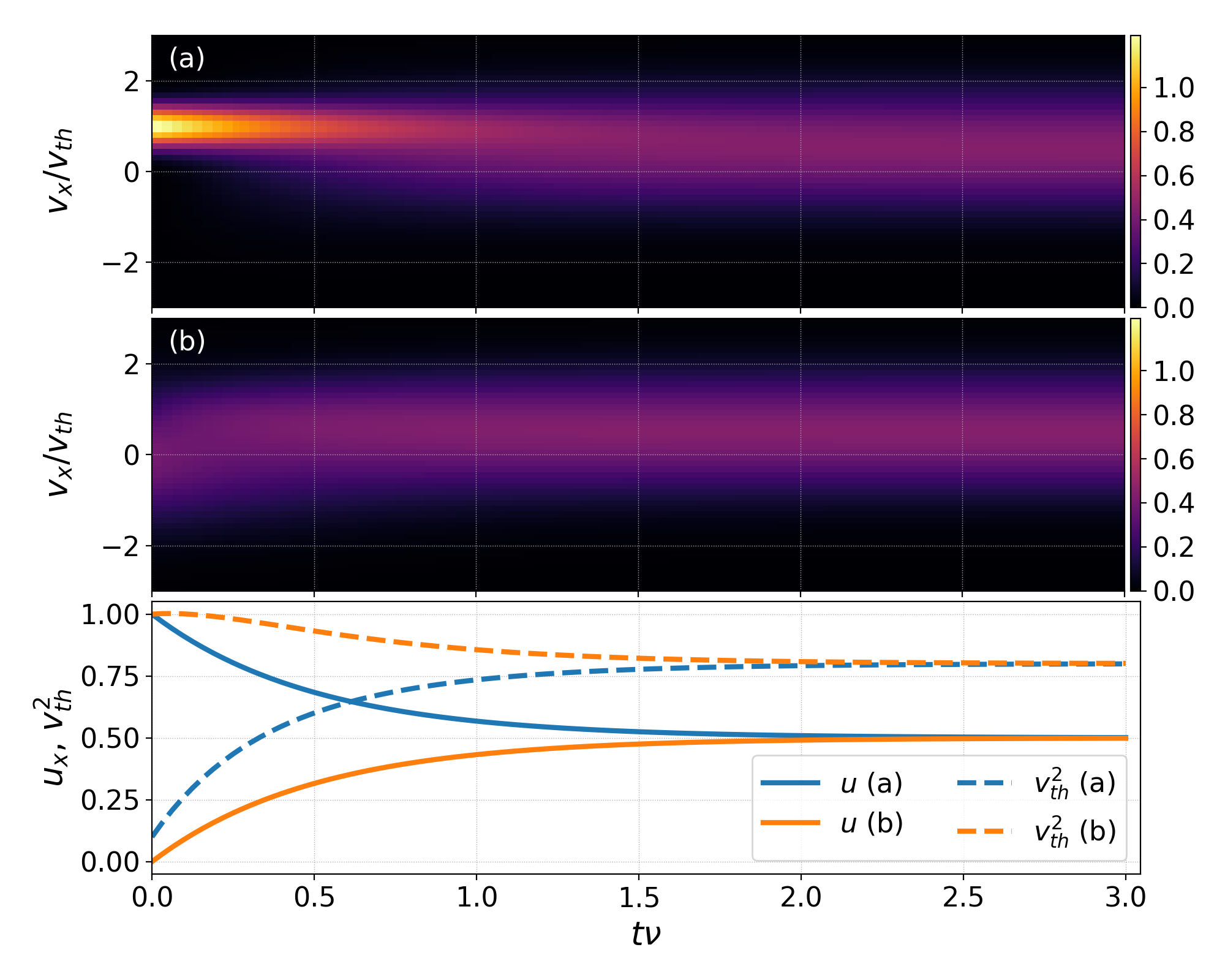}
  \caption[Mutual interaction of two populations]{Mutual interaction
    of two distinct neutral populations caused by cross-species BGK
    collisions implemented bases on \cite{Greene1973}.  The two top
    panels show individual distribution functions while the bottom
    panel captures the time evolution of bulk velocities and
    temperatures.}
  \label{fig:bounded:bgk_cross}
\end{figure}

\subsubsection{Collisional Sheath Simulations}

Classical sheath theory (\ser{bounded:sheath:theory}) assumes the
sheath itself is collision-less; however, as it was mentioned above,
collisions are required to replenish high energy tail as particles are
lost to wall.  Therefore, it is common to apply collisions only in the
presheath.  Here instead, the BGK operator is applied on the full
domain with a collision frequency from \eqr{bounded:bgk}.  This
approach naturally decreases the collisionality in the sheath as the
number density drops significantly.

A common argument against this approach is that collisions in the
sheath thermalize electrons which should be non-Maxwellian there; the
next example shows that it is not the case.  The simulation uses the
same set of Hall thruster relevant SI parameters as in
\ser{bounded:sheath:sims}.  That, however, results in a presheath
collision frequency on the order of \SI{1e6}{s^{-1}}. The frequency
corresponds to a mean free path of \SI{30}{cm} \citep{Boeuf2017},
which is much bigger than the simulation domain.  Therefore the domain
would need to be significantly increased in order to capture the
presheath thermalization, resulting in much higher computational cost.
Alternatively, the collision frequency can be artificially increased
to demonstrate the effect.  Interestingly, with the collision
frequency increased by a factor of 1000 with respect to the
\cite{Braginskii1965} formula \eqrp{bounded:bgk}, the sheath
distribution function retains the typical non-Maxwellian profile.  The
results are captured in \fgr{bounded:sheath_coll}, which shows
velocity profiles of electron distribution function $\lambda_D/6$ from
the right wall.  Solid lines mark the case with the BGK operator while
dashed lines capture simulation results from
\ser{bounded:sheath:sims}, i.e., without any collisions.  For
distribution functions in the sheath, we see the depletion of
electrons with velocities $|v_x|>\sqrt{q_e\phi/m_e}$ as these are not
reflected by the potential, $\phi$, and are lost to the wall.  Near
the right wall, the depletion is in negative (leftward) velocities.
We clearly see this effect in both collisional and collisionless cases
at $t\omega_{pe}=50$ (blue lines) with the green line marking the
critical velocities.\footnote{Note that the plot is semilogarithmic,
  i.e., exponential Maxwellian distribution looks like a quadratic
  function.}  At $t\omega_{pe} = 100$ (orange lines), which is higher
than the transition time in this simulations, both cases show
significant discrepancies.  In the collisionless case (dashed) the
depleted tail ``propagates'' through the domain until it reaches the
other side.  The depletion of positive velocities for orange dashed
line is, therefore, an effect of the opposite wall.  In the
collisional case, the tail gets replenished in the presheath and right
side of the distribution function remains Maxwellian, with minimal
differences with respect to the solution at $t\omega_{pe} = 50$.

\begin{figure}[!htb]
  \centering
  \includegraphics[width=0.9\linewidth]{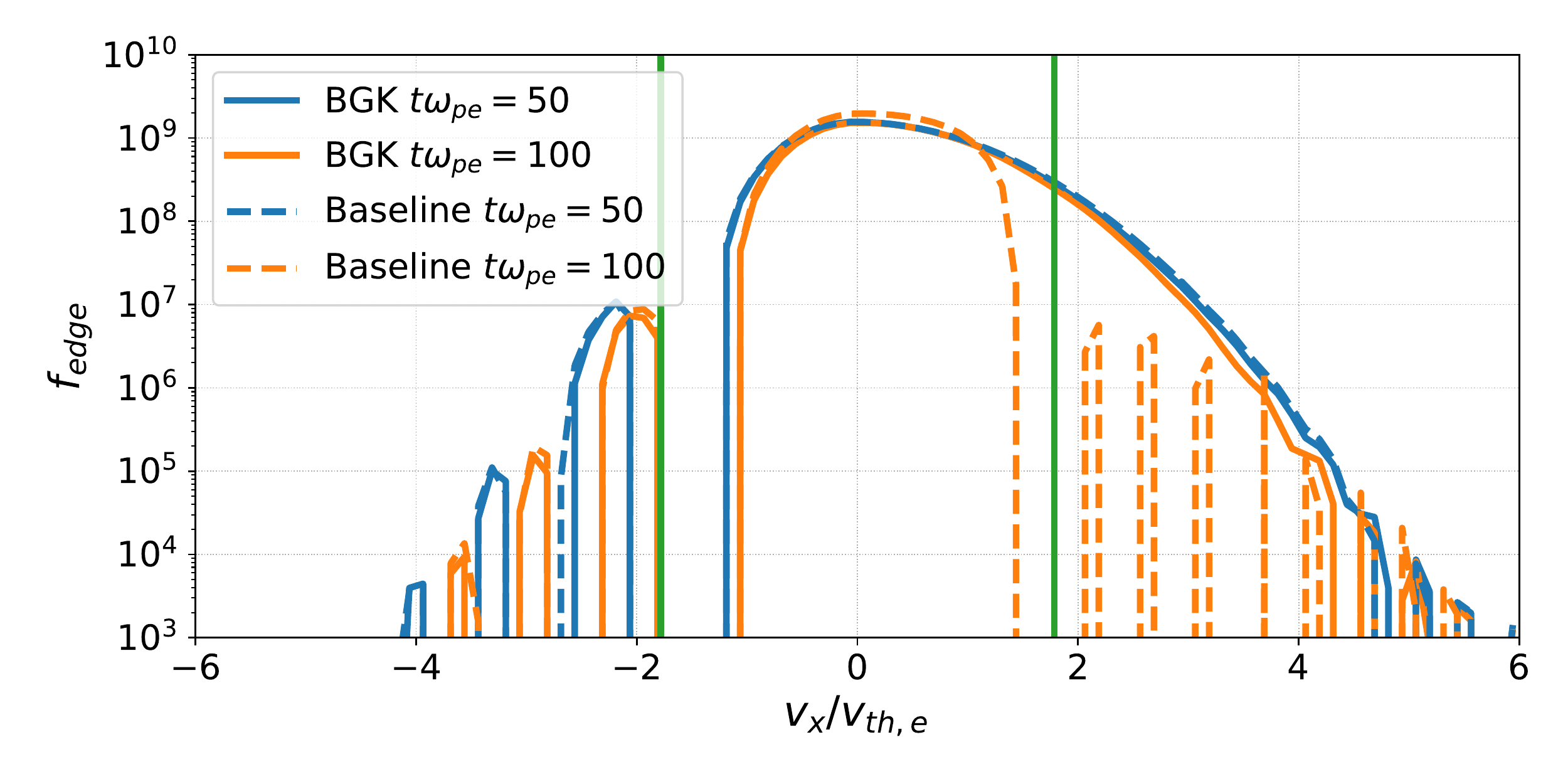}
  \caption[Sheath distribution tail in collisional and collisionless
    cases]{Comparison of electron distribution function lineouts at
    $\lambda_D/6$ from the right wall.  Figure captures both
    collisional (solid lines) and collisionless (dashed lines) cases
    together with critical velocities $\sqrt{q_e\phi/m_e}$
    corresponding to the electric field from the collisional case at
    $t\omega_{pe} = 50$.  While leftward propagating particles are
    depleted for both cases, distribution tail gets replenished in the
    presheath with the BGK operator.}
  \label{fig:bounded:sheath_coll}
\end{figure}

%---------------------------------------------------------------------
%\FloatBarrier
\subsection{Particle Source}
\label{sec:bounded:coll:ion}

Collision terms, discussed in the previous section, are not sources as
they conserve number of particles and energy.  Therefore, additional
particle sources need to be included to replenish those lost to the
wall.

Our first attempts to replenish the particles focused on forcing
inflow Maxwellian distribution directly tied to the outflow from the
domain.  This resulted in a source sheath at the inflow and positive
feedback caused simulations to crash.  Hence, we switched to
volumetric source terms, more specifically electron impact ionization,
\begin{align*}
  e^{-} + n \rightarrow i^{+} + 2e^{-} - E_{\text{ion}},
\end{align*}
which can easily be written in terms of particle distributions,
\begin{align}\begin{aligned}\label{eq:bounded:ion}
  S_{\text{ion},e} &=
  f_n(\bm{x}, \bm{v}) \int_{\mathcal{V}} \sigma(|\bm{v}-\bm{v}'|)
  |\bm{v}-\bm{v}'| f_e(\bm{x}, \bm{v}') \,d\bm{v}',\\
    S_{\text{ion},i} &=
  f_n(\bm{x}, \bm{v}) \int_{\mathcal{V}} \sigma(|\bm{v}-\bm{v}'|)
  |\bm{v}-\bm{v}'| f_e(\bm{x}, \bm{v}') \,d\bm{v}',\\
    S_{\text{ion},n} &=
  -f_n(\bm{x}, \bm{v}) \int_{\mathcal{V}} \sigma(|\bm{v}-\bm{v}'|)
   |\bm{v}-\bm{v}'| f_e(\bm{x}, \bm{v}') \,d\bm{v}',
\end{aligned}\end{align}
where $\sigma(|\textbf{v}-\textbf{v}'|)$ is the ionization
differential cross-section with units of m$^2$. Typically,
the electron thermal velocity is much higher than the
neutral thermal velocity.  With this assumption, $|\bm{v}'| \gg |\bm{v}|$,
\eqr{bounded:ion} can be simplified to,
\begin{align*}
  S_{\text{ion},ei} \approx
  f_n(\bm{x}, \bm{v}) \int_{\mathcal{V}} \sigma(|\bm{v}'|)
  |\bm{v}'| f_e(\bm{x}, \bm{v}') \,d\bm{v}',
\end{align*}
which makes the implementation much simpler.

Differential ionization cross-sections can be calculated, for example,
from the Binary-Encounter-Bethe (BEB) model \citep{Kim1994},
\begin{align}\label{eq:bounded:beb}
  \sigma_{\text{BEB}}(n,T) = \frac{4\pi a_0^2 N (R/B)^2}{t+ (t+1)/n}
  \left[\frac{Q \mathrm{ln}\,t}{2} \left(1-\frac{1}{t^2}\right) +
    (2-Q)\left(1 - \frac{1}{t} - \frac{\mathrm{ln}\,t}{t+1} \right)
    \right]
\end{align}
where $t=T/B$, $a_0=\SI{0.52918e-10}{m}$ is the Bohr radius, and
$R=\SI{13.6057}{eV}$ is the hydrogen ionization energy.  $B$ [eV], $U$
[eV], $N$, and $Q$ are constants varying for each atom or molecule and
can be found at the National Institute of Standards and Technology
(NIST) web
page.\footnote{\url{https://www.nist.gov/pml/electron-impact-cross-sections-ionization-and-excitation-database}}
Note that temperature is given in eV as in the rest of this work.
\cite{Kim1994} show a good agreement between the BEB model and
experimental data form units to thousands of eV.

At this point, everything is available to compute the integral.
However, calculation of the integrals during run-time could be
computationally expensive.  This can be alleviated by precomputing the
integrals in the same manner as boundary conditions discussed in
detail later in \ser{bounded:pmi:r}.  Alternatively, as the ionization
happens predominately in the presheath where particle distributions
are close to Maxwellian, a fluid approach can be used,\footnote{This
  is not the case for boundary conditions discussed later as the
  distributions at the wall can be significantly non-Maxwellian.}
\begin{align*}
  S_{\text{ion},ei} \approx
  f_n(\bm{x}, \bm{v}) \int_{\mathcal{V}} \sigma(|\bm{v}'|)
  |\bm{v}'| f_e(\bm{x}, \bm{v}') \,d\bm{v}' = f_n(\bm{x}, \bm{v}) n_e
  \langle \sigma v_e\rangle.
\end{align*}
\cite{Cagas2017s} uses a formula from \cite{Stangeby2000},
\begin{align*}
  \left\langle \sigma v_e\right\rangle =
  \frac{2u_B}{L}\left(\frac{\pi}{2}-1\right),
\end{align*}
which sets the ionization term to approximately match outflow from the
domain.  Another option is to use a semi-empirical model
to calculate $\langle \sigma v_e\rangle$.  \cite{Voronov1997} provides
a fitting formula for the average cross-section parameters for
elements up to $Z=28$,\footnote{Factor $10^{-6}$, which is not
  included in \cite{Voronov1997}, is due to conversion from
  centimeters to meters.}
\begin{align}\label{eq:bounded:voronov}
  \langle \sigma v_e\rangle =
  A\frac{1+P\sqrt{E_{\text{ion}}/T}}{X+E_{\text{ion}}/T}(E_{\text{ion}}/T)^K
  \exp(-E_{\text{ion}}/T)\times10^{-6}\,\mathrm{m}^{3}\mathrm{s}^{-1}.
\end{align}
Together with the fitting parameters, \cite{Voronov1997} give a usable
energy range for all the atoms and ions.  For example for hydrogen,
\eqr{bounded:voronov} is usable for electron temperatures between
\SI{1}{eV} to \SI{2e4}{eV}.

Comparison of simulations from \ser{bounded:sheath:sims} and
simulations including \cite{Voronov1997} ionization are presented in
\fgr{bounded:sheath_ion}.  The depletion is improved by the ionization
(after $1000/\omega_{pe}$, density drops by 4\% instead of previous
8\%) but the simulation does not reach a steady-state.  It is probably
caused by the fine balance of the ionization term; the density change
due to ionization is essentially $\pfracb{n}{t} \sim n$.  Therefore,
if at any point during the simulation the ionization does not
replenish the outflow, the density drops, which in turn decreases
``efficiency'' of the ionization.  On the other hand, if the
ionization produces more particles than those that leave the domain,
the number density starts growing exponentially (assuming ample supply
of neutrals).  This balance might require an additional feedback loop
and will be addressed in future work.

\begin{figure}[!htb]
  \centering
  \includegraphics[width=0.9\linewidth]{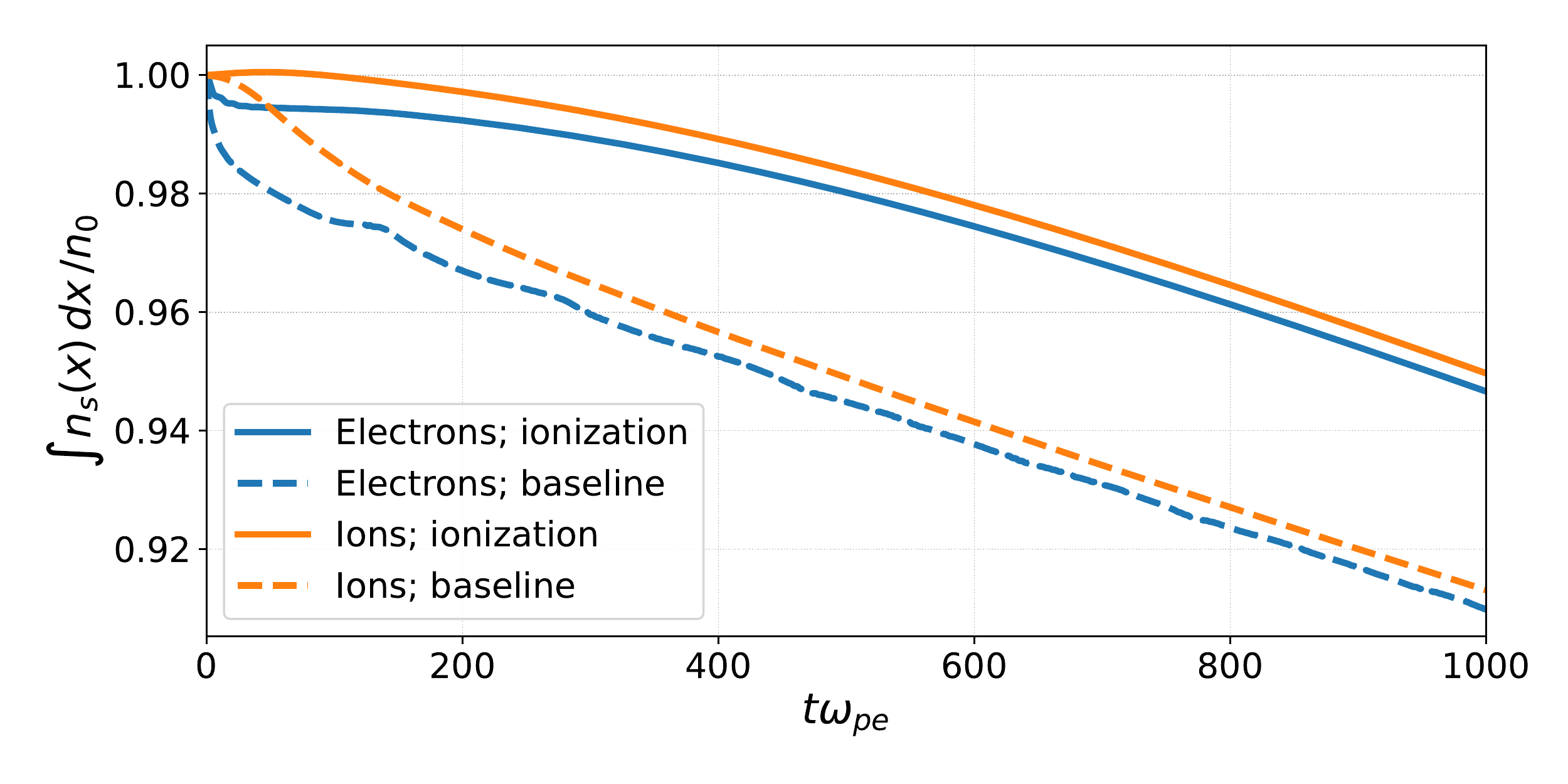}
  \caption[Integrated number densities with and without
    \cite{Voronov1997} ionization]{Comparison of relative electron and
    ion number densities with (solid lines) and without (dashed lines)
    \cite{Voronov1997} ionization term.}
  \label{fig:bounded:sheath_ion}
\end{figure}

Finally, it should be pointed out that equations in the same form as
\eqr{bounded:ion} can be used to capture recombination and charge
exchange \citep{Meier2012}; the only difference is in respective
cross-section functions.

%=====================================================================
\FloatBarrier
\section{Comparison of the Models and the Temperature Anisotropy}
\label{sec:bounded:temp}

A significant portion of \cite{Cagas2017s} is dedicated to a
comparison of the kinetic sheath model, described above, and the
five-moment two-fluid model from \ser{model:fluids}.  The kinetic
simulations are extended to 1X2V (similar to Weibel simulations in
\ser{weibel:sims}), where one direction is parallel to the wall and
the other is perpendicular, in order to capture the different
evolution of parallel and perpendicular temperatures.  An alternative
approach to capture the temperature anisotropy is to solve a
perpendicular temperature evolution equation,
\begin{align}\label{eq:bounded:Tperp}
  \frac{\partial}{\partial t}\left( nT_\perp\right) +
  \frac{\partial}{\partial x}\left(u_x nT_\perp\right)=\nu
  n\left(T-T_\perp\right),
\end{align}
along with the Vlasov equation
\eqrp{model:vlasov}. \eqr{bounded:Tperp} describes the advection of
the perpendicular temperature and its isotropization to the parallel
temperature due to collisions.  However, this approach is not used in
this work.  The particle source from \ser{bounded:coll:ion} is
implemented into the fluid model using the approach of
\cite{Meier2012}.

To simulate ideally absorbing walls in the fluid model,
\cite{Cagas2017s} use vacuum boundary conditions.  This approach is
analogous to kinetic simulations where the outgoing particle
distribution function is set to zero.  A comparisons of density, electric
field, and bulk velocity profiles are in \fgr{bounded:comp_profiles}.
With a proper Riemann solver, the fluid simulation with the vacuum boundary condition closely reproduces the kinetic solution.  The set of simulations presented in this section
is run with dimensionless units.  The electron and ion populations are initialized with $T_e/T_i =
1.0$.  \fgr{bounded:comp_profiles} shows the solution at $t\omega_{pe}
= 200$.

\begin{figure}[!htb]
  \centering
  \includegraphics[width=0.8\linewidth]{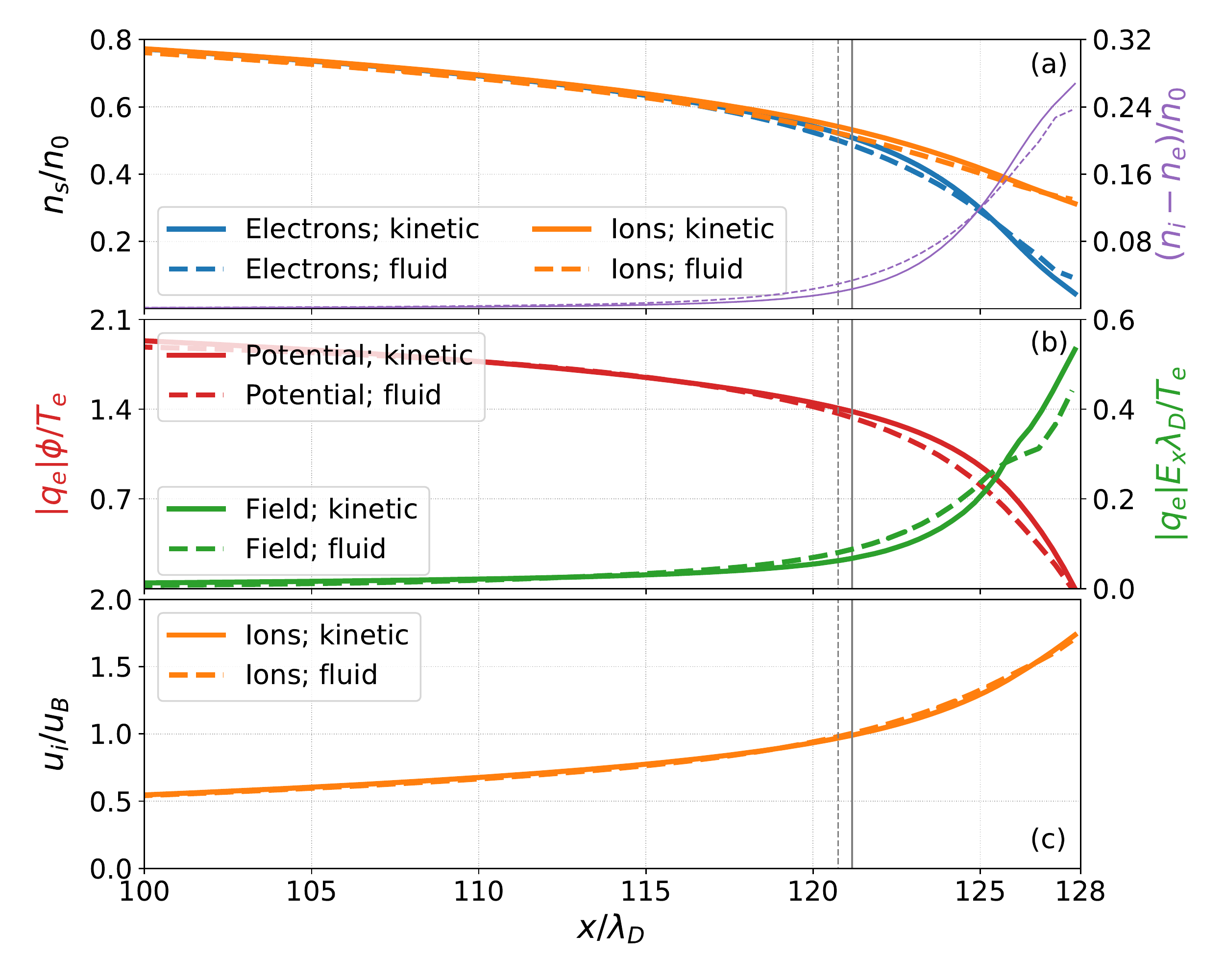}
  \caption[Comparison of sheath profiles between kinetic and fluid
    simulations]{Comparison of electron and ion number densities (a),
    electric field (b), and ion bulk velocity (c) between kinetic
    (solid lines) and fluid (dashed lines) simulations in the region
    near an ideally absorbing wall.  Vertical gray lines mark the
    crossing of the Bohm velocity \eqrp{bounded:bohm2}.  Violet lines
    in the panel (a) mark the difference between electron and ion
    densities.  Simulations are run with dimensionless units and
    initially the temperatures of the electrons and ions are set to $T_e/T_i=1$.  Solution is shown at $t\omega_{pe} = 200$.}
  \label{fig:bounded:comp_profiles}
\end{figure}

\fgr{bounded:comp_profiles} captures the crossing of the Bohm velocity
\eqrp{bounded:bohm2}, which is marked with vertical gray lines (solid
for kinetic and dashed for fluid results).  The panel (a) also shows
the absolute difference between the electron and ion number densities
normalized to the initial uniform density.  It is interesting to note
that as the electric field extends into the presheath and so does the
difference in the number densities, which is proportional to the
charge density.  There is no clear boundary between a quasi-neutral
presheath and non-neutral sheath as it is described in the two-scale
theory.  Therefore, practical use of the sheath edge definition
\eqr{bounded:edge} for noise-free continuum simulations requires
setting an arbitrary threshold, e.g., to 1\% or 10\% as in
\cite{Cagas2017s}.  For reference, at the distance where solutions
cross the Bohm velocity, the charge inequality is equal to 2.3\% for
the kinetic code and 2.9\% for the fluid code.

The density, velocity, and electric field profiles in
\fgr{bounded:comp_profiles} agree remarkably well between the two models.  However, there are discrepancies in temperatures; see
\fgr{bounded:comp_Tprofiles}.  As the five-moment fluid model uses a
scalar pressure closure, it captures only a single scalar
temperature\footnote{Temperature is calculated as $$v_{th,s}^2 =
  (\gamma-1)\left(\frac{\mathcal{E}_s}{m_sn_s} -
  \frac{1}{2}m_su_s^2\right).$$ Note that \texttt{Gkeyll} also
  includes a ten-moment fluid model with full pressure tensor but it
  is not used for this work.}  while the 1X2V kinetic code evolves the
parallel ($v_{thx}$) and perpendicular ($v_{thy}$) velocities and, as a result, both parallel and perpendicular temperatures.  As mentioned in \ser{bounded:sheath:sims}, particles leaving the
domain through the wall undergo decompression cooling in the direction
parallel to the wall \citep{Tang2011}.  The parallel temperature is
then equalized with the perpendicular temperature through collisions.
It is observed that this effect is more apparent for electrons
(\fgr{bounded:comp_Tprofiles}a) due to their higher collision
frequency with respect to ions (\fgr{bounded:comp_Tprofiles}b). The
ion temperature (\fgr{bounded:comp_Tprofiles}b) is in good agreement
between continuum kinetic and fluid simulations (isotropic fluid
temperature lies in between the ion parallel and perpendicular
temperatures). However, the electron temperature has more significant
differences between the kinetic and fluid results.  In the presheath, the kinetic parallel and perpendicular temperatures are equal to each other but lower in comparison to the fluid result.

\begin{figure}[!htb]
  \centering
  \includegraphics[width=0.8\linewidth]{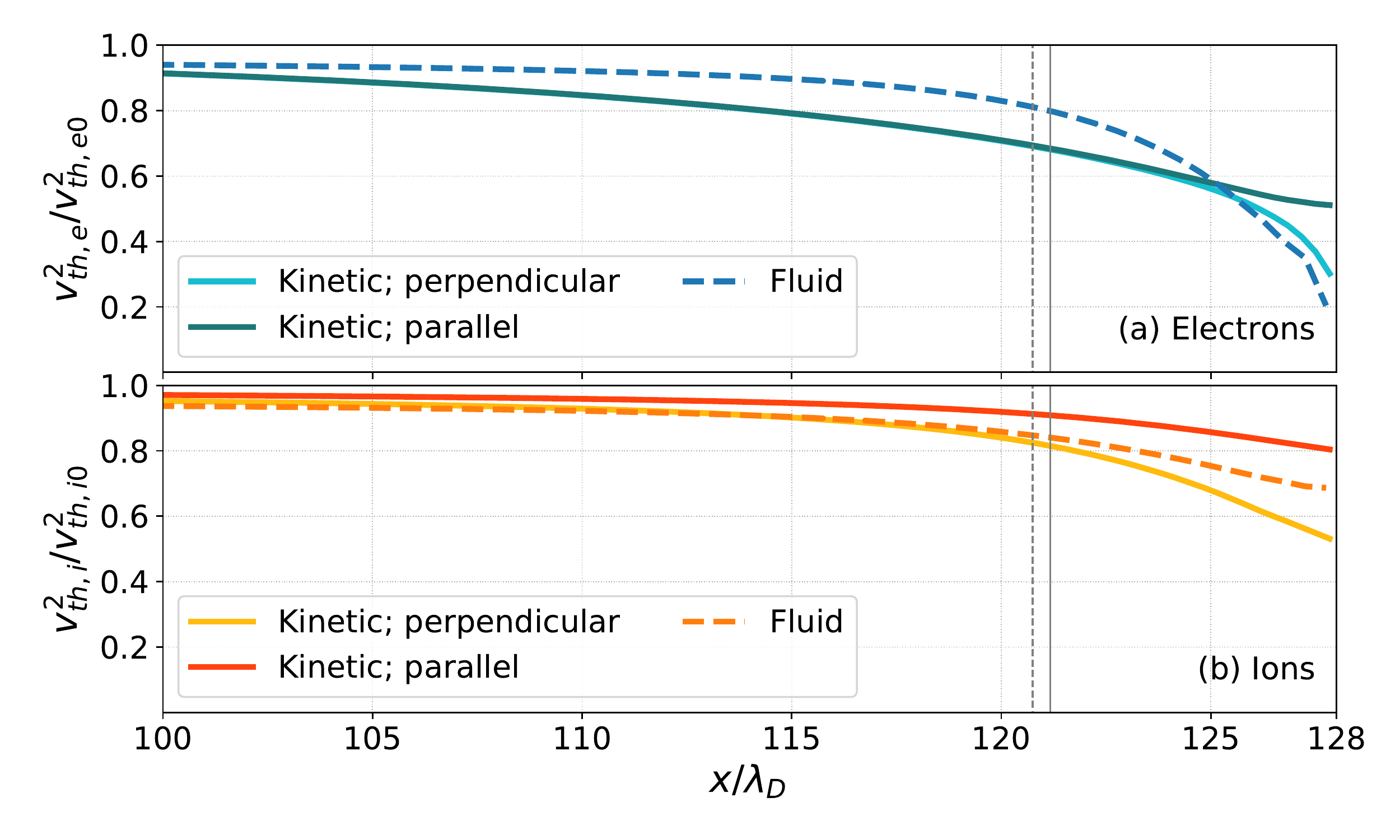}
  \caption[Sheath temperature comparison between kinetic and fluid
    simulations]{Continuation of \fgr{bounded:comp_profiles} showing a
    comparison of temperature profiles in the sheath region between
    the 1X2V kinetic model (solid and dot-dashed lines) and the
    five-moment two-fluid model (dashed lines).  The upper panel (a)
    captures electron temperatures while the bottom panel (b) shows
    ion temperatures.  Vertical gray lines mark the crossing of the
    Bohm velocity \eqrp{bounded:bohm2}. Simulations are run with
    dimensionless units and initially $T_e/T_i=1$.  Solution is shown
    at $t\omega_{pe} = 200$.}
  \label{fig:bounded:comp_Tprofiles}
\end{figure}

In \ser{model:distf}, the first three distribution function moments
are discussed (density, momentum, and energy).  However, in order to
explain the discrepancy in the electron temperatures
(\fgr{bounded:comp_Tprofiles}b), higher moments are required.  The
second moment of the Vlasov equation \eqrp{model:vlasov} leads to the
energy conservation equation,
\begin{align}
  \frac{\partial \mathcal{E}}{\partial t} + \frac{1}{2}\frac{\partial
    \mathcal{Q}_{iik}}{\partial x_k} = nq\bm{u}\cdot\bm{E},
\end{align}
where 
\begin{align*}
  \mathcal{E} = \frac{3}{2}p + \frac{1}{2} mn\bm{u}^2
\end{align*}
is the particle energy and the third moment of the particle distribution function
\begin{align*}
  \mathcal{Q}_{ijk} = m \int_{\mathcal{V}} v_i v_j v_k f \,d\bm{v}.
\end{align*}
is the heat flux tensor.  A contraction of $\mathcal{Q}_{ijk}$ gives
the particle energy-flux density and can be expanded as follows,
\begin{align}\label{eq:bounded:Q}
  \frac{1}{2} \mathcal{Q}_{iix} = \underbracket{q_x +
    u_x\Pi_{xx}}_{\text{non-ideal}} +
  \underbracket{\frac{5}{2}pu_x+\frac{1}{2}mnu_x^3}_{\text{ideal}},
\end{align}
where $\Pi_{xx}$ is the parallel component of the stress tensor, $p$
is the pressure, and
\begin{align}
  q_x =
  \frac{1}{2}m\int_{-\infty}^\infty\int_{0}^\infty\left(w_x^2 +
  v_\perp^2\right)w_xf(v_x,v_\perp)2\pi v_\perp \,dv_\perp dv_x
\end{align}
is the heat flux vector in the plasma frame ($w_x = v_x-u_x$).
Individual terms of \eqr{bounded:Q} are plotted in
\fgr{bounded:temp_Q} for the electrons. The parts of the tensor responsible for the decompression cooling (red lines) are the dominant terms and are in good agreement between the kinetic and fluid models.  The five-moment fluid model used
here does not capture the kinetic physics of the heat flux vector and
the stress tensor.  The stress tensor (green line only visible at the wall) is, in
this case, negligible.  The heat flux vector (orange line) is also negligible in the
bulk plasma where the distribution function is thermalized by
collisions; however, becomes significant within $50\,\lambda_D$ from
the wall and explains the differences in electron temperature between
the kinetic and fluid results \citep{Cagas2017s}.

\begin{figure}[!htb]
  \centering
  \includegraphics[width=0.8\linewidth]{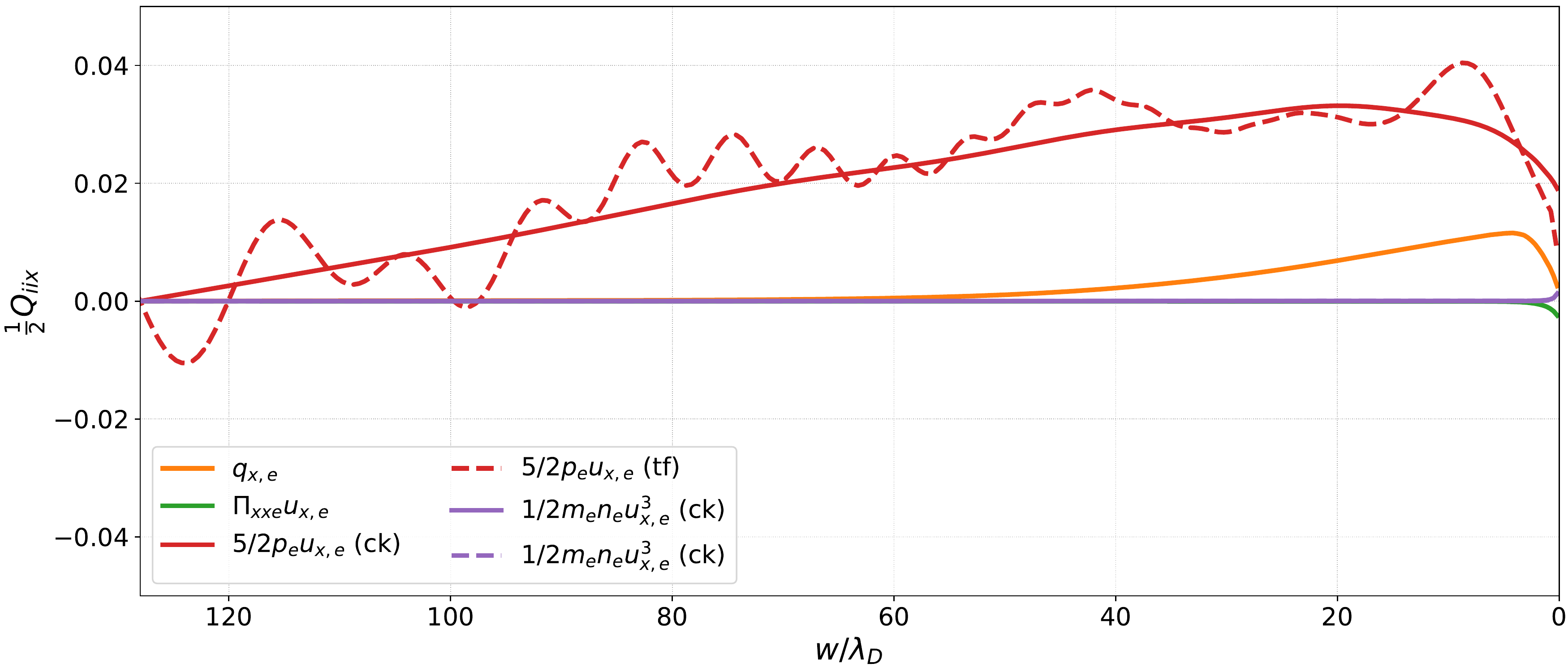}
  \caption[Comparison of heat flux terms]{Individual terms of the
    expanded heat flux \eqrp{bounded:Q}. Heat flux vector in the
    plasma frame, $q$, and stress tensor are only available in the
    continuum kinetic simulations (marked as ck). Simulations are
    evolved from initial conditions for 100$/\omega_{pe}$.}
  \label{fig:bounded:temp_Q}
\end{figure}

\cite{Cagas2017s} also extend the comparison of the kinetic and fluid
models to several different temperatures ratios.
\fgr{bounded:comp_scan} captures the second moment of the ion
distribution (flux) at the wall (a), sheath width determined based on
reaching of the Bohm velocity \eqrp{bounded:bohm2} (b), and the
potential drop over the sheath region for the initial electron to ion
temperature ratios of 0.1, 0.2, 0.5, 1.0, 2.0, 5.0, and 10.0.

\begin{figure}[!htb]
  \centering
  \includegraphics[width=0.8\linewidth]{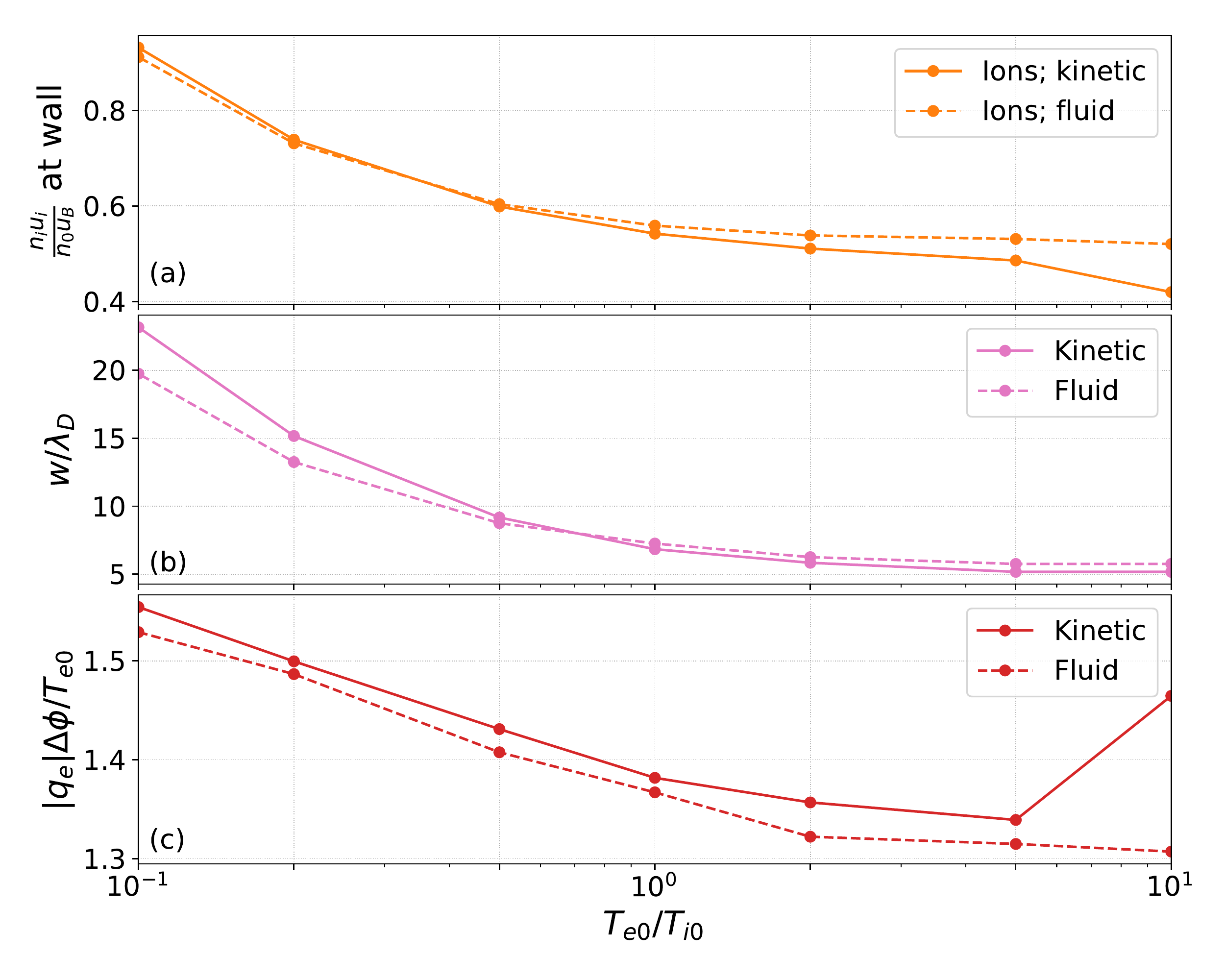}
  \caption[Comparison of the kinetic and fluid models for different
    temperature ratios]{Comparison of the kinetic and fluid models
    based on the electron to ion temperature ratios.  From top to
    bottom, the figure captures the second moment of the ion
    distribution (flux) at the wall (a), sheath width determined based
    on reaching of the Bohm velocity \eqrp{bounded:bohm2} (b), and the
    potential drop over the sheath region.  Solution is shown at
    $t\omega_{pe} = 200$.}
  \label{fig:bounded:comp_scan}
\end{figure}

As discussed in Chapter \ref{sec:weibel}, temperature anisotropy can
lead to a growth of the Weibel instability.  This growth was reported by \cite{Tang2011} using a particle-in-cell
(PIC) code.  The results are reproduced here using the continuum
kinetic code; see \fgr{bounded:weibel}.  Similar to the PIC code, the
simulation is initialized with uniform initial conditions and uses no
collision operators.  \fgr{bounded:weibel} captures the growth of the
integrated magnetic field energy (violet), which grows by roughly 8
orders of magnitude before the saturation, and the temperatures in the
last cell next to the wall.  In the 1X2V simulation, the two resolved
velocity components are perpendicular to the wall ($v_x$) and parallel
to the wall ($v_y$).  The magnetic field grows in the $z$-direction,
which is the other direction parallel to the wall.  Note that the saturation of
the instability, which in this case occurs around
$t\omega_{pe}=1\,300$, results in a decrease of the temperature
anisotropy.  

\begin{figure}[!htb]
  \centering
  \includegraphics[width=0.8\linewidth]{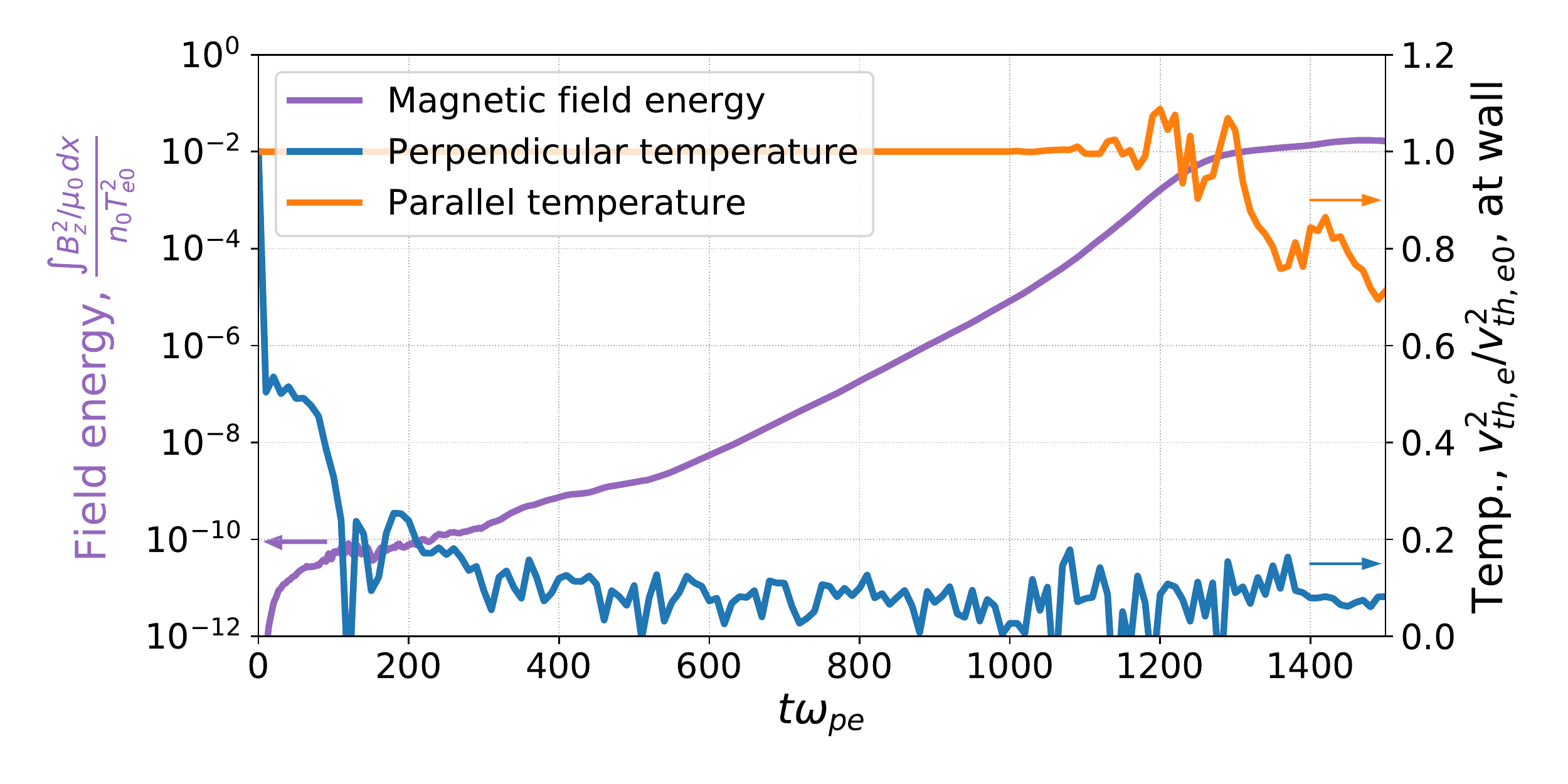}
  \caption[Magnetic field energy growing inside a sheath region]{The
    growth of the integrated magnetic field energy (violet) due to the
    Weibel instability originating from temperate anisotropy in the
    plasma sheath (blue and orange lines).  Note that the magnetic
    field energy grows by approximately eight orders of magnitude.}
  \label{fig:bounded:weibel}
\end{figure}

\begin{figure}[!htb]
  \centering
  \includegraphics[width=0.8\linewidth]{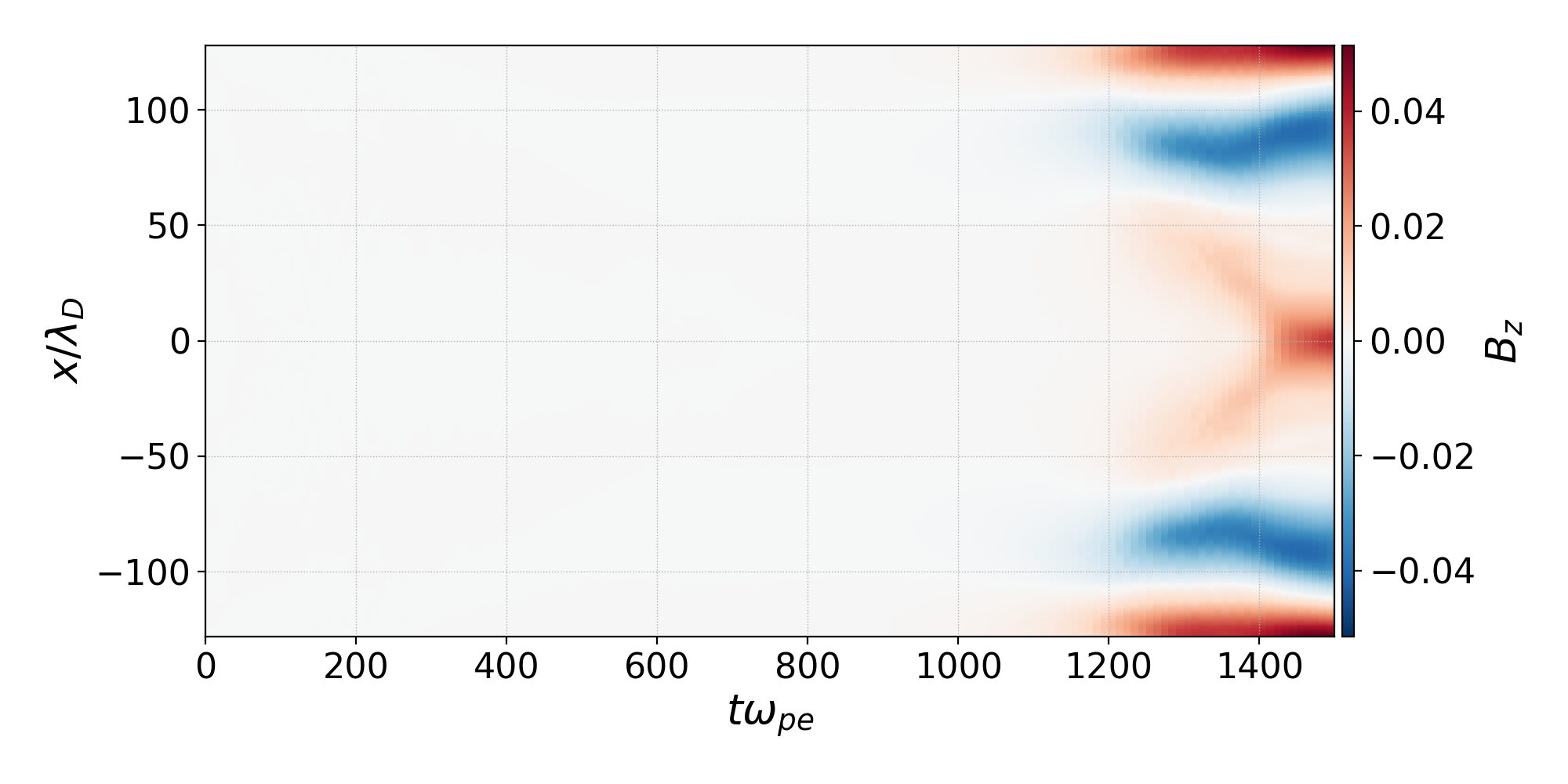}
  \caption[Self-consistent magnetic field profile from a sheath
    simulation]{Evolution of the magnetic field profile, $B_z$,
    growing due to the Weibel instability originating from the temperature
    anisotropy in the plasma sheath.}
  \label{fig:bounded:weibel_Bz}
\end{figure}

Growth of the Weibel instability and effects of initial temperature,
collisions, and preexisting magnetic fields will be the topics of a
future study.

%=====================================================================
\FloatBarrier
\section{Magnetized Sheaths}

Magnetic field can significantly alter plasma behavior near a wall.
When the field is parallel to the wall, which is relevant for tokamaks
but also for the magnetic field that self-consistently develops from a
temperature anisotropy as discussed in \ser{bounded:temp}, the
cross-field mobility of particles decreases and the plasma is
confined.

\fgr{bounded:mag_evolution} shows temporal evolution of a simulation
similar to the one depicted in \fgr{bounded:cls_evolution}.  The
difference is in the inclusion of magnetic field,
$B_z=\SI{0.02}{T}$,\footnote{Similar to the other plasma parameters,
  the magnitude of the magnetic field is relevant for Hall thrusters
  \citep{Robertson2013}.} parallel to the wall. The initial parameters correspond to
plasma $\beta = 0.001$, where
\begin{align}
  \beta = \frac{n_eT_e}{B^2/(2\mu_0)}.
\end{align}
That means that the plasma is strongly magnetized.  Consequently, the
simulation needs to be extended from 1X1V to 1X2V in order to capture
effects of the Lorentz force.

\begin{figure}[!htb]
  \centering
  \includegraphics[width=0.8\linewidth]{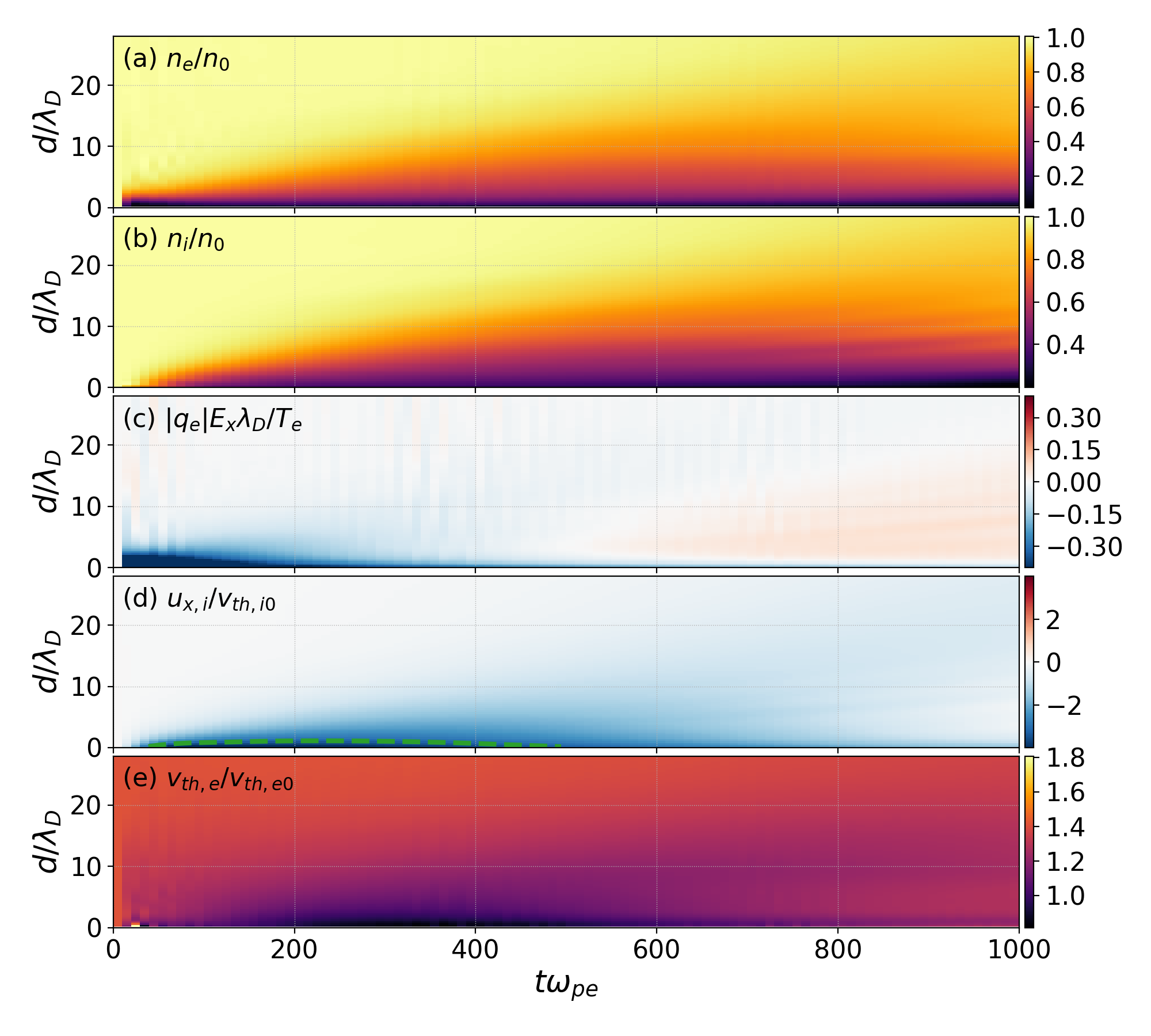}
  \caption[Temporal evolution of magnetized plasma sheath]{Temporal
    evolution of the densities (electrons top, a, and ions bottom, b),
    normalized electric field (c), ion bulk velocity (d), and the
    electron thermal velocity (e) in the region near the left wall of
    a magnetized sheath simulation ($B_z=\SI{0.02}{T}$; $d=0$ is
    directly at the left wall). The green contour in velocity panel (d) marks
    the Bohm velocity \eqrp{bounded:bohm}.}
  \label{fig:bounded:mag_evolution}
\end{figure}

The early time plasma sheath formation in  \fgr{bounded:mag_evolution}
is similar to the unmagnetized case in \fgr{bounded:cls_evolution}.
However, as the thermal flux to the absorbing wall increases, the Lorentz force
starts confining the plasma, thus limiting the electron outflow. As a
result the sheath electric field and potential decrease. Without the
field accelerating the ions from the center of the domain, the fast
electrons leave the domain and are not replaced.  Around
$t\omega_{pe} = 500$, the ion bulk velocity is below the Bohm velocity
everywhere in the domain.

The absence of the classical sheath for $\beta=0.001$ in
\fgr{bounded:mag_evolution} raises a question of the critical magnetic
field.  The previous simulation is, therefore, repeated with magnetic
fields $B_z = \SI{0.01}{T}$ and \SI{0.005}{T}, which correspond to
$\beta = 0.004$ and $0.016$.  The results are in
\fgr{bounded:mag_profiles}; the case with $\beta = 0.004$ behaves
similar to the $\beta=0.001$ case with the sheath, in this case,
disappearing for $t\omega_{pe} > 1000$.  For $\beta = 0.016$, plasma
sheath forms and does not disappear as the simulation stabilizes;
however, the sheath width is still much narrower in comparison to the
unmagnetized case \fgr{bounded:cls_evolution}d.

The oscillations of potential in \fgr{bounded:mag_evolution}b are caused by
the Langmuir waves discussed in \ser{bounded:sheath:sims}.

\begin{figure}[!htb]
  \centering
  \includegraphics[width=0.8\linewidth]{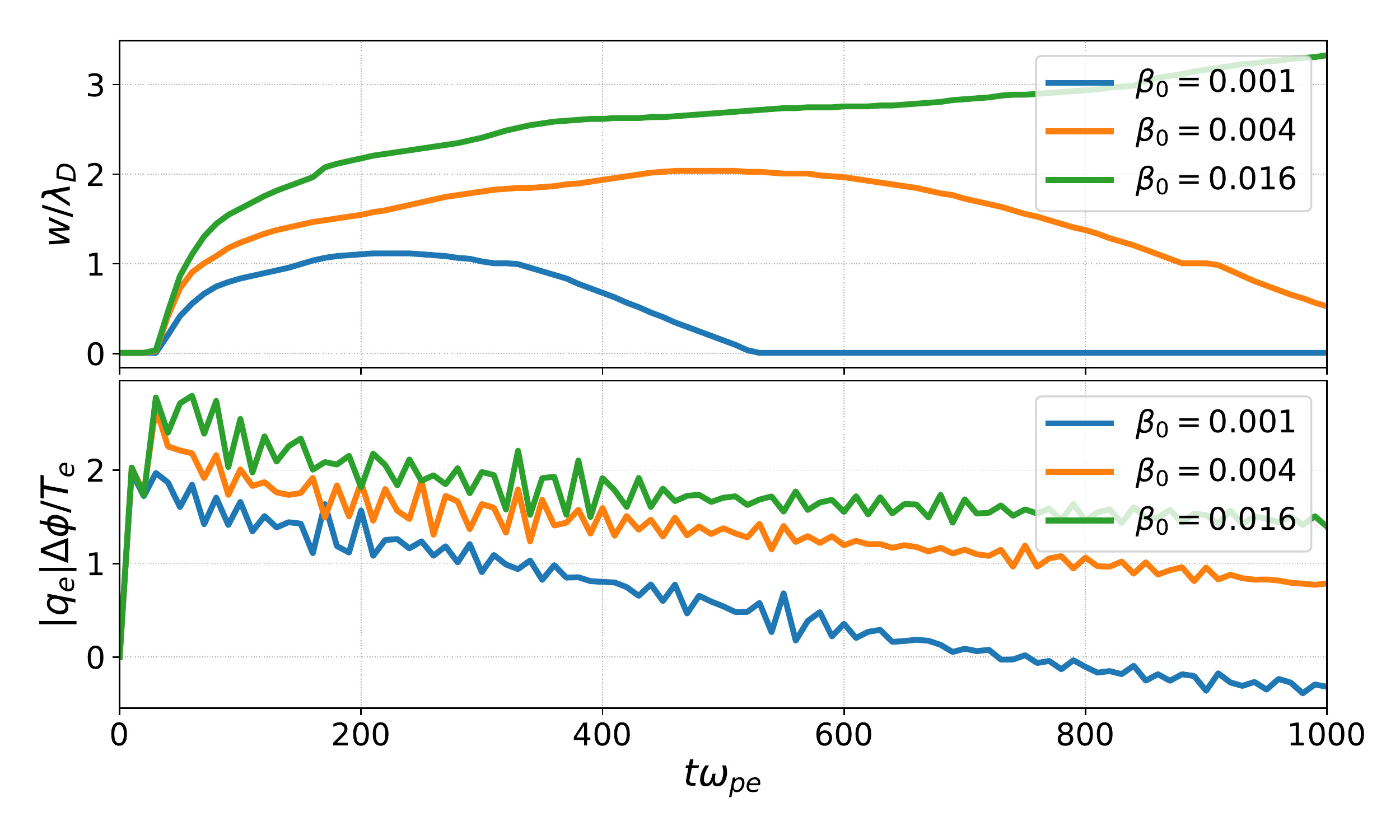}
  \caption[Comparison of sheath width and potential based on plasma
    beta]{Comparison of sheath widths (a), determined based on the
    crossing of the Bohm velocity, and potential differences between
    the wall and the center of the domain (b) based on plasma beta.}
  \label{fig:bounded:mag_profiles}
\end{figure}

%=====================================================================
\FloatBarrier
\section{Plasma-Material Interaction}
\label{sec:bounded:pmi}

In the previous sections, the wall is treated as an ideal absorber.
In reality, the situation is more complicated as the incident
electrons can be reflected back, they can penetrate the material and
then being rediffused with lower energy, or the electrons originally
in the material can gain energy and be released into the plasma.  The
latter process, known as secondary electron emission (SEE), can occur
either through direct transfer of the kinetic energy of the impacting
particle \citep{Furman2002} or through a change in its internal energy
state \citep{Bronold2018}.  While the source of the first type of SEE
can be both electrons and ions, the second type requires distinct
internal energy levels.  \cite{Bronold2018} provide an example of
helium ions colliding with a dielectric wall.  As an ion approaches
the wall, an electron from the surface can tunnel into ion shell
energy levels directly neutralizing it and releasing Auger
electron.\footnote{Auger electrons are electrons coming from a
  material that gains energy from an additional electron transferring
  to a lower energy level.  This energy might be enough for the first
  electron to leave the material.  Interestingly, energy of Auger
  electrons is, therefore, dependent on the internal structure of the
  material and is not a function of incoming energy.  For this reason,
  they typically appear at the same energy in spectra.} Alternatively,
the electron from the wall can tunnel into a metastable state,
releasing an Auger electron later through deexcitation.

SEE is critical for devices like Hall thrusters \citep{Dunaevsky2003}
and tokamak walls \citep{Takamura2004} and needs to be rigorously
modeled.  Previous work of \cite{Sydorenko2004} and
\cite{Sydorenko2006} show that the electron distribution function in
Hall thrusters is, due to the SEE, strongly anisotropic and depletes
at high energies, warranting the kinetic approach.  Further discussion
of effects of the shape of the distribution function on the wall electron flux and the electron temperature are
presented by \cite{Kaganovich2007}.  The authors also report that SEE may carry a considerable portion of the cross-field electron current. \cite{Campanell2012} study
conditions for sheath instability due to SEE and ``weakly confined
electrons'' at the boundary of the loss cone. Recently,
\cite{Campanell2017} report fundamental changes, for example reversal
of the plasma sheath potential, in cases where the emission gain
exceeds unity.

Typically, kinetic simulations of SEE couple the Particle-in-Cell
model of plasmas with Monte Carlo description of solutions with a
wall, which can be computationally expensive.  Alternatively,
\cite{Campanell2017} use a simplified continuum kinetic model.  The
goal of this section is to design a plasma-material interaction (PMI)
model with both computationally efficient physics-based approach.

%---------------------------------------------------------------------
\FloatBarrier
\subsection{General Boundary Conditions}
\label{sec:bounded:pmi:bc}

For a general case, the distribution function coming out of the wall,
$f_{\text{out}}$ is defined as integral of the incoming distribution
function, $f_{\text{in}}$ and the reflection function, $R$,
\begin{align}\label{eq:bounded:bc}
  f_{\text{out}}(t,\, \bm{x}=\bm{x}_{\text{wall}},\, \bm{v}) =
  \int_{\mathcal{V}_{\text{in}}} R(\bm{v},\,\bm{v}') f_{\text{in}}(t,\,
  \bm{x}=\bm{x}_{\text{wall}},\,\bm{v}') \,d\bm{v}', \quad \forall \bm{v} \in \mathcal{V}_{\text{out}}
\end{align}
where the integration is over half of the velocity space limited to
the incoming velocities, $\mathcal{V}_{\text{in}}$, and the relation
is defined only for the outgoing velocities from
$\mathcal{V}_{\text{out}}$.\footnote{Note that with upwind fluxes,
  only the outgoing velocities are needed to construct a boundary
  condition.}  For clarity, the $t$ dependence will be dropped from
now on.  Even though \eqr{bounded:bc} is defined only at the edge of
the domain, it is implemented in \texttt{Gkeyll} by setting the
distribution function, $f_{\text{out}}$, in the ghost cell layer based
on the distribution function, $f_{\text{in}}$, in the skin cell layer.
Therefore, the superscript cell index $j$, used in the
\ser{model:dgck}, is switched for $g$ and $s$ indexing the ghost layer
and skin layer, respectively. \eqr{bounded:bc} is then discretized as
follows,
\begin{align*}
  f_h^g(\bm{x}_{\text{wall}},\,\bm{v}) = \sum_{s} \int_{\mathcal{V}_{\text{in}}^s}
  R^{gs}(\bm{v},\,\bm{v}') f_h^{s}(\bm{x}_{\text{wall}},\,\bm{v}')
  \,d\bm{v}',
\end{align*}
where the summation is only over cells $s$ and the integration is
limited to incoming velocities within them.  Therefore, the summation
typically does not include all the cells.

Next, distribution functions are expanded onto basis functions.
However, unlike in \eqr{model:modal}, basis functions are reduced by
one dimension to surface basis functions $\varsigma$.  Assuming
without loss of generality that the boundary lies in the
$x$-direction, the distributions are expressed as
\begin{align*}
  f_h^g(\bm{x},\,\bm{v})|_{x=x_{\text{wall}}} = \sum_k \widehat{f}_k^g
  \varsigma_k(y,z,\bm{v}), \quad
  f_h^{s}(\bm{x},\,\bm{v}')|_{x=x_{\text{wall}}} = \sum_l \widehat{f}_l^{s}
  \varsigma_l(y,z,\bm{v}'),
\end{align*}
which gives the equality \eqr{bounded:bc} in the discrete weak sense,
\begin{align*}
  \sum_k \widehat{f}_k^g \varsigma_k(y,z,\bm{v}) \circeq
  \sum_{s} \sum_l \widehat{f}_l^s \int_{\mathcal{V}_{\text{in}}^{s}}
  R_x^{gs}(\bm{v},\,\bm{v}') \varsigma_l(y,z,\bm{v}')
  \,d\bm{v}'.
\end{align*}
The full equality is
\begin{multline}\label{eq:bounded:general}
  \sum_k \widehat{f}_k^g \int_{\partial_{x} K^{g}}
  \varsigma_k(y,z,\bm{v}) \varsigma_t(y,z,\bm{v}) \,dy dz d\bm{v} =\\=
  \sum_{s} \sum_l \widehat{f}_l^{s} \int_{\partial_{x} K^{s}}
  \int_{\mathcal{V}_{\text{in}}^{s}}
  R_x^{gs}(\bm{v},\,\bm{v}')\varsigma_l(y,z,\bm{v}')
  \varsigma_t(y,z,\bm{v}) \,d\bm{v}' dy dz d\bm{v}.
\end{multline}
As before, \eqr{bounded:general} needs to transformed into logical
space,
\begin{multline*}
  \sum_k \widehat{f}_k^g \int_{\partial_{x} I_p}
  \hat{\varsigma}_k(\eta_y,\eta_z,\bm{\eta}_{\bm{v}})
  \hat{\varsigma}_t(\eta_y,\eta_z,\bm{\eta}_{\bm{v}}) \,d\eta_y
  d\eta_z d\bm{\eta}_{\bm{v}} =\\= \frac{\prod_{i=1}^{d_v}\Delta
    v_i}{2^{d_v}} \sum_{s} \sum_l \widehat{f}_l^{s} \int_{\partial_{x}
    I_p} \int_{I_v}
  R_x^{gs}\big(\bm{v}^g(\bm{\eta}_{\bm{v}}),\,\bm{v}^s(\bm{\eta}_{\bm{v}}')\big)
  \hat{\varsigma}_l(\eta_y,\eta_z,\bm{\eta}_{\bm{v}}')
  \hat{\varsigma}_t(\eta_y,\eta_z,\bm{\eta}_{\bm{v}})
  \,d\bm{\eta}_{\bm{v}}'d\eta_y d\eta_z d\bm{\eta}_{\bm{v}}.
\end{multline*}
Similar to the volume basis functions, orthonormal surface basis can be
constructed and the relation simplifies to
\begin{align}\label{eq:bounded:general2}
 \widehat{f}_k^g  = \frac{\prod_{i=1}^{d_v}\Delta
    v_i}{2^{d_v}} \sum_{s,l} \widehat{f}_l^{s} \int_{\partial_{x}
    I_p} \int_{I_v}
  R_x^{gs}\big(\bm{v}^g(\bm{\eta}_{\bm{v}}),\,\bm{v}^s(\bm{\eta}_{\bm{v}}')\big)
  \hat{\varsigma}_l(\eta_y,\eta_z,\bm{\eta}_{\bm{v}}')
  \hat{\varsigma}_k(\eta_y,\eta_z,\bm{\eta}_{\bm{v}})
  \,d\bm{\eta}_{\bm{v}}'d\eta_y d\eta_z d\bm{\eta}_{\bm{v}}.
\end{align}
However, since $R^{gs}\big(\bm{v}^g(\bm{\eta}_{\bm{v}}),
\,\bm{v}^s(\bm{\eta}_{\bm{v}}')\big)$ can have a complex dependence on
$\bm{v}$ and $\bm{v}'$, the integral on the right-hand-side of
\eqr{bounded:general2} cannot be generally precomputed in the logical
space.  Instead, there are three possible options:
\begin{enumerate}
  \item The integral can be solved analytically.
  \item Tensor $\mathcal{R}_{x,lk}^{gs}$,
    \begin{align*}
      \mathcal{R}_{x,lk}^{gs} = \int_{\partial_{x} I_p} \int_{I_v}
      R_x^{gs}\big(\bm{v}^g(\bm{\eta}_{\bm{v}}),
      \,\bm{v}^s(\bm{\eta}_{\bm{v}}')\big)
      \hat{\varsigma}_l(\eta_y,\eta_z,\bm{\eta}_{\bm{v}}')
      \hat{\varsigma}_k(\eta_y,\eta_z,\bm{\eta}_{\bm{v}})
      \,d\bm{\eta}_{\bm{v}}'d\eta_y d\eta_z d\bm{\eta}_{\bm{v}},
    \end{align*}
    is precomputed and stored.
  \item The integral is solved directly during run-time using
    Gauss-Legendre quadrature.
\end{enumerate}
Naturally, the first option is the best case scenario but is also
rare.  The second option reduces the operations performed during
run-time to
\begin{align*}
  \widehat{f}_k^g = \frac{\prod_{i=1}^{d_v}\Delta v_i}{2^{d_v}}
  \sum_{s,l} R_{x,lk}^{gs}\widehat{f}_l^{s},
\end{align*}
however, it requires storing significant amounts of data; $N_p\times
N_p$ for all the $gs$ combinations.  The third option does not require
any storage but computing the quadrature might easily become the
bottleneck of the whole simulation.

Finally, it should be mentioned that \eqr{bounded:general} does not
represent the boundary condition implementation in \texttt{Gkeyll}.
We use the fact that the distribution function in the ghost layer is
used only to calculate the numerical flux.  Therefore, we can use
previously defined basis functions $\widehat{\psi}$, with $\eta_x$ in
the ghost cell and $-\eta_x$ in the skin cell, instead of
$\widehat{\varsigma}$.  In other words,
\begin{multline}\label{eq:bounded:general3}
 \widehat{f}_k^g = \frac{\prod_{i=1}^{d_v}\Delta v_i}{2^{d_v}}
 \sum_{s,l} \widehat{f}_l^{s} \int_{I_p} \int_{I_v}
 R_x^{gs}\big(\bm{v}^g(\bm{\eta}_{\bm{v}}),\,\bm{v}^s(\bm{\eta}_{\bm{v}}')\big)
 \\\widehat{\psi}_l(-\eta_x,\eta_y,\eta_z,\bm{\eta}_{\bm{v}}')
 \widehat{\psi}_k(\eta_x,\eta_y,\eta_z,\bm{\eta}_{\bm{v}})
 \,d\bm{\eta}_{\bm{v}}'d\eta_x d\eta_y d\eta_z d\bm{\eta}_{\bm{v}}.
\end{multline}

%---------------------------------------------------------------------
\FloatBarrier
\subsection{Special Cases of the Reflection Function}
\label{sec:bounded:pmi:r}

In this part we explore several special cases of the reflection
function, $R$, spanning from simple black hole boundary conditions to
complex models based on quantum mechanics.

\subsubsection{Special Case: Black Hole}

While the name might sound like a joke, it actually very well
describes what this boundary condition does.  Setting the reflection
function to zero, $R(\bm{v},\,\bm{v}') := 0$, simulates a perfectly
absorbing wall.  This simple boundary condition is successfully used
throughout \ser{bounded:sheath:sims} and by \cite{Cagas2017s} to
replicate classical sheath physics.

\subsubsection{Special Case: Specular Reflection}

The next step is specular reflection, i.e., the reflection conserving
both energy and momentum.  The reflection function is given simply as
\begin{align}
  R_x(\bm{v},\bm{v}') =
  \delta(v_x+v_x')\delta(v_y-v_y')\delta(v_z-v_z').
\end{align}
Applying this to \eqr{bounded:bc} gives the expected result,
\begin{align*}
  f_{\text{out}}(v_x,v_y,v_z) &=
  \iiint_{\mathcal{V}^{\text{in}}}
  \delta(v_x+v_x')\delta(v_y-v_y')\delta(v_z-v_z')
  f_{\text{in}}(v_x',v_y',v_z') \,dv_x'dv_y'dv_z',\\
  &= f_{\text{in}}(-v_x,v_y,v_z).
\end{align*}

In the discrete case, a specular reflection function is limited only
to the cell with ``opposite $x$-velocity'', symbolically denoted with
Kronecker delta,\footnote{Note there is a difference between the
  Kronecker delta $\delta_{ij}$ (with indices) and the Dirac delta
  function, $\delta(x)$.}
\begin{align}
  R_x^{gs}(\bm{v},\bm{v}') = \delta_{g(-s)}
  \delta(v_x+v_x')\delta(v_y-v_y')\delta(v_z-v_z').
\end{align}
Substituting this into \eqr{bounded:general3} yields
\begin{align}\label{eq:bounded:general4}
  \widehat{f}_k^g = \sum_{l} \widehat{f}_l^{-s} \int_{I_p}
  \widehat{\psi}_l(-\eta_x,\eta_y,\eta_z,-\eta_{v_x},\eta_{v_y},
  \eta_{v_z})
  \widehat{\psi}_k(\eta_x,\eta_y,\eta_z,\eta_{v_x},\eta_{v_y},
  \eta_{v_z}) \,d\bm{\eta}_{\bm{x}} d\bm{\eta}_{\bm{v}}.
\end{align}

Specifically, using the 1X1V basis \eqr{model:1x1v},
\begin{gather*}
  \mathcal{R}_{kl} = \int_{I_p}
  \widehat{\psi}_l(-\eta_x,-\eta_{v_x})
  \widehat{\psi}_k(\eta_x,\eta_{v_x}) \,d\eta_{x} d\eta_{v_x},\\
  = \begin{pmatrix}
    1 & 0 & 0 & 0 & 0 & 0 & 0 & 0 \\
    0 & -1 & 0 & 0 & 0 & 0 & 0 & 0 \\
    0 & 0 & -1 & 0 & 0 & 0 & 0 & 0 \\
    0 & 0 & 0 & 1 & 0 & 0 & 0 & 0 \\
    0 & 0 & 0 & 0 & 1 & 0 & 0 & 0 \\
    0 & 0 & 0 & 0 & 0 & 1 & 0 & 0 \\
    0 & 0 & 0 & 0 & 0 & 0 & -1 & 0 \\
    0 & 0 & 0 & 0 & 0 & 0 & 0 & -1
  \end{pmatrix}.
\end{gather*}
Precisely this boundary condition is implemented in the example at
the beginning of this work (\fgr{model:distf}), where a
``gas bouncing between two walls'' is used to demonstrate phase space
figures.  There, the $\mathcal{R}_{kl}$ listed above is used to
construct ghost cell distribution function from the skin cell.

The billiard ball boundary condition can be used to save computation
time for symmetric problems, for example plasma sheaths in
\ser{bounded:sheath:sims}.  Note that unlike for the neutral gas
simulation, the distribution function boundary condition discussed
above must be complemented with boundary conditions for fields.

For field boundary conditions the same
trick is applied as before, i.e.,
\begin{align*}
  \widehat{E}_{i,k}^g = \sum_l \widehat{E}_{i,l}^s \int_{I_c}
  \widehat{\varphi}_l(-\eta_x,\eta_y,\eta_z)
  \widehat{\varphi}_k(\eta_x,\eta_y,\eta_z)\, d\bm{\eta}_{\bm{x}},
\end{align*}
together with the appropriate physics-based boundary
conditions. Specifically, for sheath simulations that use symmetry so
that one boundary represents the center of the presheath, we require
zero normal for the electric field and zero tangent for the magnetic
field.  This gives:\footnote{Still assuming the wall in in the
  $x$-direction.}
\begin{align*}
  \widehat{E}_{x,k}^g &=- \sum_l \widehat{E}_{x,l}^s \int_{I_c}
  \widehat{\varphi}_l(-\eta_x,\eta_y,\eta_z)
  \widehat{\varphi}_k(\eta_x,\eta_y,\eta_z)\, d\bm{\eta}_{\bm{x}},\\
  \widehat{E}_{yz,k}^g &= \sum_l \widehat{E}_{yz,l}^s \int_{I_c}
  \widehat{\varphi}_l(-\eta_x,\eta_y,\eta_z)
  \widehat{\varphi}_k(\eta_x,\eta_y,\eta_z)\, d\bm{\eta}_{\bm{x}},\\
  \widehat{B}_{x,k}^g &= \sum_l \widehat{B}_{x,l}^s \int_{I_c}
  \widehat{\varphi}_l(-\eta_x,\eta_y,\eta_z)
  \widehat{\varphi}_k(\eta_x,\eta_y,\eta_z)\, d\bm{\eta}_{\bm{x}},\\
  \widehat{B}_{yz,k}^g &= - \sum_l \widehat{B}_{yz,l}^s \int_{I_c}
  \widehat{\varphi}_l(-\eta_x,\eta_y,\eta_z)
  \widehat{\varphi}_k(\eta_x,\eta_y,\eta_z)\, d\bm{\eta}_{\bm{x}}.
\end{align*}

A comparison of simulations using two absorbing wall boundaries from
\ser{bounded:sheath:sims} with half domain simulations that use a
specular boundary condition at the left edge to capture the symmetries
is in \fgr{bounded:bc_reflect}.  Values of the distribution functions
are directly subtracted.  The figure shows only the right half of the
full domain simulation to allow direct calculation of the difference.
Since there are regions where the distribution function is close to
zero, the difference is normalized to the maximal value of the
distribution.  The relative difference on the order of $10^{-13}$
gives a confidence in the implementation of the boundary condition.

\begin{figure}[!htb]
  \centering
  \includegraphics[width=0.8\linewidth]{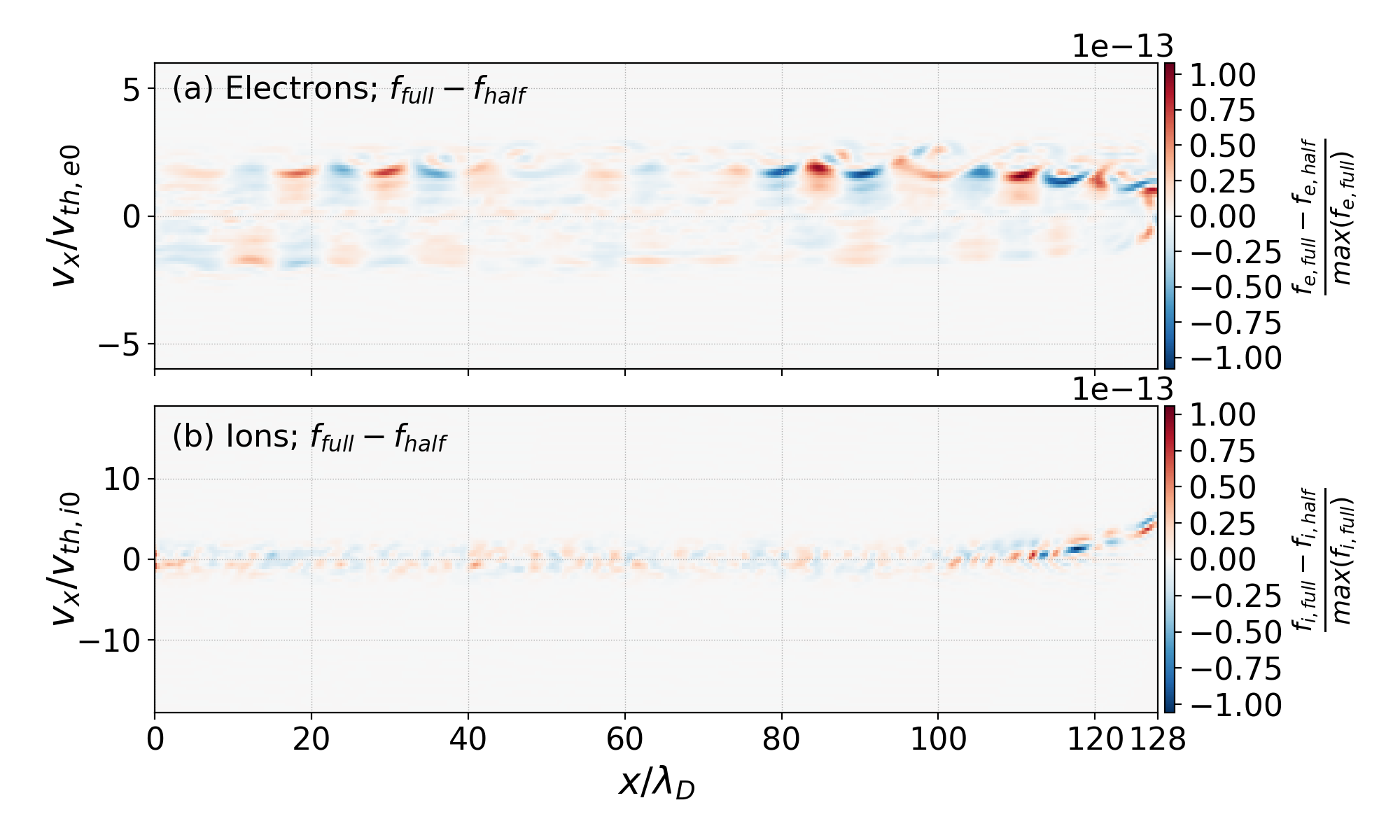}
  \caption[Difference between full domain and half domain
    simulations]{Normalized difference between distribution functions
    from the full domain (only right half is plotted) simulation using
    absorbing walls on both sides [\ref{list:bounded:cls}] and the
    half domain simulation with the specular boundary condition at the
    left edge replicating the symmetric behavior.}
  \label{fig:bounded:bc_reflect}
\end{figure}

\subsubsection{Special Case: \cite{Furman2002} Model}

While the specular reflection has many uses, it does not accurately
describe the plasma wall interaction and we need to look to literature
for more complex models.  One such example is a phenomenological model
by \cite{Furman2002}, which uses analytical descriptions for three
distinct populations of electron emission -- elastically reflected
electrons, rediffused electrons, and true-secondary
electrons.\footnote{In this case, only the true-secondary electrons
  excited with kinetic energy of incoming population are taken into
  account.}  It must be pointed out that the model assumes that the
secondary electrons are produced by a mono-energetic (cold) beam of
incoming electrons.  For each incident beam with current
$I_{\text{in}}$, the model defines energetic distribution of electron
yield, $\gamma = I_{\text{out}}/I_{\text{in}}$,
\begin{align}
  \pfrac{\gamma}{E} = \pfrac{\gamma_e}{E} + \pfrac{\gamma_r}{E} +
  \pfrac{\gamma_{ts}}{E},
\end{align}
where $\gamma_e$, $\gamma_r$, and $\gamma_{ts}$ correspond to the
three aforementioned populations.  For each population, an analytical
profile is determined based on underlining physical properties and
experimental data.

The first described group consists of primary electrons
semi-elastically reflected from the material surface.  Since they are
assumed not to lose any energy or only a small amount, the model
describes this population by a narrow half-Gaussian centered around
the incoming energy.  Note that since the secondary electrons cannot
have higher energy than the incident ones (unless the material gains
additional energy, for example by heating), the distribution is
limited by the incoming energy.  Contribution of the reflected
electrons is given as,
\begin{gather}\label{eq:bounded:furman_e}
  \pfrac{\gamma_e}{E}(E,E'\mu') =
  \theta(E)\theta(E'-E)\,\gamma_{e0}(E') \left[1 +
    e_1\left(1-\mu'^{e_2}\right)\right]
  \frac{2\exp\big(-(E-E')^2/2\sigma_e^2\big)}{\sqrt{2\pi}\sigma_e\mathrm{erf}\left(E'/\sqrt{2}\sigma_e\right)},
  \\ \gamma_{e0}(E') = P_{1,e}(\infty) + \left[\hat{P}_{1,e} -
    P_{1,e}(\infty)\right]
  \exp\left[\left(|E'-\hat{E}_e|/W\right)^p/p\right],\nonumber
\end{gather}
where $\theta()$ is the Heaviside step function ensuring that the
incoming energy is higher than the outgoing.  $e_1$, $e_2$, $\sigma_e$,
$P_{1,e}(\infty)$, $\hat{P}_{1,e}$, $W$, $\hat{E}_e$, and $p$ are
fitting parameters.  $\mu$ and $\mu'$ are direction cosines for the
outgoing and incoming angles, respectively.

The rest of the incident electrons are assumed to penetrate the
material.  As they interact with the material, they lose energy.  Part
of them eventually go through the material potential barrier and
return to the plasma.  These so-called rediffused electrons can have a
range of energies between zero and the incident energy.
\cite{Furman2002} describe them with
\begin{gather}\label{eq:bounded:furman_r}
  \pfrac{\gamma_r}{E}(E,E,\mu') = \theta(E)\theta(E'-E) \gamma_{r0}(E')
  \left[1 + r_1\left(1-\mu'^{r_2}\right)\right]
  \frac{(q+1)E^q}{E'^{q+1}}, \\ \gamma_{r0}(E') = P_{1,r}(\infty)
  \left[1 - \exp\big(-(E'/E_r)^r\big) \right], \nonumber
\end{gather}
where $r_1$, $r_2$, $q$, $P_{1,r}(\infty)$, and $E_r$ are fitting
constants.

Finally, the last part consists of the true-secondary electrons from
the material.  Since the energy of the primary beam is transferred to
the secondary electrons through a cascade, their distribution peaks at
lower energy.  However, unlike the back-scattered and rediffused
electrons, single incoming electron can produce many secondaries.
Therefore, their description is the most complicated one,
\begin{gather}\label{eq:bounded:furman_ts}
  \pfrac{\gamma_{ts}}{E}(E,E,\mu') = \sum_{n=1}^M
  \frac{nP_{n,ts}(E',\mu')(E/\epsilon_n)^{p_n-1}
    \exp(-E/\epsilon_n)}{\epsilon_n\Gamma(p_n)P(np_n,
    E'/\epsilon_n)}P\big((n-1)p_n,(E'-E)\epsilon_n\big),
  \\ P_{n,ts}(E',\mu') = \binom{M}{n} \left(\frac{ \hat{\gamma}(\mu')
    D\left[ E' / \hat{E}(\mu') \right]}{M}\right)^n \left(1- \frac{
    \hat{\gamma}(\mu') D\left[ E' / \hat{E}(\mu')
      \right]}{M}\right)^{M-n},\nonumber\\
  \hat{\gamma}(\mu') = \hat{\gamma}_{ts} \left[1 + t_1\left(1 -
    \mu'^{t_2}\right) \right], \quad \hat{E}(\mu') = \hat{E}_{ts}
  \left[1 + t_3\left(1 - \mu'^{t_4}\right) \right],\nonumber\\
  D(x) = \frac{sx}{s-1+x^s},\nonumber
\end{gather}
where $t_1$, $t_2$, $t_3$, $t_4$, $p_n$, $\epsilon_n$,
$\hat{\gamma}_{ts}$, $\hat{E}_{ts}$ and $s$ are fitting
variables. $\Gamma(\cdot)$ is the gamma function and $P(\cdot,\cdot)$
is the normalized incomplete gamma function.\footnote{$P(0,x)=1$} Note
that the summation in \eqr{bounded:furman_ts} should theoretically go
to infinity, but error from limiting it to $M=10$ is negligible
\citep{Furman2002}.

Profiles of \eqr{bounded:furman_e}, \eqr{bounded:furman_r}, and
\eqr{bounded:furman_ts} with the parameters from Tab.~I and Tab.~II of
\cite{Furman2002} are in \fgr{bounded:furman}.  This figure is
constructed for a single \SI{200}{eV} electron beam.  Integrating the
area under the curves of individual populations, we get the total
gains $\gamma_e=0.1241$ for back-scattered electrons,
$\gamma_r=0.7350$ for rediffused, and $\gamma_{ts}=1.1283$ for
true-secondary electrons. Note that $\gamma_e+\gamma_r <1$ is
required. For this case, $\gamma_{ts} > 1$ but the total kinetic
energy of particles is decreased.
\begin{figure}[!htb]
  \centering
  \includegraphics[width=0.8\linewidth]{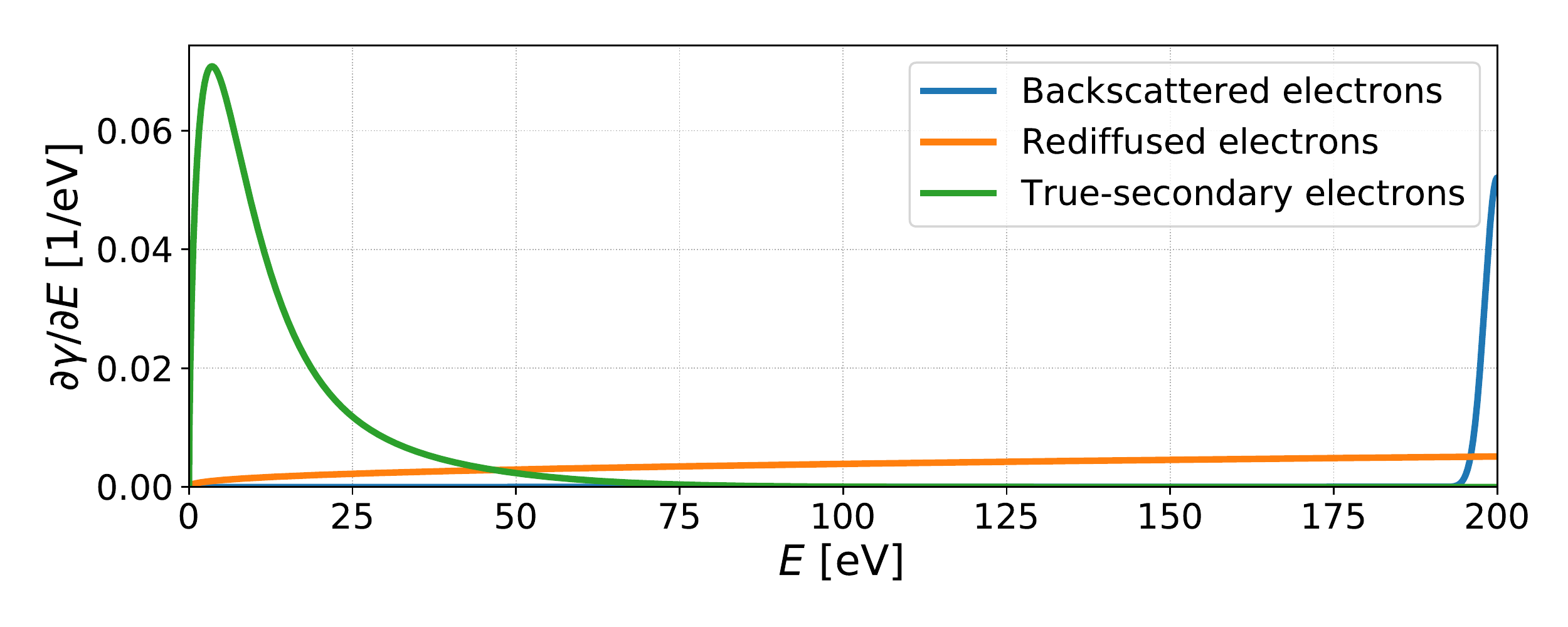}
  \caption[SEE energetic distribution from \cite{Furman2002}]{The
    energetic distribution of the three particle populations emitted
    by single \SI{200}{eV} electron mono-energetic beam with normal
    incidence.  Based on the phenomenological model fits by
    \cite{Furman2002}.}
  \label{fig:bounded:furman}
\end{figure}

In the \cite{Furman2002} model, the $\pfracb{\gamma}{E}$ is used as a
step to obtain emission probabilities for Monte-Carlo SEE codes.
However, we note that since
\begin{align*}
  \int_0^\infty \pfrac{\gamma}{E}\,dE = \gamma(E'),
\end{align*}
$\pfracb{\gamma}{E'}(E,E',\mu)$ resembles the reflection function for
the continuum kinetic code. The dependence on the outgoing angle is
the only part missing.  Experimental measurements show that the
dependence is a cosine function for the true-secondary electrons
\citep{Bruining1954}, i.e., the incoming and outgoing angles are
completely uncorrelated.  While this is not quite true for the other
two populations, \cite{Furman2002} make this assumption as well.  With
this, the model can be used as a reflection function described above,
\begin{align*}
  f_{\text{out}}(E,\mu) = \int_0^1\int_0^\infty
  \mu \pfrac{\gamma}{E}(E,E',\mu')f_{\text{in}}(E',\mu') \,dE'd\mu'.
\end{align*}
The integral can be seen as a ``summation'' over all the incoming cold
beams to extend the mono-energetic formulation to a thermal
population.  Finally, for this distribution function to be useful for
\texttt{Gkeyll} simulations, it needs to be correctly transformed from
energetic units, typical for surface physics, to phase space velocity
coordinates.  Noting that
\begin{align*}
  \pfrac{\gamma}{v_x} = \pfrac{\gamma}{E}\pfrac{E}{v_x} =
  \pfrac{\gamma}{E}mv_x,
\end{align*}
we can write in 1D
\begin{align}
  f_{\text{out}}(v_x) = \int \mu(v_x)
  \pfrac{\gamma}{E}\big(E(v_x),E(v_x'),\mu(v')\big) mv_x
  f_{\text{in}}(v_x')\,dv_x'.
\end{align}

The reflection function can be tested on a Maxwellian distribution
function.  \fgr{bounded:furman_bc} shows the results with colors of
the populations corresponding to \fgr{bounded:furman}.  At first
glance, it might be surprising that the reflected electrons contribute
the most, even though their gain in \fgr{bounded:furman} is the
smallest.  The reason for this is, that \fgr{bounded:furman} describes
a case where incoming particles have enough energy to penetrate the
material.  However, as the energy decreases, back-scattered electrons
become dominant.  This is exactly the case for electron populations
with temperatures on the order of an electron volt with bulk velocity
comparable to thermal velocity.
\begin{figure}[!htb]
  \centering
  \includegraphics[width=0.8\linewidth]{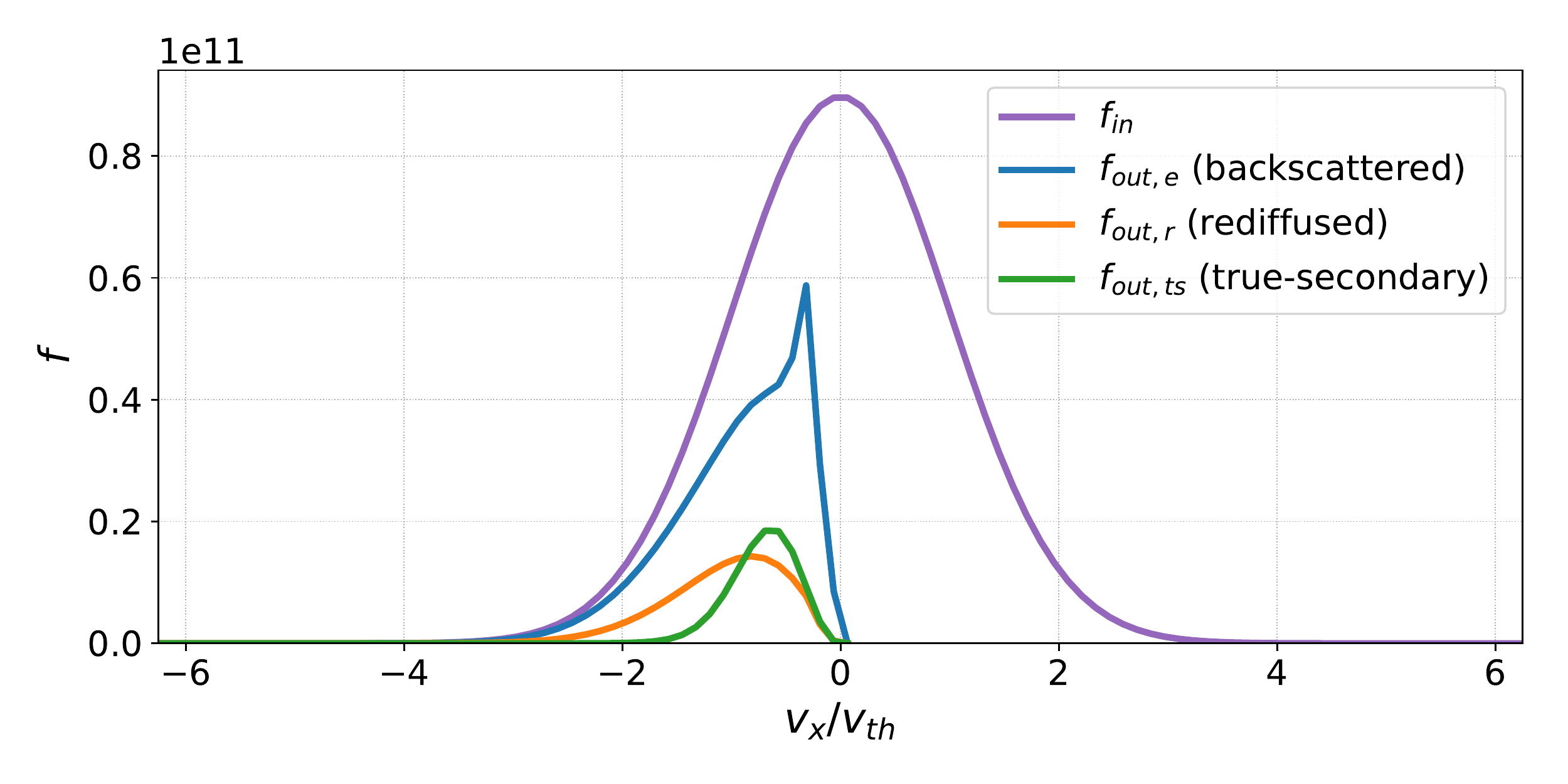}
  \caption[Application of \cite{Furman2002} on Maxwellian
    distribution]{Application of the reflection function from
    \cite{Furman2002} on Maxwellian distribution function.  Violet
    line represents simulated incoming distribution function at the
    right wall and blue, orange, and green are distributions of the
    reflected populations (colors correspond to
    \fgr{bounded:furman}).}
  \label{fig:bounded:furman_bc}
\end{figure}

\fgr{bounded:furman_scan} provides further insight into the individual
secondary populations based on the incoming beam energy.
\fgr{bounded:furman_scan} extends \fgr{bounded:furman} to include
multiple incoming beam energies, i.e., the $y$-axis of
\fgr{bounded:furman_scan} corresponds to $x$-axis of
\fgr{bounded:furman}.  However, since the outgoing energies are
limited by the incoming energy, the $y$-axis of
\fgr{bounded:furman_scan} is normalized to the incoming energy for
better visualization.  Analogously, the values of $\pfracb{\gamma}{E}$
are multiplied by $E'$ to allow for comparison of
magnitudes.\footnote{Since
  $\int_0^{E'}(\pfracb{\gamma}{E})\,dE=\gamma(E')$, normalization
  $(\pfracb{\gamma}{E})E'$ allows to compare the individual energy
  distributions.  Note that theoretically
  $\pfracb{\gamma}{E}\rightarrow\infty$ for $E'\rightarrow0$.}  This
reveals a gradually decreasing contribution of the true-secondary
emission, while rediffused electrons remain steady for a larger range
of energies before they drop for $E'<\SI{20}{eV}$.  On the other hand,
as the incoming energy decreases, the backscattered electron
population becomes more significant which corresponds to
\fgr{bounded:furman_bc}
\begin{figure}[!htb]
  \centering
  \includegraphics[width=0.8\linewidth]{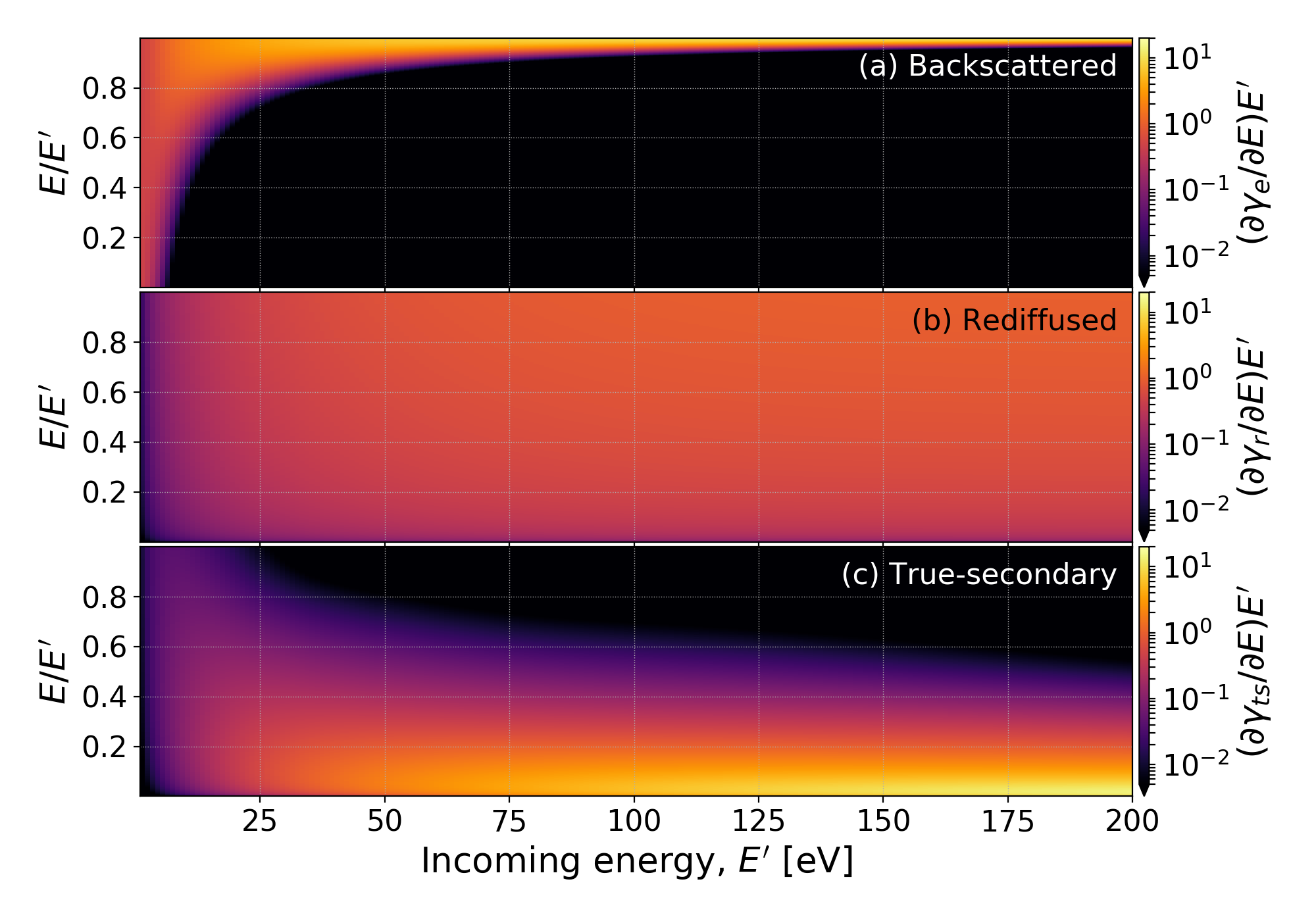}
  \caption[Relative contributions of secondary populations from
    \cite{Furman2002}]{Contributions of the secondary populations from
    the \cite{Furman2002} model, based on the incoming energy.  Both
    the values and the $y$-axis are normalized to the incoming energy
    to allow for better comparison of the relative contributions.
    From top to bottom the figure captures backscattered (elastically
    reflected) electrons, rediffused electron, and true-secondary
    electrons.}
  \label{fig:bounded:furman_scan}
\end{figure}

Finally, it should be pointed out that even though the model is
mathematically sound for incoming energies all the way to zero, the
values at the lower energy range, which are crucial as described
above, are from an extrapolation of higher energy beam data.
Therefore, for simulating $\sim \SI{10}{eV}$ electron distributions in
contact with a wall, a different model specifically tailored for these
energies might be preferred.  Another obstacle of the model is its
dependency on a significant number of fitting parameters which do not
necessarily correspond to physical quantities.  Authors provide the
values of these parameters only for copper and stainless steel, which
are not particularly useful materials for Hall thrusters nor
plasma-facing parts of fusion devices.

\subsubsection{Special Case: \cite{Bronold2015} Model}

\cite{Bronold2015} present a model for electron absorption by a
dielectric wall.  It has several advantages over the \cite{Furman2002}
model.  It is tailored for dielectrics which are more relevant for
Hall thrusters, it includes fewer parameters that are physical like
electron affinity, and it is based on first principles from quantum
mechanics.  On the other hand, \cite{Bronold2015} discuss their
model's relevance only up to incoming energies comparable to the
electron band gap $E_g\sim\SI{10}{eV}$ ($E_g=\SI{7.8}{eV}$ for MgO
used for examples here).

\cite{Bronold2015} directly define the reflection function,
\begin{align}\label{eq:bounded:bronold}
  R(E,\mu,E',\mu') =
  \underbracket{R(E',\mu')\delta(E-E')\delta(\mu-\mu')}_{\text{backscattered}}
  + \underbracket{\delta R(E,\mu,E',\mu')}_{\text{rediffused}}.
\end{align}
Note that the model assumes specular reflection for the back-scattered
electrons, i.e., the energy and angles are conserved with variable
probability $R(E',\mu')$, which is a function only of the incoming
properties.  It is given as $R(E',\mu') = 1 - \mathcal{T}(E',\mu')$,
where $\mathcal{T}(E',\mu')$ is the probability of a
quantum-mechanical reflection,
\begin{align*}
  \mathcal{T}(E',\mu') =
  \frac{4\overline{m}_ekp}{(\overline{m}_ek+p)^2}, \quad k =
  \sqrt{E'-\chi}\mu', \quad p=\sqrt{\overline{m}_eE'}\nu',
\end{align*}
where $\overline{m}_e$ is the relative mass of a conduction band
electron and $\chi$ is the electron affinity of the dielectric.  $k$
and $p$ are components of momentum perpendicular to the wall where
$\nu$ is the cosine angle inside the wall.  $\nu$ is connected with
$\mu$ through conservation of energy and lateral momentum,
\begin{align}\label{eq:bounded:bronold_conservation}
  1-\nu'^2 = \frac{E'-\chi}{\overline{m}_eE'}(1-\mu'^2).
\end{align}

The probability of reflection, $R(E',\mu')$ is captured in
\fgr{bounded:bc_R_2D}, showing several interesting regions.  First of
all, there is a region of $R(E',\mu')= 1$ in the left part.  Electrons
there have lower energy than the electron affinity of the material,
cannot penetrate the potential barrier and are all reflected.  The
second interesting region is in the bottom right.  As a direct
consequence of the conservation of energy and lateral momentum
\eqrp{bounded:bronold_conservation}, there is a critical angle given
as $\mu_c = \sqrt{1-\overline{m}_eE'/(E'-\chi)}$.  Particles entering
under this angle have the momentum vector perpendicular to the surface
after penetrating the material; particles that hit the wall with
$\mu'<\mu_c$ are reflected.\footnote{This is very similar to critical
  angle coming from the Snell's law of light refraction.}  Note that
particles with $\mu'>\mu_c$ and $E'>\SI{2}{eV}$ generally do penetrate
the material and would be lost from the plasma if back-scattering was
the only effect taken into account. They can, however, return to the
plasma through rediffusion.

\begin{figure}[!htb]
  \centering
  \includegraphics[width=0.8\linewidth]{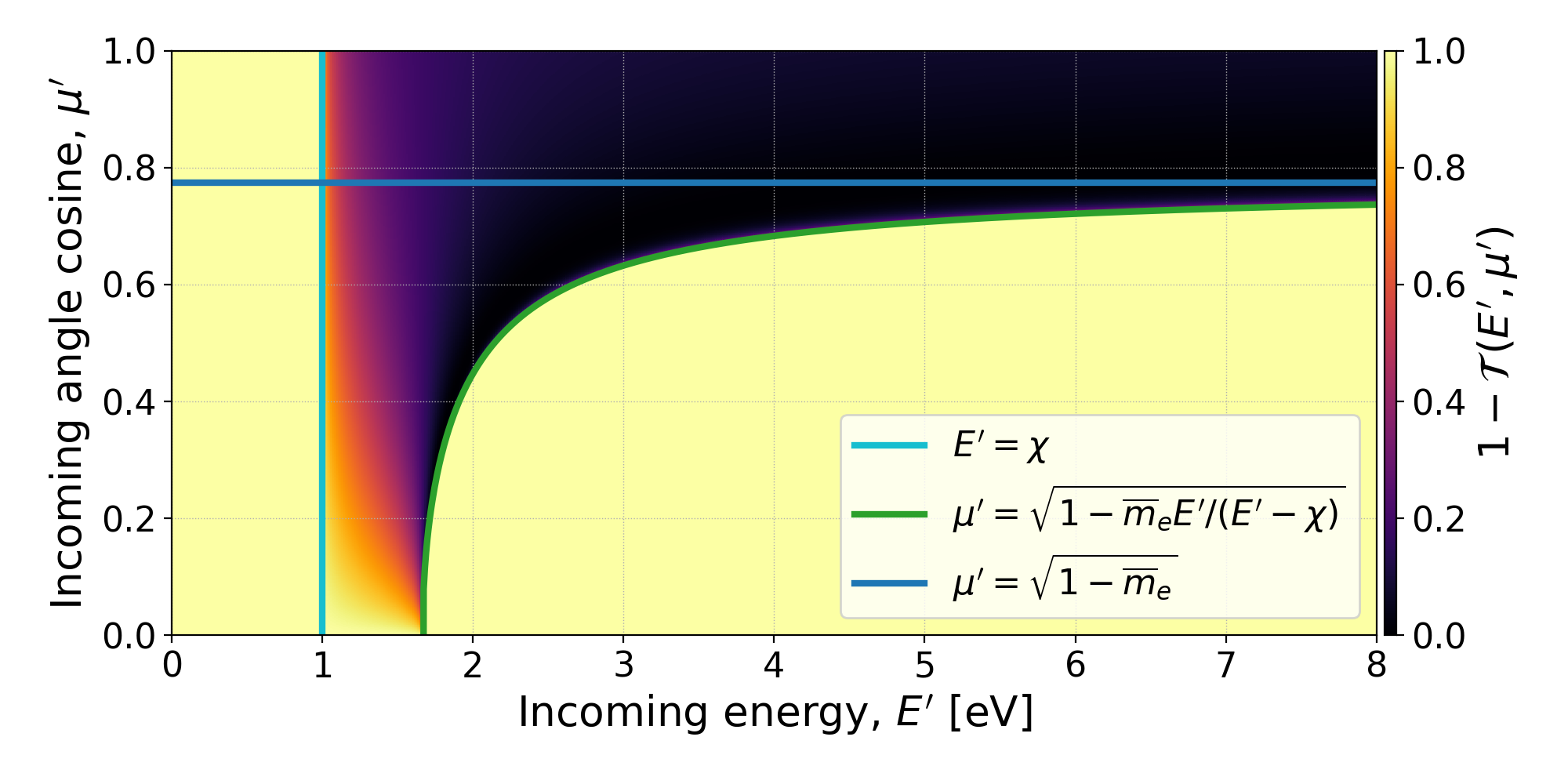}
  \caption[Probability of back-scattering from \cite{Bronold2015}
    model]{Probability of back-scattering, $R(E',\mu')$, from the
    \cite{Bronold2015} model as a function of incoming angle and
    energy. Highlighted are $E'=\chi$ (cyan line) below which are all
    particles reflected, and the critical angle $\mu_c$ (green line)
    given by the conservation laws,
    \eqrp{bounded:bronold_conservation}.  Blue line marks the angle
    above which is possible the rediffusion,
    \eqrp{bounded:bronold_rediffusion}.  Used parameters are for MgO,
    $\chi=\SI{1}{eV}$ and $\overline{m}_e=0.4$.}
  \label{fig:bounded:bc_R_2D}
\end{figure}

Description of rediffusion in \cite{Bronold2018} is much more
complicated in comparison to back-scattered electrons,
\begin{align}\label{eq:bounded:bronold_rediffusion}
  \delta R(E,\mu,E',\mu')=\pfrac{\nu}{\mu} \mathcal{T}(E',\mu')
  \rho(E) \mathcal{B}(E,\mu,E',\mu') \mathcal{T}(E,\mu)
  \theta\big(\mu-\sqrt{1-\overline{m}_e}\big),
\end{align}
where $\rho(E) = \sqrt{\overline{m}_e^3E}/2(2\pi)^3$ is the conduction
band density of states and
\begin{align*}
  \mathcal{B}(E,\mu,E',\mu') =
  \frac{Q(E,\mu,E',\mu')}{\int_0^1\int_0^{E'}\rho(E)Q(E,\mu,E',\mu')\,dEd\mu}
\end{align*}
is the probability of rediffusion.  $Q(E,\mu,E',\mu')$ is given by a
recursive relation summed over the back-scattering events inside the
material.  Note the Heaviside step function in
\eqr{bounded:bronold_rediffusion}; the limiting $\mu$ is marked by the
blue line in \fgr{bounded:bc_R_2D}.  The population with cosine angles
above this line can return to the domain after penetrating the
material, significantly influencing \eqr{bounded:bronold}.
True-secondary electrons excited by incoming electrons with energies
considered here ($<\SI{10}{eV}$) are neglected in this
model.\footnote{\cite{Bronold2018} discuss true-secondary electrons
  excited with energy coming change of internal energy levels of
  incoming ions; these effects are neglected in this work.}

The \cite{Bronold2015} model can be implemented into the simulation in
the same manner as the \cite{Furman2002} model. However,
\cite{Bronold2015} provide an interesting discussion later in the
paper.  All the relations above are derived for ideally flat walls
without any defects.  To address effects of real walls, the authors
modify the relations based on \cite{Smith1998} by adding terms with
parameter $C$, which is proportional to the density scattering
centers.  With $C=1$ and $C=2$ the results match experimental data
very well (see Fig.\thinspace3 in \cite{Bronold2015}; results are much
better than for $C=0$).  What is more, with increasing $C$, the
effects of $\delta R(E,\mu,E',\mu')$ become less important.  This
presents an interesting opportunity to develop reasonably accurate and
computational inexpensive boundary conditions by neglecting the
rediffusion and using the roughness-modified formula for the
probability of a quantum-mechanical reflection (Eq.\thinspace(13) in
\cite{Bronold2015}),
\begin{align}\label{eq:bounded:bronold_C}
  \overline{\mathcal{T}}(E',\mu') =
  \frac{\mathcal{T}(E',\mu')}{1+C/\mu'} -
  \frac{C/\mu'}{1+C/\mu'}\int_{\mu_c}^1\mathcal{T}(E',\mu'')\,d\mu''.
\end{align}
Calculating these integrals for the reflection function of each
particle would still be quite expensive.  However, as emphasized
before, the energies and angles need to be treated as coordinates and
the integrals can be precomputed.

The whole process can be performed as follows.  We define the
reflection function as
\begin{align}\label{eq:bounded:bronold_R}
  R(E,\mu,E',\mu') = \left(1 - \frac{\mathcal{T}(E',\mu')}{1+C/\mu'} -
  \frac{C/\mu'}{1+C/\mu'}\int_{\mu_c}^1\mathcal{T}(E',\mu'')\,d\mu''\right)
  \delta(E-E')\delta(\mu-\mu'),
\end{align}
and describe $E$ and $\mu$ in terms of $\bm{v}$ and
$\bm{v}'$. \eqr{bounded:bronold_C} is then substituted into the
general formula in \eqr{bounded:general} and the integration over
$\bm{v}'$ is performed, which is made simple by the Dirac delta
functions.  The rest of the integrals are precomputed numerically,
\begin{multline}
  \mathcal{R}_{x,kl}^g = \int_{I_p} \left(1 -
  \frac{\mathcal{T}\big(E^g(\bm{\eta}_{\bm{v}}),
    \mu^g(\bm{\eta}_{\bm{v}})\big)}{1+C/\mu^g(\bm{\eta}_{\bm{v}})} -
  \frac{C/\mu^g(\bm{\eta}_{\bm{v}})}{1+C/\mu^g(\bm{\eta}_{\bm{v}})}
  \int_{\mu_c^g(\bm{\eta}_{\bm{v}})}^1\mathcal{T}(E^g(\bm{\eta}_{\bm{v}}),
  \mu'')\,d\mu''\right)\times
  \\\widehat{\psi}_l(-\eta_x,\eta_y,\eta_z,-\eta_{v_x},\eta_{v_y},\eta_{v_z})
  \widehat{\psi}_k(\eta_x,\eta_y,\eta_z,\eta_{v_x},\eta_{v_y},\eta_{v_z})
  \,d\eta_x d\eta_y d\eta_zd\eta_{v_x} d\eta_{v_y} d\eta_{v_z}.
\end{multline}

The reflection function, $R$, calculated with $\overline{\mathcal{T}}$
then significantly alters \fgr{bounded:bc_R_2D}.  The modified version
is in \fgr{bounded:bc_Rt_2D}.  Particularly noticeable is the absence
of regions with absolute reflection in the bottom-right sector (higher
energies and oblique angles).
\begin{figure}[!htb]
  \centering
  \includegraphics[width=0.8\linewidth]{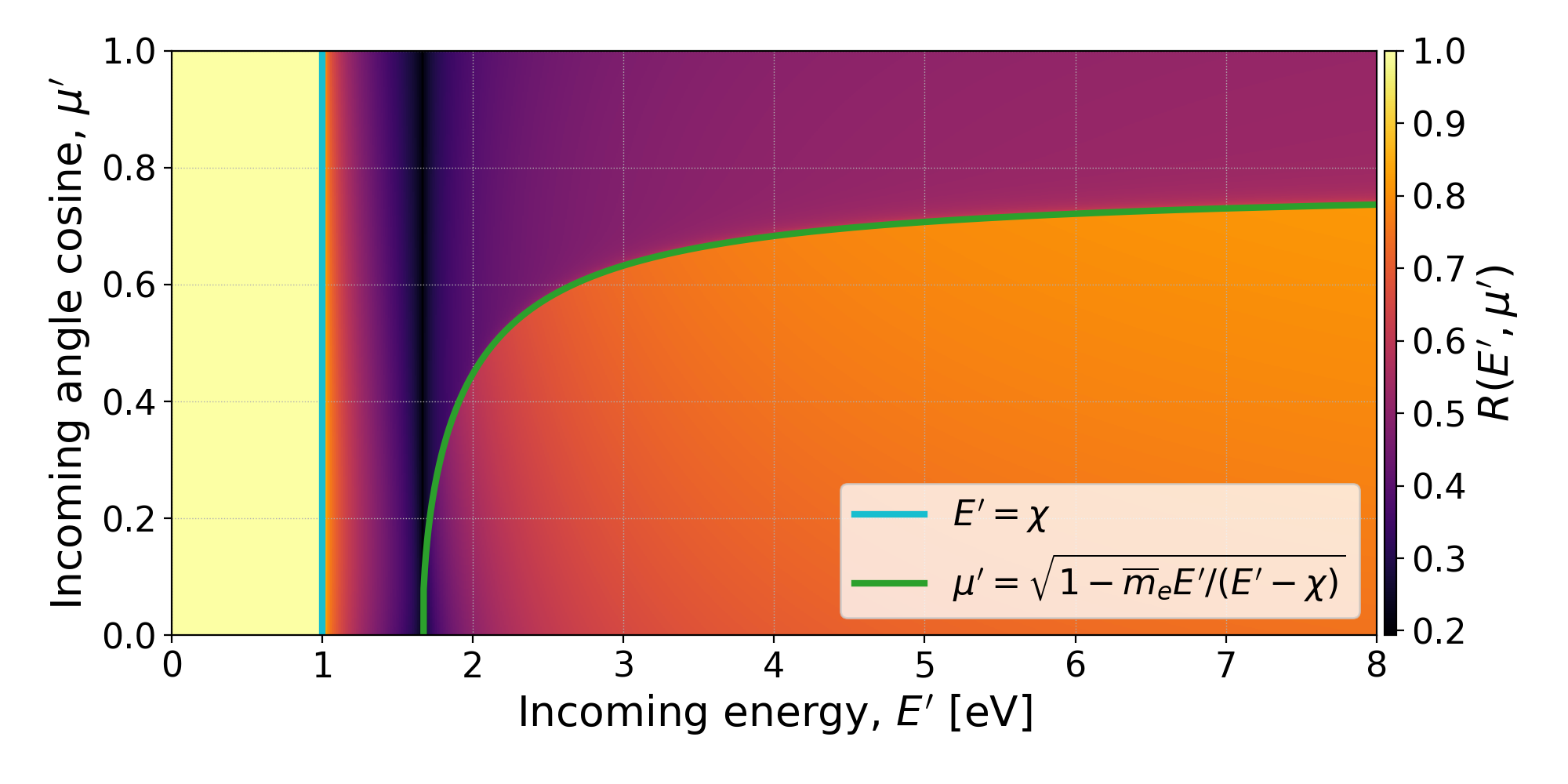}
  \caption[Modified probability of back-scattering from
    \cite{Bronold2015} model]{Probability of back-scattering,
    $R(E',\mu')$, from the \cite{Bronold2015} model modified with the
    roughness coefficient $C$ \eqrp{bounded:bronold_C}.  Using $C=2$
    and material parameters for MgO ($\chi=\SI{1}{eV}$ and
    $\overline{m}_e=0.4$).}
  \label{fig:bounded:bc_Rt_2D}
\end{figure}

Due to the complexity of $R$, the advances used to create DG kernels
and calculate moments cannot be used here and the boundary condition
needs to be precomputed for each cell.  What is more, as seen in
\fgr{bounded:bc_R_2D}, $R = 1$ for low energies and then quickly
drops.  Therefore, it is important to be careful with constructing the
velocity mesh.  The electron mesh used for previous simulations
extends from $-6\,v_{th,e}$ to $6\,v_{th,e}$ and uses 32 cells. This
puts the sharp transition at $E'=\chi$ inside the second cell
(counting from center).  As the polynomial approximation is not suited
for such sharp transitions, projection of $R$ onto this mesh results
in significant overshoot; see blue line in \fgr{bounded:bc_R}.
However, noting the strong ability of th DG method to handle
discontinuities and sharp gradients between the cells, the velocity
mesh can be tailored for the purposes of the boundary condition.  As
seen by the orange line in \fgr{bounded:bc_R_2D}, tailoring the mesh
eliminates the overshoot at $v_x\approx 0.5\,v_{th}$.

\begin{figure}[!htb]
  \centering
  \includegraphics[width=0.8\linewidth]{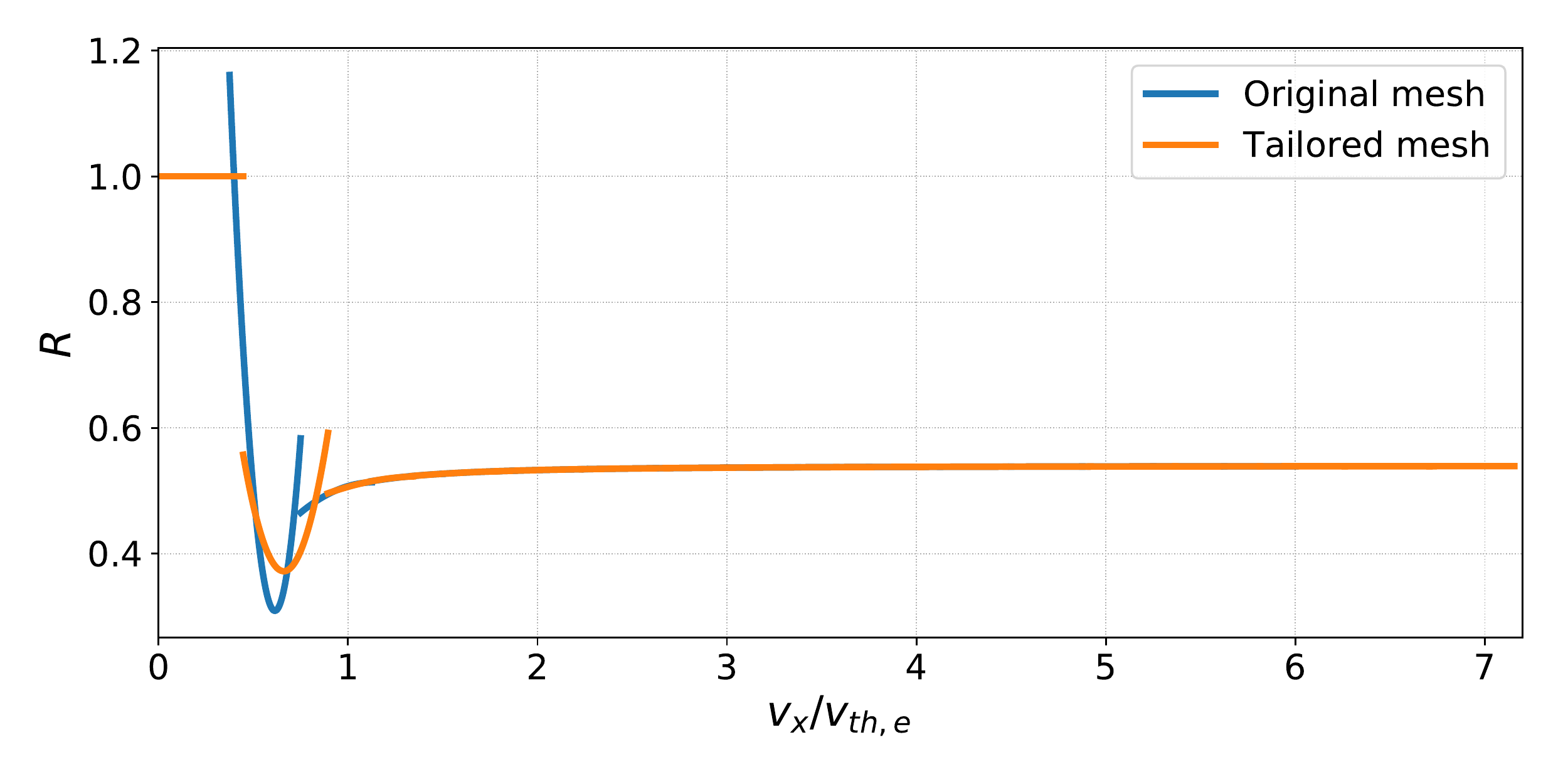}
  \caption[Projection of reflection function on DG basis]{Example of
    projecting the reflection function given by
    \eqr{bounded:bronold_R} onto the simulation mesh.  The same mesh
    is used as with previous simulations (blue line) resulting in an
    overshoot at $E=\chi$.  Orange line shows the result for mesh
    specifically tailored for material-based boundary conditions
    eliminating the overshoot at $v_x\approx 0.5\,v_{th}$.}
  \label{fig:bounded:bc_R}
\end{figure}

Similar to other key components of \texttt{Gkeyll 2.0} this boundary
condition can be precomputed and written as automatically generated
code with expanded matrix multiplications.  However, because it
changes based on the wall material and needs to be calculated for each
cell, it is stored as an external \texttt{Lua} file.  Following is an
example of a file for a second-order 1X1V simulation, i.e., with 8
basis functions in each cell.  The snipped defines each mode of the
outgoing distribution function, \texttt{fout}, in the cell with the
index 1 as a linear combination of up to eight incoming modes,
\texttt{fin}.  Note that since the coefficients are precomputed and
the matrix multiplication is expanded, the actual multiplication can
be limited only to the non-zero terms, saving computational time.
\begin{lstlisting}[language={[5.1]Lua}]
if idx[1] == 1 then
   fout[1] = 0.539057*fin[1] + 0.0000199555*fin[3] + 0.000000000000000194289*fin[5] - 0.000000501891*fin[6] - 0.00000000000000907781*fin[7]
   fout[2] = -0.539057*fin[2] - 0.0000199555*fin[4] + 0.000000501865*fin[8]
   fout[3] = -0.0000199555*fin[1] - 0.539057*fin[3] + 0.00000000000000909169*fin[5] - 0.0000178584*fin[6] - 0.000000000000401656*fin[7] + 0.00000000000000000245809*fin[8]
   fout[4] = 0.0000199555*fin[2] + 0.539057*fin[4] + 0.0000178573*fin[8]
   fout[5] = 0.000000000000000194289*fin[1] - 0.00000000000000909169*fin[3] + 0.539057*fin[5] + 0.000000000000448913*fin[6] + 0.0000199555*fin[7]
   fout[6] = -0.000000501891*fin[1] + 0.0000178584*fin[3] + 0.000000000000448913*fin[5] + 0.539057*fin[6] - 0.0000000000191165*fin[7]
   fout[7] = 0.00000000000000907781*fin[1] - 0.000000000000401656*fin[3] - 0.0000199555*fin[5] + 0.0000000000191165*fin[6] - 0.539057*fin[7] - 0.0000000000000000020854*fin[8]
   fout[8] = 0.000000501865*fin[2] + 0.00000000000000000245809*fin[3] - 0.0000178573*fin[4] - 0.0000000000000000020854*fin[7] - 0.539057*fin[8]
elseif ...
\end{lstlisting}
The \textit{Mathematica} script to create this whole file is listed in
\ref{list:scripts:bronold}.\footnote{Typically, \textit{Maxima} is the
  tool of choice in the \texttt{Gkeyll} team; however, I was unable to
  get it compute higher dimensions integrals which need to be
  calculated numerically and do not have analytical
  solution. \textit{Mathematica} handles it without much trouble.}

A possibly unexpected consequence of the dielectric boundary condition
implementation of \eqr{bounded:bronold_R} is additional cleaning of
the initial Langmuir wave.  This is caused by a smaller electron flux
to the wall as part of the electrons directly returns to the domain.
As a result, the initial relaxation of the system is less abrupt which
decreases the amplitude of the waves.  \fgr{bounded:bc_bronold_distf}
shows clear profiles of both electron and ion distribution functions
(snapshot at $t\omega_{pe} = 500$).  Note that no collisions were used
in this simulation.

\begin{figure}[!htb]
  \centering
  \includegraphics[width=0.8\linewidth]{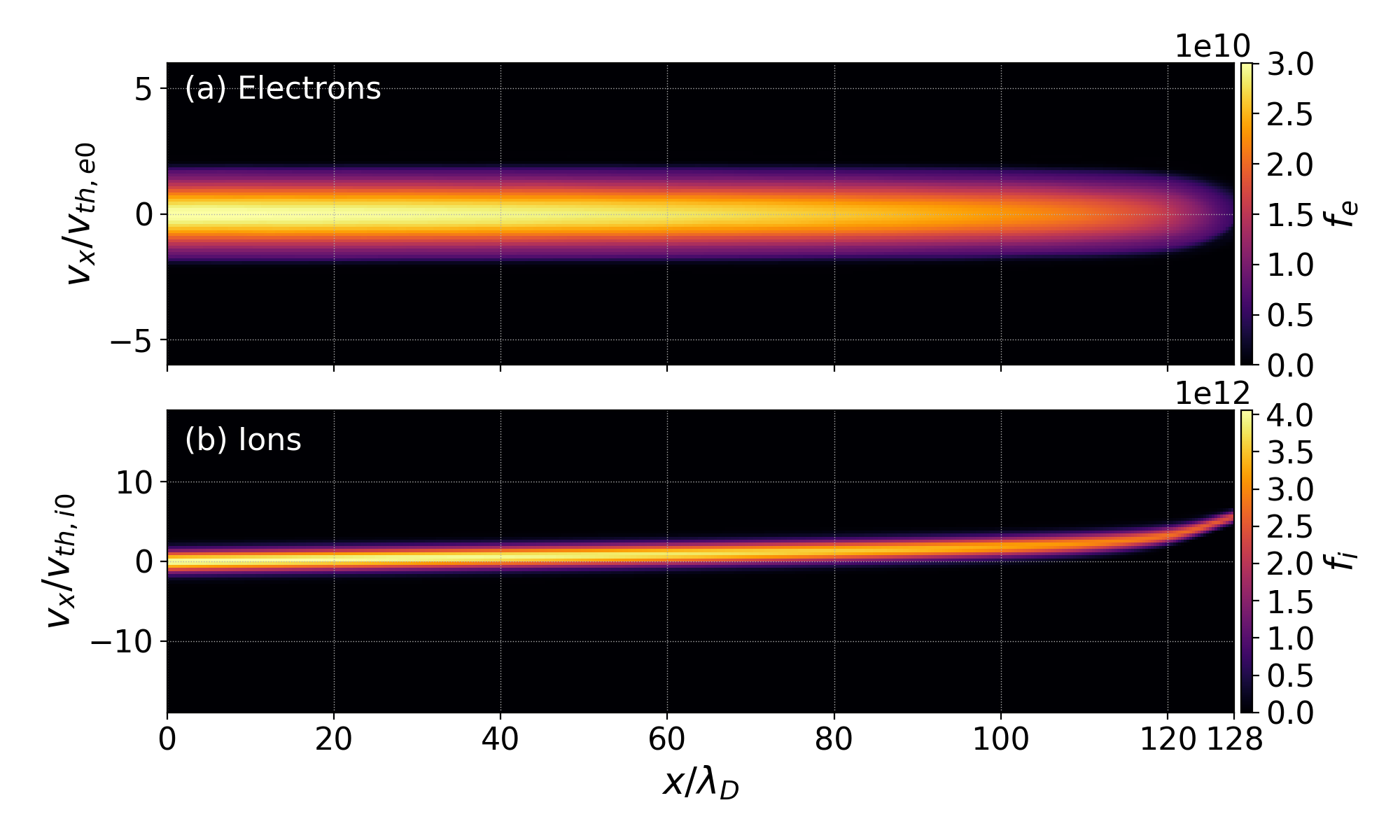}
  \caption[Electron and ion sheath distribution functions with
    dielectric wall]{Example of electron and ion distribution
    functions from a sheath simulation with dielectric wall
    \eqrp{bounded:bronold_R}.  Situation is captured at $t\omega_{pe}
    = 500$.  Note that initial Langmuir wave almost entirely
    disappeared even though no collisions were used for this run.}
  \label{fig:bounded:bc_bronold_distf}
\end{figure}

It is difficult to observe additional boundary condition features from
\fgr{bounded:bc_bronold_distf}.  Therefore,
\fgr{bounded:bc_bronold_diff} shows direct comparison (absolute
difference in the electron and ion distribution functions) of the
simulation with the dielectric boundary condition with the case that
uses ideally absorbing walls, i.e., the black hole boundary condition.
Analogous to \fgr{bounded:bc_bronold_distf}, the solution is captured
at $t\omega_{pe} = 500$ giving the simulations reasonable time to
evolve from the same initial conditions (\ser{bounded:sheath:init}).
Immediately noticeable is the periodic sign-changing structure
resulting from the absence of Langmuir waves in the case with ideally
absorbing wall.  What is more important, is the higher electron
density at the wall.  In the $v_x<0$ half of the velocity domain, we
even see the acceleration of emitted particles from the sheath
electric field.  The ion distribution (\fgr{bounded:bc_bronold_diff}b)
shows that ions reach lower velocities at the same distance from the
wall in comparison to the case with absorbing wall.

\begin{figure}[!htb]
  \centering
  \includegraphics[width=0.8\linewidth]{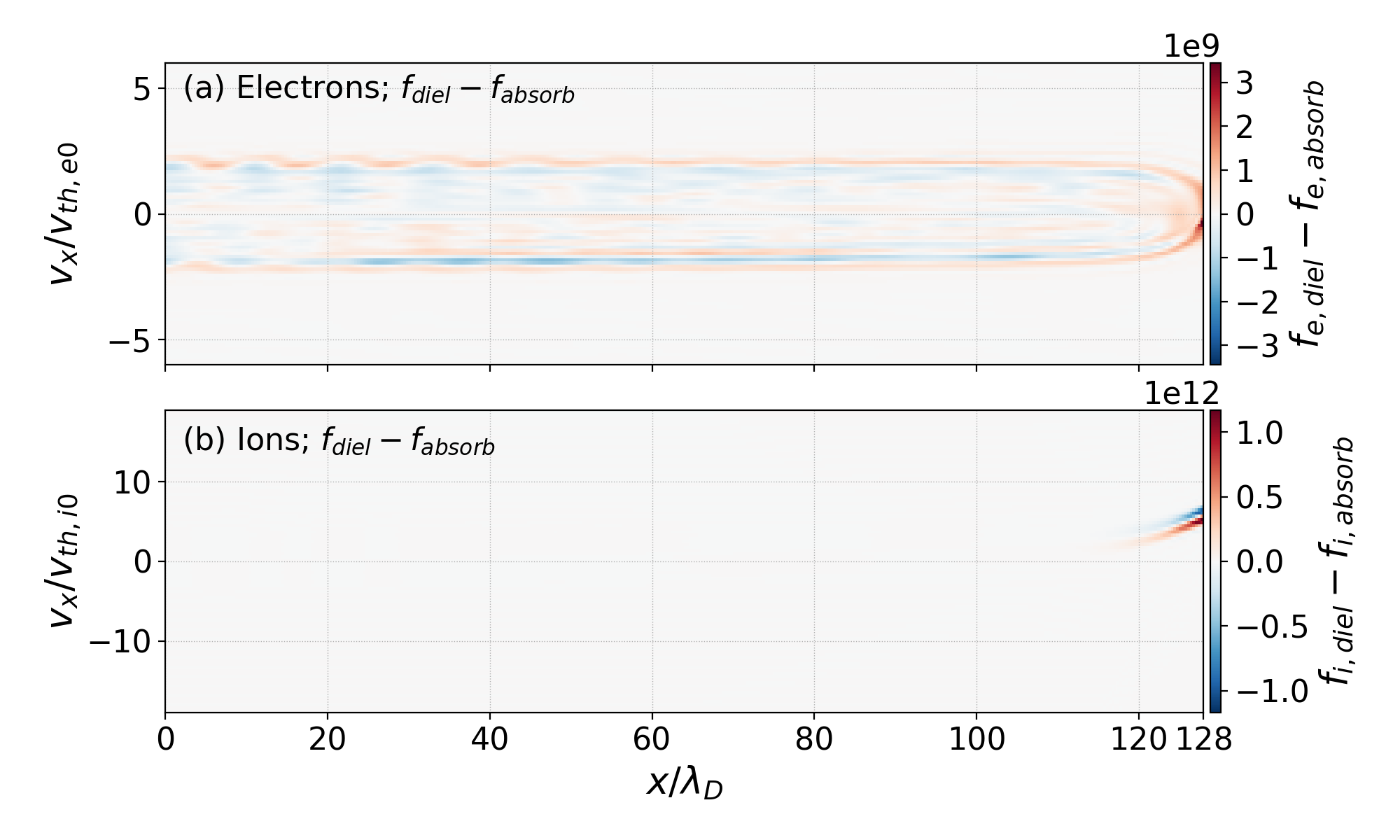}
  \caption[Direct comparison of distribution functions with absorbing
    and dielectric BCs]{Direct comparison of electron and ion
    distribution functions from sheath simulations with absorbing and
    dielectric boundary conditions ($f_{diel}-f_{absorb}$).  Red color
    denotes regions with higher particle phase space density in the
    case with dielectric wall boundary condition.  The periodic
    structure caused by the absence of Langmuir waves in the
    simulation with ideally absorbing wall.  Data are captured at
    $t\omega_{pe}=500$.}
  \label{fig:bounded:bc_bronold_diff}
\end{figure}

Plots of electron and ion densities, ion bulk velocity, electron and
ion temperatures, and electric fields are provided in
\fgr{bounded:bc_bronold_profiles}.  Simulation with the dielectric
boundary condition shows roughly doubled electron density right next
to the wall.  Returning electrons are also decreasing the overall
outflow from the domain resulting in significantly smaller electric
field needed to equalize the electron and ion fluxes.  The vertical
dashed line in \fgr{bounded:bc_bronold_profiles} marks the Bohm
velocity crossing for both cases which can be considered as the sheath
edge.  Note that the differences between the solutions for the
dielectric boundary condition and the ideally absorbing boundary
condition are localized inside the sheath region.  An exception are
small differences in the presheath electric field are caused by
Langmuir waves in the later case.  As a result, ions have the same
presheath acceleration profiles and reach the Bohm velocity at the
same distance from the wall.  The most significant difference is in
the electron temperature (\fgr{bounded:bc_bronold_profiles}d); in the
case with dielectric wall, the electron thermal velocity decrease in
the sheath region is significantly smaller.

\begin{figure}[!htb]
  \centering
  \includegraphics[width=0.8\linewidth]{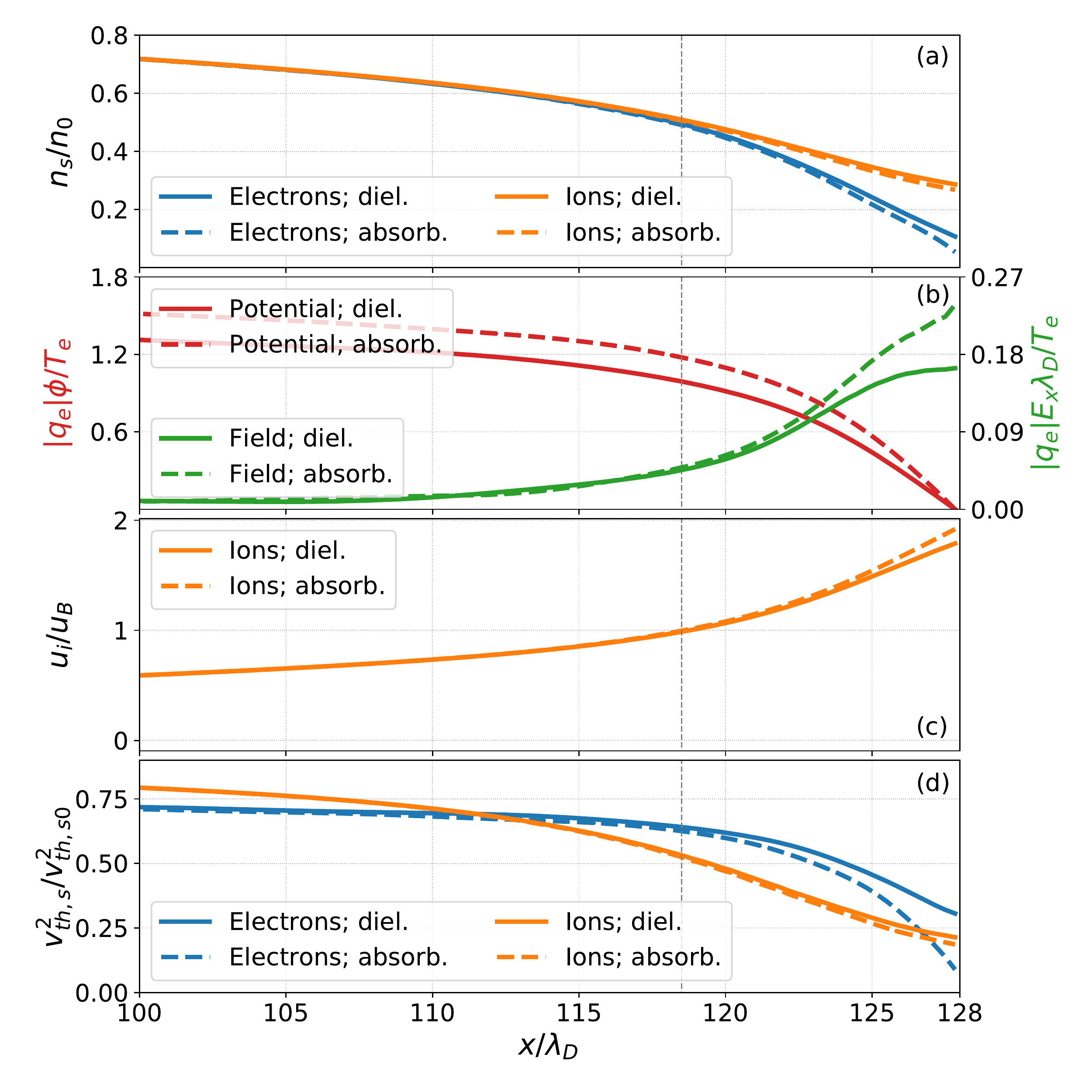}
  \caption[Comparison of sheath profiles for absorbing and dielectric
    BC]{Comparison of profiles from sheath simulations with absorbing
    and dielectric boundary conditions. From top to bottom, panels
    show density, ion bulk velocity, temperature, and electric field.
    In all the panels, solid line marks simulations with dielectric
    boundary condition based on \eqr{bounded:bronold_R} while the
    dashed lines correspond to simulation with ideally absorbing wall.
    Vertical dashed line marks crossing of the Bohm velocity
    \eqrp{bounded:bohm}.  Data are captured at $t\omega_{pe}=500$. No
    collisions or ionization are used for these runs.}
  \label{fig:bounded:bc_bronold_profiles}
\end{figure}

Similar to the discussion in \ser{bounded:temp}, explanation of the
temperature discrepancy requires higher moments of the distribution
function.  As the simulation used for
\fgr{bounded:bc_bronold_profiles} is limited to 1X1V, the third moment
gives only a scalar value, instead of the full heat flux tensor,
\begin{align*}
  q_e(x) = \frac{1}{2}m_e \int_{-\infty}^\infty v_x^3 f_e(x,v_x) \, dv_x.
\end{align*}
Normalized profile of $q_e$ in the region near the wall is shown in
\fgr{bounded:bc_bronold_qprofiles}a.  Due to the $v_x^3$ term, the
third moment is particularly sensitive to oscillations of the
distribution function like the Langmuir waves discussed in
\ser{bounded:sheath:sims}.  Therefore, the results in
\fgr{bounded:bc_bronold_qprofiles} are averaged over the full duration
of the simulation, $\Delta t\omega_{pe} = 1000$. 

\fgr{bounded:bc_bronold_qprofiles}a shows that the heat flux to the
wall is higher for the case with the dielectric wall BC, which might
seem to contradict the higher temperature shown in
\fgr{bounded:bc_bronold_profiles}d.  However, one needs to keep in
mind that $q_e$ describes an energy flux, i.e., it includes the local
particle density which is much higher for the case with the dielectric
wall.  The quantity plotted in \fgr{bounded:bc_bronold_qprofiles}a is
normalized to the initial number density in the center of the domain
so the result is dimensionless.  Alternatively, the third moment can
be normalized to the local number density, $q_e(x)/n_e(x)$, thus
removing the dependence; results then provide information about
``temperature flux''.  \fgr{bounded:bc_bronold_qprofiles}b shows the
comparison of the ``temperature fluxes'' for both of the dielectric
and absorbing cases.  The lower flux in the dielectric case is in
agreement with the higher electron temperature inside the sheath (see
\fgr{bounded:bc_bronold_profiles}d).

\begin{figure}[!htb]
  \centering
  \includegraphics[width=0.8\linewidth]{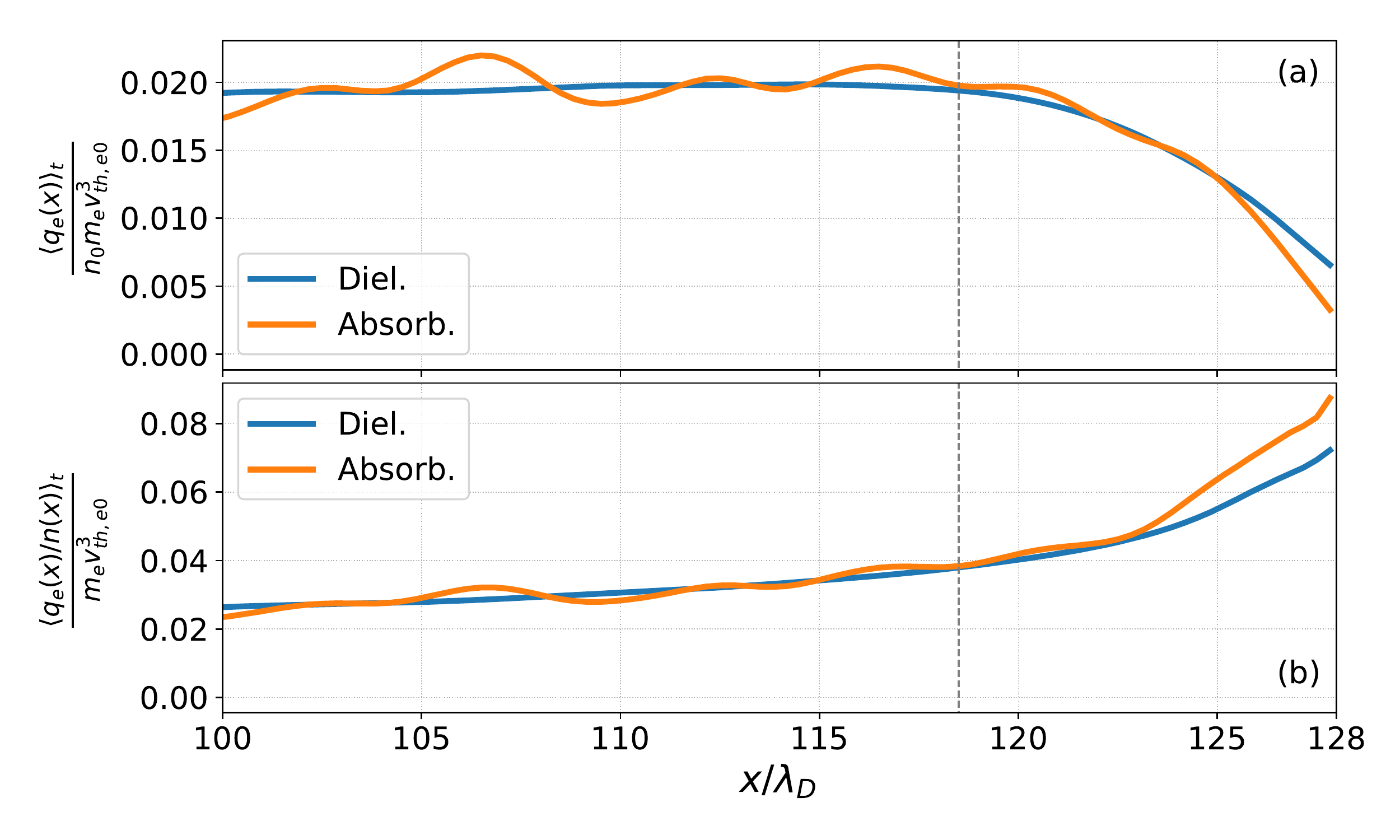}
  \caption[Comparison of heat flux profiles for absorbing and
    dielectric BC]{Comparison of heat flux profiles from sheath
    simulations with absorbing and dielectric boundary conditions.
    The top panel (a) shows the third moment of the distribution
    function, $q_e=\frac{1}{2}m_e\int v_x^3 f_e\,dv_x$, normalized to
    initial temperature and density, while the bottom panel (b)
    captures $q_e$ normalized to local density, $n_e(x)$. The profiles
    are averaged over the whole course of the simulation, $\Delta
    t\omega_{pe} = 1000$.}
  \label{fig:bounded:bc_bronold_qprofiles}
\end{figure}

%% file: conclusion.tex
\chapter{Conclusions}

Numerical algorithms are an important part of this work even though
the main motivation of this work is physics, because the understanding
of the underlying math is important for proper implementation of various
physics modules.  Therefore, a significant emphasis is placed on the
description of the discontinuous Galerkin method and its
implementation into the \texttt{Gkeyll}
framework.\footnote{\url{http://gkeyll.readthedocs.io/en/latest/}}
More precisely, this work uses and contributes to a brand new code
\texttt{Gkeyll 2.0}.  The new version of the code is based on a modal DG implementation that features highly optimized machine-generated kernels producing several orders of magnitude speed-up compared to the previous version of the code. \texttt{Gkeyll 2.0} also features an overhaul of input
files using the new \texttt{App} system which allows for compact and
easy-to-understand input files as seen in
Appendix\thinspace\ref{app:inputs}.  Therefore, results presented in
this work are easy to reproduce.

In the spirit of the discussion from the Introduction, the numerical
model is built from bottom up.  Discontinuous Galerkin implementation
\citep{Cockburn2001} of the Vlasov equation \eqrp{model:vlasov} is 
first tested on simple collisionless plasma simulations studying
Landau damping of electron Langmuir waves and growth of two-stream
instability as these text book problems are well suited for
benchmarking.  Growth and damping rates are extracted from the
simulations using reproducible ``sweeping fitting'' and have good agreement
to theoretical values predicted by linear theory.  For the set
parameters, growth rates of the two-stream instability are within
0.3\% of the predicted values.  The two-stream instability is also used to
demonstrate that the numerical method can converge even for relatively
coarse velocity resolutions.

The Weibel instability, discussed in Chapter 3, was originally meant as an
electromagnetic benchmark as both Langmuir waves and two-stream
instability are electrostatic in nature.  However, discrepancies
between the observed growth rates and the kinetic linear theory
dispersion relation, which was derived for this work, warranted a more
careful study.  The discrepancy is explained by non-negligible
increase in temperature even during the ``linear growth'' phase which
is not captured in the linear theory.  The analysis also confirmed
magnetic trapping as the mechanism of nonlinear saturation in case of
relatively warm electron beams.  What is more, we have reported
previously unpublished results on the role electric fields plays in saturation of the
instability for colder beams \citep{Cagas2017c}.  Our research of the
Weibel instability also resulted in a follow-up study of the interplay the
instability can have with the two-stream instability, which is being
pursued within the \texttt{Gkeyll} collaboration.

Study of plasma-material interactions is focused on plasma sheaths.
First of all, it is shown that even baseline collisionless
simulations, initialized with uniform density, and with ideally
absorbing walls self-consistently evolve the sheath profiles
\citep{Cagas2017s}. Within a few plasma oscillation periods, ion
population is accelerated to the Bohm velocity several Debye lengths
from the wall as predicted by the sheath theory.  Uniform
initialization, however, results in excitation of Langmuir waves as
shown by Fourier analysis.  Excitation of these waves can be decreased
by initializing simulations with approximate profiles obtained from a
semi-empirical model.  BGK collisions are added to repopulate
high-energy electrons from the tail of the distribution that are lost
to the wall. It is demonstrated that with local collision frequencies
the collisional term does not thermalize the distribution inside a
sheath region.  Significant temperature anisotropy exists in a sheath,
which can trigger the growth of the Weibel instability for certain
plasma regimes.

Finally, a new way to form boundary conditions through general
reflection functions is derived and implemented.  The concept is proven
by implementing a specular reflection boundary condition and then using
it to reproduce symmetry of a sheath simulation with absorbing walls
on both sides.  After letting the simulations evolve for
$1000/\omega_{pe}$, maximal relative differences between the
full-domain case with two walls and the half-domain case with
reflecting boundary conditions is on the order or $10^{-13}$.  With
confidence in the process, more complex reflection functions based on
\cite{Furman2002} and \cite{Bronold2015} models are discussed.  The
\cite{Bronold2015} absorption model is then self-consistently
implemented in a slightly approximated but efficient manner, showing
significant impact even on the simplest case of a 1X1V sheath.  With the
\cite{Bronold2015} based boundary condition, electron density next to
the wall is doubled and electric field magnitude is roughly 60\% in
comparison to the case with ideally absorbing walls.

%% file: appendix_postgkyl.tex
\chapter{Postgkyl: Gkeyll Postprocessing Suite}

\epigraph{There are two ways of constructing a software design. One
  way is to make it so simple that there are obviously no
  deficiencies. And the other way is to make it so complicated that
  there are no obvious deficiencies}{\textit{C. A. R. Hoare}}

Having a good postprocessing and visualization package can make a life
of a researcher much easier.  While most people write specialized
scripts for publication level figures, a tool to quickly probe
simulation data, ideally directly from a command line terminal of a
super-computer, is very valuable.  \texttt{Gkeyll 1.0} had a
Python script called \texttt{gkeplot} which allowed for easy plotting
together with features like DG interpolation.  Its downside was that
it was designed primary as a plotting script and, therefore, adding new
features became a bit difficult. For example, one
of the often used functions was to load time-evolution data from
several \texttt{hdf5} files, merge them, and save them as an ASCII
file.  Even though this feature was very useful, one can argue that it ideally should not be a part a plotting script.  For those and other reasons, a
decision has been made to create a new tool for \texttt{Gkeyll 2.0}
and \texttt{Postgkyl} has been started.  This new tool has been redesigned from bottom up and is now based on a system of modular pieces which can be arbitrarily chain together for various postprocessing tasks. \texttt{Postgkyl} also adapts current trends, for example the perceptually uniform color maps.\footnote{See \url{https://bids.github.io/colormap/} for more details.}  Full documentation of \texttt{Postgkyl} is now part of the \texttt{Gkeyll} project website
\url{http://gkeyll.readthedocs.io/en/latest/}.

\section{Installation}

\texttt{Postgkyl} is a Python package and can be cloned from its
repository \url{http://bitbucket.org/ammarhakim/postgkyl} or can be
directly installed from the Anaconda cloud with its \texttt{conda}
package manager
\begin{lstlisting}
conda install -c gkyl postgkyl
\end{lstlisting}
Note that the first case requires from user to manually install all the
dependencies (\texttt{numpy}, \texttt{scipy}, \texttt{matplotlib},
\texttt{click}, and \texttt{sympy}) and correctly set
\texttt{\$PYTHONPATH}.\footnote{The \texttt{\$PYTHONPATH} should point
  to the upper \texttt{postgkyl} directory.}  Installation through the
\texttt{conda} package manager performs the required setup automatically.

\section{Basic Terminal Functionality}

\texttt{Postgkyl}'s terminal mode can be quickly
used to probe simulation outputs and perform basic diagnostics.  The
call consists of the baseline script, \texttt{pgkyl}, followed by
various commands.  The baseline script also takes various flags, for
example \texttt{-f} for specifying a \texttt{Gkeyll} output file, common \texttt{-v}
for verbosity, but also flags for partial load of the files, which can be useful for high-dimensional distribution function data.
At any point, \texttt{--help} can be used to produce a output similar to the following:
\begin{lstlisting}
$ pgkyl --help
Usage: pgkyl [OPTIONS] COMMAND1 [ARGS]... [COMMAND2 [ARGS]...]...

Options:
  -f, --filename TEXT   Specify one or more files to work with.
  -s, --savechain       Save command chain for quick repetition.
  --stack / --no-stack  Turn the Postgkyl stack capabilities ON/OFF
  -v, --verbose         Turn on verbosity.
  --version             Print the version information.
  --c0 TEXT             Partial file load: 0th coord (either int or slice)
  --c1 TEXT             Partial file load: 1st coord (either int or slice)
  --c2 TEXT             Partial file load: 2nd coord (either int or slice)
  --c3 TEXT             Partial file load: 3rd coord (either int or slice)
  --c4 TEXT             Partial file load: 4th coord (either int or slice)
  --c5 TEXT             Partial file load: 5th coord (either int or slice)
  -c, --comp TEXT       Partial file load: comps (either int or slice)
  --help                Show this message and exit.

Commands:
  abs          Calculate absolute values of data
  agyro        Compute a measure of agyrotropy.
  collect      Collect data from the active datasets
  dataset      Select data sets(s)
  euler        Extract Euler (five-moment) primitive...
  fft          Calculate the Fast Fourier Transformartion
  growth       Fit e^(2x) to the data
  info         Print info of the current top of stack
  integrate    Integrate data over a specified axis or axes
  interpolate  Interpolate DG data on a uniform mesh
  log          Calculate natural log of data
  mult         Multiply data by a factor
  norm         Normalize data
  plot         Plot the data
  pop          Pop the data stack
  pow          Calculate power of data
  runchain     Run the saved command chain
  select       Subselect data set(s)
  tenmoment    Extract ten-moment primitive variables from...
  write        Write data into a file
\end{lstlisting}

One of the simpler but useful tasks is to load a
data file and follow it up with the \texttt{info} command for printing
information about data.  The following examples uses two-stream instability data from
\ser{benchmark:two-stream} [\ref{list:benchmark:two-stream}],
\begin{lstlisting}
$ pgkyl -f two-stream64_elc_0.bp info
Dataset #0
- Time: 0.000000e+00
- Frame: 0
- Number of components: 8
- Number of dimensions: 2
- Grid type: uniform
  - Dim 0: Num. cells: 64; Lower: -6.283185e+00; Upper: 6.283185e+00
  - Dim 1: Num. cells: 64; Lower: -6.000000e+00; Upper: 6.000000e+00
- Maximum: 1.902326e+00 at (31, 26) component 0
- Minimum: -3.119872e-01 at (31, 38) component 2
\end{lstlisting}
Note the information about the number of \texttt{components}.  That is how \texttt{Postgkyl} refers to the one extra dimension in \texttt{Gkeyll} data, which can represent many
things like components of an electromagnetic field, expansion coefficients of DG
data, or both.  In this case, the simulation uses second order 1X1V modal
Serendipity basis \eqrp{model:1x1v} thus the eight components.  It
should be also pointed out that the \texttt{--help} can be called for
a command rather than for the base script,
\begin{lstlisting}
$ pgkyl info --help
Usage: pgkyl info [OPTIONS]

  Print info of the current top of stack

Options:
  -a, --allsets  All data sets
  --help         Show this message and exit.
\end{lstlisting}

While getting information about output data is useful, probably the
most common use of \texttt{Postgkyl} is to quickly plot simulation results.  The following
command produces a figure result shown in \fgr{postgkyl:plot}
\begin{lstlisting}
$ pgkyl -f two-stream64_elc_100.bp plot
\end{lstlisting}
\begin{figure}[!htb]
  \centering
  \includegraphics[width=0.7\linewidth]{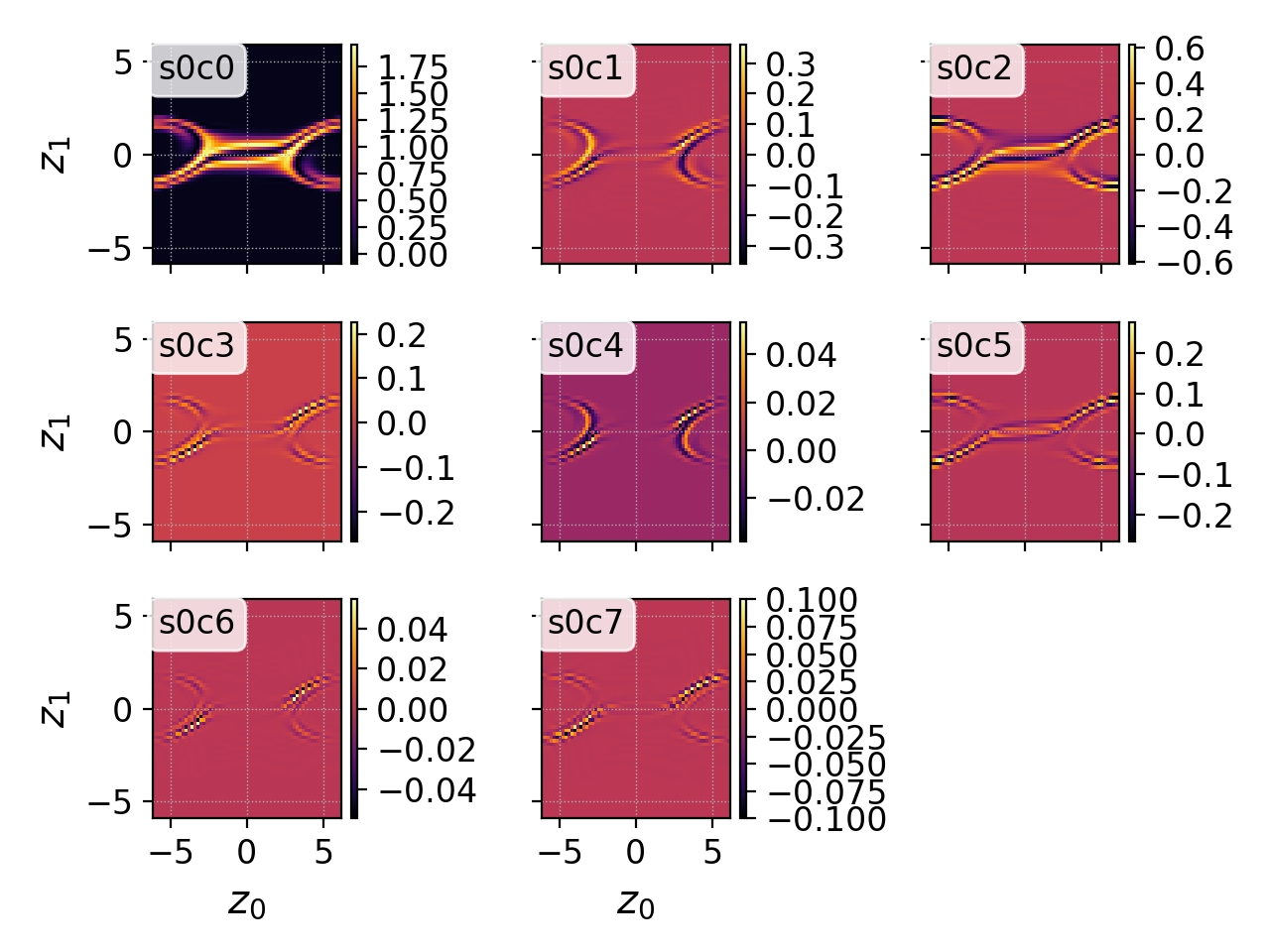}
  \caption{\texttt{pgkyl -f two-stream64\_elc\_100.bp plot}}
  \label{fig:postgkyl:plot}
\end{figure}

By default, \texttt{Postgkyl} assumes that when plotting data with
multiple components, user wants to compare them.  Therefore, in this
case, it outputs 2D plots for eight expansion coefficients for the basis from
\eqr{model:1x1v}. Plotting just a single component can be achieved
with the \texttt{select} command.  Alternatively, \texttt{Postgkyl}
can use the DG expansion coefficients to create uniform data with finer
resolution using the \texttt{interpolate} command
\begin{lstlisting}
$ pgkyl -f two-stream64_elc_100.bp interpolate -p2 -b ms plot
\end{lstlisting}
The result is in \fgr{postgkyl:plot_int}.  Note that the
\texttt{interpolate} command is included before \texttt{plot} (more on
that later) and has its own flags: \texttt{-p} for polynomial order
and \texttt{-b} for basis.
\begin{figure}[!htb]
  \centering
  \includegraphics[width=0.7\linewidth]{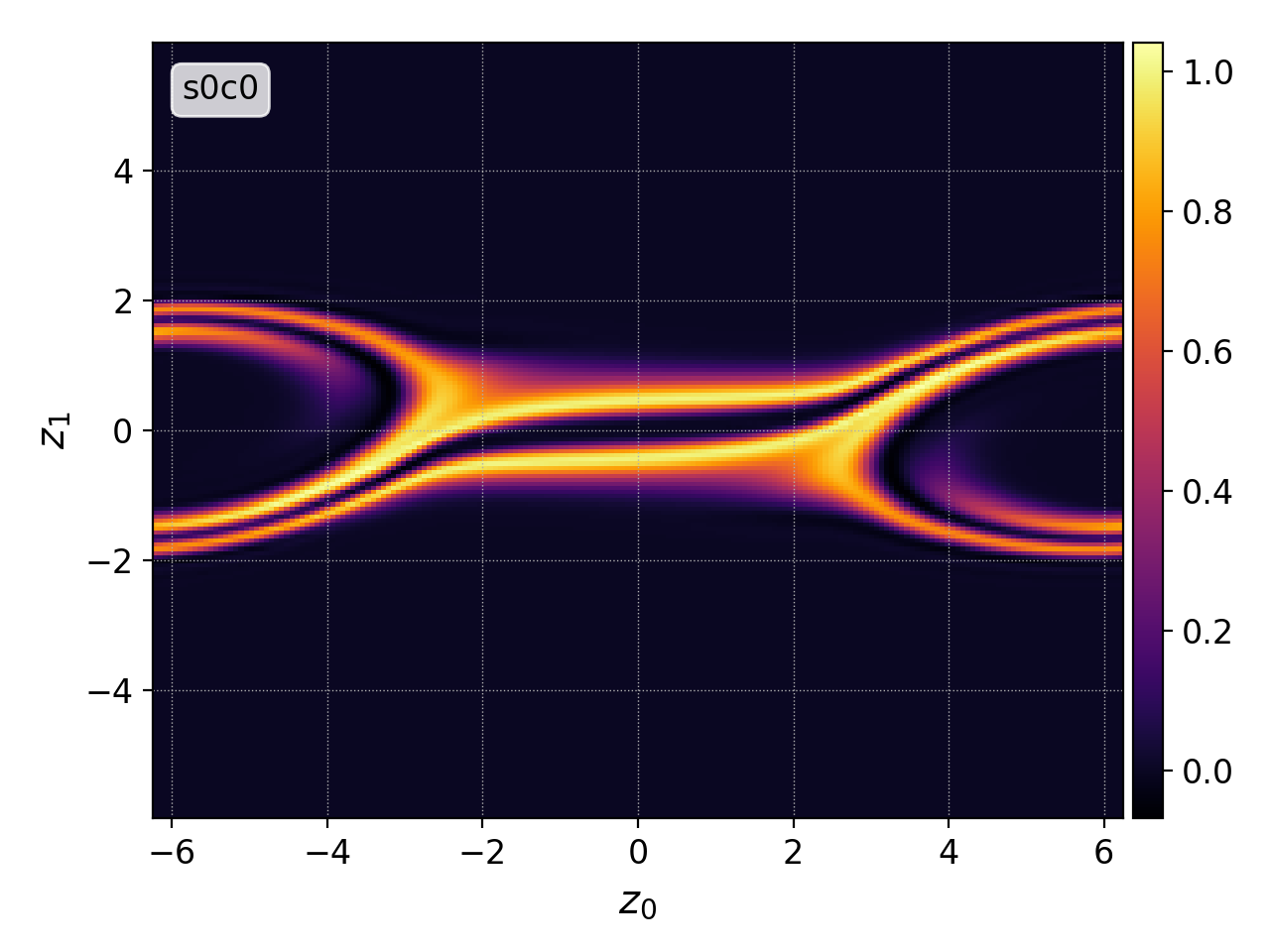}
  \caption{\texttt{pgkyl -f two-stream64\_elc\_100.bp interpolate -p2 -b
      ms plot}}
  \label{fig:postgkyl:plot_int}
\end{figure}

Note that the axis are by default labeled with neutral phase space
coordinate $z_i$.  These labels are set at the beginning and are
persistent through the \texttt{Postgkyl} command; therefore, if the
initial data are 1X2V and the second component is integrated out for
plotting (as it is the case for Weible instability in
\ser{weibel:phase}), \texttt{Postgkyl} will plot $z_0$ versus $z_2$.
Alternatively, labels can be manually set by \texttt{plot} flags
\begin{lstlisting}
$ pgkyl -f two-stream64_elc_100.bp interpolate -p2 -b ms plot -x'$x$' -y'$v_x$'
\end{lstlisting}

\FloatBarrier
\section{Terminal/Script Duality}

\texttt{Postgkyl} is designed with the
terminal/script duality in mind.  This is important because even
thought the majority of \texttt{Postgkyl} usage probably comes from
terminal mode, it is common to carefully craft plotting scripts for
publications\footnote{For example, all the plotting scripts for this
  work are stored in Jupyter notebooks for each of the chapters.} and these
scripts are required to load data as well.  Therefore, most of the
\texttt{Postgkyl} functionality comes from Python functions, which are
accessible from a Python script as well.  The terminal commands are only
\texttt{click}\footnote{\texttt{click} is a useful Python package for
  handling command line inputs; \url{http://click.pocoo.org/5/}}
wrappers for these functions.  For example, data file loading
\begin{lstlisting}
$ pgkyl -f two-stream64_elc_100.bp
\end{lstlisting}
can be exactly reproduced in a script with
\begin{lstlisting}[language=Python]
import postgkyl as pg
data = pg.GData('two-stream64_elc_100.bp')
\end{lstlisting}

\FloatBarrier
\section{Chaining of Commands}

The true strength of \texttt{Postgkyl} comes from the almost
unlimited command chaining options, which is enabled by the \texttt{click}
package mentioned above.  In essence, the command chains are  similar to Unix pipes, As data are
pushed through series of commands, result of one command being feed to the
next one.  This feature has been already shown in this section, when the
\texttt{interpolate} command was inserted before \texttt{plot}.  This allows
to build potentially complex diagnostics out of simple pieces,
enhancing usefulness of each command.

For example, \texttt{Postgkyl} can be used to probe an evolution of the velocity
profiles in the middle of the domain during the course of instability.
A single lineout can be created with the
\texttt{select}\footnote{\texttt{select} is used to specify a
  coordinate and/or a component for lineout.  For example, selecting
  component 1 and creating 1D data for $x=0$ (zeroth coordinate since
  Python is zero-index language)) is done with \texttt{select --c0 0.0
    --component 1}. Note that when the coordinate flag is given
  integer instead of float-point number, lineout is done based on
  coordinate index rather than value.}  command.  An example of this
is in \fgr{postgkyl:lineout}.
\begin{figure}[!htb]
  \centering
  \includegraphics[width=0.7\linewidth]{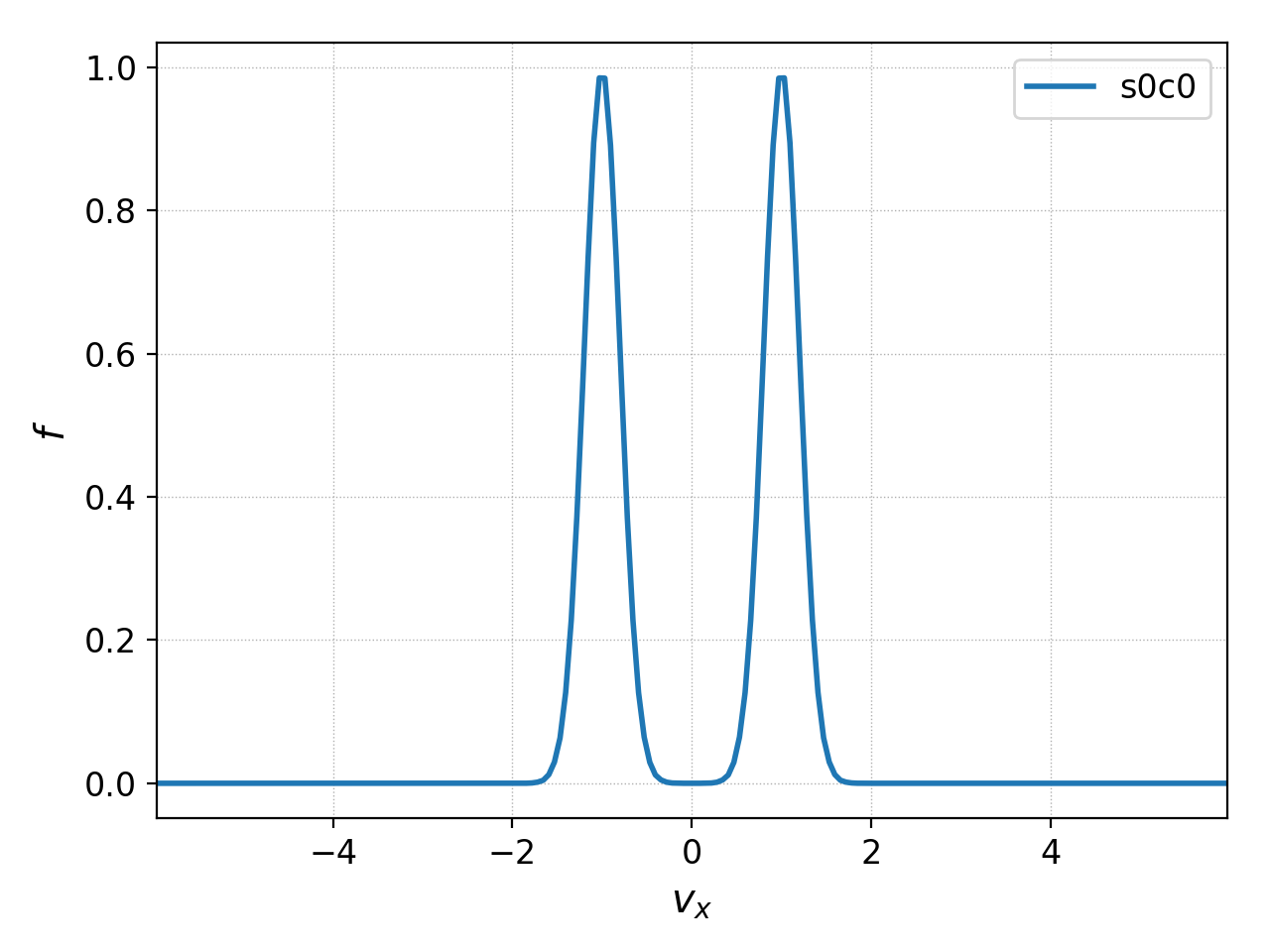}
  \caption{\texttt{pgkyl -f two-stream64\_elc\_0.bp interpolate -p2 -b ms select --c0 0.0 plot -x'\$v\_x\$' -y'\$f\$'}}
  \label{fig:postgkyl:lineout}
\end{figure}

Chaining of commands can become even more useful by loading multiple
files at once with wild-card characters and then using the
\texttt{collect} command, which stacks multiple data one after
another, e.g., creates 2D data out of a set of 1D lineouts:
\begin{lstlisting}
pgkyl -f 'two-stream64_elc_[0-9]*.bp' interpolate -p2 -b ms select --c0 0.0 collect plot -x'$t$' -y'$v_x$'
\end{lstlisting}
Note that for wild-card loading, the file name must be in quotes.  The
\texttt{[0-9]} wild-card stands for any single number; this is required
to exclude diagnostic moment data like
\texttt{two-stream64\_elc\_M0\_0.bp}.  The result of this command is
in \fgr{postgkyl:evolution}, which provides quite a lot of insight
into the instability.  It captures the decrease of kinetic energy of
the beams, as it is transformed into the electric field energy, and
the nonlinear phase space mixing later in time.

\begin{figure}[!htb]
  \centering
  \includegraphics[width=0.7\linewidth]{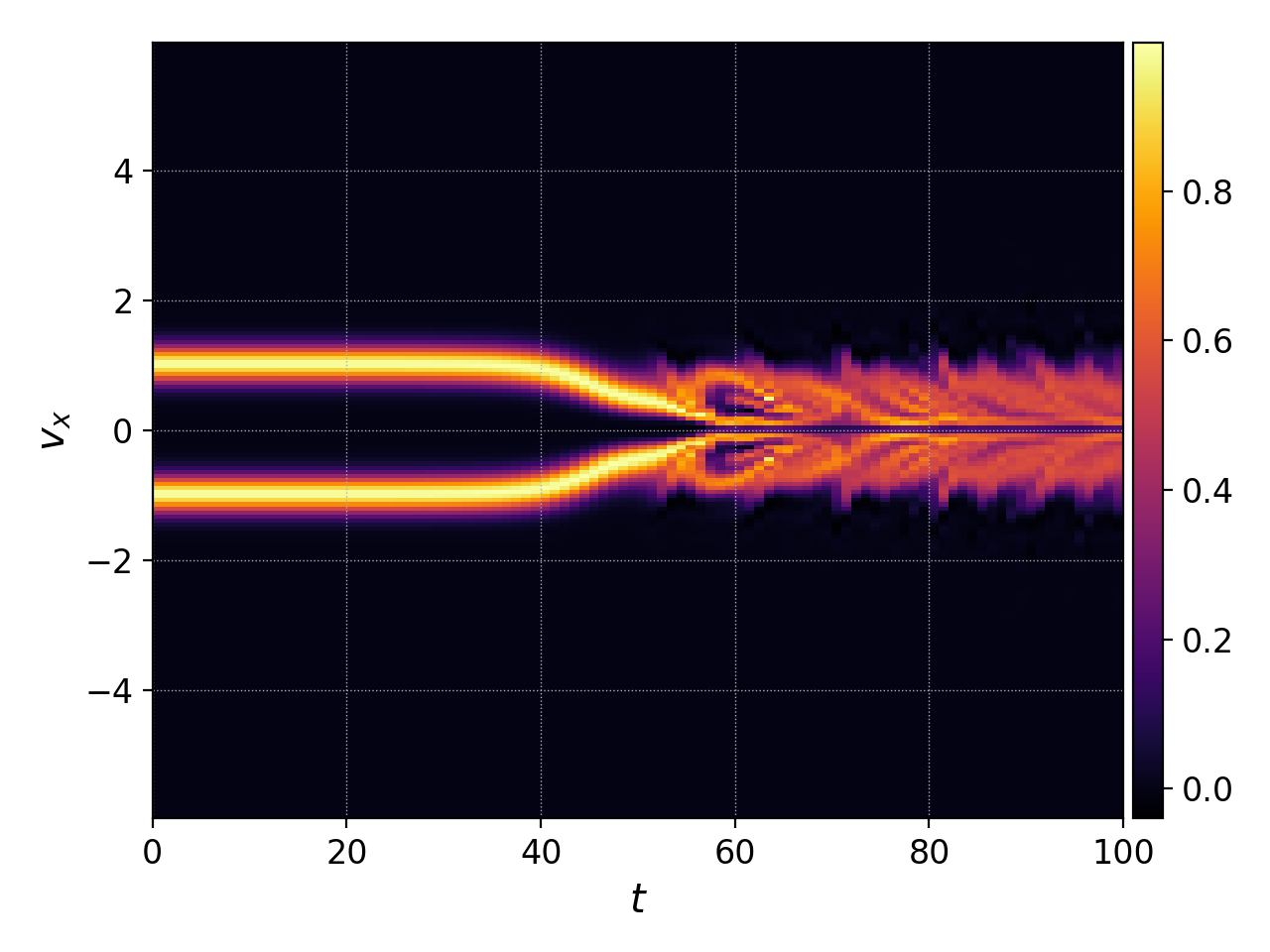}
  \caption{\texttt{pgkyl -f 'two-stream64\_elc\_[0-9]*.bp' interpolate -p2 -b ms select --c0 0.0 collect plot -x'\$t\$' -y'\$v\_x\$'}}
  \label{fig:postgkyl:evolution}
\end{figure}

With a simple change, \texttt{Postgkyl} can also integrate data along
$x$ instead of selecting $x=0$ before stacking.  As seen in
\fgr{postgkyl:evolution2}, this produces quite a different view.
\begin{figure}[!htb]
  \centering
  \includegraphics[width=0.7\linewidth]{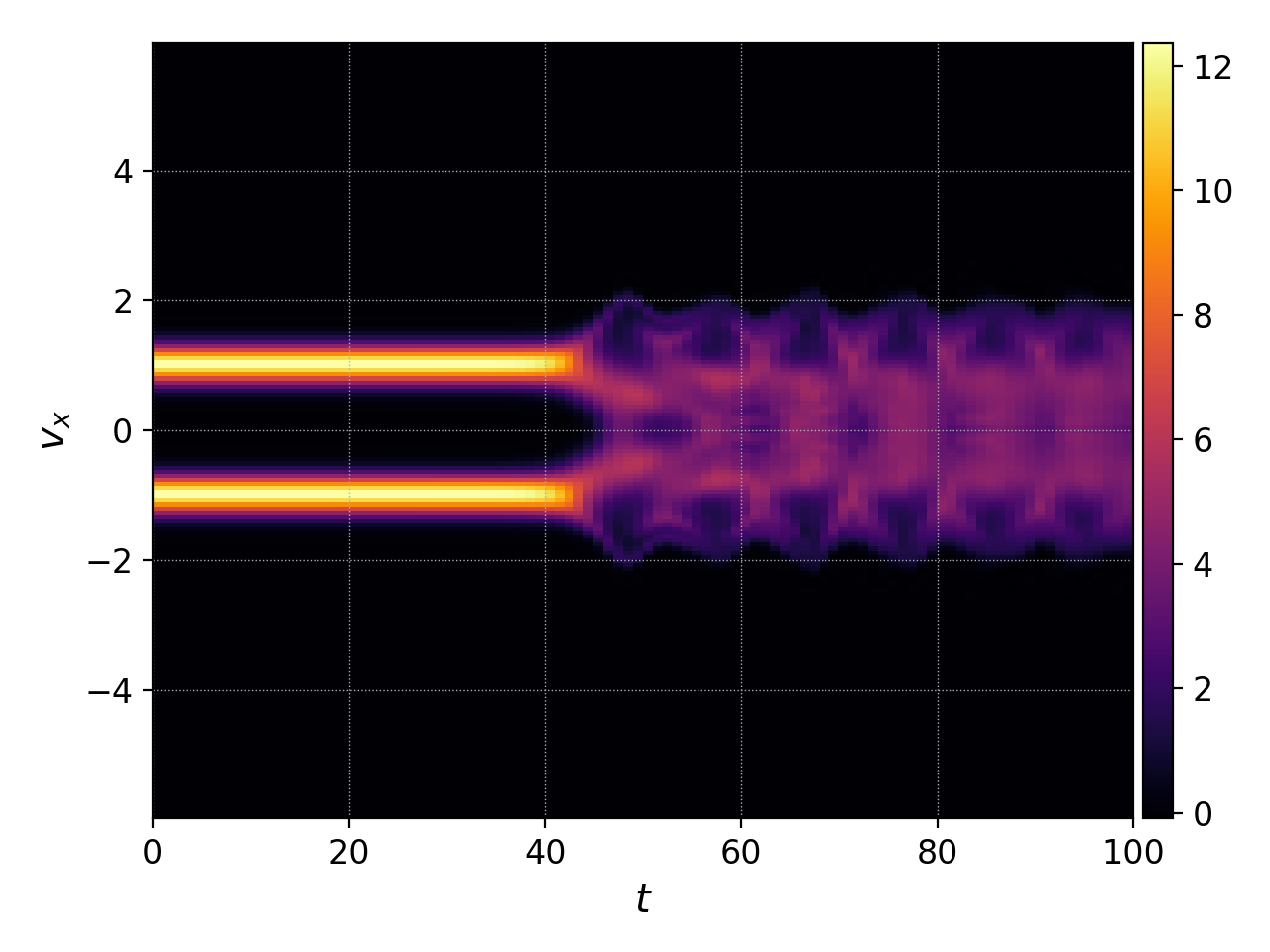}
  \caption{\texttt{pgkyl -f 'two-stream64\_elc\_[0-9]*.bp' interpolate -p2 -b ms integrate 0 collect plot -x'\$t\$' -y'\$v\_x\$'}}
  \label{fig:postgkyl:evolution2}
\end{figure}

With the same ease, \texttt{Postgkyl} can, for example, combine
\texttt{collect} with \texttt{fft} to create spectrograms.
There are no enforced limits for the chaining. This puts
a responsibility on the user to make sure that the command chain is
meaningful.  For example, repeating the \texttt{interpolate} command results in an
error because after the first command, data no longer have the valid form of
DG expansion coefficients.  As a final note, command chains can be
made more compact and less readable by using only the shortest unique
combination of starting letters for each command
\begin{lstlisting}
pgkyl -f 'two-stream64_elc_[0-9]*.bp' inter -p2 -b ms s --c0 0.0 c pl -x'$t$' -y'$v_x$'
\end{lstlisting}

%% file: appendix_cfi.tex
\chapter{Derivation of the Kinetic Weibel Instability Dispersion Relation}\label{app:weibel}

\epigraph{`Obvious' is the most dangerous word in
  mathematics.}{\textit{E. T. Bell}}

This Appendix provides the comprehensive derivation of the kinetic
Weibel instability dispersion relation \eqr{weibel:disp}.

\section{Linearized Equations}
 
Similar to the other instabilities, the derivation starts with the
linearizion of the Vlasov equation \eqrp{model:vlasov}. However, the
full Lorentz force \eqrp{model:motion2} must be taken into the
account. Since $\bm{B}_1=(0,0,B_{z,1})$ and $\bm{k} = (k_x,0,0)$, we
get
\begin{align*}
  -i\omega f_{1}^\pm + iv_xk_xf_{1}^\pm +
  \frac{q}{m}\left[(E_{x,1}+v_yB_{z,1})\partial_{v_x}
    f_{0}^\pm + (E_{y,1}-v_xB_{z,1})\partial_{v_y} f_{0}^\pm\right] = 0.
\end{align*}
Then,\footnote{Minus sign
  coming from moving $i$ to the numerator.}
\begin{align}\label{eq:app_weibel:perturbf}
  f_{1}^\pm = -\frac{iq}{m(\omega-v_xk_x)}
  \left[(E_{x,1}+v_yB_{z,1})\partial_{v_x} f_{0}^\pm +
    (E_{y,1}-v_xB_{z,1})\partial_{v_y} f_{0}^\pm\right] .
\end{align}

The Faraday's law \eqrp{model:faraday},
\begin{align*}
  i\begin{pmatrix}
    k_yE_{z,1}-k_zE_{y,1}\\
    k_zE_{x,1}-k_xE_{z,1}\\
    k_xE_{y,1}-k_yE_{x,1}
  \end{pmatrix} = i\omega \begin{pmatrix}
    B_{x,1}\\
    B_{y,1}\\
    B_{z,1}
  \end{pmatrix},
\end{align*}
simplifies to
\begin{align}\label{eq:app_weibel:linfaraday}
  k_xE_{y,1} = \omega B_{z,1}.
\end{align}

Finally, as it is discussed in \ser{weibel:lintheory},
linearizing the Ampere's law \eqrp{model:ampere} gives
\begin{align}\label{eq:app_weibel:linampere}
  -ik_xB_{z,1} = \mu_0q \sum_{\pm}\left(\int
  v_yf_{1}^{\pm}\,d\bm{v}\right) - \frac{i\omega}{c^2} E_{y,1}.
\end{align}

%---------------------------------------------------------------------
\FloatBarrier
\section{Dispersion Relation}

Substituting the perturbed distribution function
\eqrp{app_weibel:perturbf} into \eqr{app_weibel:linampere} gives
\begin{align*}
  -ik_x B_{z,1} = \mu_0\frac{-iq^2}{m}
  \sum_\pm\int v_y \frac{(E_{x,1}+v_yB_{z,1})\partial_{v_x}f_{0}^\pm +
      (E_{y,1}-v_xB_{z,1})\partial_{v_y}f_{0}^\pm}{\omega-v_xk_x}
  d\bm{v}- \frac{i\omega}{c^2} E_{y,1}.
\end{align*}
The equation then can be simplified\footnote{Using
  $\omega_{pe}^2=nq^2/(\varepsilon_0m)$ and
$c^2=1/(\varepsilon_0\mu_0)$.} and slightly rearranged,
\begin{align}\label{eq:app_weibel:step1}
   B_{z,1} = \frac{\omega_{pe}^2}{c^2k_x^2 n} \sum_\pm\int
   \Big[\underbracket{\frac{v_yE_{x,1}\partial_{v_x}f_{0}^\pm}{\frac{\omega}{k_x}-v_x}}_{\text{I}}
     +
     \underbracket{\frac{v_y^2B_{z,1}\partial_{v_x}f_{0}^\pm}{\frac{\omega}{k_x}-v_x}}_{\text{II}}
     +
     \underbracket{\frac{v_yE_{y,1}\partial_{v_y}f_{0}^\pm}{\frac{\omega}{k_x}-v_x}}_{\text{III}}
     -
     \underbracket{\frac{v_yv_xB_{z,1}\partial_{v_y}f_{0}^\pm}{\frac{\omega}{k_x}-v_x}}_{\text{IV}}
     \Big]
   d\bm{v}+ \frac{\omega}{c^2k_x} E_{y,1}.
\end{align}

Similar to \ser{benchmark:landau:lintheory}, the distribution function
can be factorized,
\begin{align*}
  \int_{\mathcal{V}} f_{0}^\pm\,d\bm{v} = n
  \int_{-\infty}^\infty\int_{-\infty}^\infty f_{xy,0}^\pm
  \,dv_xdv_y \underbracket{\int_{-\infty}^\infty f_{z,0}\,dv_z}_{=1},
\end{align*}
where $f_{xy,0}^\pm$ are assumed to be 2D Maxwellians,
\begin{align*}
  f_{xy,0}^\pm = \frac{1}{2\pi v_{th}^2}\exp\left(-\frac{v_x^2+(v_y\pm
      u_y)^2}{2v_{th}^2}\right).
\end{align*}
The derivatives of the Maxwellians can be readily evaluated
\begin{align}
  \partial_{v_x}f_{xy,0}^\pm &= -\frac{v_x}{v_{th}^2}\frac{1}{2\pi
    v_{th}^2}\exp\left(-\frac{v_x^2+(v_y\pm
      u_y)^2}{2v_{th}^2}\right),\\ \partial_{v_y}f_{xy,0}^\pm &=
  -\frac{v_y\pm u_y}{v_{th}^2}\frac{1}{2\pi
    v_{th}^2}\exp\left(-\frac{v_x^2+(v_y\pm u_y)^2}{2v_{th}^2}\right).
\end{align}

The derivatives are then substituted into \eqr{app_weibel:step1} and
each term is assessed individually.

\subsubsection{First Term}
Splitting the integrals gives
\begin{align*}\begin{aligned}
  \sum_\pm \int_{-\infty}^\infty\int_{-\infty}^\infty
  \frac{v_yE_{x,1}\partial_{v_x}f_{xy,0}^\pm}{\frac{\omega}{k_x}-v_x}
  \,dv_xdv_y =& -\sum_\pm \int_{-\infty}^\infty\int_{-\infty}^\infty
  \frac{v_yE_{x,1}\frac{v_x}{v_{th}^2} \frac{1}{2\pi
      v_{th}^2}\exp\left(-\frac{v_x^2+(v_y\pm
      u_y)^2}{2v_{th}^2}\right)}{\frac{\omega}{k_x}-v_x}\,dv_xdv_y,
  \\ =& -\frac{E_{x,1}}{\pi}\sum_\pm\int_{-\infty}^\infty
  \frac{\frac{v_x}{v_{th}^2}\exp\left(-\frac{v_x^2}{2v_{th}^2}\right)}{\frac{\omega}{k_x}-v_x}
  \,\frac{dv_x}{\sqrt{2v_{th}^2}} \times\\ &\int_{-\infty}^\infty
  v_y\exp\left(-\frac{(v_y\pm
    u_y)^2}{2v_{th}^2}\right)\,\frac{dv_y}{\sqrt{2v_{th}^2}}.
\end{aligned}\end{align*}

Now the $v_y$ before the exponential function in the second integral
can be substituted with $v_y\pm u_y$.\footnote{The added terms cancel
  out because of the summation over the both
  populations,
\begin{align*}
\sum_\pm\int_{-\infty}^\infty (v_y\pm
  u_y)\exp\left(-\frac{(v_y\pm u_y)^2}{2v_{th}^2}\right)\,dv_y =&
  \sum_\pm\int_{-\infty}^\infty v_y\exp\left(-\frac{(v_y\pm
      u_y)^2}{2v_{th}^2}\right)\,dv_y +\\&
  \underbracket{u_y\int_{-\infty}^\infty \exp\left(-\frac{(v_y+
        u_y)^2}{2v_{th}^2}\right)\,dv_y - u_y\int_{-\infty}^\infty
    \exp\left(-\frac{(v_y- u_y)^2}{2v_{th}^2}\right)\,dv_y}_{=0}.
\end{align*}}
The integral
is then calculated from an odd function of $v_y\pm u_y$, which is defined for both $\pm \infty$,
\begin{align*}
  \int_{-\infty}^\infty (v_y\pm u_y)\exp\left(-\frac{(v_y\pm
      u_y)^2}{2v_{th}^2}\right)\,d(v_y\pm u_y) = 0,
\end{align*}
and, therefore, the entire first term is zero.

\subsubsection{Second Term}
Splitting the second term gives
\begin{align*}\begin{aligned}
  \sum_\pm \int_{-\infty}^\infty\int_{-\infty}^\infty
  \frac{v_y^2B_{z,1}\partial_{v_x}f_{xy,0}^\pm}{\frac{\omega}{k_x}-v_x}\,dv_xdv_y
  =& -\sum_\pm \int_{-\infty}^\infty\int_{-\infty}^\infty
  \frac{v_y^2B_{z,1}\frac{v_x}{v_{th}^2}\frac{1}{2\pi v_{th}^2}
    \exp\left(-\frac{v_x^2+(v_y\pm
      u_y)^2}{2v_{th}^2}\right)}{\frac{\omega}{k_x}-v_x}
  \,dv_xdv_y,\\ =& -\frac{B_{z,1}}{\pi}\sum_\pm\int_{-\infty}^\infty
  \frac{\frac{v_x}{v_{th}^2}
    \exp\left(-\frac{v_x^2}{2v_{th}^2}\right)}{\frac{\omega}{k_x}-v_x}\frac{dv_x}{\sqrt{2v_{th}^2}}
  \times \\ &\int_{-\infty}^\infty v_y^2\exp\left(-\frac{(v_y\pm
    u_y)^2}{2v_{th}^2}\right)\frac{dv_y}{\sqrt{2v_{th}^2}},\\ =&
  -\frac{B_{z,1}}{\pi}\sum_\pm\int_{-\infty}^\infty
  \frac{\sqrt{\frac{2}{v_{th}^2}}\frac{v_x}{\sqrt{2v_{th}^2}}
    \exp\left(-\frac{v_x^2}{2v_{th}^2}\right)}{\frac{\sqrt{2v_{th}^2}}{\sqrt{2v_{th}^2}}
    \left(\frac{\omega}{k_x}-v_x\right)}\frac{dv_x}{\sqrt{2v_{th}^2}}
  \times \\ &\int_{-\infty}^\infty v_y^2\exp\left(-\frac{(v_y\pm
      u_y)^2}{2v_{th}^2}\right)\frac{dv_y}{\sqrt{2v_{th}^2}}.
\end{aligned}\end{align*}
Now using the substitution,
\begin{align*}
  \chi_x := \frac{v_x}{\sqrt{2v_{th}^2}}, \quad
    \zeta := \frac{\omega/k_x}{\sqrt{2v_{th}^2}}, \quad
    \chi_{y}^\pm := \frac{v_y\pm u_y}{\sqrt{2v_{th}^2}}.
\end{align*}
the equation simplifies to\footnote{Note that this substitution leads to
  $d\chi_x = dv_x/\sqrt{2v_{th}^2}$ and $d\chi_y =
  dv_y/\sqrt{2v_{th}^2}$.}
\begin{align*}
  \sum_\pm \int_{-\infty}^\infty\int_{-\infty}^\infty
  \frac{v_y^2B_{z,1}\partial_{v_x}f_{xy,0}^\pm}{\frac{\omega}{k_x}-v_x}\,dv_xdv_y
  &= -\frac{B_{z,1}}{\pi v_{th}^2}\sum_\pm \int_{-\infty}^\infty
  \frac{\chi_x\exp(-\chi_x^2)}{\zeta-\chi_x}\,d\chi_x
  \int_{-\infty}^\infty v_y^2\exp(-(\chi_y^\pm)^2)\,d\chi_{y}^\pm,\\
  &= -\frac{B_{z,1}}{2\sqrt{\pi} v_{th}^2}\sum_\pm Z'(\zeta)
   \int_{-\infty}^\infty
   v_y^2\exp(-(\chi_{y}^\pm)^2)\,d\chi_{y}^\pm.
\end{align*}

Finally, $v_y^2$ needs to be expressed in terms of  $\chi_y^2$ in
order to perform the integration,
\begin{align*}
  v_y^2 &= (v_y\pm u_y)^2 \mp2v_yu_y-u_y^2,\\
  &= (v_y\pm u_y)^2 \mp 2(v_y\pm u_y)u_y + 2u_y^2-u_y^2,\\
  &= 2v_{th}^2(\chi_{y}^\pm)^2 \mp 2\sqrt{(2v_{th}^2)}\chi_{y}^\pm u_y+u_y^2.
\end{align*}
The middle term, $2\sqrt(2v_{th}^2)\chi_{y}^\pm u_y$, will have no
effect on the result, because it integrates to zero as an odd
function.  Substituting the rest and using Gauss
integrals\footnote{$\int_{-\infty}^\infty\exp(-x^2)=\sqrt{\pi}$ and
  $\int_{-\infty}^\infty x^2\exp(-x^2)=\sqrt{\pi}/2$.} gives
\begin{align*}\begin{aligned}
    \sum_\pm \int_{-\infty}^\infty\int_{-\infty}^\infty
    \frac{v_y^2B_{z,1}\partial_{v_x}f_{xy,0}^\pm}{\frac{\omega}{k_x}-v_x}\,dv_xdv_y
    &=  -\frac{B_{z,1}}{2\sqrt{\pi} v_{th}^2}\sum_\pm Z'(\zeta)
    \int_{-\infty}^\infty
   \left[(2v_{th}^2(\chi_{y}^\pm)^2+u_y^2\right] \exp(-(\chi_{y}^\pm)^2)\,d\chi_{y}^\pm,\\
    &=- \sum_\pm\frac{B_{z,1}}{2}Z'(\zeta)\left(\frac{u_y^2}{v_{th}^2}+1\right),\\
    &=- B_{z,1}Z'(\zeta)\left(\frac{u_y^2}{v_{th}^2}+1\right).
\end{aligned}\end{align*}

\subsubsection{Third Term}
Splitting the third term gives
\begin{align*}\begin{aligned}
  \sum_\pm \int_{-\infty}^\infty\int_{-\infty}^\infty \frac{v_y
    E_{y,1}\partial_{v_y}f_{xy,0}^\pm}{\frac{\omega}{k_x}-v_x}
  \,dv_xdv_y =& -\sum_\pm \int_{-\infty}^\infty\int_{-\infty}^\infty \frac{v_yE_{y,1}\frac{v_y\pm
      u_y}{v_{th}^2}\frac{1}{2\pi v_{th}^2}\exp\left(-\frac{v_x^2+(v_y\pm
        u_y)^2}{2v_{th}^2}\right)}{\frac{\omega}{k_x}-v_x}
  \,dv_xdv_y,\\
  =& -\frac{E_{y,1}}{\pi}\sum_\pm\int_{-\infty}^\infty
  \frac{\exp\left(-\frac{v_x^2}{2v_{th}^2}\right)}{\frac{\omega}{k_x}-v_x}
  \frac{dv_x}{\sqrt{2v_{th}^2}} \times \\ &\int_{-\infty}^\infty v_y\frac{v_y\pm
    u_y}{v_{th}^2}\exp\left(-\frac{(v_y\pm
      u_y)^2}{2v_{th}^2}\right)\frac{dv_y}{\sqrt{2v_{th}^2}}
\end{aligned}\end{align*}
Similar to the first term, we substitute $v_y\pm u_y$ for $v_y$ in the second integral,
\begin{align*}\begin{aligned}
  \sum_\pm \int_{-\infty}^\infty\int_{-\infty}^\infty \frac{v_y
    E_{y,1}\partial_{v_y}f_{xy,0}^\pm}{\frac{\omega}{k_x}-v_x}
  \,dv_xdv_y &= -\frac{E_{y,1}}{\pi}\sum_\pm\int_{-\infty}^\infty
  \frac{\exp\left(-\frac{v_x^2}{2v_{th}^2}\right)}{\frac{\omega}{k_x}-v_x}
  \frac{dv_x}{\sqrt{2v_{th}^2}} \times \\ &\int_{-\infty}^\infty 2\frac{(v_y\pm
    u_y)^2}{2v_{th}^2}\exp\left(-\frac{(v_y\pm
      u_y)^2}{2v_{th}^2}\right)\frac{dv_y}{\sqrt{2v_{th}^2}},\\
  =&  -\frac{E_{y,1}}{\sqrt{2v_{th}^2}\pi}\sum_\pm \int_{-\infty}^\infty 
  \frac{\exp(-\chi_x^2)}{\zeta-\chi_x} \,d\chi_x
  \int_{-\infty}^\infty 2(\chi_{y}^\pm)^2
  \exp(-(\chi_{y}^\pm)^2)\,d\chi_{y}^\pm,\\
  =& \sum_\pm\frac{E_{y,1}}{\sqrt{2v_{th}^2}}Z(\zeta),\\
  =& \frac{2E_{y,1}}{\sqrt{2v_{th}^2}}Z(\zeta).
\end{aligned}\end{align*}

\subsubsection{Fourth Term}
Finally, splitting the fourth term
\begin{align*}\begin{aligned}
  \sum_\pm \int_{-\infty}^\infty\int_{-\infty}^\infty \frac{-v_yv_xB_{z,1}\partial_{v_y}f_{xy,0}^\pm}{\frac{\omega}{k_x}-v_x}
  \,dv_xdv_y
  =& \sum_\pm \int_{-\infty}^\infty\int_{-\infty}^\infty \frac{v_yv_xB_{z,1}\frac{v_y\pm u_y}{v_{th}^2}
    \frac{1}{2\pi v_{th}^2}\exp\left(-\frac{v_x^2+(v_y\pm 
      u_y)^2}{2v_{th}^2}\right)}{\frac{\omega}{k_x}-v_x} \,dv_xdv_y,\\
  =& \frac{B_{z,1}}{\pi}\sum_\pm \int_{-\infty}^\infty 
  \frac{v_x\exp\left(-\frac{v_x^2}{2v_{th}^2}\right)}{\frac{\omega}{k_x}-v_x}
  \frac{dv_x}{\sqrt{2v_{th}^2}} \times \\ &
  \int_{-\infty}^\infty
  v_y\frac{v_y\pm u_y}{v_{th}^2}\exp\left(-\frac{(v_y\pm
      u_y)^2}{2v_{th}^2}\right)\frac{dv_y}{\sqrt{2v_{th}^2}},\\
   =& \frac{B_{z,1}}{\pi}\sum_\pm \int_{-\infty}^\infty 
  \frac{v_x\exp\left(-\frac{v_x^2}{2v_{th}^2}\right)}{\frac{\omega}{k_x}-v_x}
  \frac{dv_x}{\sqrt{2v_{th}^2}} \times\\&
  \int_{-\infty}^\infty
  2\frac{(v_y\pm u_y)^2}{2v_{th}^2}\exp\left(-\frac{(v_y\pm
      u_y)^2}{2v_{th}^2}\right)\frac{dv_y}{\sqrt{2v_{th}^2}},\\
  =&  \frac{B_{z,1}}{\pi}\sum_\pm \int_{-\infty}^\infty 
  \frac{\chi_x \exp(-\chi_x^2)}{\zeta-\chi_x} \,d\chi_x
  \int_{-\infty}^\infty 2\chi_{y}^\pm
  \exp(-(\chi_{y}^\pm)^2) \,d\chi_{y}^\pm,\\
  =& \sum_\pm \frac{B_{z,1}}{2}Z'(\zeta).\\
  =& B_{z,1}Z'(\zeta).
\end{aligned}\end{align*}

\subsubsection{Combining the Terms}
Putting everything together gives
\begin{align*}
  B_{z,1} = \frac{\omega_{pe}^2}{c^2k_x^2}
  \left[-B_{z,1}Z'(\zeta)\left(\frac{u_y^2}{v_{th}^2}+1\right)+
    \frac{2E_{y,1}}{\sqrt{2v_{th}^2}}Z(\zeta)+B_{z,1}Z'(\zeta)\right]+
  \frac{\omega}{c^2k_x}E_{y,1}.
\end{align*}

As a final step, $E_{y,1}$ is substituted for from Faraday's law
\eqrp{app_weibel:linfaraday},
\begin{align*}\begin{aligned}
  B_{z,1} &= \frac{\omega_{pe}^2}{c^2k_x^2}
  \left[-B_{z,1}Z'(\zeta)\frac{u_y^2}{v_{th}^2}+ 2\frac{\omega
      B_{z,1}}{\sqrt{2v_{th}^2}k_x} Z(\zeta)\right]+
  \frac{\omega^2}{c^2k_x^2}B_{z,1},\\ 1 &=
  \frac{\omega_{pe}^2}{c^2k_x^2}
  \left[-Z'(\zeta)\frac{u_d^2}{v_{th}^2}+ 2\zeta
    Z(\zeta)\right]+\frac{\omega^2}{c^2k_x^2},
\end{aligned}\end{align*}
which is a single equation with one unknown $\omega$.
\eqr{benchmark:derdispf} can be used for the final rearrangement to
obtain \eqr{weibel:disp},
\begin{align*}
 \frac{1}{2} = \frac{\omega_{pe}^2}{c^2k_x^2} \left[\zeta
   Z(\zeta)\left(1+\frac{u_y^2}{v_{th}^2}\right) +
   \frac{u_y^2}{v_{th}^2}\right] + \frac{v_{th}^2}{c^2}\zeta^2
\end{align*}

%% file: appendix_inputs.tex
\chapter{Gkeyll Input Files}\label{app:inputs}

This appendix provides a list of the \texttt{Gkeyll 2.0} input files
for the simulations mentioned in the text.  For testing purposes,
\texttt{Gkeyll} executable is readily available through Anaconda
package manager \texttt{conda} and can be installed with just a
single command:
\begin{lstlisting}[language=Bash]
conda install -c gkyl gkyl
\end{lstlisting}

\pagebreak
\section{Reflection of the Neutral Gas}\label{list:input:bounce}
\begin{lstlisting}[language={[5.1]Lua}]
-- Gkyl --------------------------------------------------------------
-- Neutral gass bouncing between two walls ---------------------------
local Plasma = require "App.PlasmaOnCartGrid"

vth = 0.1
vd = 1.0

-- initialization function
local function maxwellian(n, vd, vth, v)
   return n / math.sqrt(2*math.pi*vth*vth) * 
      math.exp(-(v-vd)^2/(2*vth*vth))
end

sim = Plasma.App {
   logToFile = false,

   tEnd = 20, -- end time
   nFrame = 5, -- number of output frames
   lower = {-5.0}, -- configuration space lower left
   upper = {5.0}, -- configuration space upper right
   cells = {128}, -- configuration space cells
   basis = "serendipity", -- one of "serendipity" or "maximal-order"
   polyOrder = 2, -- polynomial order
   cflFrac = 1.0, -- CFL "fraction". Usually 1.0
   timeStepper = "rk3", -- one of "rk2" or "rk3"

   -- decomposition for configuration space
   decompCuts = {2}, -- cuts in each configuration direction
   useShared = false, -- if to use shared memory

   -- boundary conditions for configuration space
   periodicDirs = {}, -- periodic directions

   -- electrons
   neut = Plasma.VlasovSpecies {
      charge = 0.0, mass = 1.0,
      -- velocity space grid
      lower = {-2.0},
      upper = {2.0},
      cells = {48},
      decompCuts = {1},
      -- initial conditions
      init = function (t, xn)
         local x, v = xn[1], xn[2]
         return maxwellian(1.0, vd, vth, v)*math.exp(-x^2)
      end,
      evolve = true, -- evolve species?
      bcx = { Plasma.VlasovSpecies.bcReflect,
              Plasma.VlasovSpecies.bcReflect },
      diagnosticMoments = { "M0", "M1i", "M2" },
   },
}
-- run application
sim:run()
\end{lstlisting}

\pagebreak
\section{Landau Damping of Langmuir Wave}\label{list:benchmark:landau}
\begin{lstlisting}[language={[5.1]Lua}]
-- Gkyl --------------------------------------------------------------
-- Landau damping of Langmuir waves ----------------------------------
local Plasma = require "App.PlasmaOnCartGrid"

-- Parameters
local epsilon_0, mu_0 = 1.0, 1.0
local elemCharge = 1.0
local q_e, q_i = -elemCharge, elemCharge
local m_e, m_i = 1.0, 1836.0
local n_e, n_i = 1.0, 1.0
local T_e, T_i = 1.0, 1.0
local vth_e, vth_i = math.sqrt(T_e/m_e), math.sqrt(T_i/m_i)
local vd_e, vd_i = 0.0, 0.0

local omega_pe = math.sqrt((n_e * q_e^2)/(epsilon_0*m_e))
local lambda_D = math.sqrt((epsilon_0 * T_e)/(n_e * q_e^2))
-- artificially decrease the speed of light
local c = 10*vth_e
mu_0 = 1.0/(epsilon_0 * c^2)

-- Perturbations
kNumber = 0.5/lambda_D
A = 1e-1

-- initialization function
local function maxwellian1D(n, vd, vth, vx)
   return n / math.sqrt(2*math.pi*vth*vth) * 
      math.exp(-(vx-vd)^2/(2*vth*vth))
end

local tEnd = 20/omega_pe

sim = Plasma.App {
   logToFile = false,

   tEnd = tEnd, -- end time
   nFrame = 100, -- number of output frames
   lower = {0.0}, -- configuration space lower left
   upper = {2*math.pi/kNumber}, -- configuration space upper right
   cells = {64}, -- configuration space cells
   basis = "serendipity", -- one of "serendipity" or "maximal-order"
   polyOrder = 2, -- polynomial order
   cflFrac = 1.0, -- CFL "fraction". Usually 1.0
   timeStepper = "rk3", -- one of "rk2" or "rk3"

   -- decomposition for configuration space
   decompCuts = {2}, -- cuts in each configuration direction
   useShared = false, -- if to use shared memory
   -- boundary conditions for configuration space
   periodicDirs = {1},

   -- electrons
   elc = Plasma.VlasovSpecies {
      charge = q_e, mass = m_e,
      -- velocity space grid
      lower = {-6.0*vth_e},
      upper = {6.0*vth_e},
      cells = {128},
      decompCuts = {1},
      -- initial conditions
      init = function (t, xn)
         local x, vx = xn[1], xn[2]
         return  maxwellian1D(n_e, vd_e, vth_e, vx) *
            (1 + A * math.cos(kNumber*x))
      end,
      evolve = true,
      diagnosticMoments = { "M0", "M1i", "M2ij" },
   },

   ion = Plasma.VlasovSpecies {
      charge = q_i, mass = m_i,
      -- velocity space grid
      lower = {-6.0*vth_i},
      upper = {6.0*vth_i},
      cells = {128},
      decompCuts = {1},
      -- initial conditions
      init = function (t, xn)
         local x, vx = xn[1], xn[2]
         return maxwellian1D(n_i, vd_i, vth_i, vx)
      end,
      evolve = false, -- evolve species?
      diagnosticMoments = { "M0", "M1i", "M2ij" },
   },
   
   -- field solver
   field = Plasma.MaxwellField {
      epsilon0 = epsilon_0, mu0 = mu_0,
      init = function (t, xn)
         local x = xn[1]
         local E = A*q_e*n_e/kNumber/epsilon_0 * math.sin(kNumber*x)
         return E, 0.0, 0.0, 0.0, 0.0, 0.0
      end,
      evolve = true, -- evolve field?
   },
}
-- run application
sim:run()
\end{lstlisting}

\pagebreak
\section{Electron Two-stream Instability}\label{list:benchmark:two-stream}
\begin{lstlisting}[language={[5.1]Lua}]
-- Gkyl --------------------------------------------------------------
-- Electron Two-stream instability -----------------------------------
local Plasma = require "App.PlasmaOnCartGrid"

kNumber = 0.5 -- wave-number
vth_e = 0.2 -- electron thermal velocity
vd_e = 1.0 -- drift velocity
perturbation = 1.0e-6 -- distribution function perturbation

-- initialization function
local function maxwellian1D(n, vd, vth, vx)
   return n / math.sqrt(2*math.pi*vth*vth) * 
      math.exp(-(vx-vd)^2/(2*vth*vth))
end

sim = Plasma.App {
   logToFile = false,

   tEnd = 50.0, -- end time
   nFrame = 100, -- number of output frames
   lower = {-math.pi/kNumber}, -- configuration space lower left
   upper = {math.pi/kNumber}, -- configuration space upper right
   cells = {64}, -- configuration space cells
   basis = "serendipity", -- one of "serendipity" or "maximal-order"
   polyOrder = 2, -- polynomial order
   timeStepper = "rk3", -- one of "rk2", "rk3" or "rk3s4"

   -- decomposition for configuration space
   decompCuts = {8}, -- cuts in each configuration direction
   useShared = false, -- if to use shared memory
   -- boundary conditions for configuration space
   periodicDirs = {1}, -- periodic directions

   -- electrons
   elc = Plasma.VlasovSpecies {
      charge = -1.0, mass = 1.0,
      -- velocity space grid
      lower = {-6},
      upper = {6},
      cells = {32},
      decompCuts = {1},
      -- initial conditions
      init = function (t, xn)
	 local x, v = xn[1], xn[2]
	 local f = maxwellian1D(0.5, vd_e, vth_e, v) + maxwellian1D(0.5, -vd_e, vth_e, v) 
	 return (1 + perturbation*math.cos(kNumber*x)) * f
      end,
      evolve = true, -- evolve species?

      diagnosticMoments = { "M0", "M1i", "M2" }
   },

   -- field solver
   field = Plasma.MaxwellField {
      epsilon0 = 1.0, mu0 = 1.0,
      init = function (t, xn)
	 local x = xn[1]
	 local Ex = -perturbation * math.sin(kNumber*x) / kNumber
	 return Ex, 0.0, 0.0, 0.0, 0.0, 0.0
      end,
      evolve = true, -- evolve field?
   },
}
-- run application
sim:run()
\end{lstlisting}

\pagebreak
\section{Weibel Instability}\label{list:weibel:weibel}
\begin{lstlisting}[language={[5.1]Lua}]
-- Gkyl --------------------------------------------------------------
-- Electron Weibel instability ---------------------------------------
local Plasma = require "App.PlasmaOnCartGrid"

-- Constants
q_e, m_e = -1.0, 1.0
epsilon_0, mu_0 = 1.0, 1.0

-- Initial conditions ('p' for plus population and 'm' for minus population)
n_ep, n_em = 0.5, 0.5
ux_ep, ux_em = 0.0, 0.0
uy_ep, uy_em = 0.15, -0.15
vth_ep, vth_em = 0.03, 0.03
-- T_ep, T_em = 0.01, 0.01
-- vth_ep, vth_em = math.sqrt(T_ep/m_e), math.sqrt(T_ep/m_e)

kx = 0.4-- wave-number
perturb = 1.0e-4 -- distribution function perturbation

-- Maxwellian in 1x2v
local function maxwellian2D(vx, vy, n, ux, uy, vth)
   local v2 = (vx - ux)^2 + (vy - uy)^2
   return n/(2*math.pi*vth^2)*math.exp(-v2/(2*vth^2))
end

sim = Plasma.App {
   logToFile = false,

   tEnd = 500.0, -- end time
   nFrame = 500, -- number of output frames
   lower = { 0.0 }, -- configuration space lower left
   upper = { 2*math.pi/kx }, -- configuration space upper right
   cells = {64}, -- configuration space cells
   basis = "serendipity", -- one of "serendipity" or "maximal-order"
   polyOrder = 2, -- polynomial order
   timeStepper = "rk3", -- one of "rk2" or "rk3"

   -- decomposition for configuration space
   decompCuts = {2}, -- cuts in each configuration direction
   useShared = false, -- if to use shared memory
   -- boundary conditions for configuration space
   periodicDirs = {1}, -- periodic directions

   -- electrons
   elc = Plasma.VlasovSpecies {
      charge = q_e, mass = m_e,
      -- velocity space grid
      lower = {-1.0, -1.0},
      upper = {1.0, 1.0},
      cells = {32, 32},
      decompCuts = {1, 1},
      -- initial conditions
      init = function (t, xn)
	 local x, vx, vy = xn[1], xn[2], xn[3]
	 return maxwellian2D(vx, vy, n_ep, ux_ep, uy_ep, vth_ep) +
	    maxwellian2D(vx, vy, n_em, ux_em, uy_em, vth_em)
      end,
      evolve = true, -- evolve species?
      diagnosticMoments = { "M0", "M1i", "M2" }
   },

   -- field solver
   field = Plasma.MaxwellField {
      epsilon0 = epsilon_0, mu0 = mu_0,
      init = function (t, xn)
	 local x = xn[1]
	 local Bz = perturb*math.sin(kx*x)
	 return 0.0, 0.0, 0.0, 0.0, 0.0, Bz
      end,
      evolve = true, -- evolve field?
   },
}
-- run application
sim:run()
\end{lstlisting}

\pagebreak
\section{Collisionless Sheath}\label{list:bounded:cls}
\begin{lstlisting}[language={[5.1]Lua}]
-- Gkyl --------------------------------------------------------------
-- Basic sheath simulation -------------------------------------------
local Plasma = require "App.PlasmaOnCartGrid"

-- SI units
local epsilon_0, mu_0 = 8.854e-12, 1.257e-6
local q_e, q_i = -1.6021766e-19, 1.6021766e-19
local m_e, m_i = 9.109383e-31, 1.6726218e-27
local n_e, n_i = 1.0e17, 1.0e17
local u_e, u_i = 0.0, 0.0
local T_e, T_i = 10*1.6021766e-19, 1*1.6021766e-19

local vth_e, vth_i = math.sqrt(T_e/m_e), math.sqrt(T_i/m_i)
local uB = math.sqrt(T_e/m_i)
-- artificially decrease the speed of light
mu_0 = 1.0/(epsilon_0 * (10*vth_e)^2)
--epsilon_0 = 1.0/(mu_0 * (10*vth_e)^2)

local omega_pe = math.sqrt((n_e * q_e^2)/(epsilon_0*m_e))
local lambda_D = math.sqrt((epsilon_0 * T_e)/(n_e * q_e^2))

-- initialization function
local function maxwellian(n, u, vth, v)
   return n / math.sqrt(2*math.pi*vth*vth) * 
      math.exp(-(v-u)^2/(2*vth*vth))
end

sim = Plasma.App {
   logToFile = false,

   tEnd = 1000/omega_pe, -- end time
   nFrame = 100, -- number of output frames
   lower = {-128.0*lambda_D}, -- configuration space lower left
   upper = {128.0*lambda_D}, -- configuration space upper right
   cells = {256}, -- configuration space cells
   basis = "serendipity", -- one of "serendipity" or "maximal-order"
   polyOrder = 2, -- polynomial order
   cflFrac = 1.0, -- CFL "fraction". Usually 1.0
   timeStepper = "rk3", -- one of "rk2" or "rk3"

   -- decomposition for configuration space
   decompCuts = {8}, -- cuts in each configuration direction
   useShared = false, -- if to use shared memory

   -- boundary conditions for configuration space
   periodicDirs = {}, -- periodic directions

   -- electrons
   elc = Plasma.VlasovSpecies {
      charge = q_e, mass = m_e,
      -- velocity space grid
      lower = {-6.0*vth_e},
      upper = {6.0*vth_e},
      cells = {32},
      decompCuts = {1},
      -- initial conditions
      init = function (t, xn)
         local x, v = xn[1], xn[2]
         return maxwellian(n_e, u_e, vth_e, v)
      end,
      evolve = true, -- evolve species?
      bcx = { Plasma.VlasovSpecies.bcAbsorb,
              Plasma.VlasovSpecies.bcAbsorb },
      diagnosticMoments = { "M0", "M1i", "M2" },
   },

   ion = Plasma.VlasovSpecies {
      charge = q_i, mass = m_i,
      -- velocity space grid
      lower = {-6.0*uB},
      upper = {6.0*uB},
      cells = {32},
      decompCuts = {1},
      -- initial conditions
      init = function (t, xn)
         local x, v = xn[1], xn[2]
         return maxwellian(n_i, u_i, vth_i, v)
      end,
      evolve = true, -- evolve species?
      bcx = { Plasma.VlasovSpecies.bcAbsorb,
              Plasma.VlasovSpecies.bcAbsorb },
      diagnosticMoments = { "M0", "M1i", "M2" },
   },
   
   -- field solver
   field = Plasma.MaxwellField {
      epsilon0 = epsilon_0, mu0 = mu_0,
      init = function (t, xn)
         return 0.0, 0.0, 0.0, 0.0, 0.0, 0.0
      end,
      evolve = true, -- evolve field?
      bcx = { Plasma.MaxwellField.bcReflect,
              Plasma.MaxwellField.bcReflect },
   },
}
-- run application
sim:run()
\end{lstlisting}

\pagebreak
\section{\cite{Sod1978} Shock Tube}\label{list:bounded:sod}
\begin{lstlisting}[language={[5.1]Lua}]
-- Gkyl --------------------------------------------------------------
-- BGK Sod Shock Tube Test -------------------------------------------
local Plasma = require "App.PlasmaOnCartGrid"

-- left/right state for shock
local nl, ul, pl = 1.0, 0.0, 1.0
local nr, ur, pr = 0.125, 0.0, 0.1

local vthl = math.sqrt(pl/nl)
local vthr = math.sqrt(pr/nr)

local Kn = 0.001

local function maxwellian(n, u, vth, v)
   return n / math.sqrt(2*math.pi*vth^2) * 
      math.exp(-(v-u)^2 / (2*vth^2))
end

sim = Plasma.App {
   logToFile = false,

   tEnd = 0.3, -- end time
   nFrame = 8, -- number of frames to write
   lower = {0.0}, -- configuration space lower left
   upper = {1.0}, -- configuration space upper right
   cells = {128}, -- configuration space cells
   basis = "serendipity", -- one of "serendipity" or "maximal-order"
   polyOrder = 2, -- polynomial order
   timeStepper = "rk3", -- one of "rk2", "rk3" or "rk3s4"

   -- decomposition for configuration space
   decompCuts = {1}, -- cuts in each configuration direction
   useShared = false, -- if to use shared memory

   -- boundary conditions for configuration space
   periodicDirs = {}, -- periodic directions

   -- electrons
   neut = Plasma.VlasovSpecies {
      --nDiagnosticFrame = 2,
      charge = 0.0, mass = 1.0,
      -- velocity space grid
      lower = {-6.0},
      upper = {6.0},
      cells = {32},
      decompCuts = {1},
      -- initial conditions
      init = function (t, xn)
	 local x, v = xn[1], xn[2]

	 if math.abs(x) < 0.5 then
	    return maxwellian(nl, ul, vthl, v)
	 else
	    return maxwellian(nr, ur, vthr, v)
	 end
      end,
      bcx = { Plasma.VlasovSpecies.bcCopy,
	      Plasma.VlasovSpecies.bcCopy },
       -- evolve species?
      evolve = true,
      -- diagnostic moments
      diagnosticMoments = { "M0", "M1i", "M2" },

      -- collisions
      coll = Plasma.BgkCollisions {
      	 collFreq = vthl/Kn,
      },
   },
}
-- run application
sim:run()
\end{lstlisting}

%% file: appendix_scripts.tex
\chapter{Scripts and Algorithms}

This appendix provides listings of miscellaneous scripts and algorithms
used through this work.

\section{Field Particle Correlation (\textit{Python})}\label{list:scripts:fpc}
\begin{lstlisting}[language=Python]
import numpy as np
import matplotlib.pyplot as plt
import matplotlib.cm as cm
import postgkyl as pg

xIdx = 48
vIdx1 = 280
vIdx2 = 224
numFrame = 101
numV = 384
E = np.zeros(numFrame)
f = np.zeros((numFrame, numV))
C0 = np.zeros((numFrame, numV))
t = np.linspace(0., 20., numFrame)
v = np.linspace(-6., 6., numV)
q = -1.0
dt = 20/(numFrame-1)

for fIdx in range(numFrame):
    data = pg.GData('landau_field_{:d}.bp'.format(fIdx))
    dg = pg.GInterpModal(data, 2, 'ms')
    grid, tmp = dg.interpolate(0)
    E[fIdx] = tmp[xIdx, 0]
    
    data = pg.GData('landau_elc_{:d}.bp'.format(fIdx))
    dg = pg.GInterpModal(data, 2, 'ms')
    grid, tmp = dg.differentiate(1)
    df[fIdx, :] = tmp[xIdx, :, 0]
    
for vIdx in range(numV):
    C0[:, vIdx] = -q * v[vIdx]**2 / 2 * (df[:, vIdx] - df[0, vIdx]) * E

fig, ax = plt.subplots(2, 2, sharex=True, sharey="row", figsize=(10,8))

minN = 10
maxN = 70
for N in range(minN, maxN):
    tmp00 = np.full(numFrame, np.nan)
    tmp01 = np.full(numFrame, np.nan)
    tmp10 = np.zeros(numFrame)
    tmp11 = np.zeros(numFrame)
    for i in range(numFrame - N):
        tmp00[i] = np.sum(C0[i : i+N, vIdx1]) / (N*dt)
        tmp01[i] = np.sum(C0[i : i+N, vIdx2]) / (N*dt)
    for i in range(1, numFrame):
        tmp10[i] = np.sum(tmp00[:i])
        tmp11[i] = np.sum(tmp01[:i])
    ax[0, 0].plot(t, tmp00, color=cm.inferno((N-minN)/float(maxN-minN)))
    ax[0, 1].plot(t, tmp01, color=cm.inferno((N-minN)/float(maxN-minN)))
    tmp10 = tmp10 * dt
    tmp11 = tmp11 * dt
    ax[1, 0].plot(t, tmp10, color=cm.inferno((N-minN)/float(maxN-minN)))
    ax[1, 1].plot(t, tmp11, color=cm.inferno((N-minN)/float(maxN-minN)))
\end{lstlisting}

\section{Precomputation of Material-Based Boundary Condition File (\textit{Mathematica})}\label{list:scripts:bronold}
\begin{lstlisting}[language=Mathematica]
basis[x_, vx_] := {1/2, (Sqrt[3] x)/2, (Sqrt[3] vx)/2, (3 vx x)/2, (
  3 Sqrt[5] (x^2 - 1/3))/4, (3 Sqrt[5] (vx^2 - 1/3))/4, (
  3 Sqrt[15] (vx x^2 - vx/3))/4, (3 Sqrt[15] (vx^2 x - x/3))/4}
numBasis = 8;
\[Chi] = 1.0;
mc = 0.4;
\[Eta][E_, \[Xi]_] := Sqrt[1 - (E - \[Chi])/(mc E) (1 - \[Xi]^2)]
\[Xi]c[E_] := 
 If[E < \[Chi]/(1 - mc), 0.0, Sqrt[1 - (mc E)/(E - \[Chi])]]
T[E_, \[Xi]_] := (
 4 mc Sqrt[E - \[Chi]] \[Xi] Sqrt[
  mc E] \[Eta][E, \[Xi]])/(mc Sqrt[E - \[Chi]] \[Xi] + 
   Sqrt[mc E] \[Eta][E, \[Xi]])^2
R[E_, \[Xi]_, C_] := 
 If[E < \[Chi], 1., 
  If[\[Xi] > \[Xi]c[E], 
   1. - T[E, \[Xi]]/(
    1 + C/\[Xi]) - (C/\[Xi])/(1 + C/\[Xi])
      NIntegrate[T[E, t], {t, \[Xi]c[E], 1.0}], 
   1. - (C/\[Xi])/(1 + C/\[Xi])
      NIntegrate[T[E, t], {t, \[Xi]c[E], 1.0}]]]
elemCharge = 1.6021766*^-19;
me = 9.109383*^-31;
roughness = 2.0;
vth = Sqrt[(10*elemCharge)/me];
Nv = 32;
dv = Sqrt[(2*\[Chi]*elemCharge)/me]
vc = Table[(-(Nv/2) + 1/2 + i)*dv, {i, 0, Nv - 1}]
getE[vx_, vc_, dv_] := 1/2 me (vc + vx dv/2)^2/elemCharge;
fh = OpenWrite[NotebookDirectory[] <> "wall_1X1V.lua"];
WriteLine[fh, "local _M = {}"];
WriteLine[fh, "_M[1] = function (idx, f, out)"];
For[j = 1, j <= Nv, j++,
 If[j == 1, 
  WriteLine[fh, "   if idx[1] == " <> ToString[j] <> " then"], 
  WriteLine[fh, "   elseif idx[1] == " <> ToString[j] <> " then"]];
 For[k = 1, k <= numBasis, k++,
  str = "";
  For[l = 1, l <= numBasis, l++,
   temp = 
    NIntegrate[
     R[getE[vx, vc[[j]], dv], 1.0, roughness] * basis[-x, -vx][[l]]*
      basis[x, vx][[k]], {x, -1, 1}, {vx, -1, 1}, MaxPoints -> 9];
   If[temp != 0.0, 
    str = str <> " + " <> TextString[temp] <> "*f[" <> TextString[l] <>
       "]"];
   ];
  WriteLine[fh, "      out[" <> TextString[k] <> "] = 0.0" <> str];
  ]
 ]
WriteLine[fh, "   end"];
WriteLine[fh, "end"];
WriteLine[fh, "return _M"];
Close[fh];
\end{lstlisting}